\documentclass[sigconf, 10pt]{acmart}
\usepackage{geometry}
\usepackage[utf8]{inputenc}
\inputencoding{latin2}

\usepackage{amsmath}
\usepackage{amsthm}

\usepackage{comment}
\usepackage{tabu}
\usepackage{float}
\usepackage{listings}
\usepackage{environ}
\usepackage{tabularx}

\usepackage{scalefnt}
\usepackage{amsfonts}
\usepackage{subcaption}
\usepackage{thm-restate}

\usepackage{enumitem}

\usepackage{tikz}
\usetikzlibrary{matrix, shapes, automata, positioning, calc}

\usepackage{cleveref}

\usepackage{sty/comments}
\usepackage{sty/envs}
\usepackage{sty/math}
\usepackage{sty/grammars}
\usepackage{sty/vas}
\usepackage{sty/reach}
\usepackage{sty/dcmp}
\usepackage{sty/cov}

\author{Roland Guttenberg}
\affiliation{%
\institution{University of Warsaw}
\city{}\country{}\postcode{}
}
\email{r.guttenberg@uw.edu.pl}

\author{Eren Keskin}
\affiliation{%
  \institution{TU Braunschweig}
  \city{}\country{}\postcode{}
  }
\email{e.keskin@tu-bs.de}

\author{Roland Meyer}
\affiliation{%
  \institution{TU Braunschweig}
  \city{}\country{}\postcode{}
  }
\email{roland.meyer@tu-bs.de}

\begin{document}

\title{PVASS Reachability is Decidable}



\date{}

\begin{abstract}
Reachability in pushdown vector addition systems with states (PVASS) is among the longest standing open problems in Theoretical Computer Science. 
We show that the problem is decidable in full generality.
Our decision procedure is similar in spirit to the KLMST algorithm for VASS reachability, but works over objects that support an elaborate form of procedure summarization as known from pushdown reachability. 
\end{abstract}

\maketitle

\section{Introduction}
The most basic models of computation beyond finite automata are pushdowns (acting on a stack) and vector addition systems with states (acting on a finite set of non-negative counters).
Despite their simplicity, the interaction of the two models is not understood.
Given a pushdown $P$ and a VASS $V$, it is not known whether $L(P)\cap L(V)=\emptyset$ is decidable. 
We prove here that this is the case. Our algorithm has a complexity of Hyper-Ackermann, which is conjectured to be optimal \cite{Czerwinski23}.

The problem is most well-known in the community in its equivalent formulation, as reachability problem in pushdown vector addition systems with states (PVASS). 
A PVASS is a finite automaton manipulating both a stack and non-negative counters. 
To solve PVASS reachability, one could imagine that simply combining the procedure summaries for pushdown reachability \cite{SharirP78,RHS95} with the KLMST decomposition for VASS reachability \cite{SacerdoteT77,Mayr81,Kosaraju82, Lambert92} would be enough. 
However, despite a long list of attempts over the last two decades, the problem could only be solved in special cases, and all attempts to generalize the solutions have failed~ \cite{Reinhardt08, AtigG11, GantyM12, LerouxST15, LerouxPSS19, BaumannGMTZ23, BaumannGMTZOther23, BiziereC25}. 
The combination 
requires new techniques.

\textbf{Techniques}: 
The first idea is to generalize the PVASS reachability problem in an appropriate fashion. 
We study emptiness in a new model called \emph{nested grammar vector addition systems} (NGVAS). 
NGVAS are nested context-free grammars: 
at level $0$ terminals are counter updates in $\Z^{\adim}$, at level $n+1$ terminals are NGVAS of level $n$.
NGVAS generalize PVASS via the standard conversion from pushdown automata to context-free grammars.
The nesting allows us to approach algorithmic problems we encounter while solving reachability using induction.

The second idea is to prove a version of the Steinitz Lemma for NGVAS, which will play an important role in a sufficient condition for reachability. 
It will be a corollary of our \emph{wide tree theorem}, which is purely for context-free grammars \(G\) and we believe of independent interest. Simply put, if \(G\) is non-linear and strongly-connected, then for many words \(w \in L(G)\), there is a permutation \(\pi(w) \in L(G)\) that has a derivation tree of \emph{height logarithmic} in the length of \(w\). 

The third and final idea takes up most of the paper.
To apply the aforementioned sufficient condition, we need pumpability properties. 
Despite similarities at first glance, pumpability in NGVAS carries a different flavor than in VASS. 
In VASS, the reachability algorithm checks for the existence of a cycle \(\mathbf{up}\) which has a positive effect on (for simplicity) all counters, and similarly a cycle \(\mathbf{dwn}\) which has a negative effect on all counters.
The cycles may be different.  
Pumping in NGVAS instead asks for \emph{one} derivation \(S \to w_{\mathbf{up}}.S.w_{\mathbf{dwn}}\). 
To see how the one derivation leads to a different flavor, consider the NGVAS \(S \to +1.S.-1 \mid \varepsilon \).
Although it satisfies the pumping condition, 
the reachability set is bounded.
This is different from pumping in VASS. 

We distinguish two types of unboundedness and call a counter \emph{vertically unbounded} if the pumpability condition holds, and \emph{horizontally unbounded} if it is unbounded in the reachability set, like for normal VASS. 
To decide vertical unboundedness, we first reduce it to computing the downward-closure of reachability sets, and hence horizontal unboundedness. 
The construction is similar to the Karp-Miller trees that are used to decide boundedness in VASS~\cite{KM69}. Afterwards, using Rackoff-like ideas \cite{Rackoff78}, we reduce horizontal unboundedness to vertical unboundedness in NGVAS with fewer counters, which we can solve by induction.

\textbf{Related Work}: Already the emptiness problem for VASS turned out to be one of the hardest problems in Theoretical Computer Science and was studied for around 50 years. It was proved decidable in the 1980s \cite{SacerdoteT77,Mayr81,Kosaraju82, Lambert92}, but its complexity (Ackermann-complete) could only be determined very recently \cite{LerouxS19, CzerwinskiLLLM19, CzerwinskiO21, Leroux21}. 
Landmark results in this process were Leroux's new decision procedure based on inductive invariants~\cite{Leroux09, Leroux11}, and the insight that the classical KLMST decomposition, named after its inventors Sacerdote, Tenney, Mayr, Kosaraju, and Lambert \cite{SacerdoteT77, Mayr81, Kosaraju82, Lambert92}, actually computes run ideals \cite{LerouxS15}, an object that even shows up by purely combinatorial reasoning.

Given this progress, problems that were considered out of reach moved into the focus of the community. 
One line of research asks whether reachability is decidable in models that generalize VASS, including (restricted) VASS with nested zero tests and resets~\cite{Reinhardt08, Bonnet11, Bonnet12,FinkelLS18,LerouxS20,Guttenberg24}, branching variants of VASS~\cite{VermaG05,Lazic10,DemriJLL13,ClementeLLM17}, VASS with tokens that carry data~\cite{RosaVelardoF11,HofmanLLLST16,LazicS16,BlondinL23}, valence systems~\cite{Zetzsche16}, and amalgamation systems~\cite{AnandSSZ24}. 
Another line asks whether we can compute information even more precise than reachability, like the downward closure of the reachability language~\cite{HabermehlMW10}, or separability by regular~\cite{Keskin024} or subregular languages~\cite{PZ14, ClementeCLP17}.


\textbf{Structure}: 
We explain the decidability of PVASS reachability at three levels of detail. 
The main paper gives the overall argumentation and the key techniques.
The second level, starting from \Cref{Section:IntroToAppendix},  introduces missing definitions, makes results formal, and gives proofs or sketches for the main findings. 
The third level, starting from \Cref{Section:AppendixL3}, is devoted to proofs. 

\newcommand{\ri}{\mathit{RI}}
\newcommand{\initNGVAS}{\anngvas_{\ri}}
\newcommand{\nrunsof}[1]{\mathit{R}_{\N}(#1)}
\newcommand{\zrunsof}[1]{\mathit{R}_{\Z}(#1)}
\newcommand{\adeconst}{\mathcal{D}}
\newcommand{\perffun}{\textsf{perf}}
\newcommand{\perfof}[1]{\perffun(#1)}

\newcommand{\aprop}{\textsf{(X)}}
\newcommand{\apropp}{\textsf{(Y)}}
\newcommand{\perfconds}{\textsf{Perfect}}

\newcommand{\prepar}[1]{\mathsf{pre}_{#1}}
\newcommand{\preparof}[2]{\prepar{#1}(#2)}
\newcommand{\preeqof}[1]{\preparof{\textsf{\perfectnessprodsnospace}}{#1}}
\newcommand{\preintpumpof}[1]{\preparof{\textsf{\perfectnesspumpingintnospace}}{#1}}
\newcommand{\prepumpof}[1]{\preparof{\textsf{\perfectnesspumpingnospace}}{#1}}
\newcommand{\preXof}[1]{\mathsf{pre}_{\aprop}}
\newcommand{\postXof}[1]{\mathsf{post}_{\aprop}}
\newcommand{\postpar}[1]{\mathsf{post}_{#1}}
\newcommand{\postparof}[2]{\postpar{#1}(#2)}
\newcommand{\posteqof}[1]{\postparof{\textsf{\perfectnessprodsnospace}}{#1}}
\newcommand{\postintpumpof}[1]{\postparof{\textsf{\perfectnesspumpingintnospace}}{#1}}
\newcommand{\postpumpof}[1]{\postparof{\textsf{\perfectnesspumpingnospace}}{#1}}

\newcommand{\refpostparof}[2]{\mathsf{refpost}_{#1}(#2)}

\newcommand{\towardspred}[1]{\mathsf{twd}_{#1}}
\newcommand{\towardspredof}[2]{\towardspred{#1}(#2)}
\newcommand{\towardsXof}{\towardspredof{\aprop}}
\newcommand{\towardseq}{\towardspred{\textsf{\perfectnessprodsnospace}}}
\newcommand{\towardseqof}[1]{\towardseq{#1}}
\newcommand{\towardsintpumpof}[1]{\towardspredof{\textsf{\perfectnesspumpingintnospace}}{#1}}
\newcommand{\towardspumpof}[1]{\towardspredof{\textsf{\perfectnesspumpingnospace}}{#1}}

\newcommand{\postof}[1]{\mathsf{post}(#1)}

\newcommand{\decpar}[1]{\textsf{dec}_{#1}}
\newcommand{\decparof}[2]{\decpar{#1}(#2)}

\newcommand{\decsub}{\decpar{\textsf{\perfectnesschildren}}}
\newcommand{\decsubof}[1]{\decsub(#1)}
\newcommand{\deceq}{\decpar{\textsf{\perfectnessprodsnospace}}}
\newcommand{\deceqof}[1]{\deceq(#1)}
\newcommand{\decpump}{\decpar{\textsf{\perfectnesspumpingnospace}}}
\newcommand{\decpumpof}[1]{\decpump(#1)}
\newcommand{\decintpump}{\decpar{\textsf{\perfectnesspumpingintnospace}}}
\newcommand{\decintpumpof}[1]{\decintpump(#1)}

\newcommand{\refinepar}[1]{\mathsf{ref}_{#1}}
\newcommand{\refineparof}[2]{\refinepar{#1}(#2)}
\newcommand{\refinesub}{\mathsf{ref}_{\textsf{\perfectnesschildrennospace}}}
\newcommand{\refinesubof}[1]{\refinesub(#1)}
\newcommand{\refineeq}{\refinepar{\textsf{\perfectnessprodsnospace}}}
\newcommand{\refineeqof}[1]{\refineeq(#1)}
\newcommand{\refinepump}{\refinepar{\textsf{\perfectnesspumpingnospace}}}
\newcommand{\refinepumpof}[1]{\refinepump(#1)}
\newcommand{\refineintpump}{\refinepar{\textsf{\perfectnesspumpingintnospace}}}
\newcommand{\refineintpumpof}[1]{\refineintpump(#1)}
\newcommand{\anngvasdom}{\anngvas_{\domtag}}

\section{Proof Outline} \label{SectionProofOutline}
We reduce PVASS reachability to emptiness for a new model called \emph{nested grammar vector addition system}. 
For now, we can view an NGVAS as a context-free grammar whose terminals are VAS updates, and then $\nrunsof{\anngvas}$ is the set of terminal words that lead from the initial to the final marking while staying non-negative. 
The reduction is similar to the translation of pushdowns to context-free grammars.
Our main result is this. 
\begin{theorem}\label{Theorem:EmptinessDecidable}
$\nrunsof{\anngvas}\neq\emptyset$ is decidable in \(\mathfrak{F}_{\omega^{d+3}}\), where \(d\) is the number of counters.
\end{theorem}

Our decision procedure follows the KLMST approach to VASS reachability. 
We relax the non-negativity constraint in $\nrunsof{\anngvas}\neq\emptyset$ and admit runs that may fall below zero.
We can then answer $\zrunsof{\anngvas}\neq \emptyset$ with standard techniques from integer linear programming. 
Of course the relaxation is not sound in general.
We identify a subclass of \emph{perfect} NGVAS for which it is. 
It is convenient to incorporate the emptiness check into perfectness. 
\begin{proposition}[Iteration]\label{Proposition:Iteration} 
If $\anngvas$ is perfect, $\nrunsof{\anngvas}\neq\emptyset$. 
\end{proposition}

The importance of perfectness stems from the fact that we can compute from every NGVAS a finite set of perfect NGVAS that reflects the set of non-negative runs. 
We call a set of NGVAS $\adeconst$ a \emph{deconstruction} of $\anngvas$, if it is finite and satisfies $\nrunsof{\anngvas}=\nrunsof{\adeconst}$. 
We also speak of a \emph{perfect deconstruction}, if all $\anngvasp\in\adeconst$ are perfect. 
Our second main result is this.
\begin{proposition}[Deconstruction]\label{Proposition:Deconst}
There is a function $\perffun:\setNGVAS\rightarrow\powof{\setNGVAS}$ so that \(\perffun \in \mathfrak{F}_{\omega^{d+3}}\) and, for all~$\anngvas$, $\perfof{\anngvas}$ is a perfect deconstruction of $\anngvas$. 
\end{proposition}
\Cref{Theorem:EmptinessDecidable} now follows from Propositions~\ref{Proposition:Iteration} and 3:
\begin{align*}
\nrunsof{\anngvas}\neq\emptyset \qquad \text{iff} \qquad \perfof{\anngvas}\neq \emptyset\ .
\end{align*}

We define a function $\perffun$ as required by \Cref{Proposition:Deconst}. 
For this purpose, similar to the simplified information on NGVAS, it is enough to view perfectness as the conjunction of five as of yet not relevant conditions we call \emph{clean}, \(\perfectnesschildrennospace\), \(\perfectnessprodsnospace, \perfectnesspumpingintnospace\), and \( \perfectnesspumpingnospace\).

Similar to VASS reachability, \(\perffun\) is a standard recursive algorithm: 
first it checks whether the input \(\anngvas\) is perfect. Intuitively, perfect \(\anngvas\) are the base case of the recursion. If yes, it outputs \(\{\anngvas\}\). Otherwise it computes a deconstruction \(\adeconst\) of \(\anngvas\) s.t. \(\rankof{\anngvasp}<\rankof{\anngvas}\) for every \(\anngvasp \in \adeconst\), where \(\rankof{\anngvas}\) is a well-founded rank. On the new NGVAS of lower rank, \(\perffun\) continues recursively. Since the rank is well-founded, the recursion terminates by K\"onig's Lemma.

To reduce the rank, \(\perffun\) has four subprocedures $\decsub, \deceq, \decintpump, \decpump:\setNGVAS\rightarrow\powof{\setNGVAS}$, which check whether the conditions \(\perfectnesschildrennospace\), \(\perfectnessprodsnospace, \perfectnesspumpingintnospace\), resp. \( \perfectnesspumpingnospace\) hold and, if not, return a deconstruction of smaller rank. 
Note that cleaning is not in this list. 
Cleaning can only guarantee that the rank does not increase, and all NGVAS 
in the output 
fulfill property \emph{clean}. 
Being clean is a precondition for 
the other decompositions: all the procedures \(\text{dec}_{(R i)}=\text{clean.ref}_{(R i)}\) first clean the input and then decrease the rank. 

Formally, the specifications are as follows, with two main differences to VASS: all our subprocedures rely on calling perf on inputs of smaller rank, which means all of them have the precondition that \(\perffun\) is reliable up to (and excluding) \(\rankof{\anngvas}\). 

Moreover, \(\text{ref}_{(R i)}\) assumes that \((R0)\) to \((Ri-1)\) hold.

\begin{lemma}
If $\perffun$ is reliable up to $\rankof{\anngvas}$ 
then $\cleanof{\anngvas}$ terminates with output $\adeconst$ that is a deconstruction of $\anngvas$, clean, and satisfies $\rankof{\adeconst}\leq\rankof{\anngvas}$.\label{Lemma:CleannessGuaranteeOverview}
\end{lemma}

\begin{lemma}\label{Lemma:RefineFunctions}
If $\perffun$ is reliable up to $\rankof{\anngvas}$ and the  NGVAS \(\anngvas\) is clean and fulfills \((R0)-(Ri-1)\), then $\refineparof{(R i)}{\anngvas}$ decides whether condition \((R i)\) holds and, if not, outputs a deconstruction $\adeconst$ so that $\rankof{\adeconst}<\rankof{\anngvas}$.
\end{lemma}


These lemmas imply termination of \(\perffun\) in \Cref{Proposition:Deconst} as sketched before.
We elaborate on iteration, refinement, rank, and complexity.

\section{Iteration Lemma}\label{Section:OutlineIterationLemma}
We give details on the NGVAS model and perfectness, state the iteration lemma, and explain its proof.
\subsection{NGVAS}\label{Section:OutlineNGVAS}
\newcommand{\procof}[1]{\text{Proc}(#1)}
A \emph{nested grammar vector addition system} $\anngvas=(\agram, \contextinformation, \boundednessinformation)$ is a context-free grammar $\agram=(\nonterms, \trms, \prods, \startnonterm)$ with a particular form of terminal symbols $\trms$ and extra information.  
An NGVAS of nesting depth zero has as terminals counter updates in $\Z^{\adim}$. 
An NGVAS of nesting depth $n+1$ has as terminals NGVAS of depth at most $n$, also called childNGVAS.
NGVAS are expected to be strongly connected in that every two non-terminals/procedures can call each other.
If the grammar of an NGVAS does not contain a rule that produces two non-terminals, we call it \emph{linear}, and if it does, we call it \emph{non-linear}.

The first piece of extra information $\contextinformation=(\acontext, \restrictions, \unconstrained)$, breaks into three components.
First is the pair of \emph{source} and \emph{target} markings $\acontext=(\acontextin,\acontextout)$ from $\Nomega^{\adim}$. 
Source and target not only define the reachability problem, they also summarize the effect of calls to this procedure. 

The second component is a linear set \(\restrictions \subseteq \Z^d\) of effects, called the \emph{restriction} which overapproximates the possible effects of runs. 
The idea is that if some counter \(\acounter\) is~$\omega$, source and target are not very informative. 
Last component is a set of counters $\unconstrained\subseteq [1,d]$ that are guaranteed to be unbounded anywhere in the NGVAS.
To ensure this, we require $\unconstrained\subseteq\omegaof{\acontextin}, \omegaof{\acontextout}$, and that children only extend the $\unconstrained$ component. 


We call $\contextinformation$ the \emph{(enforced) context} of the NGVAS. 
Enforced means the \emph{semantics} forbids runs that do not respect the context or that call childNGVAS out-of-context (e.g. with the wrong counter valuations at the entrance or exit, with the wrong effect etc). 
We use $\nrunsof{\anngvas}$ resp. $\zrunsof{\anngvas}$ for the set of all $\N$-runs resp. $\Z$-runs. 
Importantly, also $\Z$-runs have to respect the context. 

The last piece is called the \emph{boundedness information} $\boundednessinformation=(\abdinfomid, \infun, \outfun)$. 
For the sake of simplicity, we only present the boundedness information of a non-linear NGVAS here.
It consists of a set $\abdinfomid$, and two functions \(\infun, \outfun \colon \nonterms \to \omega^{\abdinfomid}\times\N^{[1,d]\setminus\abdinfomid}\).
The set $\adimset$ contains counters that are unbounded in the context of this NGVAS. 
The functions \(\infun, \outfun\) track for every non-terminal $\anonterm$ and bounded counter $\acounter \not \in \adimset$ the values at the start resp. end of procedure $\anonterm$. If \(\anonterm \to \anontermp.\anontermpp\) is a production, then we must have \(\infun(\anonterm)=\infun(\anontermp), \outfun(\anontermp)=\infun(\anontermpp)\) and \(\outfun(\anontermpp)=\outfun(\anonterm)\). This is similar to tracking counters in the control-state of a VASS \cite{Lambert92}.
We expect that context information of the children complies with boundedness information of the parent, i.e.\ children have the $\adimset$ of the parent as their $\unconstrained$ set, and 
if a non-terminal $\anonterm\in\nonterms$ produces a terminal $\asymbol\in\trms$ with a rule $\anonterm\to\aword\asymbol\awordp$, then the markings must fit the boundedness information: $\asymbol.\acontextin=\inof{\anonterm}$ if $\aword=\varepsilon$ and $\asymbol.\acontextout=\outof{\anonterm}$ if $\awordp=\varepsilon$.


\subsection{Perfectness}\label{Section:OutlinePerfectness}
We focus on the more complicated and interesting case of non-linear NGVAS. The important parts are: 
\begin{enumerate}[itemsep=-1pt]
\item[\perfectnesssolnospace] For every effect in the restriction, $\baseeffectchoice \in \restrictions$, there is a $\Z$-run with this effect. 
\item[\perfectnesschildrennospace] All childNGVAS are perfect.
\item[\perfectnessprodsnospace] For every $n \in \N$, there is a $\Z$-run $\rho_n$ which uses every production at least $n$ times and, for every childNGVAS $\anngvasp$, uses every period in the restriction of $\anngvasp$ at least $n$ times.
\item[\perfectnesspumpingnospace] If a counter should be unbounded, $\acounter \in D$, then it actually is unbounded: there exists a derivation $S \to^{\ast} \upseq .S.\downseq$ so that $\upseq$ can be fired at $\acontextin$ and increases $\acounter$, and $\downseq$ can be backwards fired at $\acontextout$ and backwards firing increases $\acounter$. 
The latter has the purpose of pumping down the counter $\acounter$.
\end{enumerate}

\noindent Here, \perfectnesschildrennospace, \perfectnessprods and \perfectnesspumping are the refinement conditions we mentioned when explaining the computation of $\perffun$. 
We omitted \perfectnesspumpingint because it only concerns the linear case. 



\subsection{Proof of Proposition \ref{Proposition:Iteration}}\label{Section:OutlineILProof}
To prove the iteration result, we use induction on the nesting depth.  
We slightly strengthen the inductive statement as follows:
\begin{restatable}{theorem}{TheoremIterationLemmaMain} \label{TheoremIterationLemmaNonLinearOverview}
Let $\anngvas$ be perfect with $\restrictions = \baseeffect+\periodeffect^*$. 
Let $\baseeffectchoice\in\restrictions$ and $\periodeffectchoice\in\N^{\periodeffect}$ with $\periodeffectchoice\geq 1$. 
There is $\initconst\geq 1$ so that for every $\aconst\geq \initconst$ there is $\iterrun{\aconst}\in\runsof{\anngvas}$ with effect \(\vaseffof{\iterrun{\aconst}}=\baseeffectchoice+\aconst \cdot \periodeffect \cdot \periodeffectchoice\).
\end{restatable}
\noindent This says that every effect in the restriction, $\baseeffectchoice \in \restrictions$, can be implemented by an $\N$-run as long as we add a given loop $\periodeffectchoice$ often enough. 
The loop $\periodeffectchoice$ may even be chosen arbitrarily, except that it has to use every period of $\restrictions$.
We now give a proof sketch of the induction step. 
We assume that the NGVAS \(\anngvas\) is non-linear in this overview, but many aspects carry over to the linear case. 
Some ideas are based on KLMST, with extra difficulties though. 

By \perfectnesssolnospace, there is a $\Z$-run $\reachrun$ with the desired effect~$\baseeffectchoice$. 
Our goal is to turn $\reachrun$ into an $\N$-run. 
We discuss some of the problems, solutions, and the shape of the result. 
\begin{problem}
Counters may go negative on \emph{the current level} of nesting.
\end{problem}
To solve this problem, we adapt the run to \(\upseq^{\aconst}.\reachrun.\downseq^{\aconst}\), where \(\upseq, \downseq\) are the loops from \perfectnesspumpingnospace. 
This pumps all counters to high values to ensure \(\reachrun\) does not go negative if \(\aconst\) is large enough.
\begin{problem}
Counters may go negative on a \emph{lower level of nesting}.
\end{problem}
%
The problematic situation is that a counter \(\acounter\) is bounded on the current but unbounded on a lower nesting level.
Then the solution to Problem 1 does not apply, 
but we have to pump the counter \emph{inside the child}. 
To do so, we rely on \perfectnesschildrennospace. 
Since all childNGVAS $\anngvasp$ are perfect, we can invoke the induction hypothesis.
It modifies $\reachrun$ to a run $\reachrun^{(\aconst)}$ by adding periods of $\anngvasp.\restrictions$ that make sure the run is enabled on the lower level.
At this point, the shape is $\upseq^{\aconst}.\reachrun^{(\aconst)}.\downseq^{\aconst}$. 
There are two problems left. 
\begin{problem}
Both additions change the effect of $\reachrun$.
\end{problem}

We may have $\vaseffof{\upseq}+\vaseffof{\downseq} \neq 0$, and also adding the periods to obtain \(\reachrun^{(\aconst)}\) may change the effect. 
To still reach the target \(\acontextout\), we add a loop \(S \rightarrow^{\ast} \diffrunleft.S.\diffrunright\) to the run which cancels out the changes.  
We do not elaborate on the loop construction, the ideas are similar to KLMST.  
At this point, our run is \(\upseq^{\aconst}.\diffrunleft^{\aconst}.\reachrun^{(\aconst)}.\diffrunright^{\aconst}.\downseq^{\aconst}\). 

We focus on a new problem. 
\begin{problem}
$\diffrunleft$ may have a negative effect which $\upseq$ cannot make up for. 
\end{problem}
The problem does not have an analogy in the VASS setting. 
To solve it, 
we modify the run so that it can still be derived in the grammar but also stays non-negative: 
\begin{align*}
\iterrun{\aconst}\quad =\quad \upseq^{\aconst}.\pi^{(k)}(\diffrunleft^{\aconst}.\reachrun^{(\aconst)}.\diffrunright^{\aconst}).\downseq^{\aconst}\ .
\end{align*} 
Here, $\pi^{(\aconst)}$ is a permutation of the inner sequence that we obtain with a new result for context-free grammars (no VASS involved) called \emph{wide tree theorem}. 
To explain it, consider the derivation $S\rightarrow \diffrunleft.S.\diffrunright$ that we want to repeat $\aconst$-times.
The theorem says we can find a parse tree for the repetition that only grows logarithmic in~$\aconst$. 
Behind this is a binary tree construction, illustrated below, which fundamentally needs non-linearity and strong connectedness. 
\vspace{0.1cm}


\begin{center}
\includegraphics[scale=0.7,page=1]{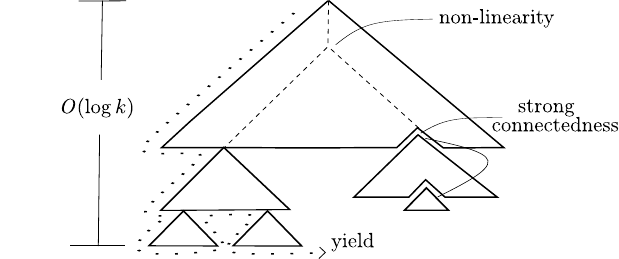}
\end{center}


The relevance to our problem is this. 
As $\diffrunleft+\diffrunright$ was chosen to compensate $\upseq+\downseq$, 
we have the equality $\vaseffof{\upseq+\diffrunleft+\diffrunright}=-\vaseffof{\downseq}>0$. 
Hence, whenever we finish executing a pair $(\diffrunleft, \diffrunright)$, we know this caused a positive effect.  
Only copies of $\diffrunleft$ that are incomplete in that we have not yet executed~$\diffrunright$ may lead to negative counters. 
In the parse tree constructed by the wide tree theorem, only \emph{logarithmically} in $\aconst$ many copies of $\diffrunleft$ can be incomplete.  
Since $\upseq^{\aconst}$ produces a linear in $\aconst$ positive value on every counter, and linear beats logarithmic, $\iterrun{\aconst}$ is an $\N$-run for $\aconst$ large enough.

\newcommand{\extraa}{\textsf{(a)}}
\newcommand{\extrab}{\textsf{(b)}}
\newcommand{\extrac}{\textsf{(c)}}
\newcommand{\extrad}{\textsf{(d)}}
\newcommand{\downseqj}{\downseq^{\acounterp}}
\newcommand{\upseqi}{\upseq^{\acounter}}
\newcommand{\mylast}{\mathit{last}}
\newcommand{\aset}{\mathit{X}}
\newcommand{\asetp}{\mathit{Y}}
\newcommand{\relc}{\mathit{r}}
\newcommand{\downseqrack}{\downseq^{\aset}}
\newcommand{\upseqrack}{\upseq^{\aset}}
\newcommand{\downseqz}{\downseq^{\Z}}
\newcommand{\upseqz}{\upseq^{\Z}}
\newcommand{\aconstz}{\aconst}
\newcommand{\aconstij}{\aconst_{\acounter, \acounterp}}
\newcommand{\aconstj}{\aconst_{\acounterp}}
\newcommand{\rackfun}{\mathit{f}}
\newcommand{\rackfunof}[1]{\rackfun(#1)}
\newcommand{\anumber}{\mathit{l}}
\section{Up- and Down-Pumping}\label{Section:OutlinePumping}
With respect to computing a decomposition of lower rank, most aspects are implemented similar to VASS, the only major difficulty is \(\refinepump\) in \Cref{Lemma:RefineFunctions}. 
Recall that we may assume the input 
fulfills all perfectness conditions except $\perfectnesspumping$ and that $\perffun$ is reliable up to the rank of the input.
The function has to check whether there are up- and down-pumping runs as required by $\perfectnesspumpingnospace$, and if not it has to compute a decomposition $\adeconst$ that is of smaller rank than the input $\anngvas$. 

We have not defined the rank yet, as we only need the following two facts about the rank: That childNGVAS have lower rank, and that more unconstrained counters implies a lower rank. The intuition is that unconstrained counter are easy (in VASS terms \(\Z\)-counters), since we cannot have reachability constraints on them.

\begin{lemma} \label{LemmaBasicRankProperties}
Let \(\anngvasp\) be a childNGVAS of a \emph{non-linear} NGVAS \(\anngvas\). 
Then \(\rankof{\anngvasp}<\rankof{\anngvas}\).

Let \(\anngvas_1, \anngvas_2\) be two NGVAS with \(|\anngvas_1.\unconstrained|>|\anngvas_2.\unconstrained|\). Then \(\rankof{\anngvas_1} < \rankof{\anngvas_2}\).
\end{lemma} 

Our insight is that both the check for pumping runs and the decomposition can be reduced to computing two functions.
Let $\omegamrkdomainof{} = \setcond{\amarking\in \Nomega^{\adim}}{\unconstrained\subseteq\omegaof{\amarking}\subseteq \abdinfomid}$ contain the generalized markings of interest, and $\afuncdomain=\omegamrkdomainof{}\times(\nonterms\cup\trms)$.
We let $\postfuncN{\anngvas}, \prefuncN{\anngvas}:\ \afuncdomain\rightarrow\powof{\Nomega^{\adim}}$ with
\begin{align*}
 \postfuncNof{\anngvas}{\amarking, \asymbol} &=\iddecompof{\downclsof{\setcompact{\amarkingppp\in\N^{d}}{(\amarkingpp, \arun, \amarkingppp)\in\runsof{\asymbol},\;\amarkingpp\sqsubseteq\amarking}}}\\
    \prefuncNof{\anngvas}{\amarkingp, \asymbol}&=\iddecompof{\downclsof{\setcompact{\amarkingpp\in\N^{d}}{(\amarkingpp, \arun, \amarkingppp)\in\runsof{\asymbol},\;\amarkingppp\sqsubseteq\amarkingp}}}.
\end{align*}
The specialization order $\sqsubseteq$ on $\N_{\omega}$ is $\omega \sqsubseteq \omega$, $k \sqsubseteq k$, and $k \sqsubseteq \omega$ for all $k \in \N$.  
We lift it to $\N_{\omega}^d$ component-wise.  

Function $\postfuncN{\anngvas}$ takes as input a generalized marking~$\amarking$ and a terminal or non-terminal symbol~$\asymbol$. 
It considers all concrete markings $\amarkingpp$ represented by $\amarking$ and computes the set of reachable markings, more precisely, the set of all markings reachable via runs that can be derived from $\asymbol$. 
To obtain a finite representation, the function closes the set downwards ($\downclsof{-}$) and computes the decomposition into maximal ideals ($\iddecomp$).
The downward closure does not lose information when it comes to pumping. 
The set of maximal ideals in a well-quasi order is guaranteed to be finite, and in our case maximal ideals correspond to generalized markings~\cite{FinkelG09}. 

We give the reduction and then the computability of $\prefuncN{\anngvas}$ and $\postfuncN{\anngvas}$. The latter is a main achievement.
\subsection{Computing $\refinepump$}\label{Section:OutlinePumpingDec}
The function first has to check the existence of an up-pumping run $\upseq$ from $\acontextin$ and a down-pumping run $\downseq$ from $\acontextout$ so that $\startnonterm\rightarrow^*\upseq.\startnonterm.\downseq$.  
Assuming the computability of $\prefuncN{\anngvas}$ and $\postfuncN{\anngvas}$, we can use a mild generalization of the Karp-Miller tree.  

The Karp-Miller tree is labeled by triples $(\amarking_1, \anonterm, \amarking_2)$ from $\Nomega^{\adim}\times\nonterms\times\Nomega^{\adim}$. 
The non-terminal $\anonterm$ should be understood as to lie on the path $\startnonterm\rightarrow^*\upseq.\startnonterm.\downseq$. 
The markings $\amarking_1$ and $\amarking_2$ represent the current outcome of $\upseq$ and $\downseq$, computed using $\postfuncN{\anngvas}$ on $\acontextin$ resp.  $\prefuncN{\anngvas}$ on $\acontextout$. 
To be precise, the root is $(\acontextin, \startnonterm, \acontextout)$. 
We extend a node $(\amarking_1, \anonterm, \amarking_2)$ by considering all rules $\anonterm\rightarrow \anontermp_1.\anontermp_2$. 
This can be understood as selecting the parse tree that contains the pumping situation. 
We consider both choices $\anontermp=\anontermp_1$ and $\anontermp=\anontermp_2$, which corresponds to selecting the path in this parse tree.
Let $\anontermp=\anontermp_2$.  
We consider all $\amarkingp_1\in \postfuncNof{\anngvas}{\amarking_1, \anontermp_1}$. 
With these choices made, we define the successor node in the tree as $(\amarkingp_1, \anontermp, \amarking_2)$. 
We achieve termination with an acceleration step.
We look for a node $(\amarkingp_1', \anontermp, \amarking_2')$ on the path to the new successor so that $(\amarkingp_1', \amarking_2')< (\amarkingp_1, \amarking_2)$.
Then we replace the entries in our successor by $\omega$ wherever the inequality is strict. 
\begin{restatable}{lemma}{LemmaKarpMillerTreeOverview}\label{Lemma:WitnessGrammarTermination}\label{Lemma:WitnessGrammarProp}
Assume $\postfuncN{\anngvas}$ and $\prefuncN{\anngvas}$ are computable for~$\anngvas$.  
Then the Karp-Miller tree construction terminates. 
Moreover, $\anngvas$ satisfies \perfectnesspumping if and only if the tree contains a node $(\amarking_1, \anonterm, \amarking_2)$ s.t. \(\amarking_j[i]=\omega\) for all \(i \in D\) and \(j \in \{1,2\}\).%
\end{restatable}

If there is no pumping run, we have to decompose the NGVAS.  
Unfortunately, the Karp-Miller tree is not precise enough for this: 
it mixes parse trees, has different views on the same parse tree, and only has incomplete information about runs. 

We compute the decomposition from a new type of context-free grammar called a  \emph{coverability grammar}. 
The idea is to restrict the NGVAS by $\postfuncN{\anngvas}$ and $\prefuncN{\anngvas}$. 
We implement this restriction with two triple constructions that are inspired by the intersection of a context-free with a regular language. 
The non-terminals in the coverability grammar are $5$-tuples $(\amarking_1, \amarking_2, \anonterm, \amarkingp_2, \amarkingp_1)$ from $\Nomega^{\adim}\times\Nomega^{\adim}\times\nonterms\times\Nomega^{\adim}\times\Nomega^{\adim}$. 
When denoted as the more familiar triple, $(\amarking_1, 
\anonterm, \amarking_2)$ says that $\anonterm$ can produce a run which transforms $\amarking_1$ into~$\amarking_2$, meaning $\amarking_2\in\postfuncNof{\anngvas}{\amarking_1, \anonterm}$.  
The triple $(\amarkingp_2, \anonterm, \amarkingp_1)$ should be read similarly but with $\amarkingp_2\in\prefuncNof{\anngvas}{\amarkingp_1, \anonterm}$. 
For a production $\anonterm\rightarrow\anontermp_1.\anontermp_2$ in the NGVAS, the coverability grammar will have
\begin{align*}
(\amarking_1, \amarking_2, \anonterm, \amarkingp_2, \amarkingp_1)&\rightarrow\\
& (\amarking_1, \amarking_2', \anontermp_1, \amarkingp_3, \amarkingp_{2}').(\amarking_2', \amarking_3, \anontermp_2, \amarkingp_2', \amarkingp_1).
\end{align*}
The markings are chosen as expected, $\amarking_2'\in \postfuncNof{\anngvas}{\amarking_1, \anontermp_1}$ is the start marking for the first triple in the second non-terminal, and we again compute $\amarking_3\in \postfuncNof{\anngvas}{\amarking_2', \anontermp_2}$. 
For the other triple, the reasoning is similar. 
The new source and target markings have to respect the reachability that was promised initially, meaning we only add the production if $\amarking_3\sqsubseteq \amarking_2$ and $\amarkingp_3\sqsubseteq \amarkingp_2$. 
There is again an acceleration step that guarantees finiteness. 
Note that we do not restrict $\amarking_3$ to $\amarkingp_1$ and $\amarkingp_3$ to $\amarking_1$. 
With such a restriction, the second triple would eliminate $\omega$ entries from the first and vice versa. 
Then the number of $\omega$ entries would not grow monotonically, and we would lose termination.  
As we have defined it, the two triple constructions do not influence each other. 

The \emph{decomposition} consists of a single NGVAS. We split the coverability grammar into its strongly-connected components, and use them to create the nesting structure for the new NGVAS.

Recall that generalized markings represent downward-closed sets. 
We can implement the intersection of these sets by an elementwise minimum wrt. $\sqsubseteq$ on the generalized markings.  
To determine the boundedness information for $(\amarking_1, \amarking_2, \anonterm, \amarkingp_2, \amarkingp_1)$, we use $\amarking_1\cap \amarkingp_2$ and $\amarking_2\cap \amarkingp_1$.
This uses the most precise information we can obtain from the two triples. 

To see that the rank decreases, we prove that in every SCC some counter \(\acounter\) is bounded (as opposed to VASS however, there might not be a counter \(\acounter\) which is bounded in every SCC). To see this, observe that the corresponding node in the Karp-Miller graph would be of the form \((\amarking_1, \anonterm, \amarking_2)\) with \(\amarking_j[i]=\omega\) for all \(\acounter \in \adimset\) and \(j \in \{1,2\}\), and hence \perfectnesspumping holds. Similar to \cite{LerouxS19} for VASS, we will use the cycle space of SCCs as the rank, and bounding a counter decreases the cycle space. Decreasing all cycle spaces decreases the rank.

Before we state the formal guarantees given by the decomposition, we generalize the definition of the coverability grammar.
\subsection{Generalization}\label{Section:OutlineApproximations}
Rather than defining the coverability grammar only for the functions $\prefuncN{\anngvas}$ and $\postfuncN{\anngvas}$, we take two such functions as parameters and define $\wtgrammarof{\anngvas, \preapprox, \postapprox}$ relative to them. 
We expect that $\preapprox$ and $\postapprox$ are \emph{sound approximations} of $\prefuncN{\annagvas}$ resp. $\postfuncN{\anngvas}$ for $\anngvas$, meaning they are computable and over-approximate the original functions (in a precise sense).  
We use the notation $\anngvas_{\agram}$ for the decomposition computed from $\agram=\wtgrammarof{\anngvas, \preapprox, \postapprox}$.
\begin{restatable}{lemma}{LemmaDecompositionPumpingOverview}\label{Lemma:Decomp}
Let $\preapprox$ and $\postapprox$ be sound approximations of $\prefuncN{\annagvas}$ resp.~$\postfuncN{\anngvas}$ for $\anngvas$.  Then the computation of $\anngvas_{\agram}$ terminates. 
Moreover, if $\anngvas$ does not satisfy $\perfectnesspumpingnospace$ even in this approximation, 
then $\anngvas_{\agram}$ is a decomposition of $\anngvas$.
\end{restatable}
\subsection{Computing $\postprefunc$: Simple Cases}\label{Section:OutlineSimpleCases}
We now prove the computability of $\prefuncN{\anngvas}$ and $\postfuncN{\anngvas}$ that is behind our implementation of $\refinepump$. 
We focus on $\postfuncN{\anngvas}$ as the reasoning for $\prefuncN{\anngvas}$ is similar.   
Our strategy is again to establish assumptions as strong as possible. 
For the input NGVAS, we already have all perfectness conditions except $\perfectnesspumpingnospace$ and the reliability of $\perffun$ up to its rank. 
We now introduce assumptions $\extraa$ to $\extrad$ and show that every single one of them makes $\postfuncN{\anngvas}$ easy to compute.  
What we gain by this is that we can assume $\neg \extraa\wedge\neg \extrab\wedge\neg \extrac \wedge \neg \extrad$ when proving the computability for the remaining hard cases. 
Fix a marking-symbol pair $(\amarking,\anonterm)\in \afuncdomain\subseteq\omegamrkdomainof{}\times(\trms\cup\nonterms)$.
We show how to compute $\postfuncNof{\anngvas}{\amarking, \asymbol}$ under each of the assumptions. 

Assumption \extraa\ is that $\asymbol\in\trms$, we have a childNGVAS. 
ChildNGVAS have a lower rank, so the reliability of $\perffun$ allows us to compute a perfect decomposition. 
From this decomposition, we can read-off the values of $\postfuncN{\anngvas}$.   

Assumption \extrab\ is that $\unconstrained\subsetneq \omegaof{\amarking}\subseteq \abdinfomid$.
We construct an NGVAS $\otherctxNGVAS{\anngvas}{\amarking, \asymbol, \outof{\asymbol}}$ that still has $\postfuncNof{\anngvas}{\amarking, \asymbol}$ as the reachable markings but a smaller rank. 
On this NGVAS, we reason as in \extraa.
The construction simply changes the input marking to $\amarking$, the initial non-terminal to $\asymbol$, and replaces the set of unconstrained counters by $\omegaof{\amarking}$. 
The latter is what makes the rank go down. 

From now on, we can assume $\unconstrained= \omegaof{\amarking}\subseteq \abdinfomid$, the negation of $\extrab$. 
We call counters from $\abdinfomid\setminus \unconstrained$ \emph{relevant} as they should be pumped.  
Assumption \extrac\ is that \perfectnesspumping even fails in an approximation that ignores relevant counters.  
More precisely, we have a family of approximations, one for each pair of functions $\postapprox_{\acounter}$, $\preapprox_{\acounterp}$ that ignore the relevant counters  $\acounter$ resp. $\acounterp$. 
Ignoring these counters means we set them to $\omega$ in the given marking.  
The computability of $\postapprox_{\acounter}$, $\preapprox_{\acounterp}$ follows as in \extrab. 
Given their computability, we can check the failure of~\perfectnesspumping\ using the Karp-Miller tree. 
\Cref{Lemma:Decomp} now yields a decomposition that reduces the rank. 
We then argue as in \extraa. 

Assumption \extrad\ is that \perfectnesspumping fails in an approximation that does not have to keep the counters non-negative.  
The underlying functions $\intpreapprox$ and $\intpostapprox$ are trivially computable. 
To check \perfectnesspumpingnospace, we do not need the Karp-Miller tree but can use integer linear programming.
We then reason as in \extrac.

The approximations in \extrac\ and \extrad\ are incomparable. 
The approximation based on $\postapprox_{\acounter}, \preapprox_{\acounterp}$ looks for runs that pump all relevant counters except $\acounter$ and~$\acounterp$. 
These pumping runs are guaranteed to remain non-negative on all counters except $\acounter$ resp.~$\acounterp$.  
On $\acounter$ resp.~$\acounterp$, they are allowed to fall below zero and even have a negative total effect. 
The approximation based on $\intpreapprox, \intpostapprox$ looks for runs that pump all relevant counters. 
These pumping runs are allowed to fall below zero on any counter, but guarantee a non-negative effect. 
\subsection{Computing $\postprefunc$: Hard Case 1}\label{Section:PostHardCase1}
We show how to compute $\postfuncNof{\anngvas}{\amarking, \anonterm}$ with $\omegaof{\amarking}=\unconstrained$ for an NGVAS that is almost perfect: we even have up- and down-pumping runs as soon as we ignore an arbitrary pair of relevant counters (since $\extrac$ failed), and we have up- and down-pumping runs for all relevant counters if we admit integer values (since $\extrad$ failed). 

We proceed with a Rackoff-like case distinction~\cite{Rackoff78}.  
We show that we can compute a bound~$\abigbound$ that satisfies the following.  
If in the input marking $\amarking$ the value of \emph{some} relevant counter exceeds $\abigbound$, then we show that we can find pumping runs by stitching together the almost pumping runs. 
This means the NGVAS is perfect and we can read-off $\postfuncNof{\anngvas}{\amarking, \anonterm}$ as in \extraa. 
If in marking $\amarking$ all relevant counters are bounded by $\abigbound$, then we show how to compute $\postfuncNof{\anngvas}{\amarking, \anonterm}$ in the next section. 

We explain how the almost pumping runs lead to a full pumping run provided we have a high value, say in the relevant counter $\acounter$. 
Consider marking $\amarking = \acontextin$ and non-terminal $\anonterm=\startnonterm$. 
Let $\upseqi$ be the up-pumping run that ignores counter $\acounter$ and let $\upseqz$ be the up-pumping run that may fall below zero. 
For the counters that are tracked concretely in the NGVAS, $\upseqi$ and $\upseqz$ are enabled and have effect zero. 
We can ignore these counters in our analysis. 
The same holds for the counters in $\omegaof{\acontextin}$. 
For simplicity, let the counter updates stem from $\set{-1, 0, 1}$. 
With this assumption, $\upseqz$ is enabled as soon as we have a value of $\cardof{\upseqz}$ in every relevant counter.
Moreover, $\cardof{\upseqi}$ is a bound on the negative effect that $\upseqi$ may have on counter~$\acounter$. 
Then 
\begin{align*}
\upseq\quad = \quad (\upseqi)^{\cardof{\upseqz}}.(\upseqz)^{\cardof{\upseqz}\cdot\cardof{\upseqi}+1}
\end{align*}
has a positive effect on all relevant counters. 
Indeed, the only negative effect is due to $\upseqi$, and this is compensated by the repetition of $\upseqz$. 
The run is enabled as soon as we have a value of $\cardof{\upseqz}\cdot (\cardof{\upseqi} + 1)$ in counter $\acounter$. 
Remember that we have to derive the up- and down-pumping runs together. 
Let $\aconst$ be the maximum of $\cardof{\upseqz}$ and $\cardof{\downseqz}$, 
and let $\aconstij$ be the maximum of $\cardof{\upseqi}$ and $\cardof{\downseqj}$.
\begin{restatable}{lemma}{LemmaSimpleBoundPumpingOverview}\label{Lemma:SimpleBound}
Let $\anngvas$ be perfect except for \perfectnesspumpingnospace\ and assume $\neg \extrab\wedge \neg \extrac\wedge \neg\extrad$ holds. 
If there are relevant counters $\acounter$ and $\acounterp$ so that $\at{\acontextin}{\acounter}, \at{\acontextout}{\acounterp}\geq \aconstz\cdot (\aconstij+1)$, then $\anngvas$ is perfect. 
\end{restatable}

The lemma refers to a single NGVAS. 
To compute the function $\postfuncN{\anngvas}$, we need a lower bound on the values of (arbitrary pairs) $\acounter$ and $\acounterp$ that works for all NGVAS $\otherctxNGVAS{\anngvas}{\amarking, \anonterm, \outof{\anonterm}}$, where the given $\amarking$ is the initial marking and the given $\anonterm$ is the initial non-terminal. 
As there are only finitely many non-terminals, we can consider a fixed~$\anonterm$. 
The runs $\upseqz$ and $\downseqz$ do not have to stay non-negative.
This means they only depend on the set of relevant counters, but not on the choice of~$\amarking$. 
As the set of relevant counters is also fixed, we can consider an arbitrary pair of up- and down-pumping $\Z$-runs.
Let $\aconstz$ be a bound on their length.

The challenge is to compute a bound on the length of $\upseqi$ and $\downseqj$. 
We follow Rackoff~\cite{Rackoff78} and proceed by an induction on the number of relevant counters that can be pumped. 
Let $\relc$ be the number of relevant counters. 
We define a function $\rackfun:[0, \relc-1]\rightarrow\N$ so that $\rackfunof{\anumber}$ is a bound on the maximal length of shortest runs $\upseqrack$ and $\downseqrack$ that pump the relevant counters in $\aset$ and ignore the remaining relevant counters by setting them to $\omega$. 
The bound is taken over all initial markings, and we show in the proof that it exists. 
Moreover, it is taken over all sets $\aset$ with $\cardof{\aset}\leq \anumber$.
Note that we are only interested in sets that ignore at least one relevant counter. 
This is to match the definition of $\upseqi$ and $\downseqj$. 
Clearly, $\rackfunof{0}=1$, if there are no counters to pump, the empty run works.
In the induction step, assume we want to pump a set $\aset$ of $\anumber+1$ counters. 
The set of markings that enable pumping runs is upward-closed.
By the well-quasi order on $\N^{\adim}$, it contains a finite set of minimal elements. 
We show how to compute the minimal elements.
Once we have them, we also have the corresponding runs, and hence a bound on $\rackfunof{\anumber+1}$.

We compute the set of minimal markings that admit pumping runs for $\aset$.
We do have pumping runs from the marking with value $\aconst$ in all counters from $\aset$, namely $\upseqz$ and~$\downseqz$. 
For every marking that is strictly smaller than a marking we already have, we check the existence of pumping runs using the Karp-Miller tree, and the tree provides a run if the answer is positive. 
We can rely on the computability of $\postfuncN{\anngvas}$ and $\prefuncN{\anngvas}$ by \extrab, there is at least one relevant counter we ignore.  
The challenge lies in the configurations that are incomparable to the minimal ones we already have. 
Consider a configuration that is strictly smaller than a minimal one in the counters $\asetp$.
We set the remaining counters to $\omega$, and check in the Karp-Miller tree the existence of up- and down-pumping runs. 
If the answer is positive, the induction hypothesis applies and gives us runs of length at most $\rackfunof{\cardof{\asetp}}$. 
These runs may not be pumping on $\aset\setminus\asetp$, and even have a negative effect there.  
We construct pumping runs for the full set $\aset$ with the technique from \Cref{Lemma:SimpleBound}.
The lemma gives us a bound for the values of the counters in $\aset\setminus\asetp$ that is needed to enable the constructed runs, $k\cdot(\rackfunof{\cardof{\asetp}}+1)$. 
Note how the bound on the length of the runs from the induction hypothesis is crucial to obtain the bound on the counter values. 
\begin{restatable}{lemma}{LemmaBoundPostPumpingOverview}\label{Lemma:Bound}
Let $\anngvas$ be perfect except for \perfectnesspumpingnospace\ and assume $\neg \extrab\wedge \neg \extrac\wedge \neg\extrad$ holds. 
We can compute a bound $\abigbound$ so that $\postfuncNof{\anngvas}{\amarking, \anonterm}$ is computable for all markings $\amarking$ with $\at{\amarking}{\acounter}> \abigbound$ for some $\acounter$. 
\end{restatable}
\subsection{Computing $\postprefunc$: Hard Case 2} \label{Section:PostHardCase2}
It remains to compute $\postfuncNof{\anngvas}{\amarking, \anonterm}$ for markings $\amarking$ in which all relevant counters are bounded by $\abigbound$. 
The idea is to conduct an exhaustive search. 
The objects computed by this search are so-called \emph{marked parse trees}, parse trees in the context-free grammar underlying the NGVAS whose nodes are decorated by input and output markings $(\amarking, \asymbol, \amarkingp)$,   very similar to how the NGVAS annotates the terminals and non-terminals by input and output markings. 
There are two techniques that guarantee the termination of our search.

The first technique is that once we encounter a node whose input marking $\amarking$ exceeds the bound $\abigbound$ in some counter, we do not further expand the node but rely on the computability of $\postfuncNof{\anngvas}{\amarking, \anonterm}$ in \Cref{Lemma:Bound}, instead.  
This, however, is not enough for termination.

The second technique is an acceleration that introduces $\omega$ entries when an output marking is detected to grow unboundedly. 
Note that we cannot use the same acceleration as in the Karp-Miller tree. 
There, our goal was to find $\omega$ entries in the middle of the derivation.
Here, our goal is to find the $\omega$ entries that are the result of the derivation.
We proceed as follows.
We construct a search tree whose nodes are labeled by marked parse trees. 
A child node expands the parse tree of its parent, meaning the parent-child relationship is one of being a subtree. 
Assume the first technique does not apply.
Then the set of possible input markings is finite, and there are only finitely many marked parse trees for each height.
As a result, the search tree has finite outdegree. 
Hence, if the search does not terminate, it contains an infinite path.
\vspace{0.1cm}We consider the labeling of the root nodes $(\amarking_i, \anonterm_i, \amarkingp_i)_{i\in\N}$ in the marked parse trees on this path. 
As there are finitely many input markings and non-terminals, the sequence contains an infinite subsequence of the form $(\amarking, \anonterm, \amarkingp_{\varphi(i)})_{i\in\N}$. 
By the well-quasi order of~$\Nomega^{\adim}$ and the fact that the output markings grow unboundedly, we can even assume $\amarkingp_{\varphi(0)}<\amarkingp_{\varphi(1)}<\ldots$
Like in the Karp-Miller tree, we introduce an $\omega$ in the output counters where the inequality is strict. 
The acceleration situation is illustrated below.
\vspace{-0.3cm}
\begin{center}
\includegraphics[scale=0.8,page=1]{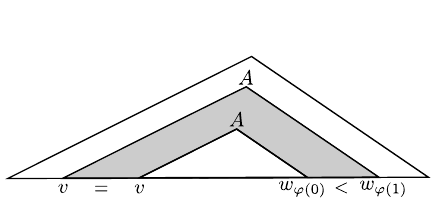}
\end{center}
\vspace{-0.3cm}
\vspace{0.2cm}
We can apply the acceleration at most $\adim$ times along a branch.
Since marked parse trees are not allowed to repeat nodes along a branch, we have termination. 

As for soundness of the acceleration, the pumping lemma for context-free languages allows us to repeat the derivation between two nodes, say from $(\amarking, \anonterm, \amarkingp_{\varphi(1)})$ to $(\amarking, \anonterm, \amarkingp_{\varphi(0)})$.
The fact that the initial marking coincides and the final markings satisfy $\amarkingp_{\varphi(0)}<\amarkingp_{\varphi(1)}$, 
combined with the monotonicity of firing in VAS, guarantees that the repetition is enabled.
\begin{restatable}{lemma}{LemmaPostComputableOverview}\label{Lemma:PostComputableOverview}
Let $\anngvas$ be perfect except for \perfectnesspumpingnospace\ and assume $\neg \extrab\wedge \neg \extrac\wedge \neg\extrad$ holds. 
Then $\postfuncNof{\anngvas}{\amarking, \anonterm}$ is computable for all markings $\amarking\leq\abigbound^{\adim}$. 
\end{restatable}
\section{Complexity} \label{SectionLevelOneComplexity}
It remains to define the rank of an NGVAS and analyze the complexity of our decision procedure. 
For the latter, we employ a precise analysis of recursive procedures, which we believe will be of independent interest.
\subsection{Bad Partially Nested Sequences} \label{SubSubsectionPartiallyNestedSequences}

Recursions with ranks from \(\N^{\beta}\) for  \(\beta \in \N\) are common in the context of infinite-state systems, and there are generic tools for analyzing their complexity: \emph{controlled bad sequences} for tail recursions~\cite{FigueiraFSS11} and \emph{controlled bad nested sequences} for full (non-tail) recursions~\cite{LerouxPS14}. 
Theorem VI.1 in \cite{LerouxPS14} gives Hyper-Ackermann time bounds for full recursions. 
There is, however, a problem with these generic tools: they analyze all procedures equally but do not consider the case where some procedures are tail-recursive and others are fully recursive. 
To take into account this distinction, we introduce controlled bad \emph{partially nested} sequences, a subclass of controlled bad nested sequences that admit a finer complexity analysis.  

\newcommand{\astack}{\mathit{s}}
\newcommand{\ansf}{\mathit{y}}
\newcommand{\setfun}{\mathit{P}}
\newcommand{\afun}{\mathit{f}}
\newcommand{\afunp}{\mathit{g}}
\newcommand{\aglob}{\mathit{u}}
\newcommand{\aglobp}{\mathit{v}}
\newcommand{\aloc}{\mathit{a}}
\newcommand{\alocp}{\mathit{b}}
\newcommand{\grank}{\mathit{r}}
\newcommand{\grankof}[1]{\mathit{r}(#1)}
\newcommand{\gdom}{\mathsf{U}}
\newcommand{\ldomof}[1]{\mathsf{A}_{#1}}
\newcommand{\ldomdef}{\ldomof{\afun}}
\newcommand{\wfdom}{\mathsf{X}}
\newcommand{\lwfdomof}[1]{\wfdom_{#1}}
\newcommand{\lwfdomdef}{\lwfdomof{\afun}}
\newcommand{\funrank}[1]{\mathit{r}_{#1}}
\newcommand{\funrankof}[2]{\funrank{#1}(#2)}
\newcommand{\funrankdef}{\funrank{\afun}}
\newcommand{\funrankdefof}[1]{\funrankdef(#1)}

\newcommand{\domsf}{\mathsf{Y}}
\newcommand{\aheight}{\mathit{h}}

\newcommand{\controlfun}{\mathit{c}}
\newcommand{\controlfunof}[1]{\controlfun(#1)}
\newcommand{\initbound}{\mathit{m}}

\newcommand{\infnormof}[1]{|\!|#1|\!|_{\infty}}
\newcommand{\leqlex}{\leq_{\mathit{lex}}}

\newcommand{\fastgrowof}[1]{\mathfrak{F}_{#1}}

\newcommand{\amainrank}{\mathit{x}}
\newcommand{\anauxrank}{\mathit{z}}

\newcommand{\acomp}{\tau}
\newcommand{\absof}[1]{\alpha(#1)}
\newcommand{\aseq}{\sigma}

We explain the programming model that we use for our complexity analysis.  
The essential characteristic of our decision procedure is that (i) it is heavily recursive but (ii) the recursive functions all have the same type, they map NGVAS to sets of NGVAS.
To capture this, we define a recursive program as a finite set of functions $\setfun$ that all have type $\gdom\rightarrow \powof{\gdom}$ for some domain~$\gdom$.  
During its computation, each function $\afun\in\setfun$ may transform the input, meaning it has access to state from $\gdom$. 
Moreover, the function may maintain auxiliary state from some domain $\ldomdef$. 
An example for such auxiliary state is the Karp-Miller tree. 
We do not need to be precise about the syntax of programs and can also be rough about the semantics.  
In fact, even the domains we just introduced will be abstracted away by ranking functions. 
This high level of abstraction is needed to deal with objects as complicated as NGVAS and functions as complicated as $\postfuncN{\anngvas}$. 
We only need stack frames $\ansf=(\aglob, \afun, \aloc)$ which consist of state $\aglob\in\gdom$, the current function $\afun\in\setfun$, and auxiliary state $\aloc\in\ldomdef$. 
A stack is then a finite sequence of stack frames $\astack=\ansf_0\ldots \ansf_k$ with $\ansf_k$ being topmost.  

We are interested in proving the termination of such recursive programs. 
To this end, we assume to have a ranking function \(\grank:\gdom\rightarrow \wfdom\) that maps the state from $\gdom$ into a well-founded total order $\wfdom$. 
With $\gdom$ being NGVAS, the function $\grank$ will define their rank. 
Moreover, for every \(\afun \in \setfun\), we need a ranking function \(\funrankdef:\ldomdef\rightarrow\lwfdomdef\) that maps the auxiliary state into another well-founded total order. The auxiliary ranking function $\funrankdef$ may be chosen depending on $\afun\in\setfun$, but we expect that all functions in the program agree on the main rank $\grank$.
Finally, we assume a total order $\prec$ on the functions in $\setfun$. 

The goal of our complexity analysis is to separate recursive from tail-recursive behavior.
We achieve this by constraining the transition relation $\astack_1\rightarrow \astack_2$ between stacks, and analyzing computations that adhere to these constraints.  
An invariant of our transitions is that they do not increase the main rank.  
Moreover, recursive calls have to (i) reduce the main rank or (ii) call a function that is smaller in the total order \(\prec\).
The consequence is that transitions which only reduce the auxiliary rank cannot be recursive calls: computation on the auxiliary state is tail-recursive. 
We admit the following.     
\begin{enumerate}
\item A tail-recursive step $\astack.(\aglob, \afun, \aloc)\rightarrow \astack.(\aglobp, \afun, \alocp)$ and a return step from a recursive function $\astack.(\aglob, \afun, \aloc).(\aglob', \afun', \aloc')\rightarrow \astack.(\aglobp, \afun, \alocp)$, provided the main rank does not go up and a rank goes down: $\grankof{\aglob}\geq\grankof{\aglobp}\wedge (\grankof{\aglob}>\grankof{\aglobp}\vee \funrankdefof{\aloc} >\funrankdefof{\alocp})$. 
\item A recursive call  $\astack.(\aglob, \afun, \aloc)\rightarrow \astack.(\aglob, \afun, \aloc).(\aglobp, \afunp, \alocp)$, provided the main rank does not go up, and the main rank goes down or the callee is smaller,  $\grankof{\aglob}\geq\grankof{\aglobp}\wedge (\grankof{\aglob}>\grankof{\aglobp}\vee \afun\succ \afunp)$. 
\end{enumerate}
Note that the auxiliary state can be chosen freely if we can ensure the main rank goes down.

It is readily checked that from every initial stack every computation in our model terminates. 
The reason is that a computation is a sequence in \(\wfdom \times \setfun \times \bigcup_{f} \lwfdomdef\) that is strictly decreasing in the lexicographic ordering, and this ordering is well-founded. 
This is enough to show termination of our decision procedure. 

To estimate the complexity, we need notions from fast-growing complexity. 
By further abstracting from stack frames and even stacks, we arrive at the nested sequences from~\cite{LerouxPS14}. 
Towards this abstraction, note that our termination analysis does not need the domains $\gdom$ and $\ldomdef$, but only the well-founded total orders. 
We fix them to $(\N^{\beta}, \leqlex)$.   
Then (the abstract version of) a stack frame is an element of \(\domsf=\N^{\beta_1} \times \setfun \times \N^{\beta_2}\). 
We also abstract away the stack content and only keep the topmost stack frame and the height of the stack. 
With this, (the abstract version of) a stack is an element $(\ansf, \aheight)\in\domsf\times\N$. 
The abstraction of a computation is a \emph{nested sequence}, a sequence of stacks $(\ansf_0, \aheight_0).(\ansf_1, \aheight_1)\ldots$ whose height changes as expected, $\aheight_{k+1} \in \aheight_k + \{-1,0,1\}$ for all $k\in\N$. 
To capture the requirement that every transition decreases the rank, we first define the contrary, the moment of termination. 
A nested sequence is \emph{good}~\cite{LerouxPS14}, if there are indices $k<m$ so that $\ansf_k \leqlex \ansf_m$ and $\aheight_{k+1}, \dots, h_m \geq \aheight_k$. 
It is \emph{bad} if it is not good. 

Our complexity analysis determines a bound on the length of bad nested sequences. 
This also yields a bound on the termination time for our programming model, in particular for \(\perffun\), as follows.
For every computation $\acomp=\astack_0\rightarrow\astack_1\rightarrow\ldots$ in a program, the abstraction $\absof{\acomp}=(\ansf_0, \aheight_0).(\ansf_1, \aheight_1)\ldots$ is a bad nested sequence.

Clearly, without any restrictions even non-nested bad sequences like \((1,0) \to (0,B) \to (0,B-1) \to \dots\) in \(\N^2\) may be arbitarily long. Hence we control the way in which transitions may modify the ranks~\cite{FigueiraFSS11}. 
Let $\controlfun \colon \N \to \N$ be a strictly increasing (control) function and let $\initbound \in \N$ be a bound on the size of the initial stack. 
The nested sequence $(\ansf_0, \aheight_0).(\ansf_1, \aheight_1)\ldots$ is \emph{$(\controlfun, \initbound)$-controlled}~\cite{LerouxPS14} if $\sizeof{\ansf_k} \leq \controlfun^{(k)}(\initbound)\) for all \(k\in\N\). Here, we use \(\controlfun^{(k)}\) for the \(k\)-fold application of \(\controlfun\). 
Note that the boundedness of the stack frames is not built into our programming model. 
So to employ controlled sequences for the complexity analysis, we need to argue that \(\perffun\) respect the bounds. 
This, however, will be unproblematic as we will admit a general primitive-recursive control \(\controlfun\).

Recall that $\domsf=\N^{\beta_1} \times \setfun \times \N^{\beta_2}$. For controlled bad \emph{nested} sequences, the only upper bound on the length in the literature is $\fastgrowof{\cardof{\setfun}\cdot \omega^{\beta_1+\beta_2}}$, obtained from \cite[Theorem VI.1]{LerouxPS14}, which assumes that \(\controlfun(\initbound)=\initbound+1\).
Here, $\fastgrowof{\alpha}$ is the hierarchy of fast-growing functions, see e.g.~\cite{Schmitz16} for details.  
With our analysis, transitions are allowed to take an arbitrary primitive-recursive amount of time, and we can improve the upper bound to \(\fastgrowof{\omega^{\beta_1}}\). Observe that this eliminates the dependence on \(\beta_2\) and on \(\cardof{\setfun}\). 

What allows us to improve the analysis is the fact that, in our programming model, computation on the auxiliary rank is tail-recursive. 
Note that this is not yet included in the above definition of nested sequences. 
\begin{definition}
A nested sequence $(\ansf_0, \aheight_0).(\ansf_1, \aheight_1)\ldots$ is \emph{partially nested} if $\ansf_k=(\amainrank, \afun, \anauxrank_k)$ and $\ansf_{k+1}=(\amainrank, \afun, \anauxrank_{k+1})\) implies \(\aheight_{k+1} \leq \aheight_k\), for all $k\in\N$.
\end{definition}
We are allowed to analyze the length of controlled bad  partially nested sequences, because the abstraction $\absof{\acomp}$ of a computation $\acomp$ in our model is not only a bad nested sequence but actually a bad partially nested sequence.
The bound we obtain with this analysis will also be precise, as we have the following. 
For every bad partially nested sequence $\aseq=(\ansf_0, \aheight_0).(\ansf_1, \aheight_1)\ldots$ there is a computation $\acomp=\astack_0\rightarrow\astack_1\rightarrow\ldots$ so that $\aseq=\absof{\acomp}$.

We explain how partially nested sequences eliminate the dependence on \(\beta_2\) and \(\cardof{\setfun}\). 
Consider the case where \(\beta_1=1\) and so \(\domsf=\N \times \setfun \times \N^{\beta_2}\). 
Assume we start with value \(k\in\N\) for the main rank. 
Then a partially nested sequence can be modelled as a tail recursion with a rank in \(\N^{(k+1) \cardof{\setfun} \cdot \beta_2}\), with the following argument.
At every moment in the sequence, the call stack has height at most \((k+1) \cardof{\setfun}\). 
We can store the \((k+1)\cardof{\setfun}\) auxiliary ranks that lie on the stack in a long tuple. 
Using the bad sequence analysis of \cite[Prop. 5.2]{FigueiraFSS11}, we arrive at a complexity of \(\fastgrowof{(k+1)\cardof{\setfun} \cdot \beta_2}\subseteq \fastgrowof{\omega}\).  
In fact, already \(\beta_1 = 1 = \beta_2=\cardof{\setfun}\) would lead to \(\fastgrowof{\omega}\), as one can directly implement the Ackermann function. 
The consequence is that the details \(\beta_2\) and \(\cardof{\setfun}\) vanish once we round up to the next power of \(\omega\), which will always be \(\omega^{\beta_1}\).

\begin{restatable}{theorem}{TheoremLengthBoundBadSequences} \label{TheoremLengthBoundBadPartiallyNestedSequences}
Let \(\beta_1, \initbound, \beta_2 \in \N\) and let \(\cardof{\setfun}\) be finite. 
Let \(\alpha\) be an ordinal and let \(\controlfun\in \fastgrowof{\alpha}\) be strictly increasing. 
Then the length of \((\controlfun, \initbound)\)-controlled bad partially nested sequences over \(\N^{\beta_1} \times \setfun \times \N^{\beta_2}\) is in $\fastgrowof{\omega^{\beta_1}+\alpha}$.
\end{restatable}

\newcommand{\LexDec}[0]{\mathsf{LexDec}}
\newcommand{\LexDecof}[1]{\LexDec(#1)}
\newcommand{\maxlength}{\mathit{L}}
\newcommand{\maxlengthof}[2]{\maxlength_{#1}(#2)}
\newcommand{\vect}[1]{\mathbf{#1}}
\begin{proof}
Note that we can get rid of $\setfun$ by setting \(\beta_2'=\beta_2 \cdot \cardof{\setfun}\) and storing the auxiliary ranks of one call for every function \(\afun\) in a common larger rank from \(\N^{\beta_2'}\). 

The proof now has two steps, we first analyze the shape of a worst-case controlled bad partially nested sequence to obtain a recursive formula, and then we simplify this formula to be able to analyze it.

To obtain long bad partially nested sequences, we need to reduce ordinals in the least possible way. 
Given an ordinal \(\beta< \omega^{\omega}\) in Cantor normal form, we define the worst-case \emph{lexicographic decrement} as the function that finds the monomial $\omega^i\cdot k_i$ with the least exponent $\min(\beta)=i$, decrements the coefficient $k_i$, and fills up the missing monomials with smaller exponents. 
The coefficient that should be used for the added monomials is given as a parameter $n$ to the lexicographic decrement.  
We let \(\LexDecof{\beta, n}=\beta-\omega^{min(\beta)}+\sum_{j<\min(\beta)} \omega^{j}\cdot n\). 
For example, $\LexDecof{\omega^3+\omega^2\cdot 5, n}=\omega^3+\omega^2\cdot 4 + \omega\cdot n + n$.


Let \(\maxlengthof{\amainrank, \anauxrank,\controlfun}{\initbound}\) be the maximal length of a \((\controlfun,\initbound)\)-controlled bad partially nested sequence starting at \((\amainrank, \anauxrank) \in \N^{\beta_1} \times \N^{\beta_2'}\). 
We obtain a recursive formula for \(\maxlengthof{\amainrank, \anauxrank, \controlfun}{\initbound}\) by considering the worst-case behavior the bad sequence can show. 
A step in the sequence may execute a push on the stack, and thereby turn \((\amainrank, \anauxrank, \aheight=0)\) into \((\LexDecof{\amainrank, \controlfunof{\initbound}}, \controlfunof{\initbound}, \aheight=1)\). 
So we reduce the main rank, but in the least possible way. 
We then execute a bad nested sequence on this smaller new value.  
After a large number of steps \(t\), we return to \((\amainrank, \LexDecof{\anauxrank, \controlfun^{(t)}(\initbound)}, \aheight=0)\). 
We have only reduced the auxiliary rank, and again only in the least possible way, which now even includes the large number $\controlfun^{(t)}(\initbound)$ as the coefficient for the new monomials.
We now perform a bad nested sequence from this new stack. 

By formalizing this argument, we obtain what is often called a \emph{descent equation} \(\maxlengthof{\amainrank, \anauxrank, \controlfun}{\initbound}=1+l_1+l_2\), where 
\begin{align*}
l_1 &=\maxlengthof{\LexDec(\amainrank, \controlfun(\initbound)), \controlfun(\initbound), \controlfun}{\controlfun(\initbound)} \\
l_2 &=\maxlengthof{\amainrank, \LexDec(\anauxrank, \controlfun^{(l_1+1)}(\initbound)),\controlfun}{\controlfun^{l_1+1}(\initbound)}
\end{align*} are respectively the lengths of the first and second half of the worst-case sequence above.

Our descent equation is difficult to analyze because it refers to nested sequences. 
We therefore translate nested sequences into ordinary (non-nested) sequences, using summarization plus stuttering.
By summarization, we mean that we over-approximate the rank modification that may happen while the stack height is non-zero using a new control function and a single step. 
By stuttering, we mean that the high counter values introduced by the summary step allow us to execute a bad sequence that is at least as long as what we had. 
We illustrate the idea in the case $\beta_1=1$ studied above.
Starting at \(k\), we can bound \(\maxlength_{k, \anauxrank, \controlfun}\) by a \emph{non-nested (nn)} sequence in \(\N^{\beta_2'}\) whose control function is \(\controlfun^{(2+\maxlength_{k-1, \controlfun(\initbound), \controlfun})}\). 
Indeed, in the worst-case the first step arrives at \((k-1, \controlfun(\initbound),h=1)\), and after an additional at most \(\maxlength_{k-1, \controlfun(\initbound), \controlfun}+1\) steps we are back at \(h=0\). For any nested sequence, the subsequence defined by \(h=0\) is a standard bad sequence that is controlled by \(\controlfun^{(2+\maxlength_{k-1, \controlfun(\initbound), \controlfun})}\). 

In the general case, we define the new control function \(\controlfun_{\amainrank}= \controlfun^{(2 + l_1)}\), and emphasize that parameter $\amainrank$ is used as the initial main rank of the $\LexDec$ function that appears within $l_1$. 
Let \(\maxlengthof{nn,\beta_2',\controlfun_{\amainrank}}{\initbound}\) denote the length of the longest \((\controlfun_{\amainrank}, \initbound)\)-controlled bad non-nested sequence in \(\N^{\beta_2'}\). 
We can show \(\maxlengthof{\amainrank, \anauxrank, \controlfun}{\initbound} \leq (\maxlengthof{nn, \beta_2', \controlfun_{\amainrank}}{\initbound}+1)^2\), and call this the bridge inequality. 

To obtain a bound on $\maxlengthof{nn, \beta_2', \controlfun_{\amainrank}}{\initbound}$, we need to know how complicated \(\controlfun_{\amainrank}\) is in terms of fast-growing functions. 
Using the bridge inequality, we have: 
\begin{align*}
\controlfun_{\amainrank}(\initbound)&=\controlfun^{(2+\maxlengthof{\LexDec(\amainrank, \controlfun(\initbound)), \controlfun(\initbound), \controlfun}{\controlfun(\initbound)})}\\
&\leq \controlfun^{(2+(\maxlengthof{nn, \beta_2', \controlfun_{\LexDec(\amainrank, \controlfun(\initbound))}}{\controlfun(\initbound)}+1)^2)}
\end{align*}
This recursion is similar to $F_{\beta}(\initbound)=F_{\LexDec(\beta, \initbound-1)}^{(\initbound)}(\initbound)$, which can be used as the definition of the fast-growing functions. 
The complexity classes \(\mathfrak{F}_{\beta}\) only 
depend on the style of the recursion, but are robust against the differences. 
The consequence is that \(\controlfun_{\amainrank}\) is at level \(\amainrank \beta_2'+\alpha\) of the fast-growing hierarchy, using a calculation that is inspired by \cite[Theorem VI.1]{LerouxPS14}, and the fact \cite[Prop. 5.2]{FigueiraFSS11} that the \(\maxlength_{nn, \beta_2', \controlfun}\) function is at level \(\alpha+\beta_2'\) of the hierarchy, if \(\controlfun\) is at level \(\alpha\).
\end{proof}

\subsection{Rank of NGVAS}\label{Section:OutlineRanks}
To analyze the complexity of our decision procedure using \Cref{TheoremLengthBoundBadPartiallyNestedSequences}, we have to define a main rank (from~$\N^{\beta_1}$) and an auxiliary rank (from $\N^{\beta_2}$).
Recall that the main rank is the one in which we admit recursion, and the auxiliary rank is for non-recursive computations.  
Our decision procedure is recursive in NGVAS objects: the functions expect NGVAS as input and return sets of NGVAS as output. 
This means the main rank from $\N^{\beta_1}$ will measure the depth of the recursion that is needed to analyze an NGVAS. 
In fact, also the first $\gamma$-components of the auxiliary rank $\beta_2=\gamma+\delta$ contribute to the complexity of an NGVAS, but are not needed for the recursion. 
We now define the complexity measure from $\N^{\beta_1}\times \N^{\gamma}$, and call it the \emph{rank of an NGVAS}.

We begin by explaining why the complexity of an NGVAS is split between the main and the auxiliary rank. 
Recall that an NGVAS is a nesting of (extended) context-free grammars, and every such grammar can either be non-linear or linear.
The main rank $\N^{\beta_1}$ will (almost) only consider non-linear grammars.
The auxiliary rank $\N^{\gamma}$ will be for linear grammars. 
We use the notation $\rankof{\anngvas}=(\recrankof{\anngvas}, \itrankof{\anngvas})\in\N^{\beta_1}\times\N^{\gamma}$.
Let \(d\) be the number of counters. 

\subsubsection*{Main Rank $\recrankof{\anngvas}$}
The main rank is
\begin{align*}
    \recrankof{\anngvas}\ =\ (\cardof{\constrained}, \srankof{\anngvas}, \localgrmindexof{\anngvas})\in\N\times\N^{d+1}\times\N\ .
\end{align*}
There are three components that together yield $\beta_1=d+3$. 
We explain them starting from the most significant one. 
%
\paragraph{Component $\cardof{\constrained}\in\N$}
This is the number of counters that are subject to reachability constraints. 
Like KLMST \cite{LerouxS19}, our decomposition never removes a reachability constraint from an NGVAS. 
However, in Section \ref{Section:OutlinePumping} we often considered variants of the given NGVAS with smaller sets $\constrained$. The component $\cardof{\constrained}$ is the most important in order to ensure Lemma \ref{LemmaBasicRankProperties}.
%
\paragraph{Component $\srankof{\anngvas}\in\N^{d+1}$}
This so-called system rank is inspired by the rank for KLMST sequences~\cite{LerouxS19}. 
The main difference is that NGVAS form a branching structure. 
To bridge the gap, we identify the branch in the NGVAS that can lead to the deepest recursion. 
A \emph{branch} of $\anngvas_{0}$ is a sequence of NGVAS $\abranch=\anngvas_{0}\ldots\anngvas_{k}$ where $\anngvas_{i+1}$ is the child of $\anngvas_{i}$ for all $i<k$. 
We thus have
\begin{align*}
    \srankof{\anngvas}\ =\ \max_{\abranch\text{ branch in }\anngvas}\brankof{\abranch}\ .
\end{align*}
The depth of the recursion that is needed to analyze a branch is the sum of the depths that are needed to analyze the NGVAS that lie on this branch: 
\begin{align*}
    \brankof{\abranch}\ =\ \sum_{\anngvas\in\abranch}\lrankof{\anngvas}\ . 
\end{align*}
The idea behind the local rank of an NGVAS is to replace the syntactic notion of counters by a semantic one, and determine the number of semantic counters that the NGVAS uses. 
This means the recursion that depends on the local rank removes the semantic counters.  
To understand the idea of semantic counters, assume the NGVAS has two (syntactic) counters and every terminal with effect $k\in\Z$ on the first counter has effect $2k$ on the second. 
Then the counter values move along  $(1, 2)\cdot\N$, so only one semantic counter is needed.
The idea of semantic counters, or the dependence among syntacic counters, is naturally captured as the dimension of a suitable vector space~\cite{LerouxS19}. 
We adapt it to our needs. 

A \emph{cycle} in an NGVAS is a derivation $\anonterm\to^{*}\aword.\anonterm.\awordp$ that reproduces the non-terminal $\anonterm$, along with $\aword, \awordp\in\trms^{*}$.
The terminals $\anngvasp\in \trms$ are childNGVAS, and the effect $\ceffof{\anngvasp}\subseteq\Z^{d}$ of such a child is $\anngvasp.\restrictions$, its linear restriction.
This overapproximates the effects of runs in $\runsof{\anngvasp}$. 
For a terminal sequence $\aword=\anngvasp_{0}\ldots\anngvasp_{l}$, the effect is the sum $\ceffof{\aword}=\ceffof{\anngvasp_{0}}+\ldots+\ceffof{\anngvasp_{l}}$.
We define the \emph{left cycle space} as the vector space 
\begin{align*}
    \lcyclespaceof{\anngvas}\ =\ \mathsf{span}\setcond{\amarking\in\ceffof{\aword}}{\exists \anonterm.\ \anonterm\to^{*}\aword.\anonterm.\awordp}\subseteq\mathbb{R}^{d}.
\end{align*}
Assume $\anngvas$ is non-linear and recall that the local rank is a vector of dimension $d+1$.
This vector only has one non-zero entry, namely ($d$ minus) the dimension of the cycle space. 
This entry is the number of non-terminals:
\begin{align*}
    \lrankof{\anngvas}\ =\ 1_{d-\dimensionof{\lcyclespaceof{\anngvas}}}\cdot\cardof{\nonterms}.
\end{align*}
The subtraction moves higher dimensional spaces to the front, and in the lexicographic order $1_{i}>1_{j}$ if $i<j$. 
If $\anngvas$ is linear, $\lrankof{\anngvas}=0$.

The definition calls for comments. 
We only use the left cycle space as, for non-linear NGVAS, the left and the right cycle space coincide.
Moreover, we do not need to define the dimension for each non-terminal: it will be the same due to strong connectedness. 
Finally, we add up the local ranks when determining the branch rank, and this will fill further dimensions with values. 
\paragraph{Component $\localgrmindexof{\anngvas}\in\N$} 
%
As we have defined it so far, the main rank is for the  analysis of non-linear NGVAS. 
To also support linear NGVAS, we add the \emph{local index} as the last component of the main rank. 
We explain the recursion on linear NGVAS that needs this index. 

To decompose linear NGVAS, we fold them into VASS with nested zero tests (VASSnz)~\cite{Reinhardt08}, and then apply the recent decomposition for VASSnz~\cite{GuttenbergCL25}
We explain both steps. 
Consider a linear and non-nested NGVAS. 
Using convolution, we fold it into an ordinary VAS with twice the number of counters, 
for the run on the left and the run on the right of the single non-terminal. 
Adding the nesting that an NGVAS has, we obtain a VASSnz~\cite{AtigG11}. 
The decomposition for VASSnz~\cite{GuttenbergCL25} proceeds recursively on the nesting depth. 
The benefit of the recursion is that it can treat each nesting level as an ordinary VASS.

The \emph{local index} measures the nesting depth of the VASSnz that results from an NGVAS by folding the linear components. 
Since linear and non-linear components may alternate  in an NGVAS, the formal definition is \emph{modulo the system rank}. 
This means we treat subNGVAS thar are non-linear or have lower system rank as terminals. 
The local index $\localgrmindexof{\anngvas}\in\N$ is then the index of the resulting grammar.
If \(\anngvas\) is non-linear, then $\localgrmindexof{\anngvas}=0$. 

\subsubsection*{Auxiliary Rank $\itrankof{\anngvas}$}
The definition of the local index leads to the notion of a \emph{main branch}, which behaves like the top level VASSnz in \cite{GuttenbergCL25}.
We define the main branch $\mainbranchof{\anngvas}$ to consist of all subNGVAS of $\anngvas$ with the same system rank and the same local index as $\anngvas$.
This is unique.
%
The auxiliary rank $\itrankof{\anngvas}\in\N^{\gamma}$ is then the rank of the system represented by $\mainbranchof{\anngvas}$, as it would be ranked in \cite{GuttenbergCL25}.
%
%
%
The details can be found in the appendix.

\subsection{Analyzing PVASS using Theorem \ref{TheoremLengthBoundBadPartiallyNestedSequences}}

We now apply Theorem~\ref{TheoremLengthBoundBadPartiallyNestedSequences} to determine a bound on the complexity of our decision procedure.  
The theorem expects stack frames from $\domsf=\N^{\beta_1}\times \setfun\times\N^{\beta_2}$.  
We already discussed the main rank \(\recrankof{\anngvas}\in\N^{\beta_1}\) that controls the recursion.  
The recursive functions $\setfun$ in our decision procedure are depicted in the following graph: 
\begin{center}
\includegraphics[scale=1,page=1]{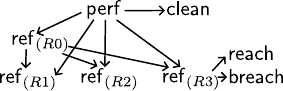}
\end{center}
An edge $\afun\rightarrow \afunp$ indicates that function $\afun(\anngvas)$ will invoke function $\afunp(\anngvasp)$, and the NGVAS argument $\anngvasp$ it passes may still have the same main rank, $\recrankof{\anngvasp}=\recrankof{\anngvas}$. 
This means the functions have to be ordered $\afun\succ\afunp$, and so the total order on the functions that Theorem~\ref{TheoremLengthBoundBadPartiallyNestedSequences} expects can be an arbitrary linearization of this graph. 
As for the auxiliary rank $\N^{\beta_2}$ with $\beta_2=\gamma+\delta$, we already defined the linear rank $\itrankof{\anngvas}\in\N^{\gamma}$. 
We now focus on the implementation of the functions, and the auxiliary rank from $\N^{\delta}$ they need. 

Function $\refinesub$ establishes perfectness condition $\perfectnesschildrennospace$ stating that all children are perfect.
Thanks to omitted cleanness conditions, we only need to invoke this for linear NGVAS.
For linear NGVAS, we do not have the guarantee that children have a smaller rank. 
Here, we rely on the recent VASSnz decomposition to avoid calls to perf and remain within our programming model.

Function $\refinepump$ computes a Karp-Miller tree to check whether a pumping situation exists and, if the check fails, a coverability grammar for the decomposition, see Section \ref{Section:OutlinePumping}. 
Both computations are non-recursive, which means the auxiliary rank~$\N^{\delta}$ has to be large enough to accomodate them. 
Before we get into the details, we note that both computations invoke the functions $\postfunc$ and $\prefunc$. 
These functions are smaller in the total order of functions, and so we can invoke them without having to reduce the main rank. 


\newcommand{\aconfig}{\mathit{c}}
\newcommand{\adcset}{\mathit{D}}
\newcommand{\typeof}[1]{\mathsf{type}(#1)}

We explain how we compute the Karp-Miller tree, 
for the coverability grammar we proceed similarly. 
We store configurations $\aconfig$ in the tree as elements from \(\N^{2d+2}\), with \(d\) counters for the input marking, \(d\) counters for the output marking, and \(2\) counters for the state as \((i, \cardof{Q}-i)\). 
For every branch \(\abranch=\aconfig_0, \dots, \aconfig_r\) in the tree, we define a rank, namely \(\cardof{\N^{2d+2} \setminus \upclsof{\{\aconfig_0, \dots, \aconfig_r\}}}\).
This is the size of the downward-closed set of configurations which may still be added to the branch.  
Using the notion of \emph{types}~\cite{FigueiraFSS11}, we can represent the size of a downward-closed set $\adcset \subseteq \N^{2d+2}$ as an element of \((\N^{2d+2}, \leqlex)\).
For this, we write the downward-closed set $\adcset$ as a union of \(\omega\)-markings/copies of \(\N^{n_i}\), i.e.\ \(\adcset=\bigcup_{i=1}^r \N^{n_i}\).
Now \(\typeof{\adcset} \in \N^{2d+2}\) counts the number of \(\omega\)-markings with a given number $k$ of $\omega$-entries: \(\typeof{\adcset}[k]=\cardof{\setcompact{i}{2d+2-n_i=k}}\). 
For example, \(\downclsof{(1, \omega)}=(0, \omega) \cup (1,\omega)\) has $\typeof{\downclsof{(1, \omega)}}[0]=0$ and $\typeof{\downclsof{(1, \omega)}}[1]=2$. 
To simplify the notation, we just write $\typeof{\abranch}$ for $\typeof{\N^{2d+2}\setminus\upclsof{\abranch}}$. 

The auxiliary rank in $\N^{\delta}$ that we need for the tree is $(\max_{\abranch}\typeof{\abranch}, \#\text{branches in current round}) \in \N^{2d+3}$. 
We build the tree in breadth-first fashion, 
and the last component is a loop counter, like above. 


It remains to compute $\postfunc$ and $\prefunc$. 
In the simple cases, and similar to $\refinesub$, we can use a recursion that calls $\perfect$.
This is justified because $\cardof{\constrained}$ decreased. 
For Hard Case~1, we have to compute \(f(r-1)\) with $r$ being the number of relevant counters.
We have an outer loop for $r$, and an inner loop that maintains the currently known minimal markings fulfilling the pumping property.
Moreover, we construct several Karp-Miller trees, but reuse the space once we have the result.
Finally, we obtain \(Bd\) using integer linear programming. 
For Hard Case 2, we compute marked parse trees similar to Karp-Miller trees, and refer to Section \ref{SectionComplexity} for details.
\vspace{-0.15cm}
\section{Conclusion}
We presented a decision procedure for 
the reachability problem in pushdown vector addition systems. 
It relies on several insights that may be of independent interest. 
The first insight is a wide tree theorem for context-free grammars that guarantees the existence of derivation trees of height logarithmic in the length of the word, at the cost of only preserving Parikh equivalence. 
The second insight is the notion of vertical pumping which, intuitively, harmonizes the manipulation of the stack and the counters.
The third insight is that Rackoff's induction generalizes beyond coverability.
The fourth insight is that the problem is best reduced to itself, with an input that is smaller in a well-founded order. 
The last insight is that a combination of recursive and tail-recursive functions over ordinals admits a better complexity analysis than what was known in the literature on fast-growing complexity.

\bibliography{main}
\newpage
\appendix

\section{Introduction to the Appendix}\label{Section:IntroToAppendix}
The appendix is divided into two parts.
The first part serves as an extended version of the main paper, and features a detailed and concrete treatment of the sections from the main paper.
We go over the individual sections.
\Cref{Section:Prelims} provides the essential definitions.
In \Cref{Section:GVAS} we state our main result with the developped terminology.
In \Cref{Section:NGVAS}, we define NGVAS in detail. 
We extensively discuss the compatibility requirements in \Cref{Section:NGVASDefinition}.
This extends the definition in \Cref{Section:OutlineNGVAS}.
In \Cref{Section:CharEqNL,Section:CharEqNL}, we define the characteristic equation systems.
We discuss the ranks in \Cref{Section:Ranks}, which extends \Cref{Section:OutlineRanks}.
We state the complete list of perfectness conditions in \Cref{Section:Perfectness}, which extends \Cref{Section:OutlinePerfectness}.
In \Cref{Section:IterationLemma}, we prove the iteration lemma in detail, extending the argument in \Cref{Section:OutlineIterationLemma}.
In \Cref{Section:Decomposition}, we discuss the details of decomposition: cleaning and simple refinements.
In \Cref{Section:PumpingMP}, we treat the pumpability problem in detail.
This extends \Cref{Section:OutlinePumping}.
Finally, in \Cref{SectionComplexity}, we consider the complexity of our algorithm, and fill in the missing details from \Cref{SectionLevelOneComplexity}.

The second part of the appendix, starting from \Cref{Section:AppendixL3}, contains the proofs and details omitted from the extended version.

\section{Preliminaries} \label{Section:Prelims}
\subsection{Basic Notation}
We fix some notation.
Let $\anindexset$ be a finite set of indices.  
Given a vector $\avec\in\Z^\anindexset$, we use $\at{\avec}{\anindex}$ for the entry at dimension $\anindex\in\anindexset$. 
We use $\normof{\avec}=\sum_{\anindex\in\anindexset} \cardof{\at{\avec}{\anindex}}$ for the size. 
We write $0, 1\in\Z^\anindexset$ for the vector with $0$ resp. $1$ in all dimensions. 
We write~$1_{\anindex}$ for the $\anindex$-th unit vector that has $1$ at dimension~$\anindex\in\anindexset$ and~$0$~otherwise. 
We write $\seqat{\asentform}{\anindexp}$ for the $\anindexp$-th entry in a sequence~$\asentform$.
We use $\leq$ to refer to the componentwise order $\mathord{\leq}\subseteq\N^{\anindexset}\times\N^{\anindexset}$, where $\amarking\leq\amarkingp$ holds if $\amarking[i]\leq\amarkingp[i]$ for all $i\in\anindexset$.  
The componentwise order on $\N^{\anindexset}$ forms a \emph{well quasi-order}.
In our development, we occasionally rely on the properties of well quasi-orders.
The most important property is the following.
For any infinite sequence $[\amarking_{i}]_{i\in\N}\in(\N^{\anindexset})^{\omega}$, there is an increasing subsequence $[\amarking_{\phi(i)}]_{i\in\N}$, that is, $\amarking_{\phi(i)}\leq\amarking_{\phi(i+1)}$ for all $i\in\N$.
We refer the reader to \cite{FinkelG09} for more information on the topic.

%
%
%

We use a Parikh image that is parameterized in the set it wishes to count. 
Fix a set $A$ and let $B\subseteq A$. 
We define $\paramparikh{B}:A^*\rightarrow \N^{B}$ 
as the function that maps a word $\asentform\in A^*$ to the vector $\paramparikhof{B}{\asentform}\in \N^{B}$ which says how often each letter from $B$ occurs in $\asentform$. 
We wish to capture the effect of productions in a context-free grammar on the number of symbols.  
For an abstract account, let $C\subseteq A$ and consider a relation $\prods\subseteq C\times A^*$. 
The effect of a pair $\aprod = (c, \asentform)$ on the number of elements from $B$ is 
$\effof{B, \aprod} = \paramparikhof{B}{\asentform} - \aconst$, where $\aconst = 1_{c}$ if $c\in B$ and $\aconst = 0$ otherwise. 
The matrix $\effof{B}\in\Z^{B\times \prods}$ has $\effof{B, \aprod}$ as column $\aprod$. 

An (ordered) \emph{tree structure} $P\subseteq\N^{*}$ is a prefix-closed set of strings whose entries are natural numbers.
This means we identify a node $\anode$ in $\N^{*}$ with the path that leads to it.
We call $\varepsilon\in P$ the root and $\anode.a$ a child of node $\anode$.
%
%
The leftmost child of $\anode$ is the node $\anode.a\in P$, where $a\in\N$ is the smallest number for which there is such an element.
%
%
The helper function $\childnodes:P\to P^{*}$ lists the children of a given node, $\childnodesof{\anode}=(\anode.a_0)\ldots(\anode.a_i)$ with $a_0<\ldots<a_i$.
A \emph{$K$-labeled tree} is a pair $\atree=(P, \nodemarking)$ consisting of a tree structure $P\subseteq\N^{*}$ and a labeling $\nodemarking:P\to K$.
%
%
%
%
The yield of a $K$-labeled tree $\atree$ is the sequence of $K$ elements labeling the leaves as encountered in a left-first traversal.
Formally, we set $\yieldof{\atree}=\yieldof{\varepsilon}$.
For a node $\anode$ with $\childnodesof{\anode}=(\anode.a_0)\ldots(\anode.a_i)$, 
we have $\yieldof{\anode}=\yieldof{\anode.a_0}\ldots \yieldof{\anode.a_i}$. 
For a leaf $\anode$, we have $\yieldof{\anode}=\nodemarking(\anode)$.
%
We say that $\atreep$ is a subtree of $\atree$ rooted at the node $\anodep\in\atree$, if $\anodep.\atreep\subseteq\atree$ and $\nodemarking(\anodep.\anode)=\nodemarking(\anode)$ for all $\anode\in\atreep$.  
We say that $\atreep$ is a subtree of $\atree$, if it is a subtree of $\atree$ rooted at some node.
We refer to the subtree rooted at the leftmost child of the root as the left-subtree.
\subsection{Context-Free Grammars}\label{Section:CFG}
A context-free grammar $\agram = (\nonterms, \trms, \prods, \startnonterm)$ consists of a finite set of non-terminal symbols~$\nonterms$, a finite set of terminal symbols $\trms$ with $\nonterms\cap\trms=\emptyset$, a start non-terminal $\startnonterm\in\nonterms$, and a finite set of  productions $\prods\subseteq \nonterms\times(\nonterms\discup\trms)^*$. 
We call sequences of non-terminals and terminals $\alpha, \asentformp\in (\nonterms\discup\trms)^*$ sentential forms. 
We use $\asentform\xrightarrow{\aprod} \asentformp$ for the derivation relation between sentential forms, which says that $\asentformp$ can be obtained from $\asentform$ by an application of the production $\aprod\in \prods$. 
We extend the relation to sequences of productions. 
We write $\asentform\xrightarrow{\aprodseq}$ if there is $\asentformp$ so that $\asentform\xrightarrow{\aprodseq}\asentformp$ holds.
We write $\asentform\rightarrow^* \asentformp$ if there is $\aprodseq$ so that $\asentform\xrightarrow{\aprodseq}\asentformp$ holds.
The language $\langof{\agram}=\setcond{\asentform\in\trms^*}{\startnonterm\rightarrow^*\asentform}$ consists of the terminal words that can be derived from the start non-terminal. 

%
Strong connectedness and branching play central roles in our development.
We say that $\nonterms_{sc}\subseteq\nonterms$ is \emph{strongly-connected}, if for all $\anonterm, \anontermp\in\nonterms_{sc}$, there is a derivation $\anonterm\to^{*}\aword.\anontermp.\awordp$ for some $\asentform, \asentformp\in (\nonterms\discup\trms)^*$. 
Note that each singleton $\set{\anonterm}$ is strongly connected under this definition.
We call strongly connected $\nonterms_{scc}\subseteq\nonterms$ a \emph{strongly-connected component (SCC)}, if there is no $\nonterms_{scc}\subsetneq\nonterms_{sc}\subseteq\nonterms$ that is strongly connected.
The grammar is strongly-connected, if $\nonterms$ is strongly-connected. 
We assign each non-terminal $\anonterm\in\nonterms$ a call set and a strongly connected component, $\callof{\anonterm}, \sccof{\anonterm}\subseteq\nonterms$.
The call set $\callof{\anonterm}$ of $\anonterm$, consists of symbols $\asymbol\in\nonterms\discup\trms$ that can be reached from $\anonterm$, that is $\anonterm\to^{*}\asentform.\asymbol.\asentformp$ for some $\asentform, \asentformp\in(\nonterms\discup\trms)^{*}$.
Note that $\anonterm\in\callof{\anonterm}$ holds because of the empty derivation.
For $\anonterm\in\nonterms$, we define $\sccof{\anonterm}$ to be the SCC that includes $\set{\anonterm}$.
Note that $\sccof{\anonterm}\subseteq\callof{\anonterm}$.
We classify productions based on whether they allow for further derivation.
A production $\aprod$ is an \emph{exit-production}, if the rule does not produce any non-terminals, $\aprod\in\nonterms\times\trms^{*}$.
If $\aprod$ is not an exit-production, it is a \emph{persisting-production}.
We call the grammar \emph{non-branching}, if every persisting production has exactly one non-terminal symbol on the right-hand side, $\prods\subseteq \nonterms\times \trms^*.\nonterms.\trms^*$.
If this is not the case, we call the grammar \emph{branching}.
We reserve the terms \emph{linear} and \emph{non-linear} for grammars that are strongly connected and non-branching resp. branching.
%

The grammar is in weak Chomsky normal form (wCNF), if every production has at most two symbols on the right-hand side, $\prods\subseteq \nonterms\times (\trms \discup\nonterms)^{\leq 2}$, and every terminal occurs on the right-hand side of a production.   
The advantage over the classical Chomsky normal form is that the notion of linearity applies without change.  
A non-terminal $\anonterm$ is useful, if it can be used to derive a terminal word: there are $\asentform_1, \asentform_2\in(\nonterms\discup\trms)^*$ and $\asentform\in\trms^*$ so that $\startnonterm\rightarrow^*\asentform_1.\anonterm.\asentform_2\rightarrow^*\asentform$. 

We need a linear-algebraic description of the derivations in a context-free grammar that can serve as an interface to VASS arguments. 
In particular, we need a way to capture the effect of derivations on the number of terminals and non-terminals. 
The notation from above helps, and we can rely on a powerful theorem due to Esparza.
\begin{theorem}[Theorem 3.1 in~\cite{Esparza97}]\label{Theorem:EEK}
Let $\agram$ only have useful non-terminals and let $1\leq \prodvec\in\N^{\prods}$. 
If $\effof{\nonterms}\cdot \prodvec\geq -1_{\startnonterm}$, then 
there is $\aprodseq\in\prods^*$ with $\startnonterm\xrightarrow{\aprodseq}$ and $\paramparikhof{\prods}{\aprodseq}=\prodvec$ .  
\end{theorem}

Theorem~\ref{Theorem:EEK} allows us to turn a solution to a system of linear inequalities into a feasible production sequence.   
The following is the converse. 
\begin{lemma}\label{Lemma:EEKConverse}
Consider $\startnonterm\xrightarrow{\aprodseq}\asentform$ with $\paramparikhof{\prods}{\aprodseq}=\prodvec$.  
Then $\paramparikhof{\nonterms}{\asentform} = 1_{\startnonterm}+\effof{\nonterms}\cdot \prodvec$ and $\paramparikhof{\trms}{\asentform} = \effof{\trms}\cdot \prodvec$. 
\end{lemma}
For easy reference, we name some equations. 
By \emph{Esparza-Euler-Kirchhoff} and its homogeneous variant, we mean 
\begin{alignat*}{5}
\eekof{\agram}{\prodvar}:&\qquad &\effof{\nonterms}\cdot \prodvar\ &=\ -1_{\startnonterm}\\
\homeekof{\agram}{\prodvar}:&\qquad &\effof{\nonterms}\cdot \prodvar\ &=\ 0\ . 
\end{alignat*}
If $\prodvec\geq 1$ solves $\eek{}$, Theorem~\ref{Theorem:EEK} yields a feasible production sequence. 
By Lemma~\ref{Lemma:EEKConverse}, the resulting sentential form only consists of terminals. 
If $\prodvec$ solves $\homeek{}$, the resulting sentential form has a copy of the start non-terminal. 
We also have equations that convert a number of productions into the number of terminals they produce. 
With $\termvar\in\N^{\trms}$, we define  
\begin{align*}
\ptof{}{\prodvar, \termvar}:\qquad \termvar - \effof{\trms}\cdot \prodvar&= 0\ . 
\end{align*} 

\subsection{Wide Tree Theorem}
A parse tree organizes the application of productions into a tree.
The root is the start non-terminal, the inner nodes are non-terminals, and the yield of the parse tree is the sentential form that has been derived with the productions in the tree.  
We use $\treesof{\agram}$ for the set of all parse trees in $\agram$. 
The link to the derivation relation is 
$\langof{\agram}=\yieldof{\treesof{\agram}}$. 
We use $\paramparikhof{\prods}{\atree}\in\N^{\prods}$ for the productions used in parse tree $\atree$.  
The height of the tree is~$\heightof{\atree}$ and a single node has height $0$.

We show a consequence of Theorem~\ref{Theorem:EEK} and Lemma~\ref{Lemma:EEKConverse} that is interesting in its own right.
If we have a strongly connected non-linear grammar and a solution to the homogeneous variant of Esparza-Euler-Kirchhoff, then we can iterate this solution and obtain parse trees that only grow logarithmically in height.
Actually, we will be more precise, but this needs terminology.

Consider a parse tree $\atree$ with $\aconst\geq 1$ copies of a production vector $\prodvec\in\N^{\prods}$, $\paramparikhof{\prods}{\atree}=\aconst\cdot \prodvec$. 
We introduce functions that track the provenance of $\trms$-labeled leaves, the copy of $\prodvec$ that generated the terminal. 
Formally, $\prov:\leaves_{\trms}\rightarrow [1, \aconst]$ \emph{tracks provenance}, if $\provinvof{i}=\effof{\trms}\cdot \prodvec$ for all $i\in[1, \aconst]$. 
Consider a prefix $\asentform$ of $\yieldof{\atree}$. 
We say that copy $i$ of $\prodvec$ is \emph{complete} in $\asentform$, if $\provinvof{i}\subseteq \asentform$. 
Otherwise, if also the remainder of the yield contains terminals that stem from this copy, we say that copy $i$ is \emph{incomplete} in $\asentform$. 
The \emph{order} of the provenance tracking function is the maximal number of incomplete copies of $\prodvec$ in any prefix of the yield.

\begin{restatable}{theorem}{WideTreeTheorem} \label{TheoremWideTree}
Consider a context-free grammar $\agram$ that is non-linear, strongly connected, and only has useful non-terminals. 
Let $\prodvec\geq 1$ solve $\homeekof{\agram}{\prodvar}$. 
For every $\aconst\geq 1$ there is $\atreedef\in\treesof{\agram}$ with $\paramparikhof{\prods}{\atreedef}= \aconst\cdot \prodvec$ and $\heightof{\atreedef}\leq \ceilof{1 + \ld\aconst}\cdot \normof{\avec_{\prods}}$. 
Moreover, $\atreedef$ admits a provenance tracking function of order at most $\ceilof{1+\ld\aconst}$. 
\end{restatable}
The proof is illustrated below.
We build a binary tree and, by strong connectedness, can strengthen the statement so that it holds for all start non-terminals. \\
\begin{center}
\includegraphics[scale=0.7,page=1]{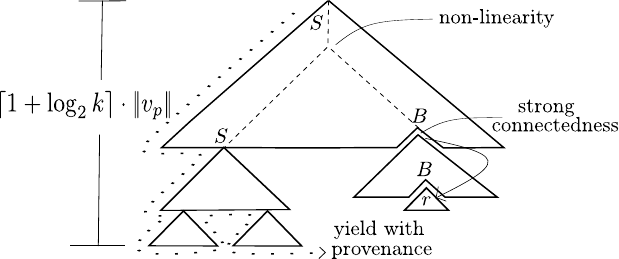}
\end{center}

\subsection{VAS}
A \emph{vector addition system} (VAS) of dimension $\adim$ is a finite set of so-called updates $\updates\subseteq \Z^{\adim}$. 
A marking of the VAS is a vector $\amark\in\N^\adim$. 
The firing relation $\fires{-}\subseteq \N^{\adim}\times\updates\times\N^{\adim}$ contains all triples $\amark \fires{\anupd}\amark'$ with $\amark'=\amark+\anupd$.
Note that $\amark'$ has to remain non-negative. 
We extend the firing relation to runs $\arun\in \updates^*$, so $\amark\fires{\arun}\amark'$ means firing the updates from $\arun$ one after the other transforms $\amark$ into~$\amark'$, and all markings encountered on the way are non-negative. 
A run $\arun$ is enabled in marking~$\amark$, denoted by~$\amark\fires{\arun}$, if there is a marking $\amark'$ so that $\amark\fires{\arun}\amark'$ holds.  
The hurdle of the run, $\hurdleof{\arun}$, is the unique least marking that enables the run. 
The effect of the run, $\vaseffof{\arun}$, is the sum of all updates in the run.
The link between the two is this. 
If on every prefix $\arun'$ of $\arun$ and on every counter $\acounter\in [1, \adim]$ we have $\at{\vaseffof{\arun'}}{\acounter}\geq 0$, then $\hurdleof{\arun}=0$.  
The reversal operation $\revof{\arun}$ reverses the run and changes the sign of all updates, $\revof{(\arun_1.\arun_2)}= \revof{\arun_2}.\revof{\arun_1}$ and $\revof{\anupd}= -\anupd$. 

The \emph{marking equations} approximate VAS reachability.
Let the matrix $\effof{\updates}\in\Z^{d\times \updates}$ have~$\anupd$ as the column for dimension~$\anupd$. 
With $\avar, \avar'\in\N^d$ and $\avarp\in\N^{\updates}$, we define
\begin{align*}
\meof{\updates}{\avar, \avarp, \avar'}:\qquad
\avar+\effof{\updates}\cdot \avarp - \avar' = 0\ . 
\end{align*}
\begin{lemma}
If $\amark\fires{\arun}\amark'$, then $\amark$, $\paramparikhof{\updates}{\arun}$, $\amark'$ solve $\meof{\updates}{\avar, \avarp, \avar'}$.
\end{lemma}

We augment the natural numbers with a top element $\omega$, and write $\Nomega$ for $\N\discup\set{\omega}$. 
We extend addition to $\Nomega$ by $a+\omega=\omega+a=\omega$ for all $a\in\Nomega$.
A generalized marking is an element~$\amark\in\Nomega^{\adim}$. 
We use $\omegaof{\amark}\subseteq[1, \adim]$ for the dimensions where $\amark$ carries $\omega$. 
A well-known concept from VAS reachability~\cite{Lambert92} that we will also need is the specialization quasi order on $\Nomega$. 
It is defined by $\omega\sqsubseteq \omega$, $k\sqsubseteq k$, and $k\sqsubseteq\omega$ for all $k\in \N$, and lifted to generalized markings in a componentwise fashion.
For a marking $\amarking\in\Nomega^{d}$, $J\subseteq [1,d]$, and $a\in\N$ we write $\setter{\amarking}{J}{a}$ for the marking that agrees with $\amarking$ on all components outside $J$, and has value $a$ for the components in $J$.
The zero-version of a generalized marking $\zeroof{\amark}$ is defined by $\at{\zeroof{\amark}}{\anindex}=0$ if $\at{\amark}{\anindex}\in\N$, and $\at{\zeroof{\amark}}{\anindex}=\omega$ if $\at{\amark}{\anindex}=\omega$.  
The $\omega$-version $\settoomega{X}{\amark}$ depends on $X\subseteq [1, \adim]$ and is defined as $\settoomega{X}{\amark}=\setter{\amarking}{X}{\omega}$.
%
It is also useful to have a notion of compatibility between generalized markings.
For two vectors $\amarking, \amarkingp\in \Nomega^{d}$ and $i\in [1,d]$, we write $\amarking\compwithat{i}\amarkingp$ if $\amarking[i]=\omega$, $\amarkingp[i]=\omega$, or $\amarking[i]=\amarkingp[i]$.
This means that $\amarking\not\compwithat{i}\amarkingp$ holds if and only if $\amarking[i], \amarkingp[i]\in \N$ and $\amarking[i]\neq\amarkingp[i]$.
We write $\amarking\compwith\amarkingp$ if $\amarking\compwithat{i}\amarkingp$ for all $i\in[1,d]$.
For $\amarking, \amarkingp\in\Nomega^{d}$, we define $\amarking\specmeet\amarkingp\in\Nomega^{d}$ to be the largest vector that is a specialization of $\amarking$ and $\amarkingp$, $(\amarking\specmeet\amarkingp)\sqsubseteq\amarking, \amarkingp$. 
That is, for all $i\in\omegaof{\amarking}\cap\omegaof{\amarkingp}$, $(\amarking\specmeet\amarkingp)[i]=\omega$ must hold.
This is well-defined if and only if $\amarking\compwith\amarkingp$.
We extend the firing relation, enabledness, the notion of a hurdle, and also the marking equations to generalized markings. 
\section{Main Result} \label{Section:GVAS}
We study PVASS reachability through a grammar formalism. 
A $\adim$-dimensional \emph{grammar vector addition system} (GVAS) is a CFG $\agram = (\nonterms, \updates, \prods, \startnonterm)$ whose terminals are VAS updates, $\updates \subseteq \Z^{\adim}$.
The problem of interest is defined as follows.
\begin{quote}
\vspace{0.2cm}
\underline{\gvassreach}\\[0.1cm]
\given\ \ GVAS $\agram$ of dimension $\adim$ and $\amark_1, \amark_2\in\N^{\adim}$.\\
\question\ \ Is there $\arun\in\langof{\agram}$ so that $\amark_1\fires{\arun}\amark_2$? \\[-0.3cm]
\end{quote}
The following is our main result. 
\begin{theorem}\label{TheoremGVASReachDecidable}
\gvassreach\ is decidable (in Hyper-Ackermann time).
\end{theorem}

\section{Nested GVAS} \label{Section:NGVAS}
This is the data structure we rely on for our decision procedure. 
From now on, we fix the dimension to be $\adim$.
\subsection{Definition}\label{Section:NGVASDefinition}
A \emph{weak nested GVAS} (wNGVAS) is a triple $\anngvas=(\agram, \contextinformation, \boundednessinformation)$ with $\agram=(\nonterms, \trms, \prods, \startnonterm)$ a CFG in wCNF.
The choice of the set of terminal symbols depends on the \emph{nesting depth}, a natural number that is associated with the wNGVAS. 
A wNGVAS of nesting depth zero has as terminals counter updates in~$\Z^{\adim}$, and its only production is an exit-production.
A wNGVAS of nesting depth $i+1$ has as terminals wNGVAS of nesting depth at most~$i$. 
We call a wNGVAS that occurs as a terminal symbol of $\anngvas$ a child.
We call a wNGVAS that occurs as a terminal symbol in one of the wNGVAS defined within $\anngvas$ a subwNGVAS, so every child is a subwNGVAS but not the other way around. 
We expect that all subwNGVAS have mutually disjoint sets of terminals and non-terminals. 
We lift the terminology for the grammar $\agram$ to the wNGVAS $\anngvas$. 
For example, we say that $\anngvas$ is linear when we mean $\agram$ is linear.

The second component in a wNGVAS is called the \emph{context information} and takes the form $\contextinformation=(\unconstrained, \restrictions, \acontext)$.
The \emph{set of unconstrained counters} $\unconstrained\subseteq [1,d]$ consists of counters that are shielded from reachability constraints.
This means that a counter $i\in\unconstrained$ may not be required to reach some value $a\in\N$, neither at the current level nor in a subwNGVAS.
We formalize the propagation to subwNGVAS by requiring $\unconstrained\subseteq\anngvasp.\unconstrained$ for all $\anngvasp\in\trms$.
The component $\unconstrained$ shields its counters from reachability constraints by imposing conditions on other components of wNGVAS.
We will discuss these restrictions while explaining those components.

The \emph{restriction} $\restrictions$ is a linear set $\avec+\asetvec^*\subseteq \Z^d$ that will limit the effects of possible derivations and serve as an interface between the wNGVAS and its children. 
%
%
This means that the counters $i\in\unconstrained$ may be freely influenced by $\restrictions$.
We will also understand the restriction as the equation  $\updates \cdot \updvar-\asetvec\cdot \periodvar = \avec$ in $\updvar\in\N^{\updates}$ and $\periodvar\in\N^{\asetvec}$.
Here, $\updates$ is the set of all counter updates that appear in $\anngvas$ or in a subwNGVAS. 
We denote the restriction equation by~$\restrictions(\updvar, \periodvar)$.
We will also need a homogeneous variant, denoted by $\homrestrictions(\updvar, \periodvar)$ and defined as $\updates \cdot \updvar-\asetvec\cdot \periodvar = 0$, meaning we replace the base vector on the right-hand side by $0$. 
The component $\acontext=(\acontextin, \acontextout)\in\Nomega^{2\adim}$ consists of so-called \emph{input} and \emph{output markings}.
Here, $\unconstrained$ applies and requires $\unconstrained\subseteq\omegaof{\acontextin}\cap\omegaof{\acontextout}$.

The last component is called the \emph{boundedness information}. 
Its purpose is to track the values of counters which remain bounded in (an over-approximation of) the reaching runs. 
For readers familiar with KLMST for VASS, this corresponds to storing which counters are rigid together with corresponding values. 
It is however made more difficult especially in linear NGVAS because the letters derived 
on the left and right of the non-terminal may have different sets of rigid counters.
Before we move on to the definition, we give the intuition.
The goal is to track the counters concretely along any derivation $\anonterm\to^{*}\asentformp_{1}.\anontermp.\asentformp_{2}$ with $\asentformp_{1}, \asentformp_{2}\in\trms^{*}$.
The information is tracked forwards in $\asentformp_{1}$ resp.~backwards in $\asentformp_{2}$ starting from (a generalization of) $\acontextin$ resp $\acontextout$. Remember that generalization means that some counters which used to have concrete values now have value \(\omega\).

Formally, the boundedness information $\abdinfo=(\abdinfoleft, \abdinforight, \infun, \outfun)$ consists of sets $\abdinfoleft, \abdinforight\subseteq[1,d]$ so that $\unconstrained\subseteq\abdinfoleft\cap\abdinforight$, and functions
\begin{align*}
    \infun&:\nonterms\rightarrow \omega^{\abdinfoleft}\times \N^{[1, \adim]\setminus \abdinfoleft}\\
    \outfun&:\nonterms\rightarrow \omega^{\abdinforight}\times \N^{[1, \adim]\setminus \abdinforight}
\end{align*}
that assign to each non-terminal an input and an output marking.
In these markings, the counters from $\abdinfoleft$ have value $\omega$ on the input, and the counters from $\abdinforight$ have value $\omega$ on the output.
The boundedness information does not apply for nesting depth zero wNGVAS beyond this point.   
If the wNGVAS is branching, we expect $\abdinfoleft=\abdinforight$.
In this case, we just use $\abdinfomid$ for clarity. 
We extend $\infun$ and $\outfun$ to terminals by following their context information, $\inof{\anngvasp}=\anngvasp.\acontextin$, $\outof{\anngvasp}=\anngvasp.\acontextout$ for all $\anngvasp\in\trms^{*}$.
If we understand $\anngvasp.\acontextin$ and $\anngvasp.\acontextout$ as reachability information, $\inof{\anngvasp}$ and $\outof{\anngvasp}$ properly capture it. 
Note that $\abdinfoleft$ and $\abdinforight$ do not need to be respected in the extension.
However, the extended assignments should be consistent.
Let $\aprod=\anonterm\rightarrow\asentform_1.\asentform_2$ be a production.
Consistency requires $\inof{\anonterm}=\inof{\asentform_1}$, $\outof{\asentform_1}=\inof{\asentform_2}$, and $\outof{\asentform_2}=\outof{\anonterm}$ if $\aprod$ is a persisting production \emph{or} if the grammar is branching, and $\inof{\anonterm}\sqsupseteq\inof{\asentform_1}$, $\outof{\asentform_1}=\inof{\asentform_2}$, and $\outof{\asentform_2}\sqsubseteq\outof{\anonterm}$ in the remaining cases.
For technical reasons, we allow the exit production in a non-branching grammar to demand reachability for counters within $\abdinfoleft$ resp.~$\abdinforight$. 
In the case that $\aprod$ is a persisting rule or the grammar is branching, we also ensure that the untracked counters are shielded by requiring $\omegaof{\asentform_i.\acontextin}=\omegaof{\asentform_i.\acontextout}=\asentform_i.\unconstrained$ whenever $\asentform_i$ is a terminal for $i\in\set{1,2}$.
Note that consistency also implicitly ensures $\asentform_i.\unconstrained=\abdinfoleft$ if $\asentform_i$ is generated on the left by a non-terminal, and $\asentform_i.\unconstrained=\abdinforight$ if it is generated on the right.
Lastly, for the start non-terminal, consistency requires $\acontextin\sqsubseteq \inof{\startnonterm}$ and $\acontextout\sqsubseteq \outof{\startnonterm}$, the 
input and output markings we track take the values from the context information, or they are more abstract and use~$\omega$. 
A consequence is $\omegaof{\acontextin}, \omegaof{\acontextout}\subseteq \adimset$.

In branching grammars, exit productions are treated in the same way as persisting productions for the following reason.
Recall our goal of tracking counters along $\asentformp_{1},\asentformp_{2}\in\trms^{*}$ in $\anonterm\to^{*}\asentformp_{1}.\anontermp.\asentformp_{2}$.
If the grammar is branching, $\anonterm\to^{*}\asentformp_{1}.\anontermp.\asentformp_{2}$ can have the intermediary step $\anonterm\to^{*}\asentformp_{3}.\anontermpp.\asentformp_{4}.\anontermp.\asentformp_{2}\to\asentformp_{1}.\anontermp.\asentformp_{2}$ with $\asentformp_{3}, \asentformp_{4}\in\trms^{*}$.
In the last step, we use an exit production.
However, the result is also a part of the string $\asentformp_{1}$ where we want an unbroken chain of concretely tracked information.
Hence, we cannot allow an exit production to break the chain.
Such a situation is not possible in non-branching grammars.
There, exit productions are always the last possible derivation.

The wNGVAS inherits the notion of a language from the underlying grammar, $\langof{\anngvas}=\langof{\agram}$. 
We also associate with the wNGVAS a set of runs. 
The effect of these runs should satisfy the given restriction. 
Moreover, the run should be enabled in $\anngvas.\acontextin$ and lead to $\anngvas.\acontextout$. 
These requirements are made not only for $\anngvas$, but for all subNGVAS. 
The definition is by Noetherian induction:
\begin{align*}
    \runsof{\anngvas}=&\setcond{(\amarking, \arun, \amarkingp)\in\runsof{\asentform}}{\\ &\quad\asentform\in\langof{\agram}\ \wedge\ \updates\cdot\paramparikhof{\updates}{\arun}\in\restrictions\ \wedge \\
    &\quad\amarking\sqsubseteq\acontextin\ \wedge\ \amarkingp\sqsubseteq\acontextout}\ .
\end{align*}
Here $\runsof{\asentform}$ is defined by merging the runs of $\asentform[1] \ldots \asentform[\cardof{\asentform}]$ that agree on the intermediary markings.
That is, $\runsof{\varepsilon}=\setcond{(\amarking, \varepsilon, \amarking)}{\amarking\in\N^{d}}$, and 
\begin{align*}
    \runsof{\asentform_{0}.\asentform_{1}}=\setcond{(\amarking_{0}, \arun_{0}.\arun_{1}, \amarkingp_{1})}{(\amarking_{0}, \arun_{0}, \amarkingp_{0})\in\runsof{\asentform_{0}}\\ \qquad\text{ and }(\amarkingp_{0}, \arun_{1}, \amarkingp_{1})\in\runsof{\asentform_{1}}&}\ .
\end{align*}
%
%
%
If a terminal is an update, we use $\runsof{\anupd}=\set{(\amarking, \anupd, \amarking+\anupd) \mid \amarking, \amarking+\anupd \in \N^d}$. 
If a terminal is a childNGVAS, it has a lower nesting depth and so the set of runs is defined by the induction hypothesis.
Note that the definition yields a strong form of mononoticity on counters in $\unconstrained$.
We have $(\amarking+\amarkingpp, \arun, \amarkingp+\amarkingpp)\in\runsof{\anngvas}$ for $\amarkingpp\in\N^{d}$, if $(\amarking,\arun,\amarkingp)\in\runsof{\anngvas}$ \emph{and} $\amarkingpp[i]=0$ for all $i\not\in\unconstrained$.
This is easy to verify, since no subNGVAS is allowed to constrain the absolute value of a counter $i\in\unconstrained$.
To fix the notation, for a non-terminal $\anonterm$, we also define $\runsof{\anonterm}$ as the union of $\runsof{\asentform}$ over all $\asentform$ that can be derived from $\anonterm$. 
We also write $\updateseqof{\asentform}=\setcond{\arun\in\updates^{*}}{\exists \amarking, \amarkingp.\; (\amarking, \arun, \amarkingp)\in\runsof{\asentform}}$ for the sake of convenience if we are only interested in the update sequences of runs.

%
A wNGVAS is \emph{strong}, if it is strongly connected, all its non-terminals are useful, and in the case of a linear grammar, it has exactly one exit production.
In the rest of the development, whenever we use the term NGVAS, we mean a strong NGVAS. 
In a linear NGVAS with the (hence unique) exit production $\anonterm\to\anngvaspcenterleft.\anngvaspcenterright$, we refer to the children generated on the left resp.~on the right in this rule by $\anngvaspcenterleft$ and $\anngvaspcenterright$.
By standard arguments, we can transform a GVAS to an NGVAS while preserving the language.
We break down the GVAS into its strongly components and track no boundedness information.
\begin{lemma}\label{Lemma:GVAStoNGVAS}
Consider GVAS $\agram$ of dimension $\adim$, $\amarking_{1}, \amarking_{2}\in\N^{d}$. 
We can construct an NGVAS $\anngvas$ with elementary resources so that $\anngvas.\acontextin=\amarking_{1}$, $\anngvas.\acontextout=\amarking_{2}$, $\runsof{\anngvas}=\setcond{(\amarking_{1}, \arun, \amarking_{2})}{\amarking_{1}\fires{\arun}\amarking_{2}\wedge \arun\in\langof{\agram}}$.
\end{lemma}

Finally, we discuss \emph{cycles}, and the notion of \emph{rigid counters}.
A \emph{cycle} $\acyc$ in an NGVAS $\anngvas$ is a derivation $\acyc=\anonterm\to^{*}\aword.\anonterm.\awordp$ with $\anonterm\in\nonterms$, and $\aword, \awordp\in\trms^{*}$.
Intuitively, a cycle captures a pump in the context-free grammar that is available when the constraints are removed.
Towards the definition of a cycle-effects, we define the \emph{effect} $\ceffof{\aword}\subseteq\Z^{d}$ of a string of NGVAS $\aword=\awordp.\awordpp\in\trms^{*}$ as follows.
We let $\ceffof{\varepsilon}=\set{0}$, $\ceffof{\aword}=\ceffof{\awordp}+\ceffof{\awordpp}$, and $\ceffof{\anngvasp}=\anngvasp.\restrictions$ for an NGVAS $\anngvasp$.
For a cycle $\acyc=\anonterm\to^{*}\aword.\anonterm.\awordp$ in $\anngvas$, we define the set of effects $\ceffof{\acyc}$ to be $\ceffof{\aword}\times\ceffof{\awordp}\subseteq\Z^{2d}$, and the set of left-effects to be $\lceffof{\acyc}=\ceffof{\aword}\subseteq\Z^{d}$. Observe in particular that the effect is now a set of vectors, since a childNGVAS \(\anngvasp \in \trms\) may have many possible derivations.
We say that a cycle $\acyc$ is $\anonterm$-\emph{centered}, if its derivation starts from $\anonterm$.

If all cycles of an NGVAS have $0$ effect on a particular counter, we call this counter \emph{fixed}.
We say that $i\in [d]$ is \emph{fixed on the left}, or $(i, \lefttag)$ is \emph{fixed} in $\anngvas$, if $\ceffof{\acyc}[i]=0$ for all cycles $\acyc$ in $\anngvas$.
We say that $i\in [d]$ is \emph{fixed on the right}, or $(i, \righttag)$ is \emph{fixed} in $\anngvas$, if $\ceffof{\acyc}[i+d]=0$ for all cycles $\acyc$ in $\anngvas$.
If the NGVAS $\anngvas$ is non-linear, we say that a counter $i$ is fixed if it is fixed on the left, or the right.
If a counter $i\in [d]$ is fixed on the left (on the right) and $\acontextin[i]\in\N$ ($\acontextout[i]\in\N$), then we call it \emph{rigid on the left (right)}.
For linear NGVAS, if a counter $i\in[d]$ is fixed on the left (on the right) and $\anngvaspcenterleft.\acontextin[i]\in\N$ ($\anngvaspcenterright.\acontextout[i]\in\N$), we call it \emph{rigid on the inside-left (inside-right)}.
If a counter is rigid on any side, we call it rigid.
Note that all counters tracked in the boundedness information are rigid, but not all rigid counters are boundedly tracked.
In our decomposition procedure we will establish the reverse inclusion.

\subsection{Characteristic Equations: Non-Linear Case}\label{Section:CharEqNL}
Consider the non-linear NGVAS $\anngvas=(\agram, \unconstrained, \restrictions, \acontext, \boundednessinformation)$. 
We approximate the set of runs that solve reachability with the \emph{characteristic system of equations} $\chareq{\anngvas}$ defined as
\begin{align*}
\eekof{\agram}{\prodvar} \wedge \evaleqof{\acontextin, \acontextout}{\prodvar, \updvar, \invar, \outvar} 
\wedge
\restrictions(\updvar, \periodvar)\tinyspace . 
\end{align*}
The equations $\evaleq{\acontextin, \acontextout}$ evaluate reachability as follows: 
\begin{align*}
\ptof{}{\prodvar, \termvar}\
&\wedge\ 
\computeupdatesof{\anngvas}{\termvar, \updvar}\
\wedge\ \meof{\updates}{\avar_{\myin}, \updvar, \avar_{\myout}}\ \\
&\wedge\ 
\invar \sqsubseteq \acontextin\ \wedge\
\outvar \sqsubseteq \acontextout\ .
\end{align*}
We have variables $\prodvar\in\N^{\prods}$ that determine how often each production should be taken. 
This choice has to satisfy Esparza-Euler-Kirchhoff. 
The productions lead to a number of terminal symbols given by~$\termvar\in\N^{\trms}$.  
We introduce variables $\updvar\in\N^{\updates}$ that store how often each update should be used, in $\anngvas$ and in the subNGVAS.  
The constraint $\computeupdatesof{\anngvas}{\termvar, \updvar}$ fills $\updvar$ with appropriate values, and we elaborate on it in a moment.
The choice of updates has to satisfy the restriction, and so $\periodvar$ states how often each period in that linear set will be used. 
We have variables $\invar$ and $\outvar$ that store the input and output markings for which the approximation of reachability holds.
They have to coincide with the given $\acontextin$ and $\acontextout$ whenever these are concrete, and can obtain arbitrary non-negative values where $\acontextin$ and $\acontextout$ hold $\omega$. 
We make sure the input and output variables are related by the marking equation.

The constraint $\computeupdatesof{\anngvas}{\termvar, \updvar}$ has to meet the following specification: the variables $\updvar$ store how often each counter update from $\updates$ is used in a run that can be derived from the number of terminal symbols $\termvar$.  
The challenge is to include the updates generated by subNGVAS.
To determine their number, we access the restriction of the immediate childNGVAS. 
Consider $\anngvasp\in\trms$ with restriction $\restrictions_{\anngvasp}=\baseeffectdef+\periodeffectdef^*\subseteq \Z^d$. 
We introduce variables $\avar_{\anngvasp, \updates}$ that store the total number of updates used in a set of runs derivable from instances of~$\anngvasp$. 
The number of instances, and so the number of runs, is given by $\termvar[\anngvasp]$.
To make sure each run can be derived, we use the restriction.  
For every run, we have a copy of the base vector. 
Moreover, we have variables $\periodeffectvardef\in\N^{\periodeffectdef}$ determining how often each period vector should be taken. 
The variables $\updvar$ are then filled by addition, and we make sure not to forget the updates $\restrictto{\termvar}{\updates}$ done by $\anngvas$. 
We define $\computeupdatesof{\anngvas}{\termvar, \updvar}$ as
\begin{align*}
\bigwedge_{\anngvasp\in\trms} \ 
\updates \cdot \updvardef - \baseeffectdef\cdot \termvar[\anngvasp] - \periodeffectdef\cdot \periodeffectvardef\ &=\ 0\\
\updvar- \restrictto{\termvar}{\updates} - \sum_{\anngvasp\in\trms} \updvardef\ &=\ 0\ .
\end{align*}

We also need a \emph{homogeneous variant} $\homchareq{\anngvas}$
 of the characteristic equations,  
\begin{align*}
\homeekof{\agram}{\prodvar}\ &\wedge \ 
\evaleqof{\zeroof{\acontextin}, \zeroof{\acontextout}}{\prodvar, \updvar,\tinyspace \invar, \outvar}\\
 &\wedge \ \homrestrictions(\updvar, \periodvar)\ .
\end{align*}
We check reachability between the zero versions of the input and output markings, and we use the homogeneous variants of Esparza-Euler-Kirchhoff and the restriction. 
The homogeneous characteristic equations are defined such that if $\asol$ solves $\chareq{\anngvas}$ and $\ahomsol$ solves $\homchareq{\anngvas}$, then also $\asol+\ahomsol$ solves $\chareq{\anngvas}$. 
To be explicit, we consider $\N$ solutions. 
Using well-quasi ordering arguments, one can then show that the variables which are unbounded in the solution space of $\chareq{\anngvas}$ are precisely the variables that receive a positive value in some homogeneous solution. 
This leads to the definition of the \emph{support} of the characteristic equations, denoted by $\suppof{\homchareq{\anngvas}}$:
\begin{align*}
\setcond{\avar\in\varsof{\homchareq{\anngvas}}}{\exists \ahomsol.\ \ahomsol\text{ solves }\homchareq{\anngvas}\wedge \ahomsol(\avar)>0}\ .
\end{align*}
As the solution space of $\homchareq{\anngvas}$ is closed under addition, there always is a \emph{full homogeneous solution} that gives a positive value to all variables in the support. 
\subsection{Characteristic Equations: Linear Case}\label{Section:CharEqLin}
Consider the linear NGVAS $\anngvas=(\agram, \unconstrained, \restrictions, \acontext, \boundednessinformation)$. 
In the linear case, we do not have the wide tree theorem available for pumping.
To overcome this problem, the first step is to let the characteristic equations require a stronger form of reachability. 
Remember the directions $\mydir\in\set{\myleft, \mycenterleft, \mycenterright, \myright}$. 
We require that the terminals produced in each direction solve reachability from $\acontextindir$ to $\acontextoutdir$. 
These markings are defined as follows (with $\anngvaspcenterleft.\acontextout=\anngvaspcenterright.\acontextin$ by consistency): 
\begin{alignat*}{5}
\acontextinleft\ &=\ \acontextin\qquad &\qquad \acontextoutleft\ &= \ 
\anngvaspcenterleft.\acontextin\\
\acontextincenterleft\ &=\ \anngvaspcenterleft.\acontextin\qquad &\qquad 
\acontextoutcenterleft\ &= \ \anngvaspcenterleft.\acontextout\\
\acontextincenterright\ &= \ \anngvaspcenterright.\acontextin
\qquad &\qquad \acontextoutcenterright\ &=\ \anngvaspcenterright.\acontextout\\
\acontextinright\ &= \ \anngvaspcenterright.\acontextout \qquad &\qquad 
\acontextoutright\ &=\ \acontextout\ .
\end{alignat*}
We have variables $\prodvarleft$ and $\prodvarright$ for the productions on the left and on the right, but we do not have a variable for the center production.
Such a variable would be bounded in the solution space, and therefore complicate a perfectness condition. 
To ease the notation, we define the following vectors: 
\begin{alignat*}{5}
\prodvecleft\ &=\ (\prodvarleft, 0, 0^{\prodsright})\qquad&\qquad
\prodveccenterleft\ &=\ \prodveccenterright\ = \ \centerprodvec\\
\prodvecright\ &=\ (0^{\prodsleft}, 0, \prodvarright) \qquad &\qquad \prodvec\ &=\ (\prodvarleft, 1, \prodvarright)\ .
\end{alignat*}
We define $\chareq{}$ as  
\begin{align*}
\bigwedge_{\mydir}&
\evaleqof{\acontextindir, \acontextoutdir}{\prodvecdir, \updvardir, \invardir, \outvardir}
\end{align*}
\vspace{-0.5cm}
\begin{alignat*}{7}
\wedge&\quad &\outvarleft - \invarcenterleft &= 0\quad &\wedge\quad &\eekof{}{\prodvec}\\
\wedge&\quad &\outvarcenterleft - \invarcenterright &= 0\quad &\wedge\quad &\restrictions(\updvar, \periodvar)\\
\wedge&\quad &\outvarcenterright - \invarright &= 0 \quad &\wedge \quad & \updvar - \sum_{\mydir} \updvardir = 0\ .
\end{alignat*}
Each instance of $\evaleq{}$ has its own copy $\termvardir, \periodeffectvardefdir, \updvardefdir$ of the variables $\termvar, \periodeffectvardef, \updvardef$. 
In $\computeupdatesof{\anngvas}{\termvardir, \updvardir}$, the conjunction only iterates over $\termsdir$, and so does the sum. 
With $\outvarleft - \invarcenterleft = 0$, the output marking we obtain for reachability on the left coincides with the input marking for the first center run, and so the runs for the two directions can be connected.
We check the restriction on the sum of the updates obtained in all four directions.

We again have a homogeneous variant of the characteristic equations. 
We remove the occurrence of the center production, 
$\prodvecleftzero=\prodvecleft, \prodvecrightzero=\prodvecright, \prodveccenterleftzero=\prodveccenterrightzero=0, \prodveczero=(\prodvarleft, 0, \prodvarright)$.
The definition is then as expected:
\begin{align*}
\bigwedge_{\mydir}&
\evaleqof{\zeroof{\acontextindir}, \zeroof{\acontextoutdir}}{\prodvecdirzero, \updvardir, \invardir, \outvardir}
\end{align*}
\vspace{-0.5cm}
\begin{alignat*}{7}
\wedge&\quad &\outvarleft - \invarcenterleft &= 0\quad &\wedge\quad &\homeekof{}{\prodveczero}\\
\wedge&\quad &\outvarcenterleft - \invarcenterright &= 0\quad &\wedge\quad &\homrestrictions(\updvar, \periodvar)\\
\wedge&\quad &\outvarcenterright - \invarright &= 0 \quad &\wedge \quad & \updvar - \sum_{\mydir} \updvardir = 0\ .
\end{alignat*}
Consider the homogeneous reachability constraint for $\mycenterleft$.
Since we have no productions, the variable $\termvarcenterleft$ will be zero. 
However, we will still collect updates, namely for the period vectors of $\anngvaspcenterleft$. 
The support is as before. 

\subsection{Ranks}\label{Section:Ranks}
As we discussed, a well founded rank underlies our decomposition procedure.
This rank is not only necessary for the termination proof, but also for a full understanding of the perfectness notion.
In the following we develop our rank.
This section is meant to be understood as an extended verision of the rank discussion in Section \Cref{Section:OutlineRanks}.
For the sake of completeness, we do not refrain from restating definitions of Section \Cref{Section:OutlineRanks}.
This section concerns itself with (strong) NGVAS, we extend these concepts to wNGVAS in \Cref{Section:Decomposition}.

We understand the rank $\rankof{\anngvas}$ of an NGVAS $\anngvas$ in two parts, 
$$\rankof{\anngvas}=(\recrankof{\anngvas}, \itrankof{\anngvas})\in\N^{d+3}\times\N^{2d+1}.$$ 
First is the \emph{main rank}, $\recrankof{\anngvas}$, which we also call the \emph{non-linear rank}.
It captures the recursive behaviour.
The second is the \emph{auxiliary rank}, $\itrankof{\anngvas}$, which we also call the \emph{linear rank}.
It captures the iterative / VASS-like behaviour of the NGVAS.
Both parts of the rank rely on quantifying the effects of cycles of an NGVAS in a tangible way. 
The corresponding local rank is an extension of a definition by Leroux \cite{LerouxS19} for VASS.

Below, we first define cycles in a grammar.
Then, we build up the concepts required to define the non-linear rank, and conclude the subsection by defining the linear rank.

\paragraph*{Cycles} 
Then, the vector spaces $\cyclespaceof{\anngvas}$, $\lcyclespaceof{\anngvas}$ spanned by the effects resp. left-effects of cycles are
%
\begin{align*}
    \cyclespaceof{\anngvas}=\ &\spanof{\setcond{\amarking\in\ceffof{\acyc}}{ \acyc\text{ a cycle in }\anngvas}}\subseteq\mathbb{R}^{2d}\\
    \lcyclespaceof{\anngvas}=\ &\spanof{\setcond{\amarking\in\lceffof{\acyc}}{ \acyc\text{ a cycle in }\anngvas}}\subseteq\mathbb{R}^{d}.
\end{align*}
Clearly, projecting $\cyclespaceof{\anngvas}$ to the first $d$ components yields $\lcyclespaceof{\anngvas}$.
As usual, the dimension of a vector space $\dimensionof{\mathbf{V}}$ is the minimal number of generators.

\paragraph*{Local-rank} Using the dimensionality of cycle spaces, we define the local-rank of an NGVAS \(\anngvas\).
The \emph{local-rank} of an NGVAS $\anngvas$ is
$$\lrankof{\anngvas}=1_{\dimensionof{\lcyclespaceof{\anngvas}}}\cdot|\anngvas.\nonterms|$$
if $\anngvas$ is non-linear, and 
$$\lrankof{\anngvas}=0$$
if $\anngvas$ is linear. 
In the non-linear case, the rank is based on the dimension of the vector space spanned by \emph{cycle effects} $\dimensionof{\cyclespaceof{\anngvas}}$, as adapted from \cite{LerouxS19}. 
If it is linear, the local rank is $0$.
Note that we only use the left cycle space, and not the full cycle space.
Intuitively, this is because non-linear NGVAS cannot tightly control which non-terminal appears to the left resp. right of another non-terminal.
Thus, the left and the right cycle spaces coincide.
%
%

\paragraph*{Branch-rank} The rank needs to evaluate the whole structure globally, and not only the current top-level NGVAS locally.
To this end, our rank definition considers the branches of the NGVAS.
An \emph{NGVAS branch} $\abranch=\anngvas_{0}\ldots\anngvas_{k}$ is a sequence of NGVASes linked by the relation $\anngvas_{i+1}\in\anngvas_{i}.\trms$ for all $i< k$.
We say that $\abranch=\anngvas_{0}.\abranch'$ is a branch of $\anngvas_{0}$.
Similarly to the ranks for KLM sequences \cite{LerouxS19}, the rank of a branch $\abranch=\anngvas_{0}\ldots\anngvas_{k}$ is then the sum of the local ranks along it,  
$$\brankof{\abranch}=\sum_{i\leq k}\lrankof{\anngvas_{i}}.$$

\paragraph*{System-rank}
In contrast to the KLM sequences treated in \cite{LerouxS19}, an NGVAS may contain multiple branches. 
To deal with these, we simply take the maximum, i.e.\ we define the system-rank $\srankof{\anngvas}$ as
$$\srankof{\anngvas}=\max_{\abranch\text{ a branch of }\anngvas}\brankof{\abranch}.$$
the maximum rank along all branches.
This definition is the backbone of our non-linear rank.
The intuitive idea behind using the maximum branch rank is the following.
Our procedure works bottom-up, it refines each subNGVAS and then the top level NGVAS.
If the top level NGVAS is not perfect, it can be refined, the ranks of all branches go down.
Thus, the maximum decreases. But the maximum has the large advantage that copying branches does not increase the rank, allowing us to replace one child with multiple similar children without increasing the rank.
We remark that the branch $\anngvas_{0}\ldots\anngvas_{k}$ that witnesses the rank will always be unextendable, meaning $\anngvas_{k}$ will be a nesting depth $0$ NGVAS.

\paragraph*{Local Index}
As their name implies, linear grammars can only produce linear derivation trees.
However, linear NGVAS may produce higher index derivation trees, thanks to their nested structure. 
Indeed, they can produce multiple copies of a childNGVAS, which may be of high index, or even non-linear.
Thus, they also behave non-linearly, albeit in a bounded, and locally linear way.
We capture this in our non-linear rank with the local index.

The local index $\localgrmindexof{\anngvas}$ of an NGVAS is the index of the (full) derivation trees it produces, if we cut them at non-linear NGVASes, or NGVASes of smaller system rank.
To keep the definition syntactical, we let the local index $\localgrmindexof{\anngvas}\in\N$ of linear $\anngvas$ be the maximum of 
\begin{enumerate}
    \item[(i)] $1$,
    \item[(ii)]  $\localgrmindexof{\anngvasp}$ for all $\anngvasp\in\anngvas.\trms$ with $\srankof{\anngvasp}=\srankof{\anngvas}$,
    \item[(iii)] $\localgrmindexof{\anngvasp}+1$ for all $\anngvasp\in\anngvas.\rectrms$ with $\srankof{\anngvasp}=\srankof{\anngvas}$, where \(\rectrms\) is the set of children \(\anngvasp \in \trms\) which can be derived by some persistent production, i.e.\ not only the exit production,
    \item[(iv)] $\localgrmindexof{\anngvaspcenterleft}+1$ if $\srankof{\anngvaspcenterleft}=\srankof{\anngvaspcenterright}=\srankof{\anngvas}$, and $\localgrmindexof{\anngvaspcenterleft}=\localgrmindexof{\anngvaspcenterright}$ for the exit rule $\anonterm\to\anngvaspcenterleft.\anngvaspcenterright$.
\end{enumerate}
If the NGVAS $\anngvas$ is non-linear, then $\localgrmindexof{\anngvas}=0$.
It is easy to check that this syntactical definition of index matches the usual semantical one.

\paragraph*{Non-linear rank and the main branch}
Now, we contextualize these concepts within the non-linear NGVAS rank.
We define the non-linear rank $\recrankof{\anngvas}$ of $\anngvas$ as
$$\recrankof{\anngvas}=(\cardof{\constrained}, \srankof{\anngvas}, \localgrmindexof{\anngvas})\in\N\times\N^{d+1}\times\N.$$
The most significant component of the rank is $\cardof{\constrained}\in\N$, the number of counters with reachability constraints.
The remaining counters can start from arbitrarily high values, and can be seen as not having positivity constraints.
Thus, this component captures the number of $\N$-counters.
The second most significant component of the rank is $\srankof{\anngvas}\in\N^{d+1}$, the system rank we discussed above.
Finally, the least significant component is $\localgrmindexof{\anngvas}$, the local index. 

First consequence of the recursive rank definition is that all childNGVAS of a non-linear NGVAS have a smaller recursive rank.
\begin{lemma}\label{Lemma:NonLinearChildRank}
    Let $\anngvas$ be a non-linear NGVAS, and let $\anngvasp\in\anngvas.\trms$.
    Then, $\recrankof{\anngvas}>\recrankof{\anngvasp}$. 
\end{lemma}
For linear NGVAS, this might not be the case.
The definition of the local rank, which feeds into the definition of the system rank, ignores linear NGVASes.
Furthermore, the local index definition allows the NGVAS to produce one NGVAS with the same system-rank and local index.

To understand the linear case better, we need the concept of a main branch.
We call $\abranch=\anngvas_{0}.\ldots.\anngvas_{k}$ a \emph{main branch} of $\anngvas_{0}$, if $\srankof{\anngvas_{i}}=\srankof{\anngvas_{0}}$ and $\localgrmindexof{\anngvas_{i}}=\localgrmindexof{\anngvas_{0}}$ for all $i\leq k$.
By \Cref{Lemma:NonLinearChildRank}, the only main branch of non-linear $\anngvas$ is $\abranch=\anngvas$, since all children have lower system-rank.
In the linear case, we observe that $\anngvasp_{i+1}$ can neither be recurring in $\anngvas$, nor be produced with a same index, same system-rank sibling.
If this were the case, then we would get $\localgrmindexof{\anngvas_{i}}\geq\localgrmindexof{\anngvas_{i-1}}+1$ by conditions (iii) and (iv).
By extending this analysis, we observe that there is a unique maximal main branch. 
\begin{lemma}\label{Lemma:LinearMainBranch}
    Let $\anngvas$ be an NGVAS.
    Then, there is a main branch $\abranch$ of $\anngvas$, such that all main branches $\abranch'$ of $\anngvas$ are a prefix of $\abranch$.
\end{lemma}
We call this branch \emph{the main branch} of $\anngvas$, and refer to it by $\mainbranchof{\anngvas}$ ($\mainbranch$ if $\anngvas$ is clear from the context).

\paragraph*{Linear rank}
The linear rank draws on the connection between linear grammars, and left-linear grammars (regular languages).
Using convolution, linear grammars can be turned into left-linear grammars \cite{AtigG11}.
This process does not preserve the language, but keeps enough information to decide VAS-related queries.
With the same intuition, our linear case decomposition procedure closely mirrors the KLM decomposition \cite{LerouxS19}, or the nested VASS decomposition \cite{GuttenbergCL25}.
It treats the main branch as a $2d$-counter KLM sequence that calls NGVAS with its transitions.
The linear rank captures this.
Towards its definition, we develop the notion of linear local-rank 
$$\linlrankof{\anngvas}=1_{\dimensionof{\cyclespaceof{\anngvas}}}\cdot\cardof{\anngvas.\nonterms}\in\N^{2d+1}.$$
for linear $\anngvas$, and $\linlrankof{\anngvas}=0$ otherwise.
We also define the linear branch-rank 
$$\linbrankof{\abranch}=\sum_{i\leq k}\linlrankof{\anngvas_{i}}.$$
for $\abranch=\anngvas_{0}\ldots\anngvas_{k}$ similarly to $\brank$.
The linear rank is then the rank of the main branch,
$$\itrankof{\anngvas}=\linbrankof{\mainbranchof{\anngvas}}.$$
The linear rank is similar to the system rank definition.
We find it useful to underline the three main differences.
First, in opposition to the non-linear rank, the linear rank ignores non-linear NGVAS, and only evaluates linear NGVAS.
Second, it evaluates only one branch, $\mainbranch$.
Finally, it uses the full cycle space, $\cyclespaceof{\anngvas}$, instead of the left cycle-space $\lcyclespaceof{\anngvas}$.
This has the following reasoning behind it.
Linear grammars can tightly control which terminals are produced on the left resp. right side of the central branch in the derivation tree.
Thus, the left- and right-effects remain heterogenous.
To soundly capture the cyclical behaviour of a linear NGVAS, our definition must consider the interplay between the sides.
This is only possible with $\cyclespaceof{\anngvas}\subseteq\mathbb{R}^{2d}$.
\subsection{Perfectness}\label{Section:Perfectness}
We define perfectness conditions on NGVAS that allow us to construct reaching runs. 
The corresponding iteration lemma is our first technical achievement. 
Let $\anngvas=(\agram, \restrictions, \acontext, \boundednessinformation)$ have the restriction $\restrictions = \baseeffect+\periodeffect^*$.
Let the sets of unbounded counters be $\adimsetleft, \adimsetright\subseteq [1, \adim]$ with  $\adimsetleft=\adimsetright$ in the non-linear case.  
We call $\anngvas$ \emph{clean}, if \perfectnesssol to \perfectnesssubclean hold, and \emph{perfect} if it is clean and additionally \perfectnesschildren to \perfectnesspumping hold: 
\begin{description}[leftmargin=0cm, labelsep=-0.05cm,itemsep=4pt]
\item[\perfectnesssol] For every $\baseeffectchoice\in\restrictions$ there is a solution $\asol$ to $\chareq{\anngvas}$ with $\asol[\updvar]=\baseeffectchoice$. 
For every $\periodeffectchoice\in\N^{\periodeffect}$ with $\periodeffectchoice\geq 1$ there is a full homogeneous solution $\ahomsol$ to $\homchareq{\anngvas}$ with $\ahomsol[\updvar]=\periodeffect\cdot \periodeffectchoice$. 
\item[\perfectnesscounters] All unbounded counters in a reachability constraint are in the support. 
\item[\perfectnessrecchildren] All subNGVAS $\anngvasp\in \anngvas$ not along $\mainbranchof{\anngvas}$ are perfect. 
\item[\perfectnessbase] The base effects of all recurring childNGVAS are enabled, for all $\anngvasp\in\trms$ with restriction $\baseeffectdef+\periodeffectdef^*$ there is $\baserundef\in\runsof{\anngvasp}$ with $\updates\cdot\paramparikhof{\updates}{\baserundef}=\baseeffectdef$.
\item[\perfectnesschildperiodsbd] For linear $\anngvas$, and center childNGVAS $\anngvasp\in\trms$, and $\mydir\in\set{\mycenterleft, \mycenterright}$ we have $\periodeffectvardef, \periodeffectvardefdir\subseteq\suppof{\homchareq{\anngvas}}$.
\item[\perfectnessrigid] All rigid counters are concretely tracked in the boundedness information.
\item[\perfectnesssubclean] All subNGVAS $\anngvasp\in\anngvas$ are clean. 
\item[\perfectnesschildren] All subNGVAS $\anngvasp\in\anngvas$ are perfect.
\item[\perfectnessprods]  All productions, as well as all period vectors in the restrictions of childNGVAS $\anngvasp\in\trms$ that can be produced in cycles are in support. That is, for all $\mydir$ we have $\prodvar, \periodeffectvardef, \periodeffectvardefdir\subseteq \suppof{\homchareq{\anngvas}}$.
\item[\perfectnesspumpingint]
This requirement only applies in the linear case. Intuitively, the right-linear nested VASS resulting from convolution has a down-pumping cycle. 

There exist a derivation $\startnonterm\rightarrow \asentformpumpleftint.\startnonterm.\asentformpumprightint$, runs $\downseqint \in\runsof{\asentformpumpleftint}$, \(\upseqint\in\runsof{\asentformpumprightint}\) and markings $\inmarkint, \outmarkint\in\N^{\adim}$ with $\inmarkint\sqsubseteq \acontextincenterleft, \outmarkint\sqsubseteq \acontextoutcenterright$ so that $\inmarkint\fires{\revof{\downseqint}}\amark_3$ and $\outmarkint\fires{\upseqint}\amark_4$. 
The runs have an effect $(\amark_3-\inmarkint)[\adimsetleft\setminus \omegaof{\acontextinint}]\geq 1$ and $(\amark_4-\outmarkint)[\adimsetright\setminus \omegaof{\acontextoutint}]\geq 1$. 

\item[\perfectnesspumping]
There are up-pumping and down-pumping runs. 

Formally, there exist a derivation $\startnonterm\rightarrow \asentformpumpleft.\startnonterm.\asentformpumpright$, runs $\upseq \in\runsof{\asentformpumpleft}$ and \(\downseq \in \runsof{\asentformpumpright}\) and markings  $\inmark, \outmark\in\N^{\adim}$ with $\inmark\sqsubseteq \acontextin$ and $\outmark\sqsubseteq \acontextout$ so that $\inmark\fires{\upseq}\amark_1$, $\outmark\fires{\revof{\downseq}}\amark_2$, $\at{(\amark_1-\inmark)}{\adimsetleft\setminus \omegaof{\acontextin}}\geq 1$ and $\at{(\amark_2-\outmark)}{\adimsetright\setminus \omegaof{\acontextout}}\geq 1$. 
I.e.\ the runs have a strictly positive effect on counters that become $\omega$ in $\adimsetleft$ but are not $\omega$ in the input marking respectively a strictly negative effect on the counters that are $\omega$ in $\adimsetright$ but not  $\omega$ in the output marking.
\end{description}

Observe the subtle difference between \perfectnesspumpingint and \perfectnesspumpingnospace: In \perfectnesspumpingint the left pumping sequence is fired in reverse, and in \perfectnesspumping the right pumping sequence is fired in reverse, and the starting markings are different. The idea is the following: For derivations $\startnonterm\rightarrow \asentformpumpleftint.\startnonterm.\asentformpumprightint$ and the corresponding runs \(r_{\text{left}}, r_{\text{right}}\), there are \emph{four} crucial markings: The source/target markings of \(r_{\text{left}}\), and the source/target markings of \(r_{\text{right}}\). In \perfectnesspumpingint we force \(r_{\text{left}}\) to have a negative effect and \(r_{\text{right}}\) to have a positive effect, and in \perfectnesspumping it is the other way around.
\section{Iteration Lemma}\label{Section:IterationLemma}

From now on, $\anngvas$ is for NGVAS and $\anngvasp$ is for children.

\TheoremIterationLemmaMain*
We proceed by Noetherian induction on the nesting depth. The base case of an update is obvious. For the induction step, we distinguish between the linear and the non-linear case.

Consider $\baseeffectchoice\in\restrictions$ and $\periodeffectchoice\in\N^{\periodeffect}$ with $\periodeffectchoice\geq 1$. We start as follows: Perfectness \perfectnesssol gives us a solution~$\asol$ to $\chareq{\anngvas}$ with $\updates \cdot \asol[\updvar]=\baseeffectchoice$ and a full homogeneous solution~$\ahomsol$ to $\homchareq{\anngvas}$ with $\updates \cdot \ahomsol[\updvar]=\periodeffect\cdot \periodeffectchoice$. 
We consider a non-linear \(\anngvas\) first.
\subsection{Non-Linear Case}
\subsubsection{Reaching Derivation}
We define a constant \(\maxconst\) and a new solution $\asol' = \asol+ \maxconst\cdot \ahomsol$ to $\chareq{\anngvas}$  that is large enough to embed non-negative runs for all instances of all childNGVAS.
It is a solution because $\ahomsol$ is homogeneous. 
The situation is the following. 
The derivation induced by $\asol'$ will introduce several instances of the terminal symbols $\anngvasp\in\trms$. 
Each instance $\anngvasp$ is a childNGVAS that now has to provide its own derivation. 
There are two requirements on the derivations of $\anngvasp$: 
\begin{itemize}
\item[(1)] The run given by the derivation has to be enabled (solving Problem 2 from the overview).
\item[(2)] The derivations of all $\anngvasp$ instances together have to give the number of updates $\at{\asol'}{\updvardef}$, in order to guarantee the desired effect.   
\end{itemize}
To achieve (1), our plan is to use the perfectness property~\perfectnessbase for all instances of $\anngvasp$ except one. 
Property~\perfectnessbase gives us a derivable run that is guaranteed to be enabled in the input marking and correspond to the base vector $\baseeffectdef$ of the restriction $\restrictions_{\anngvasp}=\baseeffectdef+\periodeffectdef^*$. 
This, however, does not yet guarantee~(2). 
The solution $\asol'$ may also ask for repetitions of the period vectors in~$\periodeffectdef$. 
For the one instance of $\anngvasp$ that we left out, we obtain the non-negative run by invoking the induction hypothesis with appropriate $\baseeffectchoicedef$ and~$\periodeffectchoicedef$.

We define $\maxconst$.
Let $\anngvasp\in\trms$ with $\restrictionsdef=\baseeffectdef+\periodeffectdef^*$.  
To invoke the induction hypothesis, we let $\baseeffectchoicedef$ include the base vector of the restriction $\baseeffectdef$ and repetitions of the period vectors as required by $\asol$. 
For $\periodeffectchoicedef$, we follow the full homogeneous solution:
\begin{align*}
\baseeffectchoicedef\ &=\ \baseeffectdef + \periodeffectdef\cdot \asol[\periodeffectvardef]\\
\periodeffectchoicedef\ &=\ \ahomsol[\periodeffectvardef]\ .
\end{align*}
Since $\anngvasp$ is perfect by~\perfectnesschildrennospace, $\baseeffectchoicedef\in\restrictionsdef$, and $\periodeffectchoicedef\geq 1$ by~\perfectnesschildperiodsbdnospace, we can invoke the induction hypothesis. 
It yields $\initconstdef$ so that for every $\aconst\geq \initconstdef$ we have a run $\iterrundefof{\aconst}\in \runsof{\anngvasp}$ with effect $\vaseffof{\iterrundefof{\aconst}}=\baseeffectchoice_{\anngvasp} +  \aconst \cdot \periodeffectchoicedef$.

Recall that $\asol'=\asol+\maxconst\cdot \ahomsol$, and that for \emph{every} childNGVAS \(\anngvasp\) we hence have to insert \(\maxconst\) many periods. Hence $\maxconst$ has to be larger than all \(\initconstdef\), i.e. $\maxconst:=\max_{\anngvasp} \initconstdef$. 

We show that $\asol'$ induces a derivation.
Since~$\asol'$ solves Esparza-Euler-Kirchhoff, contains a copy of $\ahomsol$, and $\at{\ahomsol}{\prodvar}\geq 1$ by~\perfectnessprodsnospace, Theorem~\ref{Theorem:EEK} yields $\startnonterm\xrightarrow{\aprodseqreachanngvas}\asentformreach$. 
By Lemma~\ref{Lemma:EEKConverse}, $\asentformreach$ is a sequence of terminal symbols in $\anngvas$. 
Thanks to $\ptof{}{\prodvar, \termvar}$, even $\paramparikhof{\trms}{\asentformreach}=\at{\asol'}{\termvar}$ holds.  
Since we use a Noetherian induction, we are not sure whether the terminals are updates $\anupd\in\Z^{\adim}$ or childNGVAS $\anngvasp$, but admit both. 

We still have to derive runs for the childNGVAS~in~$\asentformreach$. 
Together, we then have 
\begin{align*}
\startnonterm\xrightarrow{\aprodseqreachanngvas}\asentformreach
\text{ and }\reachrun\in\runsof{\asentformreach}.
\end{align*}
Recall that $\runsof{\asentformreach}$ connects runs of the terminals in $\asentformreach$, $\reachrun=\arun_1\ldots\arun_{\cardof{\asentformreach}}$.  
If~$\at{\asentform}{\mathit{i}}$ is an update $\anupd$, then $\arun_i=\anupd$. 
If~$\at{\asentform}{\mathit{i}}$ is the first instance of a child $\anngvasp$, then $\arun_i=\arun_{\anngvasp}^{(\maxconst)}$ as defined above.  
If $\at{\asentform}{\mathit{i}}$ is another instance of $\anngvasp$, then we use \perfectnessbase and get $\arun_i=\baserundef$ with effect $\baseeffectdef$. 

Thanks to the marking equation $\meof{\anngvas}{\invar, \updvar, \outvar}$, the number of updates given by $\at{\asol'}{\avar_{\updates}}$ transforms 
$\at{\asol'}{\invar}\sqsubseteq \acontextin$ into $\at{\asol'}{\outvar}\sqsubseteq\acontextout$. 
In the appendix, we show that $\reachrun$ has precisely the desired effect:
\begin{align}
\vaseffof{\reachrun}\ = \ \updates \cdot \asol'[\avar_{\updates}].\tag{\updatesreach}\label{Equation:UpdatesReach} 
\end{align}

In particular, \(\reachrun\) respects the restriction \(\restrictions\), i.e.\ respects contexts. However, $\reachrun$ is not guaranteed to stay non-negative. 
Neither do the updates in $\asentformreach$ guarantee non-negativity, nor do the runs of the childNGVAS guarantee non-negativity \emph{on the counters from} $\adimset$, since they are \emph{considered} \(\omega\) inside any childNGVAS \(\anngvasp\). 
We now address this problem (Problem 1 in overview). \\

\subsubsection{Pumping Derivation and Embedding}
The end idea is simple, surround $\reachrun$ by a number of up-pumping and down-pumping runs to increase the counters from $\adimset$.  
So we wish to repeatedly apply the pumping derivation $\startnonterm\xrightarrow{\aprodseqpump}\upseq.\startnonterm.\downseq$ that exists by~\perfectnesspumpingnospace.  
Unfortunately, we cannot even insert a single copy of $\aprodseqpump$ without running the risk of no longer satisfying $\chareq{\anngvas}$ (Problem 3 in overview).
The way out is to embed the pumping derivation into a homogeneous solution, because adding a homogeneous solution to $\asol$ remains a solution.
Embedding the pumping derivation into a homogeneous solution means we provide another derivation 
$\startnonterm\xrightarrow{\aprodseqdiff}\diffrunleft.\startnonterm.\diffrunright$ that should be understood as the difference between the homogeneous solution and the pumping derivation. 
Formally, the two derivations together,  
\begin{align*}
\startnonterm\xrightarrow{\aprodseqpump}\upseq.\startnonterm.\downseq\xrightarrow{\aprodseqdiff}\upseq.\diffrunleft.\startnonterm.\diffrunright.\downseq \ ,
\end{align*}
will have the effect prescribed by the homogeneous solution. 
In this section, we make the notion of embedding precise. 

We first inspect the pumping runs. 
There is a derivation $\startnonterm \xrightarrow{\aprodseqpumpanngvas} \asentformpumpleft.\startnonterm.\asentformpumpright$ with $\upseq.\downseq\in\runsof{\asentformpumpleft.\asentformpumpright}$. 
Moreover, by Lemma~\ref{Lemma:EEKConverse}, $\paramparikhof{\prods}{\aprodseqpumpanngvas}$ solves $\homeek{\anngvas}$. 

The homogeneous solution $\ahomsol$ we are given may not be large enough to embed $\aprodseqpump$.  
We scale it to $\sumconst\cdot\ahomsol$ with a factor $\sumconst=\embedconst+\enableconst$ we will now define.
The constant $\embedconst$ should be understood as the least natural number large enough so that $\embedconst\cdot \ahomsol$ embeds $\aprodseqpump$. 
Embedding $\aprodseqpump$ is made formal with two requirements we give next. 
These requirements are monotonic in that if they hold for $\embedconst\cdot \ahomsol$, then they will hold for $\sumconst\cdot \ahomsol$. 
The first requirement is that we want to be able to subtract $\paramparikhof{\prods}{\aprodseqpumpanngvas}$ from $\at{(\embedconst\cdot \ahomsol)}{\prodvar}$.
Even more, in the difference we want to retain a copy of each production to be able to invoke Esparza-Euler-Kirchhoff.  
This means $\embedconst$ has to be large enough so that 
\begin{align}
\at{(\embedconst\cdot \ahomsol)}{\prodvar} - \paramparikhof{\prods}{\aprodseqpumpanngvas}\ \geq\ 1\ .\tag{\embeddingprods}\label{Equation:EmbeddingProductions}
\end{align} 
This can be achieved as the productions belong to the support, due to \perfectnessprodsnospace. The next requirement is that $\embedconst\cdot \ahomsol$ has to cover the updates in $\upseq.\downseq$. 
For the updates from $\asentformpumpleft.\asentformpumpright$, we use the inequality 
\begin{align}
\at{(\embedconst\cdot \ahomsol)}{\updvar} - \paramparikhof{\updates}{\asentformpumpleft.\asentformpumpright} \geq 0. \tag{\embeddingupdates}\label{Equation:EmbeddingUpdates}
\end{align} 
We also have to cover the updates produced by the instances of the childNGVAS in $\asentformpumpleft.\asentformpumpright$.  
Consider~$\anngvasp$ with $\restrictionsdef=\baseeffectdef+\periodeffectdef^*$.  
By $\computeupdatesof{\anngvas}{\termvar, \updvar}$, we are sure $\at{(\embedconst\cdot \ahomsol)}{\updvardef}$ contains precisely one copy of the base vector $\baseeffectdef$ for every instance of~$\anngvasp$ in $\asentformpumpleft.\asentformpumpright$. 
However, $\aprodseqpump$ may also produce copies of the period vectors. 
Let $\periodeffectpumpdef\in\N^{\periodeffectdef}$ count the period vectors in all runs that belong to an instance of~$\anngvasp$ in $\asentformpumpleft.\asentformpumpright$. 
For every childNGVAS $\anngvasp\in\trms$, we want 
\begin{align}
\at{(\embedconst\cdot \ahomsol)}{\periodeffectvardef} - \periodeffectpumpdef\ \geq\ 0\ .\tag{\embeddingperiods} \label{Equation:EmbeddingEffect}
\end{align} 
This can be achieved as the period vectors of all childNGVAS belong to the support, \perfectnesschildperiodsbdnospace.\\

\subsubsection{Difference Derivation}
The goal is to turn the difference $(\sumconst\cdot \ahomsol)[\prodvar] - \paramparikhof{\prods}{\aprodseqpumpanngvas}$ into a run. 
To this end, we define $\aprodseqdiff = \aprodseqdiffanngvas.\aprodseqdiffterms$ as a production sequence in~$\anngvas$ followed by productions in the descendants. 
On the way, we define the missing $\enableconst$. 

To obtain $\aprodseqdiffanngvas$, we use Theorem~\ref{Theorem:EEK}. 
For the applicability, note that $(\sumconst\cdot \ahomsol)[\prodvar]$ solves $\homeek{\anngvas}$ because $\ahomsol[\prodvar]$ does. 
We already argued that also $\paramparikhof{\prods}{\aprodseqpumpanngvas}$ solves $\homeek{\anngvas}$. 
Hence, the difference $(\sumconst\cdot \ahomsol)[\prodvar] - \paramparikhof{\prods}{\aprodseqpumpanngvas}$ solves $\homeek{\anngvas}$. 
By Requirement \eqref{Equation:EmbeddingProductions}, $(\sumconst\cdot \ahomsol)[\prodvar] - \paramparikhof{\prods}{\aprodseqpumpanngvas}\geq 1$. 
Theorem~\ref{Theorem:EEK} yields a derivation $\startnonterm\xrightarrow{\aprodseqdiffanngvas} \asentformdiffleft.\startnonterm.\asentformdiffright$ with $\paramparikhof{\prods}{\aprodseqdiffanngvas} = (\sumconst\cdot \ahomsol)[\prodvar] - \paramparikhof{\prods}{\aprodseqpumpanngvas}$. 

We construct $\diffrunleft.\diffrunright\in\runsof{\asentformdiffleft.\asentformdiffright}$ using a sequence of productions $\aprodseqdiffterms$. 
The idea is similar to before. 
For all instances of childNGVAS~$\anngvasp$ except the first in~$\asentformdiffleft.\asentformdiffright$, we use the run $\baserundef$ from \perfectnessbasenospace, meaning we embed no periods.
For the first instance of a childNGVAS~$\anngvasp$, we use a run $\arun_{\anngvasp}'^{(\sumconst)}$ which will be given by the induction hypothesis. 
This run will compensate the $(\embedconst\cdot \ahomsol)[\periodeffectvardef] - \periodeffectpumpdef\geq 0$ excess in period vectors from the pumping derivation.


To construct $\iterrunpenable$ that compensates the excess in period vectors, we invoke the induction hypothesis with 
\begin{align*}
\baseeffectchoicedef'\ &=\ \baseeffectdef + \periodeffectdef\cdot [ (\embedconst\cdot \ahomsol)[\periodeffectvardef] - \periodeffectpumpdef]\\
\periodeffectchoicedef'\ &=\ \ahomsol[\periodeffectvardef]\ .
\end{align*}
We have $\baseeffectchoicedef'\in\restrictionsdef$ by \eqref{Equation:EmbeddingEffect} for $\embedconst$. 
As moreover $\anngvasp$ is perfect by~\perfectnesschildren and $\periodeffectchoicedef'\geq 1$ by \perfectnesschildperiodsbdnospace, the hypothesis applies and yields $\initconstdef'\geq 1$ so that for every $\aconst\geq \initconstdef'$ we have a run $\iterrunpdefof{\aconst}\in \runsof{\anngvasp}$ with updates $\paramparikhof{\updates}{\iterrunpdefof{\aconst}}=\baseeffectchoice_{\anngvasp}' +  \periodeffectdef\cdot \aconst \cdot \periodeffectchoicedef'$.
We define $\enableconst=\max_{\anngvasp} \initconstdef'$.

To sum up, $\startnonterm\rightarrow^*\upseq.\diffrunleft.\startnonterm.\diffrunright.\downseq$ using $\aprodseqpump.\aprodseqdiff$ with $\aprodseqpump = \aprodseqpumpanngvas.\aprodseqpumpterms$ and $\aprodseqdiff=\aprodseqdiffanngvas.\aprodseqdiffterms$. 
Similar to \eqref{Equation:UpdatesReach}, the effect is the one expected by $\sumconst\cdot \ahomsol$: 
\begin{align}
\vaseffof{\upseq.\diffrunleft.\diffrunright.\downseq}\tinyspace =\tinyspace
\updates \cdot \sumconst\cdot \ahomsol[\updvar]\tinyspace .\tag{\updatespumpdiff}\label{Equation:UpdatesPumpDiff}
\end{align} 

At this point there is a minor step we did not mention in the overview. So far, we can only create runs for multiples of $\sumconst\cdot \ahomsol$, instead of any $\aconst \cdot \ahomsol$ with $\aconst \geq \initconst$. 
The repair is simple, once found. 
We repeat the above paragraphs with $\sumconst'=\sumconst+1$ (remember all requirements were monotone) to obtain production sequences $\aprodseqdiff'=\aprodseqdiffanngvas'.\aprodseqdiffterms'$ and runs $\diffrunleft'.\diffrunright'$ with $\startnonterm\rightarrow^*\diffrunleft'.\startnonterm.\diffrunright'$ and $\vaseffof{\upseq.\diffrunleft'.\diffrunright'.\downseq}=\updates \cdot (\sumconst'\cdot \ahomsol)[\updvar]$, meaning we now embed $\aprodseqpump$ into $\sumconst'$ many homogeneous solutions. 
For every $\aconst \geq \sumconst^2+\maxconst$, we can write $\aconst-\maxconst = j_1 \cdot \sumconst + j_2\cdot \sumconst'$ and then $j_1$ many times embed the pumping sequence using $\aprodseqdiff$ and $j_2$ many times embed the pumping sequence using $\aprodseqdiff'$, in total using $\aconst$ many homogeneous solutions. 


\subsubsection{Pumping}
We now make sure the reaching run is enabled by surrounding the reaching derivation by repetitions of the pumping derivation and the difference derivation. Observe first that counters $\acounter \not \in \adimset$,  counters which are concretely stored in the non-terminals, cannot create problems because of the consistency conditions on the $\infun, \outfun$ functions. 
It remains to deal with counters $\acounter \in \adimset$, which can be concrete both in input and output (Case 1), concrete only in input (Case 2), concrete only in output (Case 3), or concrete neither in input nor output (Case 4). We only deal with Case 1 here. 
\paragraph*{Case 1}
The counter $\acounter \in \adimset$ is concrete in input and output.
Here is what we know about $\acounter$.
The runs $\upseq$ and~$\downseq$ have a strictly positive resp.~a strictly negative effect on \(\acounter\), by \perfectnesspumpingnospace.
Together, the runs $\upseq.\downseq$ and $\diffrunleft.\diffrunright$ from the pumping derivation resp. the difference derivation have effect zero on \(\acounter\).
This is by the use of $\zeroof{\acontextin}$ and $\zeroof{\acontextout}$ in the homogeneous variant of the characteristic equations. 
As a consequence, $\upseq.\diffrunleft.\diffrunright$ has a strictly positive effect on~$\acounter$.
Repeating the pumping and the difference derivation in a naive way then yields
\begin{align*}
\startnonterm\xrightarrow{\aprodseqpump^{j_1+j_2}}&\ \upseq^{j_1+j_2}.\startnonterm.\downseq^{j_1+j_2}\\
\xrightarrow{\aprodseqdiff^{j_1}}&\ \upseq^{j_1+j_2}.\diffrunleft^{j_1}.\startnonterm.\diffrunright^{j_1}.\downseq^{j_1+j_2} \\
\xrightarrow{\aprodseqdiff^{`j_2}}&\ \upseq^{j_1+j_2}.\diffrunleft^{j_1}.\diffrunleft'^{j_2}.\startnonterm.\diffrunright'^{j_2}\diffrunright^{j_1}.\downseq^{j_1+j_2} \\
\xrightarrow{\aprodseqreach}&\ \upseq^{j_1+j_2}.\diffrunleft^{j_1}.\diffrunleft'^{j_2}.\reachrun.\diffrunright'^{j_2}\diffrunright^{j_1}.\downseq^{j_1+j_2}\ .
\end{align*}

Unfortunately, the resulting run does not yet stay non-negative on \(\acounter\).
The problem is (as mentioned in the overview as Problem 4) that $\diffrunleft$ and $\diffrunright$ are not placed next to each other, but we first generate the copies of $\diffrunleft$ and later the copies of $\diffrunright$. 
There is no guarantee that $\upseq.\diffrunleft$ has a positive effect on the unbounded counters.

To overcome this problem, we use the wide tree theorem on the difference derivation. 
Then the negative effect contributed by the copies of the difference run will only grow logarithmically in the number of tokens produced by the pumping run.
To see that the wide tree theorem applies, note that~$\anngvas$ is non-linear, strongly connected, and only has useful non-terminals. 
We also observe that $\prodvec = \paramparikhof{\prods}{\aprodseqdiffanngvas}$ (resp. $\prodvec' = \paramparikhof{\prods}{\aprodseqdiffanngvas'}$) solves $\homeekof{\anngvas}{\prodvar}$. 
We can thus arrange $j_1$ copies of $\prodvec$ and $j_2$ copies of $\prodvec'$ in a parse tree $\atree$ of height $\ceilof{1+\ld (j_1+j_2)}\cdot \max\set{\normof{\prodvec}, \normof{\prodvec'}}$, for every $j_1+j_2\geq 1$. 
Let $\yieldof{\atree}=\asentform_{\aconst, 1}.\startnonterm.\asentform_{\aconst, 2}$. 
The theorem guarantees that for every prefix $\asentform$ of $\asentform_{\aconst, 1}.\asentform_{\aconst, 2}$, the number of incomplete copies of $\prodvec$ and $\prodvec'$ is bounded by $\ceilof{1+\ld (j_1+j_2)}$.
It remains to turn $\asentform_{\aconst, 1}.\asentform_{\aconst, 2}$ into a run $\arun_{\aconst, 1}.\arun_{\aconst, 2}$ using $\aprodseqdiffterms$ and $\aprodseqdiffterms'$. 
To be precise, we have one use of $\aprodseqdiffterms$ per copy of $\prodvec$ and similar for $\aprodseqdiffterms'$ and $\prodvec'$. 
For large \(\aconst\) a linear counter value beats logarithmically many incomplete copies, i.e.\ guarantees that the counter values remain high:

\begin{restatable}{lemma}{LemmaLowerBoundNonLinear}\label{Lemma:LowerBound}
There is $\iterlb\in\N$ so that for all $j_1+j_2\geq\iterlb$, for all prefixes $\arun$ of $\arun_{\aconst, 1}.\arun_{\aconst, 2}$, and for all $\acounter\in\adimset \setminus \omegaof{\acontextin}$ we have 
\begin{align*}
\at{\vaseffof{\upseq^{j_1+j_2}.\arun}}{\acounter}\ \geq\ \frac{1}{2} \cdot (j_1+j_2)\ .
\end{align*}
\end{restatable}

Lemma \ref{Lemma:LowerBound} shows that for all large enough \(\aconst\), the run does not go negative on \(\adimset\), and we have hence found our actual run.
\subsection{Linear Case}
We construct runs
\begin{align*}
\iterrundef\ &= \ \leftrundef.\reachruncenterleftdef.\reachruncenterrightdef.\rightrundef\\
\leftrundef\ &=\ \upseq^{j_1+j_2} \reachrunleft. \diffrunleft'^{j_2}.\diffrunleft^{j_1}.\downseqint^{j_1+j_2}\\
\rightrundef\ &=\ \upseqint^{j_1+j_2}.\diffrunright^{j_1}.\diffrunright'^{j_2} \reachrunright.\downseq^{j_1+j_2}\ .
\end{align*}
The reader familiar with Lambert's iteration lemma for VAS reachability will already see that we have created two pumping situations, on the left and right respectively. Furthermore, we have reused the $\aconst-\maxconst=j_1\cdot \sumconst + j_2\cdot (\sumconst+1)$ trick. Observe that \(\leftrundef\) and \(\rightrundef\) are seemingly mirrored, this is due to the runs otherwise not being derivable in the grammar.

What is new, not only compared to VAS reachability but also compared to the non-linear case, is that the reaching runs for the center $\reachruncenterleftdef$ and $\reachruncenterrightdef$ change with the iteration count~$\aconst$.
Here is why. 
To solve reachability, the number of updates in the overall run has to take the form $\at{(\asol+\abound\cdot \ahomsol)}{\updvar}$ for some $\abound\in\N$. 
If we only iterate the pumping and the difference derivation, however, we may not quite obtain a homogeneous solution. 
We may be missing repetitions of the period vectors for the children in the center. 
Indeed, in the homogeneous characteristic equations, the variables $\periodeffectvardefcenterleft$ and $\periodeffectvardefcenterright$ are not forced to be zero.
If they receive positive values in $\ahomsol$, the homogeneous solution expects pumping to happen in the children. 
By making $\reachruncenterleftdef$ and $\reachruncenterrightdef$ dependent on the number of iterations, we can incorporate this pumping. 

Why have we not seen the problem in the non-linear case? 
In the non-linear case, every production variable belongs to the support by \perfectnessprodsnospace, and hence receives a positive value in the (full) homogeneous solution. 
Since every terminal symbol occurs on the right-hand side of a production, every derivation corresponding to the homogeneous solution creates at least one instance of each terminal.
This in particular holds for the combination of pumping and difference derivation. 
Now, if the homogeneous solution expects pumping in a childNGVAS, we carry out this pumping in the new instance of the child. 
In the linear case, the center production cannot be repeated and we have not created a variable for it. 
Hence, the trick does not apply. 
The trick again works for the runs $\reachrunleft$ and $\reachrunright$.
We turn to the construction of the runs. 

We construct $\reachrunleft$ and $\reachrunright$ by first invoking Theorem~\ref{Theorem:EEK}.
We then have to produce runs for the children. 
For all instances of a child except one, we use \perfectnessbasenospace.
For the one instance that is left, we use a run that exists by the induction hypothesis.  
The details are like in the non-linear case except for one aspect.  
The runs $\reachrunleft, \reachruncenterleftzero, \reachruncenterrightzero, \reachrunright$ have to agree on the factor $\maxconst$ of how many homogeneous solutions to add into \(\asol\) to obtain \(\asol'\).
We achieve this by defining $\maxconst$ as the maximum over the $\initconstdef$ of the \(\anngvasp\) in \emph{all} directions.

For $\reachruncenterleftdef$ and $\reachruncenterrightdef$, we start from a childNGVAS and have to construct a reaching run. 
Constant~$\maxconst$ already allows us to invoke the induction hypothesis.
Rather than invoking it with $\maxconst$, however, we invoke it with $\aconst$.

To construct the difference runs, we first have to modify the requirements on the embedding constant so as to take into account the internal pumping sequence and the directions. 
Let $\periodeffectpumpdefdir$ and $\periodeffectpumpintdefdir$ with $\mydir\in\set{\myleft, \myright}$ be the numbers of period vectors produced by the pumping resp. the internal pumping sequence in the given direction.  
Embedding requires 
\begin{align*}
\at{(\embedconst\cdot \ahomsol)}{\prodvar} - \paramparikhof{\prods}{\aprodseqpumpanngvas.\aprodseqpumpintanngvas}\ &\geq\ 1\\
\at{(\embedconst\cdot \ahomsol)}{\periodeffectvardefdir} - (\periodeffectpumpdefdir+\periodeffectpumpintdefdir)\ &\geq\ 0\ .
\end{align*} 

For the construction of the difference runs, we start by calling Theorem~\ref{Theorem:EEK}. 
Note that this needs the first embedding requirement. 
To also construct runs for the children, we again combine \perfectnessbase with an invokation of the induction hypothesis for one instance per childNGVAS.
In this invokation,  when the child is $\anngvasp$ and the direction is $\mydir$, we use the base vector $\baseeffectdef + \at{(\embedconst\cdot \ahomsol)}{\periodeffectvardefdir} - (\periodeffectpumpdefdir+\periodeffectpumpintdefdir)$. 
Once the pumping sequences are added to the base vector, we have $\at{(\embedconst\cdot \ahomsol)}{\periodeffectvardefdir}$ repetitions of the period vectors. 
For enabledness, we determine a factor $\enableconst$, similar to the non-linear case, except it has to be common to both directions, and set $\sumconst=\embedconst+\enableconst$. 
Also, as in the non-linear case, we have to repeat this construction with $\sumconst'=\sumconst+1$.

\section{Decomposition}\label{Section:Decomposition}
\textbf{Recap}: We have seen that if an NGVAS is perfect, then the target is reachable. 
It remains to decompose a given (weak) NGVAS into a finite set of perfect NGVAS, by computing a decomposition of lower rank.

In this section, we add details for the notions of deconstruction, decomposition, which we skipped in the overview, and then define the procedures $\refinesub$,  $\refineeq$, and $\clean$.
The procedures $\refineintpump$, and $\refinepump$, require a lot of development, and are handled in \Cref{Section:PumpingMP}.
We give the intuition for how the refinement procedures satisfy the specifications given in \Cref{Proposition:Step}, \Cref{Lemma:RefineFunctions}, and \Cref{Lemma:CleannessGuaranteeOverview}.
The full proofs resp. constructions are involved, and are moved to \Cref{Appendix:Decompositions}.
\subsection{Preliminaries}\label{Section:DecompositionPrelims}
%
We formalize the structure of a decomposition wrt. the well-order $(\ranks, <)$. 
Our algorithms take one NGVAS as input, and produce a set of NGVAS.
We add details to the notions of deconstruction and refinement.
Formally, a set of NGVAS $\adecomp$ is a \emph{deconstruction} of $\anngvas$, if $\runsof{\anngvas}=\runsof{\adecomp}$, and for all $\anngvas'\in\adecomp$, (i) $\anngvas'.\unconstrained=\anngvas.\unconstrained$, (ii) $\anngvas'.\acontextin\sqsubseteq\anngvas.\acontextin$, (iii) $\anngvas'.\acontextout\sqsubseteq\anngvas.\acontextout$, (iv) $\anngvas'.\restrictions\subseteq\anngvas.\restrictions$.
We explain these conditions.
A deconstruction must preserve the enabled runs of an NGVAS. 
The same counters must be shielded from reachability constraints as stated by (i).
By (ii) and (iii), the context information of an NGVAS in the deconstruction must be a specialization of the context information for $\anngvas$.
Condition (iv) says that the new restrictions must be contained in the restrictions of $\anngvas$.
As we previously defined, a set of NGVAS $\adecomp$ is a \emph{decomposition} of $\anngvas$, if it is a deconstruction, and $\rankof{\anngvas'}<\rankof{\anngvas}$ holds for all $\anngvas'\in\adecomp$.


%

%
Recall that the system rank of an NGVAS is the maximum rank along all branches.
Thus, if $\adecomp$ is a decomposition of $\anngvasp$, then $\rankof{\anngvas\replace{\anngvasp}{\adecomp}}\leq\rankof{\anngvas}$.

%
%
%
%
%

%
\subsection{From Weak to Strong NGVAS}
We develop the necessary tools for dealing with wNGVAS.
Remember that a set \(\nonterms'\) of non-terminals is strongly-connected if for all \(\anonterm, \anonterm' \in \nonterms'\) there exists a derivation \(\anonterm \to^{*} \aword \anonterm' \awordp\) for sentence forms \(\aword, \awordp\). We write \(\sccpair{\anngvas, \anonterm}\) for the SCC of \(\anngvas\) containing \(\anonterm\).
%

We call an SCC $\sccpair{\anngvas, \anonterm}$ \emph{linear}, if $\anngvas$ contains no productions $\anontermp\to\anontermpp_{0}.\anontermpp_{1}$ with $\anontermp, \anontermpp_{0}, \anontermpp_{1}\in\sccpair{\anngvas, \anonterm}$.
Otherwise, we call it \emph{non-linear}. In case that a grammar is strongly-connected this coincides with the definition of a linear grammar.
To uniformly access the SCCs of non-terminals and top SCC of a child NGVAS, we also write $\sccpair{\anngvas, \anngvasp}$ to denote $\sccpair{\anngvasp, \anngvasp.\startnonterm}$ for a childNGVAS $\anngvasp$ of an NGVAS $\anngvas$.

To make working with SCCs cleaner, we define a few helper concepts.
First, we define a set of rules local to an SCC.
With slight abuse of notation, we write 
$$\sccpair{\anngvas, \anonterm}.\prods=\setcond{\anontermp\to\aword\in\anngvas.\prods}{\anontermp\in\sccpair{\anngvas, \anonterm}}$$
to denote the set of rules local to $\sccpair{\anngvas, \anonterm}$.

Second, we define the relation $\callscc$ between strongly connected components.
We write $\sccpair{\anngvas, \anonterm}\callscc\sccpair{\anngvas, \asymbol}$ if 
$\asymbol\in\anngvas.(\nonterms\cup\trms)\setminus\sccpair{\anngvas, \anonterm}$ can be derived from $\anonterm$ only using rules in \(\sccpair{\anngvas, \anonterm}.\prods\).
Note that if $\asymbol\in\trms$, then the right side of the relation is the top SCC of a childNGVAS.
In strong NGVAS, this relation corresponds to the child relation.
Indeed, for a strong NGVAS, $\sccpair{\anngvas, \anonterm}\callscc\sccpair{\anngvasp, \anontermp}$ iff $\anngvasp\in\anngvas.\trms$. 
We further this analogy by defining the set of local terminals
\begin{align*}
    \sccpair{\anngvas,\anonterm}.\trms=\setcond{\aword[i]}{&(\anngvasp\to\aword)\in\sccpair{\anngvas, \anonterm}.\prods,\; \\
    &\sccpair{\anngvas, \anonterm}\callscc\aword[i],\; i < \cardof{\aword}}.
\end{align*}
We also define a set of local recurring terminals
\begin{align*}
    \sccpair{\anngvas, \anonterm}.\rectrms=\setcond{\aword.\awordp[i]}{&\anonterm\to^{*}\aword.\anonterm.\awordp,\; \\
    &\aword.\awordp\in\sccpair{\anngvas, \anonterm}.\trms^{*},\; 
    i<\cardof{\aword.\awordp}}
\end{align*}
Note that both $\sccpair{\anngvas, \anonterm}.\trms$ and $\sccpair{\anngvas, \anonterm}.\rectrms$ are subsets of $\anngvas.(\nonterms\cup\trms)$.

In NGVAS, some production rules only produce terminals.
If the NGVAS is linear, this rule is unique, and can only be taken once.
To capture this behaviour in the weak case, we define $\exitpairsof{\anngvas, \anonterm}\subseteq\anngvas.\prods$ for an SCC $\sccpair{\anngvas, \anonterm}$ such that 
\begin{align*}
    \exitpairsof{\anngvas, \anonterm}=\setcond{\anontermp\to\asymbol.\asymbolp}{\anontermp\in\sccpair{\anngvas, \anonterm},\; \asymbol, \asymbolp\not\in\sccpair{\anngvas, \anonterm}}.
\end{align*}

\paragraph*{Weak NGVAS ranks} 
We extend the rank notion from \Cref{Section:Ranks} to wNGVAS.
To do this, we extend the individual rank-related notions.
We modify the definitions so that instead of working with NGVAS and the childNGVAS relation, they work with SCCs and the $\callscc$ relation.
The rank of a wNGVAS $\anngvas$ is the rank of $\sccpair{\anngvas, \anngvas.\startnonterm}$.
For the rest of these definitions, fix a wNGVAS $\anngvas$, and an SCC $\sccpair{\anngvas, \anonterm}$.

The structure of the rank is the same as the strong case
$$\rankof{\anngvas, \anonterm}=(\recrankof{\anngvas, \anonterm}, \itrankof{\anngvas, \anonterm})\in\N^{d+3}\times\N^{2d+1}.$$
It is made up of a non-linear rank $\recrankof{\anngvas, \anonterm}\in\N^{d+3}$ and a linear rank $\itrankof{\anngvas, \anonterm}\in\N^{2d+1}$.

First notion we extend is the cycle space.
Here, it becomes important which non-terminal the cycle is centered at.
This is because we can no longer assume all non-terminals can call one another.
The vector spaces $\cyclespaceof{\anngvas, \anonterm}$, $\lcyclespaceof{\anngvas, \anonterm}$ spanned by the effects resp. left-effects of $\anonterm$-centered cycles is
\begin{align*}
    \cyclespaceof{\anngvas, \anonterm}=\ &\spanof{\setcond{\amarking\in\ceffof{\acyc}}{\\ &\acyc\text{ is }\anonterm\text{-centered cycle in }\anngvas}}\subseteq\Z^{2d}\\
    \lcyclespaceof{\anngvas, \anonterm}=\ &\spanof{\setcond{\amarking\in\lceffof{\acyc}}{\\ &\acyc\text{ is }\anonterm\text{-centered cycle in }\anngvas}}\subseteq\Z^{d}.
\end{align*}

Now, we extend the local rank.
We let
$$\lrankof{\anngvas, \anonterm}=1_{\dimensionof{\lcyclespaceof{\anngvas, \anonterm}}}\cdot|\sccpair{\anngvas, \anonterm}|$$
if $\sccpair{\anngvas, \anonterm}$ is non-linear, and $\lrankof{\anngvas, \anonterm}=0$ if it is linear.

We move on to the branch rank.
To extend the branch rank to wNGVAS, we must first extend the branch notion to wNGVAS.
A \emph{wNGVAS branch} $\abranch$ is a sequence of SCC's $\abranch=\sccpair{\anngvas_0, \anonterm_0}\ldots\sccpair{\anngvas_k, \anonterm_k}$ linked by the $\callscc$ relation.
We say that $\abranch$ is a branch of the SCC $\sccpair{\anngvas, \anonterm}$, if $\abranch=\sccpair{\anngvas, \anonterm}.\abranch'$ for some branch $\abranch'$.
Same as the strong case, the rank of a branch $\abranch=\sccpair{\anngvas_{0},\anonterm_{0}}\ldots\sccpair{\anngvas_{k}, \anonterm_{k}}$ is the sum of the local ranks along it,  
$$\brankof{\abranch}=\sum_{i\leq k}\lrankof{\anngvas_{i}, \anonterm_{i}}.$$
The system rank is also defined as the maximum rank along all branches
$$\srankof{\anngvas, \anonterm}=\max_{\abranch\text{ a branch in }\sccpair{\anngvas, \anonterm}}\brankof{\abranch}.$$
Now, we extend the notion of the local index.
We define the local index $\localgrmindexof{\anngvas, \anonterm}$ to be the maximum of 
\begin{itemize}
    \item[(i)] 1,
    \item[(ii)] $\localgrmindexof{\anngvas, \asymbol}$ for all $\asymbol\in\sccpair{\anngvas, \anonterm}.\trms$ with $\srankof{\anngvas, \asymbol}=\srankof{\anngvas, \anonterm}$,
    \item[(iii)] $\localgrmindexof{\anngvas, \asymbol}+1$ for all $\asymbol\in\sccpair{\anngvas, \anonterm}.\rectrms$ with $\srankof{\anngvas, \asymbol}=\srankof{\anngvas, \anonterm}$,
    \item[(iv)] $\localgrmindexof{\anngvas, \asymbol}+1$ for all $(\anontermp\to\asymbol.\asymbolp)\in\exitpairsof{\anngvas, \anonterm}$ with $\srankof{\anngvas, \asymbol}=\srankof{\anngvas, \asymbolp}=\srankof{\anngvas, \anonterm}$ and $\localgrmindexof{\anngvas, \asymbol}=\localgrmindexof{\anngvas, \asymbolp}$.
\end{itemize}

The same components as before make up the non-linear rank
$$\recrankof{\anngvas, \anonterm}=(\cardof{\anngvas.\constrained}, \srankof{\anngvas, \anonterm}, \localgrmindexof{\anngvas, \anonterm}).$$
We adopt the definition of a main branch.
We say that the branch $\abranch=\sccpair{\anngvas_{0}, \anonterm_{0}}\ldots\sccpair{\anngvas_{k}, \anonterm_{k}}$ is a main branch, if $\srankof{\anngvas_{i}, \anonterm_{i}}=\srankof{\anngvas_{0}, \anonterm_{0}}$ and $\localgrmindexof{\anngvas_{i}, \anonterm_{i}}=\localgrmindexof{\anngvas_{0}, \anonterm_{0}}$ for all $i\leq k$.
In contrast to the strong case, weak NGVAS may have multiple main branches.
This is because one SCC may have multiple exit rules, even though they cannot be used in the same derivation tree.

For the linear rank, we extend the linear local- and branch- ranks.
We write 
\begin{align*}
    \linlrankof{\anngvas, \anonterm}&=1_{\dim\cyclespaceof{\anngvas, \anonterm}}\cdot\cardof{\sccpair{\anngvas, \anonterm}}\text{, if }\sccpair{\anngvas,\anonterm}\text{ linear, }\\
    \linlrankof{\anngvas, \anonterm}&=0\text{, else, }\\
    \linbrankof{\abranch}&=\sum_{i\leq k}\linlrankof{\anngvas_{i}, \anonterm_{i}}.
\end{align*}
for a branch $\abranch=\sccpair{\anngvas_{0}, \anonterm_{0}}\ldots\sccpair{\anngvas_{k}, \anonterm_{k}}$.
The linear rank is then the maximum rank among \emph{all main branches}.
$$\itrankof{\anngvas, \anonterm}=\max_{\abranch\text{ main branch}\text{ in }\sccpair{\anngvas, \anonterm}}\linbrankof{\abranch}.$$

\paragraph*{Reestablishing strongness}
Our refinement algorithms make use of the rich structure provided by strong NGVAS.
For this reason, reestablishing strongness is essential.
The procedure $\strongdec:\ngvas\to\powof{\ngvas}$ achieves this.
The key property of this procedure is that it preserves ranks.

\begin{lemma}\label{Lemma:StrongDecRank}
    Let $\anngvas$ be a wNGVAS.
    Then, $\strongdecof{\anngvas}$ is a deconstruction of $\anngvas$ consisting of strong NGVASes, and for all $\anngvas'\in\strongdecof{\anngvas}$, $\rankof{\anngvas'}\leq\rankof{\anngvas}$.
\end{lemma}

The procedure $\strongdec$ works on SCCs bottom up, and turns them into NGVAS without losing runs.
Let $\anngvas_{\anonterm}$ be the wNGVAS $\anngvas$, where we replace the starting symbol with $\anonterm\in\anngvas.\nonterms$, the in- and out-markings with $\inof{\anonterm}$ and $\outof{\anonterm}$, and the restrictions with $\Z^{d}$.
For $\anngvasp\in\anngvas.\trms$, let $\anngvas_{\anngvasp}=\anngvasp$.
The procedure replaces the rules with $\sccpair{\anngvas, \startnonterm}.\prods$.
Then it calls itself via $\strongdecof{\anngvas_{\asymbol}}$ for all $\asymbol\in\anngvas.(\nonterms\cup\trms)$ with $\sccpair{\anngvas, \startnonterm}\callscc\sccpair{\anngvas, \asymbol}$.
It replaces $\asymbol$ with the results of $\strongdecof{\anngvas_{\asymbol}}$ for each such $\asymbol$.
Let the result of this process be $\anngvas_{1}$.
If $\sccpair{\anngvas, \startnonterm}$ is non-linear, $\anngvas_{1}$ already a strong NGVAS thanks to the structure of wNGVAS, and we return $\anngvas_{1}$.
If it is linear, the result may not yet be strong for the following reason.
For $\anngvas_{1}$ to be a linear NGVAS, we would need $\cardof{\exitpairsof{\anngvas_{1}, \startnonterm}}=1$.
To solve this, for each $\aprod\in\exitpairsof{\anngvas_{1}, \startnonterm}$ we produce $\anngvas_{1, \aprod}$ by removing all rules in $\exitpairsof{\anngvas_{1}, \startnonterm}\setminus\set{\aprod}$.

Now, we show \Cref{Lemma:StrongDecRank}.
\begin{proof}
    First, we argue that neither of the steps lose runs.
    By standard structural induction, we can show that $\runsof{\asymbol}=\runsof{\anngvas_{\asymbol}}$ for each $\asymbol\in\anngvas.(\nonterms\cup\trms)$.
    This covers the first step.
    For the second step, consider that the derivation tree of a linear grammar may only have one branch with non-terminals.
    Thus, taking an exit rule $\aprod\in\exitpairsof{\anngvas_{1}, \startnonterm}$ concludes the derivation on the level of the current SCC.
    So, we get $\runsof{\anngvas_{1}}=\bigcup_{\aprod\in\exitpairsof{\anngvas_{1}, \startnonterm}}\runsof{\anngvas_{1, \aprod}}$.
    This concludes the argument for run preservation.

    Now we argue that the rank does not increase.
    To see this, it suffices to closely consider the definition of $\strongdec$, and the rank for NGVAS resp. wNGVAS.
    Each subNGVAS in $\anngvas'\in\strongdecof{\anngvas}$ corresponds to an SCC in $\anngvas$ with the same (linear) local rank.
    Furthermore, each branch in an NGVAS $\anngvas'\in\strongdecof{\anngvas}$ has a corresponding branch in $\anngvas$ in this sense.
    Finally, note that we only restrict the branching capabilities of the NGVAS.
    We also preserve the linearity resp. non-linearity of the SCCs.
    Then, the local index of an SCC may only decrease when moving to a corresponding subNGVAS.
    Thus, all elements that define the rank either remain the same, or decrease.
    We conclude $\rankof{\anngvas'}\leq\rankof{\anngvas}$ for all $\anngvas'\in\strongdecof{\anngvas}$.
\end{proof}

\subsection{A Closer Look at $\refinepar{\aprop}$ and $\clean$}\label{Section:TheDecompAlgorithm}
Now we take a closer look at the refinement steps and the cleaning process.
The detailed constructions, and correctness proofs have been moved to \Cref{Appendix:Decompositions}. We explain all perfectness conditions in order, i.e.\ we first explain the procedure $\clean$ followed by $\refinesub$ and $\refineeq$.
The refinement steps $\refineintpump$ and $\refinepump$ require a lot of development.
For this reason, we do not explain them here.
They are handled in their own section, \Cref{Section:PumpingMP}, dedicaded to all coverability arguments.

\paragraph*{Cleaning}
We proceed by a top-down approach.
First, we start with the main cleaning procedure, $\clean$, followed by a discussion of the subprocedures it relies on.
The call $\clean$ starts from an arbitrary wNGVAS $\anngvas$.
It assumes that $\perfect$ is reliable up to $\rankof{\anngvas}$.
As its first step, the procedure calls $\strongdecof{\anngvas}$.
Following this, the procedure works bottom-up in principle, first establishing the conditions related to childNGVAS, then moving on to the parent NGVAS.
However, modifying childNGVASes may break conditions relating to the parent, and vice-versa.
For this reason, the procedure $\clean$ applies the cleaning step $\cleanstep$ until it reaches a fixed point.
\begin{align*}
    \clean=\strongdec.\cleanstep^{*}
\end{align*}
The pre- and post-conditions of $\cleanstep$ are similar to that of $\perfect$.
For an input NGVAS $\anngvas$, it expects the reliability of $\perfect$ up to $\rankof{\anngvas}$.
As its post-condition, it guarantees that the output is a deconstruction, and that rank of the output NGVAS are smaller, $\cleanstepof{\anngvas}<\rankof{\anngvas}$, or that $\cleanstepof{\anngvas}$ is a set of clean NGVASes.
It returns the input if the input is a clean NGVAS.
This means that at the fixed-point, the result is a set of clean NGVASes.

In the following, we go into the individual subprocedures of $\cleanstep$.
The procedure behaves differently depending on whether its input is linear NGVAS.
The broad strokes of the cases are the same, but the linear case requires additional care of the main branch.
For this reason, we start with the non-linear case.

\textbf{Cleaning steps for non-linear NGVAS}.

\textbf{Step 1: Rigid Counters}. As its first step, $\cleanstep$ determines the set of counters $\acounterset\subseteq [d]$ that are rigid.
This yields a pair of markings $\amarking_{\asymbol, in}, \amarking_{\asymbol, out}\in\N^{\acounterset}\times\Nomega^{[d]\setminus \acounterset}$ for each symbol $\asymbol\in\anngvas.(\nonterms\cup\trms)$, such that in any run $\arun\in\runsof{\anngvas}$, the part of the run $\arun_{\asymbol}$ derived from $\asymbol$, starts from a counter valuation $\amarkingp_{in, \asymbol}\sqsubseteq\amarking_{in, \asymbol}$, and ends at a counter valuation $\amarkingp_{out, \asymbol}\sqsubseteq\amarking_{out, \asymbol}$.
We refer the reader to \Cref{Section:DecompositionLevel3} for the details.
If the set of counters $\acounterset$ is not concretely tracked in the boundedness information, the procedure first removes the symbols whose assigned markings are negative.
Then it specializes the marking of each terminal $\asymbol\in\trms$ with the markings $\amarking_{\asymbol, in}$ and $\amarking_{\asymbol, out}$.
Thanks to the previous observation, this does not lose any runs.
If we have removed a non-terminal in this step, we return the result.
In this case, the returned wNGVAS is of lower recursive rank.
This is because we did not modify any other part of the NGVAS, and removed one non-terminal.

\textbf{Step 2: Perfect Children}. As its next step, $\cleanstep$ makes the (modified) children $\anngvasp$ perfect, and removes them if $\runsof{\anngvasp}=\emptyset$.
To make the children perfect, $\cleanstep$ first calls $\clean$ and then calls $\perfect$, since \(\perfect\) expects a clean NGVAS as input.
This removes children with $\runsof{\anngvasp}=\emptyset$ since their perfect deconstruction is empty.

\textbf{Step 3: Base Effects Enabled}.
To establish that their base effects are enabled, the procedure uses a subroutine called $\basisfire$. 
This subroutine  expects a perfect NGVAS, and is therefore the last step in the process.
The procedure $\basisfire$ starts from a perfect NGVAS $\anngvas$, and assumes the reliability of $\perfect$ for and up to $\rankof{\anngvas}$.
It constructs the set of minimal $\amarkingpp\in\N^{\periodeffect}$, such that the effect $\baseeffect+\periodeffect\cdot\amarkingpp$ is the effect of an enabled run.
Because $\N^{\periodeffect}$ is a well quasi order, this set is guaranteed to be finite.
Finally, it returns the set of $\anngvas\minrestrto{\amarkingpp}$ for such minimal applications $\amarkingpp$.
Here, $\anngvas\minrestrto{\amarkingpp}$ is an NGVAS identical to $\anngvas$ up its base effect, and the base effect of $\anngvas\minrestrto{\amarkingpp}$ is obtained by adding $\amarkingpp[\aperiod]$ applications of $\aperiod\in\periodeffect$ to the base effect of $\anngvas$.
Clearly, this set of NGVAS has the same set of runs, and the base effect of each $\anngvas\minrestrto{\amarkingpp}$ is enabled.
To collect the set of minimal enabled applications, $\basisfire$ uses a recursive procedure that calls $\perfect$ on versions of $\anngvas$ with more and more restricted period sets.
Here, the reliability of $\perfect$ for $\anngvas$ is necessary.

As an example in \(\N^2\), we first find \emph{any} run, say corresponding to point \((2,3)\). Next we recursively use \(\perffun\) to check emptiness of the set of runs corresponding to the lines/set of points \((1,\N), (0,\N), (\N,2), (\N, 1)\) and \((\N, 0)\), i.e.\ for the decomposition of the complement of \((2,3)\uparrow\).  For each of these five calls we either immediately find emptiness, or we find a point, say \((1,7)\) for the first query. Now it only remains to check reachability for \((1,0)\) to \((1,6)\) (finitely many points), and accordingly for the other four lower calls. Then we can simply collect the results to end up with the minimal points. In particular, even though the first point \((2,3)\) might not have been minimal, we still eventually end up with all minimal points.

\textbf{Step 4: Remove Useless Non-terminals}. This is a standard construction for context-free grammars, and we do not repeat it here.

We remark that the order of these steps is important.
A non-terminal may lose usefulness, if all terminals $\anngvasp$ it can produce are found to have $\runsof{\anngvasp}=\emptyset$.
Marking rigid counters takes place before the refinement, since introducing a concrete marking may disturb perfectness.


\textbf{Step 5: Conditions \perfectnesscountersnospace, and \perfectnesssolnospace}.
These require that $\omega$-marked counters are in the support, and that the solution space of the characteristic equation corresponds to the restrictions.
To establish them, $\cleanstep$ concretizes the $\omega$-markings that are not in the support, and updates the solution space.
The procedure calls $\strongdec$ on the results, and returns.

So far, it may be unclear that this procedure fulfills the pre- and post-conditions we discussed above.
There are three key insights.
First, if removing the rules leading to non-useful (non-)terminals results in a linear component, the recursive rank of the NGVAS decreases.
To see this, let $\anngvas'$ be an NGVAS we get after applying the steps.
No branch $\anngvas'_{0}\ldots\anngvas'_{k}$ can have two subNGVASes that correspond to the same SCC in $\anngvas$, and at least one SCC is missing or linear.
Since we did not introduce new cycles, $\dimensionof{\lcyclespaceof{\anngvas'_{i}}}\leq \dimensionof{\lcyclespaceof{\anngvas}}$ for all $i\leq k$.
This means that each non-terminal in $\anngvas'$ (or related subNGVAS) counts at most as much as a non-terminal in $\anngvas$. 
Since $\lrankof{\anngvas_{i}'}=0$ for a linear $\anngvas_{i}'$, at least one non-terminal is not counted in the $\srankof{\anngvas'}$ sum, leading to a lesser rank.
Second, if the results are not clean, then it must be that (i) a further counter $j$ has become rigid in a subNGVAS or (ii) a subNGVAS is not perfect.
Assume (i) holds.
Then, for a subNGVAS $\anngvas'_{i}$, we have $\lcyclespaceof{\anngvas'_{i}}\subseteq\lcyclespaceof{\anngvas}\cap \set{0}\times\Z^{[d]\setminus j}$.
Furthermore, it must be that there are cycles in $\anngvas$ with non-zero left-effect on $j$.
This is because if $j$ is fixed in $\anngvas$, then $\acontextin[j]=\acontextout[j]=\omega$ holds, and thus $\anngvas'.\acontextin[j]=\anngvas'.\acontextout[j]=0$, contradicting rigidness.
If $j$ is not fixed in $\anngvas$, we have $\lcyclespaceof{\anngvas'_{i}}\subsetneq\lcyclespaceof{\anngvas}$ and thus $\dimensionof{\lcyclespaceof{\anngvas'_{i}}}<\dimensionof{\lcyclespaceof{\anngvas}}$.
Then, similarly to the reasoning before, one non-terminal must count less in the rank calculation, so the rank decreases.
We argue that (ii) cannot hold.
The deciding condition is \perfectnesspumpingnospace, the existence of a pumping derivation.
This holds trivially for the newly introduced subNGVAS, since the structure of the NGVAS ensures that the markings and the reachability information agree on the concrete components.

\textbf{Cleaning for linear NGVAS}.
The fundamentals of the cleaning step are the same between the linear and non-linear cases.
We only highlight the key differences here, the details can be found in the last part of this paper.
There are two differences to the non-linear case.
First, we need to consider the entire main branch $\sccpair{\anngvas_{0}, \anonterm_{0}}\ldots \sccpair{\anngvas_{k}, \anonterm_{k}}$ in order to remain within our recursion scheme. An important concept in this regard are \emph{siblings}: In a normal VASS, leaving an SCC only enters a single follow-up SCC, but in an NGVAS an exit rule may produce two childNGVAS. In such a case the \emph{sibling} of $\anngvas_{i+1}$ is the NGVAS $\anngvasp$ that is produced by the same rule in $\anngvas_{i}$.
Second, we have an additional cleanness condition \perfectnesschildperiodsbd that we need to care about.

The procedure works as follows: First it extracts the main branch.
The following step is to concretize the rigid-but-not-tracked counters.
Similarly to the non-linear case, this is followed by cleaning, perfecting, and enabling the subNGVAS that are not along the main branch.
Note that these calls are allowed by our recursion scheme, the children outside of the main branch have strictly lesser recursive rank.
Then, we remove non-useful (non-)terminals.
In analogy to the non-linear case, the next step is establishing \perfectnesssol and \perfectnesscountersnospace.
We also need to ensure \perfectnesschildperiodsbdnospace.
The steps towards \perfectnesssol and \perfectnesscounters are the same as in the non-linear case, but we apply them to each NGVAS along the main branch.
To establish \perfectnesschildperiodsbdnospace, we restrict the bounded periods of the NGVAS along the main branch, and their siblings.
Unfortunately, establishing \perfectnesschildperiodsbd for the parent may break \perfectnesscounters and \perfectnesschildperiodsbd for the child.
This is because to establish \perfectnesschildperiodsbdnospace, we restrict the periods of the child, which changes the homogenous equation system.
Furthermore, restricting the periods or counters of the sibling of a main branch NGVAS may break its perfectness.
For these reasons, this step does the following until a fixed point is reached: (i) concretize the counters and restrict the period vectors of the main branch NGVASes and their siblings, (ii) clean, perfect, and enable the siblings of the main branch NGVAS.
This loop terminates, since the number of counters are bounded, and each restriction of the period vectors decrease their dimensionality.
The latter follows from the same argument as \cite[Claim 4.7]{LerouxS19}.
The procedure $\cleanstep$ applies $\strongdec$ to the results, and returns.
The arguments for correctness wrt. pre- and post- conditions are the same.

\textbf{Refinement for \perfectnesschildrennospace}.
The procedure $\refinesub$ expects a clean NGVAS, and decomposes when \perfectnesschildren does not hold.
Note that the cleanness condition \perfectnessrecchildren readily ensures that the children of the input $\anngvas$ are perfect, unless they lie on $\mainbranch$.
Since $\mainbranchof{\anngvas}=\anngvas$ for non-linear $\anngvas$ this condition is only relevant for linear NGVAS.
The procedure could naively ensure \perfectnesschildren by calling $\perfect$ on all children, and replace them with the results.
However, this approach would call $\perfect$ on a child along $\mainbranch$, which has the same recursive rank as the parent.
Thus, this would not fit our recursion scheme.
To overcome this problem, $\refinesubof{\anngvas}$ first extracts the main branch $\mainbranchof{\anngvas}=\anngvas_{0}.\anngvas_{1}.\anngvas_{2}\ldots\anngvas_{k}$.
We remark that all NGVAS on $\abranch$ are linear.
Starting from the deepest NGVAS $\anngvas_{k}$, it uses the refinement functions $\refineeq$, $\refineintpump$, and $\refinepump$.
If these output a decomposition, inserting the decomposition to $\anngvas$ reduces the overall iterative rank, so we return this result.
If not, the NGVAS is perfect, so we proceed to the parent NGVAS and iterate the process once more.
If we establish that $\anngvas_{1}$ is perfect, we return $\anngvas$, since this means \perfectnesschildren holds.

\textbf{Refinement for \perfectnessprodsnospace}.
The procedure $\refineeq$ expects a clean NGVAS, and decomposes it when \perfectnessprods does not hold.
This means that there is a bound $b\in\N$, and sets of bounded productions $\prods_{\mathsf{bd}}$ resp. bounded child periods $\periodeffect_{\mathsf{bd}}$ that can only be taken at most $b$ times in runs of $\anngvas$.
In line with classical VASS reachability \cite{Kosaraju82,Lambert92,LerouxS19}, $\refineeq$ extends the grammar of $\anngvas$ by adding a budget component $c\in[0, \ldots, b]^{\prods_{\mathsf{bd}}\cup\periodeffect_{\mathsf{bd}}}$.

We ensure that the productions in $\prods_{\mathsf{bd}}$ and child periods in $\periodeffect_{\mathsf{bd}}$ are taken at most $b$ times by different mechanisms.
For the bounded rules, we decrement the coordinate $\aprod\in\prods_{\mathsf{bd}}$ in the counter whenever $\aprod$ is taken.
Because of the grammatical nature of our setting, the productions also partition the budget of the consumed non-terminal.
Given a production $\anonterm\to\asymbol.\asymbolp\not\in\prods_{\mathsf{bd}}$, we have the rules
$$(\anonterm, c)\to(\asymbol, c_0).(\asymbolp, c_1)\qquad \text{for all }c_0+c_1=c.$$
If $\aprod=\anonterm\to\asymbol.\asymbolp\in\prods_{\mathsf{bd}}$, we would instead require the sum to equal $c-1_{\aprod}$. 
We control the bounded applications of a child period by hard coding them into the childNGVAS.
The pair $(\anngvasp, c)$ for a terminal $\anngvasp$ refers to a version of $\anngvasp$, where the bounded period vectors can be taken exactly as often as prescribed by $c$.

We remark that it would be easier to only deal with a single bounded object at once, but similar to the corresponding decomposition for VASS \cite{LerouxS19}, this would not decrease the rank. 
\emph{Every} bounded production and period has to be removed, then we can argue as follows:
The effect of any cycle $(\anonterm, c)\to^{*}\aword.(\anonterm, c).\awordp$ cannot contain the effect of bounded components, because it must keep the budget constant.
Then, in the non-linear case, the dimension of the vector space spanned by the cycle effects must be less than that of the initial NGVAS, by the same argument as in \cite[Claim 4.7]{LerouxS19}.
There is one caveat to this argument in the linear case.
It may be that only the periods of the center children are bounded.
After applying our construction, the resulting NGVAS would have exactly the same cycles as the input.
We avoid this problem thanks to the cleanness condition \perfectnesschildperiodsbdnospace.
It states that the periods of the center children are unbounded in $\chareq{}$ of the parent.
Thus, if all other components are already unbounded, \perfectnessprods holds.

\section{Pumping}\label{Section:PumpingMP}
Up to this point, we have covered the refinement step $\refineeq$ and the process of cleaning.
We now discuss the refinement steps $\refinepump$ and $\refineintpump$.
The discussion in this section should be understood as a continuation of the discussion in \Cref{Section:OutlinePumping}.
We make our statements and definitions precise, cover the arguments in more detail, but to keep the presentation clean, we have moved the more involved constructions to the appendix.
We fix an NGVAS $\anngvas$ that has all perfectness conditions except \perfectnesspumpingnospace\ (resp. \perfectnesspumpingnospace\ and \perfectnesspumpingintnospace\ in the linear case).
This is the point from which both $\refinepump$ and $\refineintpump$ start.

We not only have to consider $\anngvas$, but versions of it where the start symbol and the context information have been modified.
For $\amarking, \amarkingp\in\Nomega^{d}$ and $\anonterm\in\nonterms$ with $\amarking\sqsubseteq\inof{\anonterm}$ and $\amarkingp\sqsubseteq\outof{\anonterm}$, we use $\otherctxNGVAS{\anngvas}{\amarking, \anonterm, \amarkingp}$ to denote the NGVAS 
\begin{align*}
    \otherctxNGVAS{\anngvas}{\amarking, \anonterm, \amarkingp}=\ &(\agram_{\anonterm}, (\amarking, \amarkingp), \Z^{d}, \omegaof{\amarking}\cap\omegaof{\amarkingp}, \abdinfo)\ .
\end{align*} 
The grammar $\agram_{\anonterm}=(\nonterms, \trms, \anonterm, \prods)$ starts from symbol~$\anonterm$ instead of $\startnonterm$.
The restrictions have been removed, all effects from $\Z^{d}$ are allowed.
If $i\in\omegaof{\amarking}\cap\omegaof{\amarkingp}$ holds, the counter $i$ is unconstrained in $\otherctxNGVAS{\anngvas}{\amarking, \anonterm, \amarkingp}$.
Since we kept the grammar, the boundedness information also remains intact.
We refer to the NGVAS $\otherctxNGVAS{\anngvas}{\amarking, \anonterm, \amarkingp}$ as \emph{variants} of $\anngvas$.

\subsection{The Karp-Miller Construction and Linear NGVAS}\label{Section:KarpMillerMP}

Recall the Karp-Miller construction from \Cref{Section:OutlinePumpingDec} that reduces the task of deciding \perfectnesspumpingnospace\ to computing $\postfunc$ and $\prefunc$.
We now strengthen this argument to include the variants of $\anngvas$, and observe that we do not need to be able to compute the entire domain of $\postfunc$ and $\prefunc$.
The relevant fragment of the domain suffices.
For
\begin{align*}
    \omegamrkdomainof{\acounterset}&=\setcond{\amarking\in\Nomega^{d}}{\acounterset\subseteq\omegaof{\amarking}\subseteq\abdinfomid},
\end{align*} 
the following holds.
\begin{restatable}{lemma}{LemmaKarpMillerCompMainPaper}\label{Lemma:ExtPumpComp}
    Let $\anngvas$ be non-linear and let $\otherctxNGVAS{\anngvas}{\amarking, \anonterm, \amarkingp}$ be a variant.
    Assume $\omegamrkdomainof{\omegaof{\amarking}}\times(\nonterms\cup\trms)$ is computable for the domain $\infuncdomain$ and $\prefuncN{\anngvas}$ for $\omegamrkdomainof{\omegaof{\amarkingp}}\times(\nonterms\cup\trms)$.
    Then the Karp-Miller tree $\extkmgraph$ can be effectively constructed, and we can decide \perfectnesspumpingnospace.
\end{restatable}
We consider the linear case in more detail. 
The first observation is that we can further restrict the above domain.  
Recall that for each rule $\anonterm\to\anontermp_0.\anontermp_1$, the Karp-Miller construction choses a non-terminal branch to explore, and abstracts the remaining symbol with $\postfunc$ and $\prefunc$ calls.
But in the linear case, each rule has at most one non-terminal.
This means we only need to compute $\postfunc$ and $\prefunc$ for terminals.
By \Cref{Lemma:NonLinearChildRank}, these have lesser non-linear rank, and we can handle them via $\perfect$ calls.

The second observation is that, in the linear case, the Karp-Miller tree not only allows us to check  $\perfectnesspumpingnospace$, it also allows us to compute a decomposition in case $\perfectnesspumpingnospace$ fails. 
Indeed, if the tree does not contain the non-terminal $(\inof{\startnonterm}, \startnonterm, \outof{\startnonterm})$, then for every derivation there is at least one counter where the run from $\acontextin$ resp. $\acontextout$ does not achieve the $\omega$-entry required by $\inof{\startnonterm}$ resp. $\outof{\startnonterm}$. 
The tree then gives us a bound that holds for all these cases.  
We incorporate this bound into the boundedness information to obtain the desired decomposition.

The final observation is that the Karp-Miller tree also allows us to compute $\refineintpump$. 
The argument is the same as for \perfectnesspumpingnospace, only the direction changes.
Instead of simulating pumps from the outside to the inside, starting at $(\acontextin, \startnonterm, \acontextout)$, we simulate pumps from the inside to the outside, starting at $(\anngvasp.\acontextin, \anontermp, \anngvasp'.\acontextout)$, where $\anontermp\to\anngvasp.\anngvasp'$ is the exit production of $\anngvas$.
We use the shorthands
$$\inmarkingdomain=\omegamrkdomainof{\omegaof{\acontextin}}\quad\text{ and }\quad\outmarkingdomain=\omegamrkdomainof{\omegaof{\acontextout}}$$
in the lemma below.

\begin{restatable}{lemma}{LemmaExternalPumpCompLinMainPaper}\label{Lemma:ExtPumpCompLin}
    Let $\anngvas$ be linear.
    Assume $\postfuncN{\anngvas}$ is computable for $\inmarkingdomain\times\trms$ and $\prefuncN{\anngvas}$ for $\outmarkingdomain\times\trms$.
    Then the Karp-Miller tree $\extkmgraph$ can be effectively constructed, and we can compute $\refinepump$ and $\refineintpump$.
\end{restatable}
To compute a decomposition in the non-linear case, we need a finer construction.
\subsection{The Coverability Grammar}\label{Section:CoverabilityGrammarMP}
The coverability grammar is inspired by the Karp-Miller tree construction, but has a novel component that we call \emph{promises}.
Each non-terminal makes a promise for the input and output values it will produce. 
The actual derivation under the non-terminal must produce a value that specializes the promise.
This structure ensures that the boundedness information we compute is compatible with the structure of an NGVAS.

Recall that coverability grammars are defined relative to functions that approximate reachability forwards resp. backwards.
We now make the requirements on these functions explicit.
\paragraph*{Approximation.}
%
We focus on post-approximators, functions that over-approximate the values obtainable on the output for a given input. 
The requirements on so-called pre-approximators are similar, and can be found in the appendix.  
A \emph{post-approximator for $\anngvas$} is a function $\postapprox:\Nomega^{d}\times(\nonterms\cup\trms)\to\powof{\Nomega^{d}}$ that always outputs a finite set, and has the following three properties.  
Let the input be $(\amarking, \asymbol)\in\Nomega^{d}\times(\nonterms\cup\trms)$.   
The first property is the soundness wrt. information about concrete values and unboundedness: for every output $\amarkingp\in\postapproxof{\amarking,\asymbol}$, we have  
$\amarkingp\sqsubseteq\outof{\asymbol}$ and  $\omegaof{\amarking}\subseteq\omegaof{\amarkingp}$. 
The second is the over-approximation of the reachability relation: for every run $(\amarking', \arun, \amarkingp')\in\runsof{\asymbol}$ with $\amarking'\sqsubseteq\amarking$, there is $\amarkingp\in\postapproxof{\amarking, \asymbol}$ with $\amarkingp'\sqsubseteq\amarkingp$.
The last property is the compatibility with the derivation relation: for every $\amarkingp'\in \postapproxof{\amarking, \asentform}$ with $\asymbol\to\asentform$, there is $\amarkingp\in\postapproxof{\amarking, \asymbol}$ with $\amarkingp'\sqsubseteq\amarkingp$.
Here, we generalize $ \postapprox$ to sentential forms by defining $\postapproxof{\amarking, \asentform.\asymbolp}=\postapproxof{\postapproxof{\amarking, \asentform}, \asymbolp}$.

The functions $\postfunc$ and $\prefunc$ from \Cref{Section:OutlinePumping} fail to be post- resp. pre-approximators.  
The problem is that the ideal decomposition may drop elements that we need to approxiate wrt. $\sqsubseteq$.
To turn $\postfunc$ and $\prefunc$ into actual actual post- resp. pre-approximators, we simply add all smaller concrete values: 
%
%
%
\begin{align*}
    \natpostapproxof{\amarking, \asymbol}=\setcond{&\amarkingp\in\Nomega^{d}}{ \exists\amarkingp'\in\postfuncNof{\anngvas}{\amarking, \asymbol}.\\
    &\;\amarkingp\leq\amarkingp',\;\omegaof{\amarkingp}=\omegaof{\amarkingp'},\; \amarkingp\sqsubseteq\outof{\asymbol}}\\
    \natpreapproxof{\amarkingp, \asymbol}=\setcond{&\amarking\in\Nomega^{d}}{\exists\amarking'\in\prefuncNof{\anngvas}{\amarkingp, \asymbol}.\\
    &\;\amarking\leq\amarking',\;\omegaof{\amarking}=\omegaof{\amarking'},\; \amarking\sqsubseteq\inof{\anonterm}}\ .
\end{align*}
%
%
\begin{restatable}{lemma}{LemmaNatPostPreApproximatorsMainPaper}\label{Lemma:NatPostPreApproximators}
    $\natpostapprox$, $\natpreapprox$ are post-/pre-approximators.
\end{restatable}

\paragraph*{Details of the Coverability Grammar.}
Consider the post- and pre-approximations $\postapprox$ and $\preapprox$, and a variant $\otherctxNGVAS{\anngvas}{\amarkingpp, \anontermp, \amarkingppp}$ of $\anngvas$.
The \emph{coverability grammar} $\wtgrammarof{\otherctxNGVAS{\anngvas}{\amarkingpp, \anonterm, \amarkingppp}, \postapprox, \preapprox}=(\wtnonterms, \wtterms, \wtstartnonterm, \wtprods)$ is a context-free grammar with non-terminals $\wtnonterminals\subseteq\Nomega^{d}\times\Nomega^{d}\times\nonterms\times\Nomega^{d}\times\Nomega^{d}\times\N$, terminals $\wtalph\subseteq\Nomega^{d}\times\Nomega^{d}\times\trms\times\Nomega^{d}\times\Nomega^{d}\times\N$, start symbol $\wtstartnonterm=((\amarkingpp, \outof{\anontermp}), \anontermp, (\inof{\anontermp}, \amarkingppp), 0)$, and production rules $\wtprods=\wtprodsinit\uplus\wtprodssim\uplus\wtprodspump$.

A symbol $((\amarking, \aprommarking), \asymbol, (\aprommarkingp, \amarkingp), \ahist)$ should be read as the conjuction of two statements: $\asymbol$ takes $\amarking$ to $\aprommarking$ forwards (with $\postapprox$), and $\amarkingp$ to $\aprommarkingp$ backwards (with $\preapprox$).

Non-terminals and terminals have a last component from $\N$ 
that keeps information about the history of the derivation.
The goal of this history information is to make sure the creation relation between non-terminals forms a tree with backlinks.
This means a non-terminal $\anonterm$ may create a non-terminal that is fresh or that can be found on the path from the root to $\anonterm$.

We construct the grammar in an iterative fashion. 
To this end, we maintain a set of non-terminals whose productions have not yet been defined, and terminate when this set runs empty.
We omit the maintenance of this set in our description. 
The construction starts from the grammar $(\wtnonterminalsinit, \emptyset, \wtstartnonterm, \wtprodsinit)$.
The productions $\wtprodsinit$ introduce the given $\amarkingpp$, $\anontermp$, and $\amarkingppp$, and have the form $\wtstartnonterm\to ((\amarkingpp, \aprommarking'), \anontermp, (\aprommarkingp', \amarkingppp), 0)$ for all $\aprommarking'\in\postapproxof{\amarkingpp, \anontermp}$ and $\aprommarkingp'\in\preapproxof{\amarkingppp, \anontermp}$. 
The set $\wtnonterminalsinit$ consists of all non-terminals we encounter in this construction. 
To explain the iteration step, let $\wtgrammar'=(\wtnonterms', \wtterms', \wtstartnonterm', \wtprods')$ be the current grammar and $\awtnonterm=((\amarking, \aprommarking), \anonterm, (\aprommarkingp, \amarkingp), \ahist)\in\wtnonterms$ the non-terminal waiting for productions. 

The productions  in $\wtprodspump$ capture the vertical pumping situations for the non-terminal $\awtnonterm$ and accelerate the input and output markings accordingly. 
A vertical pumping situation is characterized by another non-terminal $((\amarking', \aprommarking'), \anonterm, (\aprommarkingp', \amarkingp'), \ahist')\in\wtnonterms'$ that can (transitively) create $\awtnonterm$ and satisfies $(\amarking, \amarkingp)>(\amarking', \amarkingp')$.
We then have the production $((\amarking, \aprommarking), \anonterm, (\aprommarkingp, \amarkingp), \ahist)\to ((\amarking_{\omega}, \aprommarking''), \anonterm, (\aprommarkingp'',\amarkingp_{\omega}), \ahist'')$, for each $\aprommarking''\in\postapproxof{\amarking_{\omega}, \anonterm}$ and $\aprommarkingp''\in\preapproxof{\amarkingp_{\omega}, \anonterm}$. 
The accelerated markings $\amarking_{\omega}, \amarkingp_{\omega}\in\Nomega^{d}$ are defined as expected: $(\amarking_{\omega}, \amarkingp_{\omega})[i]=\omega$ if $(\amarking, \amarkingp)[i]>(\amarking', \amarkingp')[i]$, else $(\amarking_{\omega}, \amarkingp_{\omega})[i]=(\amarking, \amarkingp)[i]$.
All non-terminals created by the productions are fresh, meaning $\ahist''$ is new and unique.

The productions in the set $\wtprodssim$ simulate the original productions, including the effect on the counters in both directions.
Formally, for every production $\anonterm\to\asymbol.\asymbolp$ in $\anngvas$, for all  $\aprommarking'\in\postapproxof{\amarking, \asymbol}$ and $\aprommarking''\in\postapproxof{\aprommarking', \asymbolp}$, and for all 
$\aprommarkingp'\in\preapproxof{\amarkingp, \asymbolp}$ and 
$\aprommarkingp''\in\preapproxof{\aprommarkingp', \asymbol}$, where $\aprommarking''\sqsubseteq\aprommarking$, $\aprommarkingp''\sqsubseteq\aprommarkingp$, and $\aprommarking'\compwith\aprommarkingp'$, we have the production 
\begin{align*}
    \awtnonterm\to ((\amarking, \aprommarking'), \asymbol, (\aprommarkingp'', \aprommarkingp'), \ahist').((\aprommarking', \aprommarking''), \asymbolp, (\aprommarkingp', \amarkingp), \ahist''). 
\end{align*}
Recall that the creation relation should form a tree with backlinks.
This means a non-terminal on the right-hand side is made fresh (using $\ahist'$ resp. $\ahist'')$, if there is no non-terminal in $\wtnonterms'$ that (i) can (transitively) create $\awtnonterm$ and (ii) has the same first five components.
If such a non-terminal exists, the rule uses it, thereby introducing a backlink in the creation relation.

If the approximators are computable for the domains determined by the input resp. output markings, then the construction terminates.
\begin{restatable}{lemma}{LemmaCGTerminationMainPaper}\label{Lemma:CGTermination}
    Let $\otherctxNGVAS{\anngvas}{\amarking, \anonterm, \amarkingp}$ be a variant of $\anngvas$.
    Let the post- and pre-approximations $\postapprox$, $\preapprox$ be computable for the domains $\omegamrkdomainof{\omegaof{\amarking}}\times(\nonterms\cup\trms)$ and let $\prefuncN{\anngvas}$ be computable for the domain $\omegamrkdomainof{\omegaof{\amarkingp}}\times(\nonterms\cup\trms)$.
    Then, the construction of $\wtgrammarof{\otherctxNGVAS{\anngvas}{\amarking, \anonterm, \amarkingp}, \postapprox, \preapprox}$ terminates.
\end{restatable}

If a coverability grammar $\wtgrammarof{\otherctxNGVAS{\anngvas}{\amarking, \anonterm, \amarkingp}, \postapprox, \preapprox}$ does not contain a non-terminal of the form $(\inof{\anontermpp}, \anontermpp, \outof{\anontermpp})$, we say that it \emph{remains bounded}.
If instead, there is such a non-terminal, we say that the grammar \emph{shows unboundedness}.
We observe that, if the grammar shows unboundedness for complete approximators $\natpostapprox$ and $\natpreapprox$, then \perfectnesspumpingnospace\ holds.
This is similar to the case of the Karp-Miller graph. 
\begin{restatable}{lemma}{LemmaCGPumpabilityMainPaper}\label{Lemma:CGPumpability}
    If $\wtgrammarof{\otherctxNGVAS{\anngvas}{\amarking, \anonterm, \amarkingp}, \natpostapprox, \natpreapprox}$ shows unboundedness, then $\otherctxNGVAS{\anngvas}{\amarking, \anonterm, \amarkingp}$ fulfills \perfectnesspumpingnospace.
\end{restatable}
If the coverability grammar $\wtgrammarof{\otherctxNGVAS{\anngvas}{\amarking, \anonterm, \amarkingp}}$ remains bounded instead, and we have $(\amarking, \amarkingp)\in\inmarkingdomain\times\outmarkingdomain$, \emph{regardless of the approximation}, we get a decomposition.
This is a crucial fact that will help us in the next subsection.
\begin{restatable}{lemma}{LemmaDecompGivenCGMainPaper}\label{Lemma:DecompGivenCG}
    Let $\otherctxNGVAS{\anngvas}{\amarking, \anonterm, \amarkingp}$ be a variant of $\anngvas$ with $(\amarking, \amarkingp)\in\inmarkingdomain\times\outmarkingdomain$, and let $\wtgrammarof{\otherctxNGVAS{\anngvas}{\amarking, \anonterm, \amarkingp}, \postapprox, \preapprox}$, a coverability grammar that remains bounded, be given.
    Then, using elementary resources, we can compute a decomposition $\adecomp$ of $\otherctxNGVAS{\anngvas}{\amarking, \anonterm, \amarkingp}$ with $\recrankof{\anngvas'}<\recrankof{\anngvas}$ for all $\anngvas'\in\adecomp$.
\end{restatable} 
The proof is involved, but the intuition behind it is the following.
The promise information required by the coverability grammar is weaker than reachability, and can be lifted from a parse tree of a run in $\runsof{\otherctxNGVAS{\anngvas}{\amarking, \anonterm, \amarkingp}}$.
Every reachability derivation can thus be captured by a coverability grammar derivation.
The main challenge in obtaining the decomposition is establishing the structure needed for an NGVAS.
To compute the decomposition, we just break the coverability grammar to its SCCs, and mark it using the information from the coverability grammar.
The boundedness information we assign to the non-terminal $((\amarking, \aprommarking), \anonterm, (\aprommarkingp, \amarkingp))\in\wtnontermssim$ is $\amarking\sqcap\aprommarkingp$ on the input, and $\amarkingp\sqcap\aprommarking$ on the output.
The key property of a coverability grammar that allows an easy coversion to an NGVAS is the following.
Let $\awtnonterm^{(0)}, \awtnonterm^{(1)}, \awtnonterm^{(2)}\in\wtnonterms$ with $\awtnonterm^{(i)}=((\amarking^{(i)}, \aprommarking^{(i)}), \anonterm^{(i)}, (\aprommarkingp^{(i)}, \amarkingp^{(i)}))$ for $i\in\set{0,1,2}$ belong to the same SCC in $\wtgrammar$, and let $\awtnonterm^{(0)}\to\awtnonterm^{(1)}.\awtnonterm^{(2)}\in\wtprods$.
Then, the promise structure ensures that $I, O\subseteq [1,d]$ with $\omegaof{\amarking^{(i)}}=\omegaof{\aprommarking^{(i)}}=I$ and $\omegaof{\amarkingp^{(i)}}=\omegaof{\aprommarkingp^{(i)}}=O$ for all $i\in\set{0, 1, 2}$.
Thus, the boundedness information on input and output $\amarking^{(i)}\sqcap\aprommarkingp^{(i)}$ and $\amarkingp^{(i)}\sqcap\aprommarking^{(i)}$, has $\omegaof{\amarking^{(i)}\sqcap\aprommarkingp^{(i)}}=\omegaof{\amarkingp^{(i)}\sqcap\aprommarking^{(i)}}=I\cap O$ for all $i\in\set{0, 1, 2}$.
This is exactly the type of boundedness information we require in a non-linear NGVAS.
Note that the boundedness information in a Karp-Miller graph does not give this guarantee.
This boundedness leads to a decrease in non-linear rank, \emph{if} it is discovered in a non-rigid counter.
To ensure that this happens, we require membership to $\inmarkingdomain\times\outmarkingdomain$ in the premise of \Cref{Lemma:DecompGivenCG}.
The details can be found in \Cref{Section:CoverabilityGrammarL3}.

\subsection{Computing $\postfunc$ and $\prefunc$}\label{Section:ComputingPostAndPreMP}
By \Cref{Lemma:DecompGivenCG} and \Cref{Lemma:CGPumpability}, it suffices to show that $\natpostapprox$ and $\natpreapprox$ are computable for the domains 
\begin{align*}
    \infuncdomain=\inmarkingdomain\times(\nonterms\cup\trms)\qquad
    \outfuncdomain=\outmarkingdomain\times(\nonterms\cup\trms).
\end{align*}
However, by definition, $\natpostapprox$ and $\natpreapprox$ are only thin wrappers for the functions $\postfunc$ and $\prefunc$.
Thus, it suffices to show that $\postfunc$ and $\prefunc$ are computable for the domains $\infuncdomain$ resp. $\outfuncdomain$.
This is our goal for the rest of this section.
\begin{restatable}{lemma}{LemmaPostPreComputableMainPaper}\label{Lemma:PostPreComputable}
    Let $\anngvas$ be a non-linear NGVAS with all the perfectness properties excluding \perfectnesspumpingnospace, and let $\perfect$ be reliable up to $\rankof{\anngvas}$.
    Then functions $\postfunc$ and $\prefunc$ are computable for the respective domains $\infuncdomain$ and $\outfuncdomain$.
\end{restatable}
We make our notion of computability precise.
We treat the functions $\postfunc$ and $\prefunc$ as if they have the domain $\Nomega^{d}\times(\nonterms\cup\trms)$, but in accordance with our programming model, they also have the NGVAS $\anngvas$ as an input.
Computability for a domain should be understood as a promise problem.
The functions expect the input NGVAS $\anngvas$ to have the perfectness properties except \perfectnesspumpingnospace.
Thus, if it does not satisfy the necessary perfectness conditions, there is no guarantee on the correctness of the function.

As we discussed in \Cref{Section:OutlineSimpleCases}, our strategy in tackling this problem is to establish assumptions as strong as before actually computing $\postfunc$ and $\prefunc$.
We have discussed these assumptions \extraa-\extrad\ in \Cref{Section:OutlineSimpleCases}.
For the rest of this subsection, we fix an input $(\amarking_{in}, \asymbol_{in})\in\infuncdomain$ and focus our attention to $\postfunc$ and inputs from the domain $\infuncdomain$.
All arguments also apply to $\prefunc$ and the domain $\outfuncdomain$.

For the sake of brevity, we call a variant $\otherctxNGVAS{\anngvas}{\amarking, \anonterm, \amarkingp}$ a \emph{relevant variant}, if $(\amarking, \anonterm, \amarkingp)\in\inmarkingdomain\times\nonterms\times\outmarkingdomain$.
For notational convenience, assume $\asymbol_{in}\not\in\trms$ for the moment.
We understand the $\postfunc$ query as the NGVAS $\otherctxNGVAS{\anngvas}{\amarking_{in}, \asymbol_{in}, \outof{\asymbol_{in}}}$.
The runs of this NGVAS contains every run considered by $\postfunc$.
The assumptions \extraa-\extrad, each by different mechanisms, guarantee a workaround that allows us to refine $\otherctxNGVAS{\anngvas}{\amarking_{in}, \asymbol_{in}, \outof{\asymbol_{in}}}$.
From this refinement, we can use reliability to call $\perfect$, and get a perfect deconstruction.
We can then read the output values off the perfect deconstruction.
Here, we discuss the conditions and the mechanisms in detail.
The cases which we cannot handle using these conditions, are the hard cases.
These require special attention.

\paragraph*{Case \extraa, ChildNGVAS.}
The case \extraa\ is the case where we have a $\postfunc$ query formulated over a child NGVAS $\asymbol=\anngvasp\in\rectrms$.
We express the query as an NGVAS and call $\perfect$. 
Since \Cref{Lemma:NonLinearChildRank} holds, the NGVAS $\anngvasp$ in question is already of lower non-linear rank.
We replace the $\acontextin$ component of $\anngvasp$ with $\amarking_{in}$ to get $\anngvasp_{\amarking_{in}}$.
If $\amarking_{in}\not\compwith\anngvasp.\acontextin$, the construction might not yield a sound NGVAS.
But this is not a problem, since if $\amarking_{in}\not\compwith\anngvasp.\acontextin$, then there is no run $(\amarking_{in}', \arun, \amarkingp_{in}')\in\runsof{\anngvasp}$ with $\amarking_{in}'\sqsubseteq\amarking_{in}$, and we get $\postfuncNof{}{\amarking_{in}, \anngvasp}=\emptyset$.
The assumption $\unconstrained\subseteq\omegaof{\amarking_{in}}$ of $\infuncdomain$ ensures $\anngvas.\unconstrained\subseteq\anngvasp.\unconstrained$ is kept, so we get $\recrankof{\anngvasp_{\amarking_{in}}}<\recrankof{\anngvas}$ as well.
Then we know that $\perfect$ is reliable for $\anngvasp_{\amarking_{in}}$, which means that the call $\perfectof{\cleanof{\anngvas_{\amarking_{in}}}}$ terminates correctly.
Thanks to these observations, we get computability for the domain $\afuncdomain_{\rectrms}=\setcond{(\amarking, \asymbol)\in\Nomega^{d}\times\rectrms}{\unconstrained\subseteq\omegaof{\amarking}\subseteq\omegaof{\abdinfomid}}$.
We have ommited some technical details related to cleanness, which can be found in \Cref{Section:CoverabilityEasyCasesL3}
\begin{restatable}{lemma}{LemmaTermsPostPreComputableMainPaper}\label{Lemma:TermsPostPreComputable}
    Let $\perfect$ be reliable up to $\rankof{\anngvas}$.
    We can compute $\postfuncN{\anngvas}$ and $\prefuncN{\anngvas}$ restricted to the domain $\afuncdomain_{\rectrms}$.
\end{restatable}
For the remaining cases, we may assume $\asymbol_{in}\not\in\rectrms$.
Because we will be working with the non-linear case, where $\rectrms=\trms$, we assume $\asymbol_{in}\not\in\trms$.
We use the symbol $\anonterminalin\in\nonterms$ for $\asymbol_{in}$ to make this clear.
Note that when taken with \Cref{Lemma:ExtPumpCompLin}, \Cref{Lemma:TermsPostPreComputable} already shows that $\refineintpumpof{\anngvas}$ and $\refinepumpof{\anngvas}$ can be computed for linear $\anngvas$.
\paragraph*{Case \extrab, Lower Dimensions.}
We move on to case \extrab.
Here, one counter that admits reachability constraints is already $\omega$ on the input.
Intuitively, such a counter is not relevant for $\postfunc$, and therefore gets shielded from reachability constraints in the NGVAS $\otherctxNGVAS{\anngvas}{\amarkingin, \anonterminalin, \outof{\anonterminalin}}$.
This lowers the rank wrt. $\anngvas$, and lets us call $\perfect$ in agreement with our programming model.

Formally, this is the case $(\amarkingin, \anonterminalin)\in\easydomain$, where 
$$\easydomain=\bigcup_{\unconstrained\subsetneq\acounterset}\omegamrkdomainof{\acounterset}\times(\nonterms\cup\trms).$$
Consider the NGVAS $\otherctxNGVAS{\anngvas}{\amarkingin, \anonterminalin, \outof{\anonterminalin}}$.
By definition, we have $\otherctxNGVAS{\anngvas}{\amarkingin, \anonterminalin, \outof{\anonterminalin}}.\unconstrained=\omegaof{\amarkingin}\cap\omegaof{\outof{\anonterminalin}}=\omegaof{\amarkingin}\cap\abdinfomid$.
Since $\omegaof{\amarkingin}\subseteq\abdinfomid$, we have $\omegaof{\amarkingin}\cap\omegaof{\outof{\anonterminalin}}=\omegaof{\amarkingin}$.
However, $\unconstrained\subsetneq\omegaof{\amarkingin}$, so we get $\recrankof{\otherctxNGVAS{\anngvas}{\amarkingin, \anonterm, \outof{\anonterminalin}}}<\recrankof{\anngvas}$ by the most significant component.
So, similarly to the previous case, we can call $\perfect$ without needing to further simplify the NGVAS.
Also similarly to the previous case, cleanness details that have been moved to \Cref{Section:CoverabilityEasyCasesL3}.
We get decidability for the domain $\easydomain$.
\begin{restatable}{lemma}{LemmaLowDimPostPreComputableMainPaper}\label{Lemma:LowDimPostPreComputable}
    Let $\perfect$ be reliable up to $\rankof{\anngvas}$. 
    Then, we can compute $\prefuncN{\anngvas}$ and $\postfuncN{\anngvas}$ restricted to the domain $\easydomain$.
\end{restatable}
For the following cases, we can assume $(\amarking, \anonterm)\not\in\easydomain$.
Note that if $\unconstrained\subsetneq\omegaof{\acontextin}$ already holds, then $\infuncdomain\setminus\easydomain=\emptyset$.
Since $\unconstrained\subseteq\omegaof{\acontextin}$ is clear by the NGVAS definition, \emph{we implicitly assume $\omegaof{\acontextin}=\unconstrained=\omegaof{\amarking}$ for all statements beyond this point.}

\paragraph*{Case \extrac, No Lower Dimensional Pump.}
We move on to \extrac, we consider the case where \perfectnesspumping does not hold, even if we ignore some of the counters.
So far, our focus had been on the NGVAS $\otherctxNGVAS{\anngvas}{\amarkingin, \anonterminalin, \outof{\anonterminalin}}$.
However, for the future cases, we will need to decompose this NGVAS into $\setcond{\otherctxNGVAS{\anngvas}{\amarkingin, \anonterminalin, \amarkingp_{i}}}{i\in\anindexset}$, where potentially $\amarkingp_{i}\neq\outof{\anonterminalin}$ holds.
In order to benefit from \extrac\ in the future cases, we formulate it with respect to an arbitrary variant $\otherctxNGVAS{\anngvas}{\amarking, \anonterm, \amarkingp}$ of $\anngvas$.
The idea is the following.
If we ignore one counter $i\in\abdinfomid\setminus\unconstrained$ and $j\in\abdinfomid\setminus\unconstrained$, the $\postfunc$ and $\prefunc$ calls we need to construct the coverability grammar become simpler.
In case \perfectnesspumping does not hold even after ignoring these counters, then the coverability grammar gives us a decomposition.
We can handle the rest by calling $\perfect$.
We formalize this with a new set of approximators 
\begin{align*}
    \forgetfulpostapproxof{i}{\amarkingpp, \asymbolp}&=\natpostapproxof{\settoomega{\unconstrained\cup i}{\amarkingpp}, \asymbolp}\\
    \forgetfulpreapproxof{i}{\amarkingpp, \asymbolp}&=\natpreapproxof{\settoomega{\unconstrained\cup i}{\amarkingpp}, \asymbolp}
\end{align*}
for $i\in\abdinfomid$.
These just call $\natpostapprox$ and $\natpreapprox$ with additional $\omega$'s.
Clearly, the approximator conditions carry over from $\natpostapprox$ and $\natpreapprox$.
\begin{restatable}{lemma}{LemmaForgetfulApproximatorsMainPaper}\label{Lemma:ForgetfulApproximators}
    Let $i\in\abdinfomid\setminus\unconstrained$.
    Then, $\forgetfulpostapprox{i}$ and $\forgetfulpreapprox{i}$ are post- resp. pre-approximators.
\end{restatable}
As we discussed in \extrab, we can already compute $\postfunc$ and $\prefunc$ for $\easydomain$.
Then, by \Cref{Lemma:LowDimPostPreComputable}, we already know that these approximators are also computable.
Thanks to \Cref{Lemma:DecompGivenCG}, we know that we can compute $\wtgrammarof{\otherctxNGVAS{\anngvas}{\amarking, \anonterm, \amarkingp}, \forgetfulpostapprox{i}, \forgetfulpreapprox{j}}$ for $i, j\in\abdinfomid\setminus\unconstrained$ and the arbitrary variant $\otherctxNGVAS{\anngvas}{\amarking, \anonterm, \amarkingp}$.
\begin{restatable}{corollary}{CorollaryForgetfulApproxComputabilityMainPaper}\label{Corollary:ForgetfulApproxComputability}
    Let $\perfect$ be reliable up to $\rankof{\anngvas}$, and $\otherctxNGVAS{\anngvas}{\amarking, \anonterm, \amarkingp}$ a relevant variant of $\anngvas$.
    For $i, j\in\abdinfomid\setminus\unconstrained$, the approximators $\forgetfulpostapprox{i}$ and $\forgetfulpreapprox{j}$ are computable. 
    Furthermore, we can effectively construct $\wtgrammarof{\otherctxNGVAS{\anngvas}{\amarking, \anonterm, \amarkingp}, \forgetfulpostapprox{i}, \forgetfulpreapprox{j}}$.
\end{restatable}
Now we put these approximators to use.
We call a triple $(\amarking, \anonterm, \amarkingp)\in\inmarkingdomain\times\nonterms\times\outmarkingdomain$ \emph{simply decomposable}, denoted $(\amarking, \anonterm, \amarkingp)\in\simplydecomps$, if the coverability grammar $\wtgrammarof{\otherctxNGVAS{\anngvas}{\amarking, \anonterm, \amarkingp}, \forgetfulpostapprox{i}, \forgetfulpreapprox{j}}$ remains bounded for any $i, j\in\abdinfomid\setminus\unconstrained$.
We get two consequences.
First, if $(\amarking, \anonterm, \amarkingp)\in\simplydecomps$, then we get a perfect deconstruction.
\begin{restatable}{lemma}{LemmaSimplePerfectDecompositionMainPaper}\label{Lemma:SimplePerfectDecomposition} 
    Let $\perfect$ be reliable up to $\rankof{\anngvas}$, and let $\otherctxNGVAS{\anngvas}{\amarking, \anonterm, \amarkingp}$ be a relevant variant of $\anngvas$.
    We can decide whether $(\amarking, \anonterm, \amarkingp)\in\simplydecomps$ holds.
    If $(\amarking, \anonterm, \amarkingp)\in\simplydecomps$, we can compute a perfect deconstruction $\adecomp$ of $\otherctxNGVAS{\anngvas}{\amarking, \anonterm, \amarkingp}$.
\end{restatable}
Second, if $(\amarking, \anonterm, \amarkingp)\not\in\simplydecomps$, then the coverability grammar $\wtgrammarof{\otherctxNGVAS{\anngvas}{\amarking, \anonterm, \amarkingp}, \forgetfulpostapprox{i}, \forgetfulpreapprox{j}}$ finds unboundedness for all $i, j\in\abdinfomid\setminus\unconstrained$.
This implies we get pumping derivations whenever we ignore $i, j\in\unconstrained$ on the input resp. output.
\begin{restatable}{lemma}{LemmaSimplyUndecomposableMainPaper}\label{Lemma:SimplyUndecomposable}
    Let $\otherctxNGVAS{\anngvas}{\amarking, \anonterm, \amarkingp}$ be a relevant variant of $\anngvas$ with $(\amarking, \anonterm, \amarkingp)\not\in\simplydecomps$.
    Then, for any $i, j\in\abdinfomid\setminus\unconstrained$, $\otherctxNGVAS{\anngvas}{\settoomega{i}{\amarking}, \anonterm, \settoomega{j}{\amarkingp}}$ has a pumping derivation.
\end{restatable}

This assumption will play an important role later in Hard Case 1, \Cref{Section:HardCaseOne}.

\paragraph*{Case \extrad, No $\Z$-Pump.}

Finally, we have the case \extrad.
This is the case where we might not need to make a complete analysis to ensure that \perfectnesspumping does not hold, but an integer-based approximation suffices.
Just as in \extrac, we assume arbitrary variants.
Towards a formalization of $\Z$-approximation, let $\unbcoords:\powof{\Nomega^{d}}\to\powof{[1,d]}$ be a function that extracts the unbounded counters in $K\subseteq\Nomega^{d}$, that is 
$$\unbcoordsof{K}=\setcond{i\in[1,d]}{\forall a\in\N.\;\exists \amarking\in K.\amarking[i]\geq a}.$$
Now, let $\omegazer:\powof{\Nomega^{d}}\to\powof{\Nomega^{d}}$ be the function that abstracts the unboundedly growing components in (potentially infinite) set of markings, formally defined as $\omegazerof{K}=\setcond{\settoomega{\unbcoordsof{K}}{\amarking}}{\amarking\in K}$ for $K\subseteq\Nomega^{d}$.
We define the post and pre approximations $\intpostapprox$ and $\intpreapprox$, which formalize the notion of $\Z$-approximation.
To explain these functions, let $(\amarking, \asymbol)\in\Nomega^{d}\times(\nonterms\cup\trms)$ be an input.
The functions $\intpostapprox$ and $\intpreapprox$ reject invalid inputs.
That is, if $\amarking\not\sqsubseteq\inof{\asymbol}$, $\intpostapprox$ returns $\emptyset$, and if $\amarking\not\sqsubseteq\outof{\asymbol}$, $\intpreapprox$ returns $\emptyset$.  
Now, we assume that the input is valid.
Then, if $\asymbol\in\trms$, and $\amarking\in\Nomega^{d}$, we let 
\begin{align*}
    \intpostapproxof{\amarking, \asymbol}&=\omegazerof{\setcond{\amarking+\amarkingp\in\Nomega^{d}}{\amarkingp\in\asymbol.\restrictions}}\\
    \intpreapproxof{\amarking, \asymbol}&=\omegazerof{\setcond{\amarking-\amarkingp\in\Nomega^{d}}{\amarkingp\in\asymbol.\restrictions}}
\end{align*}
For $\asymbol=\anonterm\in\nonterms$ and $\amarking\in\Nomega^{d}$, we define the approximations similarly, but this time wrt. the characteristic equations.
\begin{align*}
    \intpostapproxof{\amarking, \anonterm}&=\omegazerof{\setcond{\asol[\outvar]}{\\ 
    &\hspace{2em}\asol\text{ solves }\otherctxNGVAS{\anngvas}{\amarking, \anonterm, \outof{\anonterm}}.\chareq{}}}\\
    \intpreapproxof{\amarking, \anonterm}&=\omegazerof{\setcond{\asol[\invar]}{\\
    &\hspace{2em}\asol\text{ solves }\otherctxNGVAS{\anngvas}{\inof{\anonterm}, \anonterm, \amarking}.\chareq{}}}
\end{align*}
It can be readily verified that these functions are suitable for constructing coverability grammars.
\begin{restatable}{lemma}{LemmaZPostPreApproximatorsMainPaper}\label{Lemma:ZPostPreApproximators}
    The functions $\intpostapprox$ and $\intpreapprox$ are post- resp. pre-approximations, and are computable.
\end{restatable}
Then, if $\wtgrammarof{\otherctxNGVAS{\anngvas}{\amarking, \anonterm, \amarkingp}, \intpostapprox, \intpreapprox}$ remains bounded, we can compute a perfect decomposition as we did in \extrac.
\begin{restatable}{lemma}{LemmaZApproxPumpDecompMainPaper}\label{Lemma:ZApproxPumpDecomp}
    Let $\otherctxNGVAS{\anngvas}{\amarking, \anonterm, \amarkingp}$ be a relevant variant.
    Further let $\wtgrammarof{\otherctxNGVAS{\anngvas}{\amarking, \anonterm, \amarkingp}, \intpostapprox, \intpreapprox}$ remain bounded, and let $\perfect$ be reliable up to $\rankof{\anngvas}$
    Then, we can compute a perfect decomposition of $\otherctxNGVAS{\anngvas}{\amarking, \anonterm, \amarkingp}$.
\end{restatable}
Similarly to the \extrac\ case, we deduce further information out of the assumption $\neg\extrad$.
If the coverability grammar under $\Z$-approximations shows unboundedness for $\otherctxNGVAS{\anngvas}{\amarking, \anonterm, \amarkingp}$, then we get a derivation $\anonterm\to\aword.\anonterm.\awordp$ that pumps all counters $\abdinfomid\setminus\omegaof{\amarking}$ and $\abdinfomid\setminus\omegaof{\amarkingp}$, but with a caveat.
The pump ignores the positivity constraints, and we need some assumptions on $\omegaof{\amarkingp}$.
This is notion of pumping is formalized by the $\Z$-pump.
This object is defined relative to $\anngvas$, and two sets of counters $I, O\subseteq[1,d]$, instead of relative to a variant.
Formally, a $(I,O)$-$\Z$-pump for $\anngvas$ a tuple $(\anonterm\to^{*}\aword.\anonterm. \awordp, \amarkingpp_{fwd}, \amarkingpp_{bck})$ consisting of a cycle $\anonterm\to^{*}\aword.\anonterm. \awordp$, and two effects $\amarkingpp_{fwd}\in\ceffof{\aword}$ and $\amarkingpp_{bck}\in\ceffof{\amarkingp}$, such that $\amarkingpp[i]\geq 1$ and $-\amarkingpp[j]\geq 1$ for all $i\in \abdinfomid\setminus I$ and $j\in \abdinfomid\setminus O$.
The assumption we need on $\omegaof{\amarkingp}$ is $\omegaof{\acontextout}\subseteq\omegaof{\amarkingp}$.
In this case, if the coverability grammar shows unboundedness under $\Z$-approximations, then we get a $\Z$-pump.
We need this assumption in order to use the full support assumptions we have from $\anngvas$, and construct connected derivations.
Without this, $\intpostapprox$ and $\intpreapprox$ might yield results that cannot be realized by connected derivations.
\begin{restatable}{lemma}{LemmaZPumpImpliesUnboundedMainPaper}\label{Lemma:ZPumpImpliesUnbounded}
    Let $\otherctxNGVAS{\anngvas}{\amarking, \anonterm, \amarkingp}$ be a relevant variant.
    Then, if $\wtgrammarof{\otherctxNGVAS{\anngvas}{\amarking, \anonterm, \amarkingp}, \intpostapprox, \intpreapprox}$ shows unboundedness, then $\otherctxNGVAS{\anngvas}{\amarking, \anonterm, \amarkingp}$ has a $(\omegaof{\amarking}, \omegaof{\amarkingp})$-$\Z$-pump. 
\end{restatable}
The $\Z$-pump is a weaker version of the derivation we search for \perfectnesspumpingnospace, since it has no positivity guarantees.
But, because $\anngvas$ is almost perfect, we are able establish the positivity for counters in $[1,d]\setminus\abdinfomid$ by only starting from a $\Z$-pump.
In other words, we can at least assume that we derive runs, instead of vectors that stem from restrictions.
The proof is similar to the proof of iteration lemma.
We refer to this kind of a pumping derivation as a free-$\N$-pump.
The free $\N$-pump is also defined relative to two sets $I, O\subseteq [1,d]$.
Formally, a $(I, O)$-free-$\N$-pump for $\otherctxNGVAS{\anngvas}{\amarking, \anonterm, \amarkingp}$ is a tuple $(\anonterm\to^{*}\aword.\anonterm.\awordp, \arun, \arunp)$ consisting of a cycle $\anonterm\to^{*}\aword.\anonterm.\awordp$, and two update sequences $\arun\in\updateseqof{\aword}$ and $\arunp\in\updateseqof{\awordp}$ with $\updates\cdot\paramparikhof{\updates}{\arun}[i]\geq 1$ and $-\updates\cdot\paramparikhof{\updates}{\arunp}[j]\geq 1$ for all $i\in\abdinfomid\setminus I$ 
and $j\in\abdinfomid\setminus O$.
The size of a free $\N$ pump $(\anonterm\to^{*}\aword.\anonterm.\awordp, \arun, \arunp)$ is $\max(\normof{\arun}, \normof{\arunp})$ where $\normof{\arun}=\sum_{i\leq \cardof{\arun}}\normof{\arun[i]}$.
An $(I, O)$-$\Z$-pump, leads to an $(I, O)$-free-$\N$-pump thanks to the almost-perfectness of $\anngvas$.
We make a stronger statement, and say that we can even compute an upper bound on the size of the runs.
\begin{restatable}{lemma}{LemmaPrecalculationIMainPaper}\label{Lemma:PrecalculationI}
    Let $\anngvas$ have all perfectness conditions except \perfectnesspumpingnospace, and let $\anngvas$ have an $(I, O)$-$\Z$-pump for $I, O\subseteq [1,d]$.
    We can compute an $\incconst\in\N$ such that for all $\anonterm\in\nonterms$, there is a $(I, O)$-free-$\N$-pump $(\anonterm\to^{*}\aword.\anonterm.\awordp, \arun, \arunp)$ of size less than $\incconst$.
\end{restatable} 
This means that whenever the $\Z$-approximation fails for an NGVAS with a suitable output marking, we can assume a free-$\N$-pump under the bound $\incconst$.

\paragraph*{Reestablishing Perfectness.}
We return back to the input $(\amarkingin, \anonterminalin)$.
As we discussed, our strategy is to make our assumptions as strong as possible.
But even though $\anngvas$ has all perfectness conditions up to \perfectnesspumpingnospace, $\otherctxNGVAS{\anngvas}{\amarkingin, \anonterminalin, \outof{\anonterminalin}}$ might not have them.
This is because the context infomation has changed.
We reestablish these conditions before moving further in our development.
Since we changed the input and output markings, condition that needs our care is \perfectnesscountersnospace.
To fix this, we define the set 
\begin{align*}
    \omegaoutcount&=\setcond{i\in[1,d]}{\\
    &\hspace{2em}\outvar[i]\in\suppof{\otherctxNGVAS{\anngvas}{\acontextin, \startnonterm, \outof{\startnonterm}}.\homchareq{}}}
\end{align*}
that consists of output counters that are in the support of homogenous the characteristic equation of $\otherctxNGVAS{\anngvas}{\acontextin, \startnonterm, \outof{\startnonterm}}$.
Since $\unconstrained=\omegaof{\acontextin}=\omegaof{\amarkingin}$, the homogeneous systems of $\otherctxNGVAS{\anngvas}{\acontextin, \startnonterm, \outof{\startnonterm}}$ and $\otherctxNGVAS{\anngvas}{\amarkingin, \anonterminalin, \outof{\anonterminalin}}$ are the same, while $\otherctxNGVAS{\anngvas}{\acontextin, \startnonterm, \acontextout}.\homchareq{}$ is stricter.
Note that this implies $\omegaof{\acontextout}\subseteq\omegaoutcount$, because all $i\in\omegaof{\acontextout}$ must already be in the support of $\anngvas.\homchareq{}$.
We expect output markings that carry $\omega$'s in $\omegaoutcount$ counters.
This means that we artificially enforce \perfectnesscountersnospace.
This is also the reason why we need assumptions from \extrac\ and \extrad\ to apply to arbitrary variants.
Now, we need to make sure that we work with a set of NGVAS $\thezdecomp=\setcond{\otherctxNGVAS{\anngvas}{\amarkingin, \anonterminalin, \amarkingp_{i}}}{i\in\anindexset}$ where $\omegaof{\amarkingp_{i}}=\omegaoutcount$ for all $i\in\anindexset$, but also $\runsof{\thezdecomp}=\runsof{\otherctxNGVAS{\anngvas}{\amarkingin, \anonterminalin, \outof{\anonterminalin}}}$.
Here, our definition of $\omegaoutcount$ comes to our rescue.
The set $[1,d]\setminus\omegaoutcount$ of counters are already those, that are not in the support of the homogeneous equation, which means they are bounded.
This means that we can colllect these values, and get the set $\thezdecomp$ we need.
Formally, we define 
\begin{align*}
    \thezdecomp=\ &\setcond{\otherctxNGVAS{\anngvas}{\amarkingin, \anonterminalin, \amarkingp}}{\\ &\asol\text{ solves }\otherctxNGVAS{\anngvas}{\amarkingin, \anonterminalin, \outof{\anonterminalin}}.\chareq{}\;\\
    & \wedge\;\asol[\outvar]\sqsubseteq\amarkingp\;\wedge\;\omegaof{\amarkingp}=\omegaoutcount}.
\end{align*}
This is precisely the set we wanted.
\begin{restatable}{lemma}{LemmaTheZDecompMainPaper}\label{Lemma:TheZDecomp}
    We have $\runsof{\otherctxNGVAS{\anngvas}{\amarkingin, \anonterminalin, \amarkingp}}=\runsof{\thezdecomp}$.
    All $\anngvas'\in\thezdecomp$ have all perfectness conditions except \perfectnesspumping and \perfectnesssolnospace. 
\end{restatable}
\begin{proof}[Proof Sketch]
    The run preservation is clear from the correspondence between the characteristic equation and runs.
    The perfectness conditions \perfectnesschildren and \perfectnessbase follow from $\anngvas$, since we did not modify the children.
    All the homogeneous-solution related conditions \perfectnesscountersnospace, \perfectnesschildperiodsbdnospace, and \perfectnessprodsnospace\ also hold.
    This is because of the definition of $\omegaoutcount$, and the following fact. 
    The full support solution to $\anngvas.\homchareq{}$, is also a support solution to $\otherctxNGVAS{\anngvas}{\amarkingin, \anonterminalin, \amarkingp}.\homchareq{}$ for all $\otherctxNGVAS{\anngvas}{\amarkingin, \anonterminalin, \amarkingp}\in\thezdecomp$, since $\otherctxNGVAS{\anngvas}{\acontextin, \startnonterm, \acontextout}.\homchareq{}$ is stricter than $\otherctxNGVAS{\anngvas}{\amarkingin, \anonterminalin, \outof{\anonterminalin}}.\homchareq{}$.
    Finally, we have no rigid counters outside of the boundedness information.
    This is because $\omegaoutcount\subseteq\omegaof{\acontextout}$ and $\omegaof{\amarkingin}\subseteq\omegaof{\acontextin}$, and counter $i$ being rigid in $\otherctxNGVAS{\anngvas}{\amarkingin, \anonterminalin, \amarkingp}\in\thezdecomp$ would imply it being rigid in $\anngvas$.
\end{proof}

Our assumptions carry over to this set.
In particular, because all $\anngvas'\in\thezdecomp$ share the same input resp. output $\omega$'s, the (mis)existence of a $\Z$-pump effects them equally.
If the $\Z$-pump does not exist, we conclude with \Cref{Lemma:ZPumpImpliesUnbounded} and \Cref{Lemma:ZApproxPumpDecomp}.
\begin{restatable}{lemma}{LemmaTheZDecompAssumption}\label{Lemma:TheZDecompAssumption}
    Let $\anngvas$ have no $(\unconstrained, \omegaoutcount)$-$\Z$-pump.
    Then, we can compute a perfect decomposition of the relevant variant $\otherctxNGVAS{\anngvas}{\amarking, \anonterm, \outof{\anonterm}}$.
\end{restatable}
For the remainder of the section, we assume that $\anngvas$ has $(\unconstrained, \omegaoutcount)$-free-$\N$-pumps, centered on each non-terminal, and with size less than 
$\incconst\in\N$ from \Cref{Lemma:PrecalculationII}.
We have covered the cases where $\asymbolin\in\trms$, or $\omegaof{\amarking}=\omegaof{\acontextin}=\unconstrained$ does not hold, or $\anngvas$ does not have a $(\unconstrained, \omegaoutcount)$-$\Z$-pump.
Then, only the case $\asymbolin\in\nonterms$, where $\anngvas$ has a $(\unconstrained, \omegaoutcount)$-$\Z$-pump is open.
We distinguish these cases based on the sizes of the values in the input counters, all small, or at least one large.
First, in \Cref{Section:HardCaseOne}, we handle the case where one counter in $\amarkingin$ is large.
Lastly, in \Cref{Section:HardCaseTwo}, we handle the remaining case by using all our previous computability assumptions.

\subsection{Hard Case 1: Large Input Counters}\label{Section:HardCaseOne}

The main observation of this case is the following.
If at least one input and one output counter in $\anngvas'\in\thezdecomp$ are large enough, we observe that we can make a Rackoff argument.
Unless $\anngvas'\in\simplydecomps$ holds, which we handled in \extrac, a pair of large input and output counters give us \perfectnesspumpingnospace.
Since $\anngvas'\in\thezdecomp$ already have all other perfectness conditions, we can read off the output values to get the values covered by runs in $\runsof{\anngvas'}$ for $\anngvas'\in\thezdecomp$.
We use this to compute $\postfunc$ and $\prefunc$ in the case of a large input counter.

In this section, we work with pumping derivations as an object with size, which we formalize as follows.
A \emph{pumping derivation} in the variant $\otherctxNGVAS{\anngvas}{\amarking, \anonterm, \amarkingp}$ is a tuple $(\anonterm\to^{*}\aword.\anonterm.\awordp, \arun, \arunp)$ consisting of a cycle $\anonterm\to^{*}\aword.\anonterm.\awordp$, and two sequences $\arun\in\updateseqof{\aword}$, and $\arunp\in\updateseqof{\awordp}$ with $\amarkingpp, \amarkingppp\in\Nomega^{d}$ where $\amarking\fires{\arun}\amarkingpp$, $\amarkingp\fires{\arunp^{rev}}\amarkingppp$, $\amarking[i]<\amarkingpp[i]$ for all $i\in\abdinfomid\setminus\omegaof{\amarking}$, and $\amarkingp[i]<\amarkingppp[i]$ for all $i\in\abdinfomid\setminus\omegaof{\amarkingp}$.
Similarly to the free-$\N$-pump, the size of a pumping derivation $(\anonterm\to^{*}\aword.\anonterm.\awordp, \arun, \arunp)$ is $\max(\normof{\arun}, \normof{\arunp})$.
The core of the argument is the following lemma.
This is the detailed version of \Cref{Lemma:SimpleBound} from the outline.
\begin{restatable}{lemma}{LemmaSimpleBoundDetailedMainPaper}\label{Lemma:SimpleBoundDetailed}
    Let $\otherctxNGVAS{\anngvas}{\amarking, \anonterm, \amarkingp}$ be a variant of $\anngvas$, let there be a $(\omegaof{\amarking}, \omegaof{\amarkingp})$-free-$\N$-pump, centered on $\anonterm$, of size at most $k$.
    Let $I, O\subseteq\abdinfomid$, where $\otherctxNGVAS{\anngvas}{\settoomega{I}{\amarking}, \anonterm, \settoomega{O}{\amarkingp}}$ has a pumping derivation of size at most $\ell\in\N$, as well as $\amarking[i]\geq k\cdot (\ell+1)$ for all $i\in I$ and $\amarkingp[j]\geq k\cdot(\ell+1)$ for all $j\in O$.
    Then, $\otherctxNGVAS{\anngvas}{\amarking, \anonterm, \amarkingp}$ has \perfectnesspumpingnospace.
\end{restatable}

\begin{proof}
    Let $\otherctxNGVAS{\anngvas}{\amarking, \anonterm, \amarkingp}$, $I, O\subseteq\abdinfomid$, and $\ell\in\N$ as given in the lemma.
    Let $(\anonterm\to^{*}\aword.\anonterm.\awordp, \arun_{f}, \arunp_{f})$ be the $(\omegaof{\amarking}, \omegaof{\amarkingp})$-free-$\N$-pump with size less than $k\in\N$, and $(\anonterm\to^{*}\aword.\anonterm.\awordp, \arun_{pd}, \arunp_{pd})$ the pumping derivation of $\otherctxNGVAS{\anngvas}{\settoomega{I}{\amarking}, \anonterm, \settoomega{O}{\amarkingp}}$ with size less than $\ell\in\N$.
    Let $I_{pmp}=\abdinfomid\setminus\omegaof{\amarking}$ and $O_{pmp}=\abdinfomid\setminus\omegaof{\amarkingp}$.
    As a shorthand, we write $\anonterm\to^{*}\arun.\anonterm.\arunp$ denote the existence of a derivation $\anonterm\to^{*}\aword.\anonterm.\awordp$ with $\arun\in\updateseqof{\aword}$ and $\arunp\in\updateseqof{\awordp}$.
    First, we argue that there is a derivation $\anonterm\to^{*}\arun_{1}.\anonterm.\arunp_{1}$ with
    where $\amarking\fires{\arun_{1}}\amarking'$, $\amarkingp\fires{\arunp_{1}^{rev}}\amarkingp'$, where $\amarking'[m]\geq k$ for all $m\in I_{pmp}$ and $\amarkingp'[m]\geq k$ for all $m\in O_{pmp}$.
    Note that this is weaker than our main goal, since the constraint in the claim is on the outcome of the derivation, and we want the effect to be positive. 
    The derivation $\anonterm\to^{*}\arun_{pd}^{k}.\anonterm.\arunp_{pd}^{k}$ readily satisfies our requirements.
    The effect of $\arun_{pd}$ is positive on counters $I_{pmp}\setminus I$, and it is enabled from these counters by the definition of the pumping derivation.
    Then, iterating it $k$ times, we get an effect of $k$.
    For the counters in $I$, we already started with at least $k\cdot (\ell+1)$ counters.
    So, anywhere along the execution of $\arun_{pd}^{k}$, the value of this counter is at least $k\cdot (\ell+1)-k\cdot \ell=k$.
    The argument for $\arunp_{pd}^{k}$ is similar.
    The free-$\N$-pump $\anonterm\to^{*}\arun_{f}\anonterm.\arunp_{f}$ handles the rest.
    Since the size of the pump is at most $k$, $\arun_{f}$ is enabled from $\amarking'$.
    The same argument applies backwards from $\amarkingp'$ for $\arunp_{f}$.
    Then, the complete derivation 
    $$\anonterm\to^{*}\arun_{pd}^{k}.\arun_{f}^{k\cdot \ell+1}.\anonterm.\arunp_{f}^{k\cdot\ell+1}.\arunp_{pd}^{k}$$
    is indeed a pumping derivation.
    The runs are enabled, and the effect (forwards from the input resp. backwards from the output) is at least $k\cdot \ell+1- k\cdot \ell=1$ on each counter.
\end{proof}
Now, we observe that we can use this argument to compute an upper bound on all lower dimensional pumping derivations.
The proof has been moved to \Cref{Section:Precalculation}.
Intuitively, we make a Rackoff argument, and apply \Cref{Lemma:SimpleBoundDetailed} repeatedly to establish upper bounds that allow more and more concrete positions.
\begin{restatable}{lemma}{LemmaPrecalculationIIMainPaper}\label{Lemma:PrecalculationII}
    Let $\anngvas$ be an NGVAS.
    Let it have all perfectness conditions excluding \perfectnesspumpingnospace, and a $(\unconstrained, \omegaoutcount)$-free-$\Z$-pump.
    Let $\otherctxNGVAS{\anngvas}{\amarking, \anonterm, \amarkingp}$ be a relevant variant with $\amarking\in\omegamrkdomainof{\unconstrained}$ and $\amarkingp\in\omegamrkdomainof{\omegaoutcount}$.
    Furthermore, let $\perfect$ be reliable up to $\rankof{\anngvas}$.
    We can compute a bound $\extcovconst\in\N$, such that the following holds for all $i,j\in\abdinfomid\setminus\omegaof{\amarking}$, $\amarking'\in\omegamrkdomainof{\omegaof{\amarking}}$, $\amarkingp'\in\omegamrkdomainof{\omegaof{\amarkingp}}$, and $\anontermp\in\nonterms$.
    If $\otherctxNGVAS{\anngvas}{\settoomega{i}{\amarking'}, \anontermp, \settoomega{j}{\amarkingp'}}$ has a pumping derivation, then it has one of size at most $\extcovconst$.
\end{restatable}
Using \Cref{Lemma:PrecalculationI}, \Cref{Lemma:PrecalculationII}, and \Cref{Lemma:SimpleBoundDetailed}, we obtain the following.
\begin{restatable}{lemma}{LemmaFreePumpNoSDBothSidesRackoffMainPaper}\label{Lemma:FreePumpNoSDBothSidesRackoff}
    Let $\anngvas$ have a $(\unconstrained, \omegaoutcount)$-$\Z$-pump.
    Let $\anngvas$ have all perfectness conditions excluding \perfectnesspumpingnospace, and let $\perfect$ be reliable up to $\rankof{\anngvas}$.
    Then, we can compute a $\abigconst\in\N$ such that the following holds for all $\amarking\in\omegamrkdomainof{\unconstrained}$ and $\amarkingp\in\intpostapproxof{\amarking, \anonterm}$.
    If $\amarking[i]\geq \abigconst$ and $\amarkingp[j]\geq\abigconst$ for some $i,j\in \abdinfomid\setminus\unconstrained$, and $(\amarking, \anonterm, \amarkingp)\not\in\simplydecomps$, then $\otherctxNGVAS{\anngvas}{\amarking, \anonterm, \amarkingp}$ has all the perfectness conditions, excluding \perfectnesssolnospace.
\end{restatable}
Using ILP techniques, we can strengthen this statement to only restrict the input side, if $\amarkingp\in\intpostapproxof{\amarking, \anonterm}$ is maximal.
The argument can be found in \Cref{Section:Rackoffesque}.
\begin{restatable}{lemma}{LemmaFreePumpNoSDRackoffMainPaper}\label{Lemma:FreePumpNoSDRackoff}
    Let $\anngvas$ have a $(\unconstrained, \omegaoutcount)$-$\Z$-pump.
    Let $\anngvas$ have all perfectness conditions excluding \perfectnesspumpingnospace, and let $\perfect$ be reliable up to $\rankof{\anngvas}$.
    Then, we can compute a $\abigbound\in\N$ such that the following holds for all $\amarking\in\omegamrkdomainof{\unconstrained}$, and \underline{maximal} $\amarkingp\in\intpostapproxof{\amarking, \anonterm}$.
    If $\amarking[i]\geq \abigbound$ for some $i\in \abdinfomid\setminus\unconstrained$, and $(\amarking, \anonterm, \amarkingp)\not\in\simplydecomps$, then $\otherctxNGVAS{\anngvas}{\amarking, \anonterm, \amarkingp}$ has all the perfectness conditions, excluding \perfectnesssolnospace.
\end{restatable}
Using \Cref{Lemma:FreePumpNoSDRackoff} and \Cref{Lemma:SimplePerfectDecomposition}, we can compute $\postfunc$ in the case that $\anngvas$ has a $(\unconstrained, \omegaoutcount)$-$\Z$-pump, and one concrete counter in the input is large. 
By iterating through each element in $\anngvas'\in\thezdecomp$, we either have \Cref{Lemma:SimplePerfectDecomposition}, or there is another $\anngvas''\in\thezdecomp$, whose output marking is larger.
In the former case, there is a perfect deconstruction, and in the latter case, $\anngvas''$ is already perfect up to \perfectnesssolnospace.
The condition \perfectnesssol can easily be established since characteristic equation of $\anngvas''$ has a solution by the definition of $\intpostapprox$.
Then, the runs of $\anngvas'$ were already covered by $\anngvas''$.
\begin{restatable}{lemma}{LemmaFreePumpRackoffMainPaper}\label{Lemma:FreePumpRackoff}
    Let $\anngvas$ be an NGVAS with a $(\unconstrained, \omegaoutcount)$-$\Z$-pump, and with all perfectness conditions excluding \perfectnesspumpingnospace.
    Let $\perfect$ be reliable up to $\rankof{\anngvas}$.
    Then, we can compute a bound $\abigbound\in\N$ such that we can compute $\postfuncN{\anngvas}$ for the domain $\setcond{(\amarking, \anonterm)\in\omegamrkdomainof{\unconstrained}\times\nonterms}{\; \exists i\in\abdinfomid\setminus\unconstrained.\; \amarking[i]\geq \abigbound}$.
\end{restatable}

\subsection{Hard Case 2: Small Input Counters, Witness Tree Search}\label{Section:HardCaseTwo}
We have dealt with the cases \extraa, \extrab, \extrac, \extrad, and the case in \Cref{Lemma:FreePumpRackoff}, where we assume a $\Z$-pump, and the case from \Cref{Section:HardCaseTwo}, where one $\abdinfomid\setminus\unconstrained$ counter has a larger value then our computed bound $\abigbound$.
It remains to deal with the remaining case, where all themselves non-$\omega$ input counters are below $\abigbound$.
Before moving further, we collect all our assumptions into one place. 
To this end, let $\fullinputdom=\omegamrkdomainof{\unconstrained}$ be the domain of input markings considered by \Cref{Lemma:PostPreComputable}.
Let $\smallinputdom_{\abigbound}=[0, \ldots, \abigbound]^{[1,d]\setminus\unconstrained}\times\omega^{\unconstrained}\subseteq\Nomega^{d}$ be the set of markings whose set of $\omega$-marked counters is $\unconstrained$, and whose concrete counters are all valued less than $\abigbound$. 
To avoid soundness problems, we can assume wlog. that $\abigbound$ is larger than any concrete counter in the boundedness information $\boundednessinformation$.
Conversely, let $\largeinputdom_{\abigbound}=\fullinputdom\setminus\smallinputdom_{\abigbound}\subseteq\Nomega^{d}$ be the set of remaining markings we consider.
Thanks to our assumptions so far, we can decide $\postfuncN{\anngvas}$ for $\largeinputdom_{\abigbound}\times\nonterms$.
For some $\amarking\in\largeinputdom_{\abigbound}$, there are two cases.
We might have $\amarking[i]\in\N$ and $\amarking[i]\geq\abigbound$, in this case \Cref{Lemma:FreePumpRackoff} applies to show computability.
If this is not the case, then we have $\amarking[i]=\omega$, and \Cref{Lemma:LowDimPostPreComputable} applies to show computability.
Combining this with \Cref{Lemma:TermsPostPreComputable} we obtain our starting assumption.
\begin{restatable}{corollary}{CorollaryWTSearchAssumptionMainPaper}\label{Corollary:WTSearchAssumption}
    Let $\anngvas$ be an NGVAS with a $\Z$-pump, and with all perfectness conditions excluding \perfectnesspumpingnospace, and let $\perfect$ be reliable up to $\rankof{\anngvas}$.
    Then, we can compute a bound $\abigbound\in\N$ such that we can compute $\postfuncN{\anngvas}$ restricted to the domain $\fullinputdom\times\trms\cup\largeinputdom_{\abigbound}\times\nonterms$.
\end{restatable}
We drop the subscripts, and fix $\abigbound\in\N$ to be this bound.
Now, we show the remaining hard case as stated below, and complete the proof of \Cref{Lemma:PostPreComputable}.
\begin{restatable}{lemma}{LemmaFreePumpPostLowMainPaper}\label{Lemma:FreePumpPostLow}
    Let $\anngvas$ be an NGVAS with a $\Z$-pump, and with all perfectness conditions excluding \perfectnesspumpingnospace, let $\perfect$ be reliable up to $\rankof{\anngvas}$.
    Then, we can compute $\postfuncN{\anngvas}$ restricted to the domain $\smallinputdom\times\nonterms$.
\end{restatable}
To show this lemma, we conduct a search for trees that witness the output values, which we aptly call witness trees.
In the following, we define witness trees and show that we can principally search them.
A witness tree is a tree that extends a parse tree in the grammar of $\anngvas$, by soundly tracked input and output markings.
The witness trees also enable a certain type of pumping we will discuss shortly.
Formally, a \emph{witness tree} $\atree$ is a $\Nomega^{d}\times(\nonterms\cup\trms)\times\Nomega^{d}$-labeled tree that yields a (possibly incomplete) parse tree when projected to its $\nonterms\cup\trms$ component, and satisfies the following for all its subtrees $\atreep$.
To ease the notation, we write $\anode.\inlabel$, $\anode.\symlabel$ and $\anode.\outlabel$ to denote the components in $\nodemarkingof{\anode}=(\anode.\inlabel, \anode.\symlabel, \anode.\outlabel)$.
This notation refers to the label of the root node, when used with a tree, i.e. $\atree.\inlabel=\atree.\arootnode.\inlabel$.
%
\begin{itemize}
    \item We have $\atreep.\inlabel\in\fullinputdom$.
    \item For any node $\anode\in\atreep$ where $\childnodesof{\anode}=\anodep.\anodepp$, $\anode.\inlabel=\anodep.\inlabel$, and $\anodep.\outlabel=\anodepp.\inlabel$.
    \item No non-root node in $\atreep$ has the label $\nodemarkingof{\atreep.\arootnode}$.
    \item We have $(\atreep.\inlabel, \atreep.\symlabel)\in\fullinputdom\times\trms\cup\largeinputdom\times\nonterms$ iff $\atreep$ is a leaf. 
    \item If $\atreep$ is a leaf, then we have $\atreep.\outlabel\in\postfuncNof{\anngvas}{\atreep.\inlabel, \atreep.\symlabel}$.
    \item For every subtree $\atreep$ of $\atree$, we have $\atreep.\outlabel=\pumpingof{\atreep}$.
\end{itemize}
The first condition says that the input of $\atree$ comes from the correct domain.
By the second condition, the inputs must be propagated soundly.
The third condition disallows redundancies in the tree.
By the next condition, a leaf are exactly those nodes who carry a childNGVAS, or a non-terminal and an input marking has one large valued counter as label.
In either case, the output marking must be consistent with respect to $\postfuncN{\anngvas}$.
Note that so far, we did not specify how the output values should be propagated from the right child to the parent. 
The final condition says that these $\omega$ entries may not be chosen arbitrarily, but have to be determined by what we call side-pumps.
The intuition for side-pumps lies in the following one dimensional example.
Consider a derivation rule $\anonterm\to \anngvasp.\anonterm.+1$, where $\arun$ has $0$ effect for $\arun\in\runsof{\anngvasp}$, and $\arunp\in\runsof{\anonterm}$ exists.
Then, if a counter value of $a\in\N$ enables $\arun$ and $\arunp$, then it also enables $\arun.\arunp, \arun^{2}.\arunp, \arun^{3}.\arunp, \ldots$.
This means that starting from $a\in\N$, we can repeat $\anonterm\to \anngvasp.\anonterm+1$, to get $\anonterm\to \arun.\arunp.+1\to (\arun)^{2}.\arunp.(+1)^{2}\to\ldots$.
The prefix remains stable, while the suffix pumps higher and higher, we conclude $\postfuncNof{\anngvas}{a, \anonterm}=\set{\omega}$.
We search for such situations in higher dimensions.
Formally, $\pumpingof{\atreep}\in\Nomega^{d}$ is a marking that is obtained via side-pumps as follows.
Let $\anode$ with $\nodemarkingof{\anode}=(\amarking, \asymbol, \amarkingp)$ be the rightmost child of the root in $\atreep$. 
We define $\at{\pumpingof{\atreep}}{\acounter}=\omega$, if there is a subtree $m$ of $\atreep$ with $m.\inlabel=\atreep.\inlabel$, $m.\outlabel \leq\amarkingp$, and $\at{m}{\acounter}<\at{\amarkingp}{\acounter}$.
We let $\at{\pumpingof{\atreep}}{\acounter}=\at{\amarkingp}{\acounter}$ otherwise. 
The following figure illustrates this situation.
\begin{center}
\includegraphics[scale=0.8,page=1]{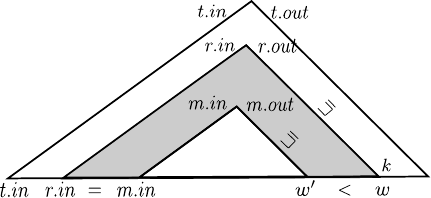}
\end{center}

We denote the set of witness trees by $\witnessset$, and the set of witness trees with maximal height $h\in\N$ by $\witnesssetof{h}$.
We also write $\witnessset(\amarking, \asymbol)=\setcond{\atree\in\witnessset}{\atree.\inlabel=\amarking,\;\atree.\symlabel=\asymbol}$, and
$\witnesssetof{h}(\amarking, \asymbol)=\witnesssetof{h}\cap\witnessset(\amarking, \asymbol)$ for $h\in\N$ to restrict the input label and root symbol by $(\amarking, \asymbol)\in\fullinputdom\times(\nonterms\cup\trms)$.

Witness trees are sound with respect to coverability.
The formal proof has ben moved to \Cref{Section:WitnessTreesL3}, but the only interesting case is the soundness of output markings introduced by $\pumpingof{\atree}$.
The argument here is the same as the one we employed in the one dimensional example.
The function $\pumpingof{\atree}$ discovers a context, such that under pumping, the input side remains stable but the output side grows unboundedly in $\omegaof{\pumpingof{\atree}}\setminus\omegaof{\atree_{\righttag}.\outlabel}$, where $\atree_{\righttag}$ is the right-subtree of $\atree$.

\begin{restatable}{lemma}{LemmaWTSoundMainPaper}\label{Lemma:WTSound}
    For each $\atree\in\witnessset$, we have $\atree.\inlabel\sqsubseteq\inof{\atree.\symlabel}$, $\atree.\outlabel\sqsubseteq\outof{\atree.\symlabel}$, $\omegaof{\atree.\inlabel}\subseteq\omegaof{\atree.\outlabel}$, and $\atree.\outlabel\in\downclsof{\postfuncNof{\anngvas}{\atree.\inlabel, \atree.\symlabel}}$.
\end{restatable}

The witness trees are also complete with respect to coverability.
The formal proof is a straightforward induction over the parse tree depth, and can be found in \Cref{Section:WitnessTreesL3}.
We use the fact that the image of $\postfunc$ overapproximate the reachable markings, and $\pumpingof{\atree}$ only introduces more $\omega$ counters.
\begin{restatable}{lemma}{LemmaWTCompleteMainPaper}\label{Lemma:WTComplete}
    Let $\asymbol\in\nonterms\cup\trms$, and $(\amarking, \arun, \amarkingp)\in\runsof{\asymbol}$.
    Then, for all $\amarking_{\omega}\in\fullinputdom$, with $\settoomega{\unconstrained}{\amarking}\sqsubseteq\amarking_{\omega}$, there is a tree $\atree\in\witnessset$ with $\atree.\inlabel=\amarking_{\omega}$, $\atree.\symlabel=\asymbol$, and $\amarkingp\leq\atree.\outlabel$. 
\end{restatable}
Witness trees under a depth bound $h$ are also effectively constructable if given an input marking.
The formal proof can be found in \Cref{Section:WitnessTreesL3}, but the broad argument is as follows.
We have only finitely many parse trees of a given depth.
We go from left-to-right, and call $\postfunc$ using our assumption \Cref{Corollary:WTSearchAssumption}.
If an input counter is outside of $\smallinputdom$, we need to close the branch.
Thanks to \Cref{Corollary:WTSearchAssumption}, we can compute $\postfunc$ for this input and achieve our goal.
A final consideration of $\pumpingof{-}$ yields the result. 
\begin{restatable}{lemma}{LemmaWTEffectivenessMainPaper}\label{Lemma:WTEffectiveness}
    Let $(\amarking, \asymbol)\in\fullinputdom\times(\nonterms\cup\trms)$ and $h\in\N$.
    Then, we can effectively construct $\witnesssetof{h}(\amarking, \asymbol)$.
\end{restatable}
Finally, we observe the most crucial property, saturation.
Saturation allows us to effectively find a maximal depth, beyond which there are no further witness trees.
This concludes the proof of \Cref{Lemma:FreePumpPostLow}.
This is thanks to the pumping property.
We suppose that we get unboundedly high witness trees, and a K\"onig's Lemma argument applies to show unboundedly many pumps along one branch.
This cannot happen, since we only have $2d$ counters.
\begin{restatable}{lemma}{LemmaWTConvergenceMainPaper}\label{Lemma:WTConvergence}
    Let $h\in\N$.
    If $\witnesssetof{h}=\witnesssetof{h+1}$, then $\witnesssetof{h}=\witnesssetof{h'}$ for all $h'\in\N$ with $h'\geq h$.
    Furthermore, there is an $h\in\N$ with $\witnesssetof{h}=\witnesssetof{h+1}$. 
\end{restatable}
\begin{proof}
    First, we show that for all $i\in\N$ if $\witnesssetof{i}=\witnesssetof{i+1}$, then $\witnesssetof{i+1}=\witnesssetof{i+2}$ by an inductive argument. 
    For any $\atree\in\witnesssetof{i+2}$, the left- and right-subtrees are witness trees.
    Therefore, they must belong to $\witnesssetof{i+1}$, but since $\witnesssetof{i+1}=\witnesssetof{i}$, the height of $\atree$ is at most $i+1$, which implies $\atree\in\witnesssetof{i+1}$.
    
    Now we show that there is an $h\in\N$ with $\witnesssetof{h}=\witnesssetof{h+1}$.
    Suppose that $\witnesssetof{h}\neq\witnesssetof{h+1}$ for all $h\in\N$.
    Then, for each $h\in\N$, there is a $\atree\in\witnesssetof{h+1}\setminus\witnesssetof{h}$.
    Consider the graph $H=(Y, E)$, where 
    \begin{align*}
        Y&=\setcond{(h, \atree)\in\N\times\witnessset}{h\neq 0,\;\atree\in\witnesssetof{h}\setminus\witnesssetof{h-1}}\\
        E&=\setcond{((h, \atree), (h', \atree'))\in Y^{2}}{\\
        &\hspace{6em} h'=h+1,\;\atree\text{ is a subtree of }\atree'}
    \end{align*}
    The nodes of the graph are witness trees annotated with their height, $(h, \atree)$.
    The node $\atree$ is an immediate successor of $\atree'$, if their heights are adjacent, and $\atree$ is a subtree of $\atree'$.
    Thus, the edges capture the bottom up construction of the tree.
    It must hold that $H$ is infinite.
    Clearly, any $(h+1, \atree)\in Y$ with $h\geq 1$ has a predecessor $(h, \atreep)\in Y$.
    If this did not hold, then the contradiction $\atree\in\witnesssetof{h}$ would hold.
    We claim that for any $h\geq 1$, $\witnesssetof{h}\setminus\witnesssetof{h-1}$ is finite.
    By \Cref{Lemma:WTEffectiveness}, $\witnesssetof{h}(\amarking, \asymbol)$ is effectively constructable and therefore finite for all $(\amarking, \asymbol)\in\fullinputdom\times(\nonterms\cup\trms)$.
    We argue that witness trees beyond depth $1$ do not have input labels resp. symbol labels outside of $\fullinputdom\times\trms\cup\largeinputdom\times\nonterms$.
    Note that if $(\amarking, \asymbol)\in\fullinputdom\times\trms\cup\largeinputdom\times\nonterms$ then, any node $\anode$ with $(\anode.\inlabel, \anode.\symlabel)=(\amarking, \asymbol)$ must be a leaf.
    Thus, $\witnesssetof{h}(\amarking, \asymbol)\subseteq\witnesssetof{0}$ for $(\amarking, \asymbol)\in\fullinputdom\times\trms\cup\largeinputdom\times\nonterms$.
    Then, $\witnesssetof{h}\setminus\witnesssetof{h-1}\subseteq\bigcup_{(\amarking, \asymbol)\in\smallinputdom\times\nonterms}\witnesssetof{h}(\amarking, \asymbol)$.
    However, $\smallinputdom\times\nonterms$ is finite, then so is $\witnesssetof{h}\setminus\witnesssetof{h-1}$.
    Then, since only edges that exist go from height $h$ to $h+1$, the graph $H$ is finitely branching.
    Since all nodes are connected to at least one node in $(\witnesssetof{1}\setminus\witnesssetof{0})\times\set{1}$, and this set is finite, $H$ only has finitely many components.
    We apply K\"onig's Lemma to get a sequence $[(h, \atree_{h})]_{h\in\N\setminus\set{0}}$ in $H$ with $((h, \atree_{h}), (h+1, \atree_{h+1}))\in E$ for all $h\in\N\setminus\set{0}$.
    This implies that for all $h\in\N\setminus\set{0}$, the tree $\atree_{h}$ has height $h$, and that it is a subtree of $\atree_{h+1}$.
    Using the fact that $\Nomega^{d}$ is a WQO and that $\smallinputdom\times\nonterms$ is finite, we get a subsequence $[(\phi(h), \atreep_{h})]_{h\in\N}$ of $[(h, \atree_{h})]_{h\in\N\setminus\set{0}}$, where $\atreep_{h}.\inlabel$ and $\atreep_{h}.\symlabel$ are constant across $h\in\N$, and $\atreep_{h}.\outlabel\leq\atreep_{h+1}.\outlabel$ for all $h\in\N$.
    Also, since $\atreep_{h}$ is a subtree of $\atreep_{h+1}$, and no node may have the same labeling as its successor, we know that $\atreep_{h}.\outlabel<\atreep_{h+1}.\outlabel$ must hold for all $h\in\N$.
    But $\atreep_{h}.\inlabel=\atreep_{h+1}.\inlabel$, and $\atreep_{h}.\symlabel=\atreep_{h+1}.\symlabel$, and we have some $j\leq d$ with $\atreep_{i}.\outlabel[j]<\atreep_{h+1}.\outlabel[j]$, which implies $\pumpingof{\atreep_{h+1}}[j]=\omega$.
    As a consequence, we get $\cardof{\omegaof{\atreep_{h}.\outlabel}}<\cardof{\omegaof{\atreep_{h+1}.\outlabel}}$.
    Then, $[|\omegaof{\atreep_{h}.\outlabel}|]_{h\in\N}$ must grow unboundedly.
    This is a contradiction to $\omegaof{\atreep_{h}.\outlabel}\subseteq [1,d]$ for all $h\in\N$. 
\end{proof}


\newcommand{\rank}[0]{\text{rank}}
\newcommand{\last}[0]{\text{last}}
\newcommand{\inc}[0]{\text{inc}}
\newcommand{\fastgrowing}{F}
\newcommand{\fastgrowingof}[1]{\fastgrowing(#1)}
\newcommand{\anelemfunc}{g}
\newcommand{\anelemnfuncof}[1]{\anelemfunc(#1)}
\newcommand{\anordinal}{\alpha}
\newcommand{\anordinalp}{\beta}
\newcommand{\fundseq}{\lambda}
\newcommand{\fundseqof}[1]{\fundseq_{#1}}
\newcommand{\funcnames}{\mathsf{Func}}
\newcommand{\names}{\mathsf{Name}}
\newcommand{\maxprogramcounter}{\mathsf{pc}}

\newcommand{\karpmillerfunc}{\mathsf{km}}
\newcommand{\cgfull}{\mathsf{cg}_{\N, \acounterset, \acountersetp}}
\newcommand{\cgint}{\mathsf{cg}_{\Z}}
\newcommand{\cgforget}{\mathsf{cg}_{low}}
\newcommand{\treesearch}{\mathsf{witness}}
\newcommand{\postperf}{\mathsf{post}_{\perfect}}
\newcommand{\preperf}{\mathsf{pre}_{\perfect}}
\newcommand{\postsearch}{\mathsf{post}_{\mathsf{search}}}
\newcommand{\presearch}{\mathsf{pre}_{\mathsf{search}}}

\newcommand{\dimrank}{\mathsf{Dim}}
\newcommand{\ngvasrank}{\mathsf{Rank}_{ngvas}}
\newcommand{\intprogrank}{\mathsf{Rank}_{prog}}
\newcommand{\asnapshot}{\sigma}
\newcommand{\asnapshotp}{\tau}
\newcommand{\ahistory}{h}
\newcommand{\atime}{t}

\newcommand{\acountersetp}{Y}

\section{Complexity} \label{SectionComplexity}

So far, we have settled the decidability of the emptiness problem of (N)GVAS.
In this section we analyze the complexity of the algorithm in terms of the fast-growing functions of the Grzegorczyk hierarchy, see \cite{Schmitz16} for a thorough introduction.
In particular, we prove the following theorem.

\begin{restatable}{theorem}{TheoremTimeBoundPerf}
The complexity of \(\perffun\) is in \(\mathfrak{F}_{\omega^{d+3}}\), where \(d\) is the number of counters of the input NGVAS. \label{TheoremLabelTimeBoundPerf}
\end{restatable}

The ideas of our complexity analysis were already explained in Section \ref{SectionLevelOneComplexity}: Let us remind the reader that in Section \ref{SectionLevelOneComplexity} we explained a programming model with support for both fully recursive and tail-recursive calls, and the goal of this section, Section \ref{SectionComplexity}, is to show that \(\perffun\) adheres to this programming model, as well as to finish the proof of Theorem \ref{TheoremLengthBoundBadPartiallyNestedSequences}. 

Remember that the programming model mainly considers ranks in \(\N^{\beta_1} \times \setfun \times \N^{\beta_2}\). Specifically, for \(\perffun\) we will use \(\beta_1=d+3\) which corresponds to the recursive rank \(\recrankof{\anngvas}\) of the current NGVAS \(\anngvas\), and \(\beta_2=2d+1+\beta_f\) which contains the iterative rank \(\itrankof{\anngvas}\) of \(\anngvas\) and an auxiliary rank \(\beta_f \leq 2d+3\) for every function \(f \in \setfun\). We will prove that any step of \(\perffun\) as in the programming model requires at most primitive-recursive time, s.t. the execution of \(\perffun\) gives rise to a \((\controlfun, |\anngvas_0|)\)-controlled bad partially nested sequence in \(\N^{d+3} \times \setfun \times \N^{3d+5}\), where \(\controlfun\) is some primitive-recursive function. At that point Theorem \ref{TheoremLengthBoundBadPartiallyNestedSequences} applies and leads to the stated bound of \(\mathfrak{F}_{\omega^{d+3}}\).

The rest of this section is hence structured as follows. First, in Section \ref{SectionFastGrowingFunctions}, we formally define the fast-growing functions. Then, in Section \ref{SubsectionLevelTwoProofOfLengthBound}, we prove Theorem \ref{TheoremLengthBoundBadPartiallyNestedSequences}. Next, in Section \ref{SectionListSubprocedures}, we provide the full list of functions \(\setfun\) called by \(\perffun\). Afterwards, in Section \ref{SectionAuxiliaryRanksForProcedures}, we will explain how to obtain a local rank for each one of the functions in sequence.

\subsection{Fast-Growing Functions} \label{SectionFastGrowingFunctions}

The fast-growing functions \(F_{\anordinal}\) for ordinals \(\anordinal \leq \omega^{\omega}\) are defined by 
\begin{align*}
&F_0(n):=n+1 && F_{\anordinal+1}(n):=F^{(n+1)}_{\anordinal}(n) && F_{\anordinal}(n):=F_{\lambda_n(\anordinal)}(n),
\end{align*}  where \(F^{(n)}\) is \((n)\)-fold application of \(F\) and \([\fundseqof{n}(\anordinal)]_{n\in\N}\) for a limit ordinal \(\alpha\) denotes the \emph{fundamental sequence} for this limit ordinal. 
The fundamental sequence of a limit ordinal $\alpha$ is a sequence of ordinals whose supremum is $\alpha$, and which is defined by the following rules for all $n\in\N$ and $k\geq 1$.
It holds that $\fundseqof{n}(\anordinalp+\omega^{k})=\anordinalp+\omega^{k-1}\cdot n$, where $\anordinalp+\omega^{k}$ is in Cantor Normal Form \cite{Schmitz16}, and $\fundseqof{n}(\omega^{\omega})=\omega^{n}$.

For example \(F_1(n)=F_0^{(n+1)}(n)=2n+1\), \(F_2(n)\) is an exponential function and \(F_3(n)=F_2^{(n+1)}(n)\) is related to the tower of exponentials. 
Finally, \(F_{\omega}(n)=F_n(n)\) is (one variant of) the Ackermann-function, defined by diagonalizing over \(F_{n}\) for \(n \in \N\). 
As we see in these examples, incrementing \(\alpha\) by a natural number corresponds to repeated application, and limit ordinals \(\alpha\) correspond to diagonalization.

Observe more generally that, reusing the notion of lexicographic decrement \(\LexDec(\alpha, n)=\alpha+\omega^{i-1} \cdot n + \dots +\omega^0 \cdot n\) as in the proof of Theorem \ref{TheoremLengthBoundBadPartiallyNestedSequences}, the definition of the fast-growing functions can be expanded to \(F_{\alpha}(n)=F_{\LexDec(\alpha, n-1)}^{(n)}(n)\): Namely 
\begin{align*}
F_{\alpha+\omega^i}(n)&=F_{\alpha+\omega^{i-1} \cdot n}(n)=F_{\alpha+\omega^{i-1} \cdot (n-1) + \omega^{i-2} \cdot n}(n)=\\
&=\dots=F_{\LexDec(\alpha, n-1)}^{(n)}(n).
\end{align*}

This observation will be used in the proof of Theorem \ref{TheoremLengthBoundBadPartiallyNestedSequences}.

%
%

The function \(F_{\omega^{\omega}}\) is called \emph{Hyper-Ackermann}. By the above definition, and since \(\lambda_n(\omega^{\omega})=\omega^n\) is the canonical fundamental sequence for this limit ordinal, we have \(F_{\omega^{\omega}}(n)=F_{\omega^n}(n)\).

\subsection{Proof of Theorem \ref{TheoremLengthBoundBadPartiallyNestedSequences}} \label{SubsectionLevelTwoProofOfLengthBound}

\TheoremLengthBoundBadSequences*

\begin{proof}
The first part is a repetition of Section \ref{SectionLevelOneComplexity}, we add it for coherence. We mark in boldface the point when the proof starts diverging.

Note that we can get rid of $\setfun$ by setting \(\beta_2'=\beta_2 \cdot \cardof{\setfun}\) and storing the auxiliary ranks of one call for every function \(\afun\) in a common larger rank from \(\N^{\beta_2'}\). 

The proof now has two steps, we first analyze the shape of a worst-case controlled bad partially nested sequence to obtain a recursive formula, and then we simplify this formula to be able to analyze it.

To obtain long bad partially nested sequences, we need to reduce ordinals in the least possible way. 
Given an ordinal \(\beta< \omega^{\omega}\) in Cantor normal form, we define the worst-case \emph{lexicographic decrement} as the function that finds the monomial $\omega^i\cdot k_i$ with the least exponent $\min(\beta)=i$, decrements the coefficient $k_i$, and fills up the missing monomials with smaller exponents. 
The coefficient that should be used for the added monomials is given as a parameter $n$ to the lexicographic decrement.  
We let \(\LexDecof{\beta, n}=\beta-\omega^{min(\beta)}+\sum_{j<\min(\beta)} \omega^{j}\cdot n\). 
For example, $\LexDecof{\omega^3+\omega^2\cdot 5, n}=\omega^3+\omega^2\cdot 4 + \omega\cdot n + n$. 

Let \(\maxlengthof{\amainrank, \anauxrank,\controlfun}{\initbound}\) be the maximal length of a \((\controlfun,\initbound)\)-controlled bad partially nested sequence starting at \((\amainrank, \anauxrank) \in \N^{\beta_1} \times \N^{\beta_2'}\). 
We obtain a recursive formula for \(\maxlengthof{\amainrank, \anauxrank, \controlfun}{\initbound}\) by considering the worst-case behavior the bad sequence can show. 
A step in the sequence may execute a push on the stack, and thereby turn \((\amainrank, \anauxrank, \aheight=0)\) into \((\LexDecof{\amainrank, \controlfunof{\initbound}}, \controlfunof{\initbound}, \aheight=1)\). 
So we reduce the main rank, but in the least possible way. 
We then execute a bad nested sequence on this smaller new value.  
After a large number of steps \(t\), we return to \((\amainrank, \LexDecof{\anauxrank, \controlfun^{(t)}(\initbound)}, \aheight=0)\). 
We have only reduced the auxiliary rank, and again only in the least possible way, which now even includes the large number $\controlfun^{(t)}(\initbound)$ as the coefficient for the new monomials.
We now perform a bad nested sequence from this new stack. 

By formalizing this argument, we obtain what is often called a \emph{descent equation} \(\maxlengthof{\amainrank, \anauxrank, \controlfun}{\initbound}=1+l_1+l_2\), where 
\begin{align*}
l_1 &=\maxlengthof{\LexDec(\amainrank, \controlfun(\initbound)), \controlfun(\initbound), \controlfun}{\controlfun(\initbound)} \\
l_2 &=\maxlengthof{\amainrank, \LexDec(\anauxrank, \controlfun^{(l_1+1)}(\initbound)),\controlfun}{\controlfun^{l_1+1}(\initbound)}
\end{align*} are respectively the lengths of the first and second half of the worst-case sequence above.

Our descent equation is difficult to analyze because it refers to nested sequences.
We therefore translate nested sequences into ordinary (non-nested) sequences, using summarization plus stuttering.
By summarization, we mean that we over-approximate the rank modification that may happen while the stack height is non-zero using a new control function and a single step. 
By stuttering, we mean that the high counter values introduced by the summary step allow us to execute a bad sequence that is at least as long as what we had. 
We illustrate the idea in the case $\beta_1=1$ studied above.
Starting at \(k\), we can bound \(\maxlength_{k, \anauxrank, \controlfun}\) by a \emph{non-nested (nn)} sequence in \(\N^{\beta_2'}\) whose control function is \(\controlfun^{(2+\maxlength_{k-1, \controlfun(\initbound), \controlfun})}\). 
Indeed, in the worst-case the first step arrives at \((k-1, \controlfun(\initbound),h=1)\), and after an additional at most \(\maxlength_{k-1, \controlfun(\initbound), \controlfun}+1\) steps we are back at \(h=0\). For any nested sequence, the subsequence defined by \(h=0\) is a standard bad sequence that is controlled by \(\controlfun^{(2+\maxlength_{k-1, \controlfun(\initbound), \controlfun})}\). 

\textbf{Starting from here} the text is diverging.

We define the new control function 
\[\controlfun_{\amainrank}= \controlfun^{(2 + \maxlengthof{\LexDec(\amainrank, inp), inp, \controlfun}{\controlfun(\initbound)})},\]
where we write \(inp\) to mean that the input value \(\controlfun(\initbound)\) of the function should be substituted here. Otherwise in the notation \(\maxlengthof{\LexDec(\amainrank, \controlfun(\initbound)), \controlfun(\initbound), \controlfun}{\controlfun(\initbound)}\), conventions might suggest that the value \(\controlfun(\initbound)\) of the subscript is fixed. We want to clarify that it is not fixed.

Let \(\maxlengthof{nn,\beta_2',\controlfun_{\amainrank}}{\initbound}\) denote the length of the longest \((\controlfun_{\amainrank}, \initbound)\)-controlled bad non-nested sequence in \(\N^{\beta_2'}\). By considering the subsequence defined by \(h=0\), we obtain \(\maxlengthof{\amainrank, \anauxrank, \controlfun}{\initbound} \leq (\maxlengthof{nn, \beta_2', \controlfun_{\amainrank}}{\initbound}+1)^2\): Here, the \(+1\) arises from the fact that ending at \((\amainrank, \vect{0})\), we still have to perform one more sequence afterwards, and squaring deals with the intermediate steps: Per step of the new non-nested sequence, in the worst-case we had \(\maxlengthof{nn, \beta_2', \controlfun_{\amainrank}}{\initbound}\) many steps in the nested sequence.

We obtain the following recursion for \(\controlfun_{\amainrank}\).
\begin{align*}
\controlfun_{\amainrank}(\initbound)&=\controlfun^{(2+\maxlengthof{\LexDec(\amainrank, inp), inp, \controlfun}{\controlfun(\initbound)})}\\
&\leq \controlfun^{(2+(\maxlengthof{nn, \beta_2', \controlfun_{\LexDec(\amainrank, inp)}}{\controlfun(\initbound)}+1)^2)}
\end{align*}
Similar to the recursion \(F_{\alpha}(n)=F_{\LexDec(\alpha, n-1)}^{(n)}(n)\), \(\controlfun_{\amainrank}\) is defined depending on a lexicographic decrement \(\controlfun_{\LexDec(\amainrank, inp)}\), though the precise recursion is different. 

Now it only remains to analyze this recursion. Let us first state the inequality, following by intuition. 
\[\controlfun_{\amainrank}(\initbound) \leq F_{\amainrank \beta_2' + \alpha}(\controlfun((\beta_1+\beta_2') \cdot \initbound)),\] 
i.e.\ \(\controlfun_{\amainrank}\) is at level \(\amainrank \beta_2'+\alpha\) of the fast-growing hierarchy. This inequality is proven by induction on \(\amainrank\). In addition to the steps of \cite[Theorem VI.1]{LerouxPS14} for nested sequences, since our recursion refers to the length of non-nested sequences, the inductive step also has to redo some of the analysis in \cite[Prop. 5.2]{FigueiraFSS11} for non-nested sequences. 

Instead of the computation, let us explain the general intuition. The fast-growing hierarchy is robust w.r.t. the exact definition, as long as the recursion scheme of refering to \(\LexDec\) is the same, one obtains hierarchy level \(\mathfrak{F}_{\amainrank}\). The reason for the factor \(\beta_2'\) is that \(\maxlength_{nn, \beta_2', \controlfun}\) is at level \(\mathfrak{F}_{\alpha+\beta_2'}\) of the hierarchy \cite[Prop. 5.2]{FigueiraFSS11}, if \(\controlfun\) is at level \(\alpha\). In our case, every single step of the recursion hence adds \(\beta_2'\) to the subscript. The \(+\alpha\) is due to \(\controlfun \in \mathfrak{F}_{\alpha}\).
\end{proof}

\subsection{List of Subprocedures} \label{SectionListSubprocedures}

The list of subprocedures is as follows, with explanations afterwards.

\begin{align*}
    \setfun=\ &\set{\perfect, \clean, \basisfire, \centerdec, \\
    &\refineeq, \refinepump, \refineintpump, \postfunc, \prefunc, \\
    &\postsearch, \presearch, \karpmillerfunc, \cgfull, \cgint}.
\end{align*}

We explain the names.
The functions from $\perfect$ to $\prefunc$ are those from our development.
The functions $\postsearch$ and $\presearch$ handle the hard case 2 of computing $\postfunc$ (\Cref{Section:PostHardCase2}), where we conduct a witness tree search.
The Karp-Miller construction is handled by the function $\karpmillerfunc$.
The coverability grammar construction, where we assume the set of $\omega$'s $\acounterset$ on the input, and $\acountersetp$ on the output is handled by $\cgfull$.
Finally, the coverability grammar construction that uses the $\Z$-approximations is handled by $\cgint$.
The rest of this section provides an auxiliary rank for every function.


\subsection{Auxiliary Ranks for Subprocedures} \label{SectionAuxiliaryRanksForProcedures}

For many procedures in our algorithm, the auxiliary ranks are simple, the sizes of working lists, the number of SCCs, etc.
This insight extends to all functions except $\clean$, $\karpmillerfunc$, $\cgfull$, $\cgint$, and $\presearch$ easily.
We argue that these also have auxiliary ranks that can be expressed in $\N^{3d+5}$.

\textbf{Auxiliary Rank for }\(\perffun\): 
A $\perfect$ call iteratively constructs a set $\adecomp$ of NGVASes by applying refinements, and after each refinement step, the rank of an NGVAS in the set decreases.
We understand $\perfectof{\anngvas}$ as a round-robin process that decomposes each NGVAS once per round.
Namely, it constructs the decomposition $\adecomp_{i}$ at round $i$, with $\adecomp_{0}=\set{\anngvas}$, where $\adecomp_{i+1}$ is constructed from $\adecomp_{i}$ by applying the decompositions once to each $\anngvasp\in\adecomp_{i}$.
Thus, the measure of progress at round $i$ is $(\max_{\anngvasp\in\adecomp_{i}}\rankof{\anngvasp}, \cardof{\adecomp_{wl}})\in\N^{3d+4}\times\N$, where $\adecomp_{wl}\subseteq\adecomp_{i}$ is the set NGVASes that are waiting for decomposition in the current round.
After the round is complete, each NGVAS is decomposed at least once, so $\max_{\anngvasp\in\adecomp_{i}}\rankof{\anngvasp}$ decreases.

%

\textbf{Auxiliary Rank for }\(\clean\): In $\clean$, we call $\clean$, $\perfect$, and $\basisfire$ on lower SCCs of the wNGVAS at hand.
Since these are lower SCCs, we do not run into issues with the comparability of the $\ngvasrank$ component between call levels.
We incorporate the number of recursive calls we need to make into the auxiliary rank: For example for the subprocedure making all children perfect, we keep a counter storing how many children have not yet been handled. This was already mentioned in Section \ref{SectionLevelOneComplexity}.
%

For functions $\karpmillerfunc$, $\cgint$, $\cgfull$, $\postsearch$, and $\presearch$, progress is less direct.
These functions call $\postfunc$ resp. $\prefunc$ until the set of explored configurations resp. trees reach a saturation.
The saturation is only ensured by a well quasi order.

\textbf{Auxiliary Rank for }\(\karpmillerfunc\): Remember that in Section \ref{SectionLevelOneComplexity} we explained types, and how they can be used as auxiliary rank for \(\karpmillerfunc\). We hence do not repeat this part of \(\karpmillerfunc\) here.

There is however an implementation detail that is crucial for the soundness of the complexity analysis.
Our algorithm depends on finding an upper bound on pumping derivations in NGVAS that are less complicated than the current input.
In the decidability proof, we make an enumeration argument for the sake of simplicity.
Here, we need to assume an implementation of $\karpmillerfunc$ that not only verifies \perfectnesspumping resp. \perfectnesspumpingintnospace, but also returns an upper bound on the derivations if these conditions hold.
We can ensure this by assuming the following two modifications.
First, we assume that $\postfunc$ and $\prefunc$ not only return the set of coverable values, but also return a set of perfect NGVAS that witness these values.
The intention is to use \Cref{TheoremIterationLemmaNonLinearOverview}, and cheaply construct runs that pump the $\omega$ counters in the images of $\postfunc$ and $\prefunc$.
In almost all cases, $\postfunc$ and $\prefunc$ already construct perfect decompositions of the NGVAS that correspond to their query.
However, this is not strictly true for the Hard Case 2 in \Cref{Section:PostHardCase2}, where we conduct a witness tree search, whose leaves are $\postfunc$ resp. $\prefunc$ calls.
But, assuming that $\postfunc$ and $\prefunc$ readily return NGVASes, these trees can be encoded as a larger NGVAS that incorporates the return values of $\postfunc$ resp. $\prefunc$.
Second, in order to at all use \Cref{TheoremIterationLemmaNonLinearOverview} for pumping, we assume an implementation where each perfect NGVAS also stores the  pumping derivations that witness their perfectness.
With this assumption, we can construct pumps in primitive recursive time as described by \Cref{TheoremIterationLemmaNonLinearOverview}.

\textbf{The remaining functions}: Even though the witness tree search conducted by $\postsearch$ and $\presearch$ seems to have a different structure, the same argument as for \(\karpmillerfunc\) can also be applied here.
In the context of e.g. $\postsearch$, we have a bound $\abigbound\in\N$ that we have already computed, i.e. part of the size of our call stack, and we consider input markings in $[0, \ldots, \abigbound, \omega]^{d}$.
Whenever we encounter a node with an input marking that does not belong to $[0, \ldots, \abigbound, \omega]^{d}$, said node is closed by a $\postfunc$ call.
This means that if a witness tree has height $\geq 2$, then at most one of its subtrees is closed by a $\postfunc$ call.
Thus, the accelerations in a witness tree have a linear structure.
To formalize this, we argue over the graph $H=(Y,E)$ constructed in the %
termination proof of Section 2.3.5. 
Each edge between trees of height $\geq 2$ incurs a blow-up caused by at most one $\postfunc$ call.
This is the same behaviour we encounter in $\karpmillerfunc$.
Thus, the branches can be similarly ranked, and progress is ensured by a measure in $\N^{2d+2}$.
The height $=1$ is a special case, because each subtree can incur a $\postfunc$ blow up.
Thanks to the program counter, this is handled without causing comparability.
This concludes our argument.

\newpage
~
\newpage

\section{The Details of Proofs and Constructions}\label{Section:AppendixL3}
The following sections of the appendix contain details omitted from the main paper, or the first part of the appendix.
These sections are meant to be consulted as they are needed, and there is not much continuity between the sections.

\section{Appendix: Wide Tree Theorem}

\WideTreeTheorem*
\begin{proof}
Let us write $\arun$ for terminal sequences in this proof.   
For every non-terminal $\anonterm\in\nonterms$, we define the grammar $\agramdef$ that coincides with $\agram$ except that it has $\anonterm$ as the start non-terminal. 
Since all non-terminals are useful in $\agram$ and $\agram$ is strongly connected, all non-terminals are useful in $\agramdef$. 
Indeed, for $\anontermp\in\nonterms$ we have
\begin{align*}
\anonterm\rightarrow^*\asentform_1.\startnonterm.\asentform_2\rightarrow^*\asentform_1.\asentformp_1\anontermp.\asentformp_2.\asentform_2\rightarrow^*\asentform_1.\arun.\asentform_2\rightarrow^*\arun'\ .
\end{align*}
The first derivation is by strong connectedness. 
The second and third derivations exist, because $\anontermp$ is useful in $\agram $. 
The last derivation uses the fact that from every non-terminal we can derive a terminal sequence, by usefulness in $\agram$. 

We strengthen the statement and show that for every number of copies $\aconst\in\N$ and for every $\anonterm\in\nonterms$, we can obtain a parse tree $\atreeparamdef\in\treesof{\agramdef}$ and a provenance tracking function on $\atreeparamdef$ as promised. 
The notation $\atreeparamdef$ is meant to indicate that the yield of this tree has the form $\arun_1.\anonterm.\arun_2$ with $\arun_1, \arun_2\in\trms^*$, so $\anonterm$ is the single non-terminal. 

{\bfseries Base case}\quad Let $\aconst=1$ and $\anonterm\in\nonterms$.  
The homogeneous variant of Esparza-Euler-Kirchhoff is independent of the choice of the start non-terminal: the equations are the same for $\agram$ and for $\agramdef$. 
Combined with the remark that all non-terminals are useful in $\agramdef$, we can invoke Theorem~\ref{Theorem:EEK} and obtain $\startnonterm\xrightarrow{\aprodseq}\asentform$ with $\paramparikhof{\prods}{\aprodseq}=\avec_{\prods}$. 
For the shape of the sentential form, Lemma~\ref{Lemma:EEKConverse} 
shows $\paramparikhof{\nonterms}{\asentform}=1_{\anonterm} + \effof{\nonterms}\cdot \avec_{\prods} = 1_{\anonterm}$. 
So $\asentform = \arun_1.\anonterm.\arun_2$ with $\arun_1, \arun_2\in\trms^*$. 
We turn this derivation sequence into a parse tree  $\atreeparam{1}{\anonterm}$. 
The height requirement is trivial and the provenance tracking can have at most $1$ incomplete copy of $\prodvec$ in a prefix. 

{\bfseries Step case}\quad Let $\aconst>1$ and consider $\anonterm\in\nonterms$. 
We determine the parse tree $\atreeparam{1}{\anonterm}$ as we have done in the base case.  
Let the yield be $\arun_1.\anonterm.\arun_2$. 
Since the grammar is non-linear and $\avec_{\prods}$ uses every production, $\atreeparam{1}{\anonterm}$ contains a node with at least two children that are non-terminals.
One of them may lead to the leaf $\anonterm$. 
The other, however, will lead to a production $\anontermp\rightarrow \arun$ that adds $\arun\in\trms^*$ to $\arun_1$ or $\arun_2$, say $\arun_2$. 
Then the parse tree can be written as $\atreepparam{1}{\anonterm, \anontermp\rightarrow \arun}$.  
We define $\aconst_1=\floorof{\frac{k-1}{2}}$ and $\aconst_2=\ceilof{\frac{k-1}{2}}$ so that $\aconst = \aconst_1+\aconst_2+1$. 
We invoke the induction hypothesis twice, for~$\anonterm$ with~$\aconst_1$ and for~$\anontermp$ with~$\aconst_2$.
There is the special case $\aconst_1=0$ in which we skip the first invokation. 
The hypothesis yields parse trees $\atreeparam{\aconst_1}{\anonterm}\in \treesof{\agramof{\anonterm}}$ and $\atreeparam{\aconst_2}{\anontermp}\in \treesof{\agramof{\anontermp}}$ together with provenance functions that have the properties in the strengthened statement. 
We insert these trees into $\atreepparam{1}{\anonterm, \anontermp\rightarrow \arun}$ and obtain
\begin{align*}
\atreeparam{\aconst}{\anonterm}\ =\ \atreepparam{1}{\atreeparam{\aconst_1}{\anonterm}, \atreeparam{\aconst_2}{\anontermp\rightarrow \arun}}\ .
\end{align*} 
Note that we moved $\anontermp\rightarrow \arun$ to the yield of $\atreeparam{\aconst_2}{\anontermp}$. 
So $\atreeparam{\aconst}{\anonterm}$ indeed has $\anonterm$ as the single non-terminal in the yield.

For the number of productions, we have
\begin{align*}
\paramparikhof{\prods}{\atreeparam{\aconst}{\anonterm}}\ \overset{\text(IH)}{=}
 \ \prodvec + \aconst_1 \cdot \prodvec + \aconst_2\cdot \prodvec\  =\ \aconst\cdot \prodvec\ .
\end{align*}
For the height, we argue similarly
\begin{align*}
\heightof{\atreeparam{\aconst}{\anonterm}}\  \leq\ &\ \heightof{\atreeparam{1}{\anonterm}}+ \max\set{\heightof{\atreeparam{\aconst_1}{\anonterm}}, \heightof{\atreeparam{\aconst_2}{\anontermp}}} \\
\argument{(IH)}\leq\ &\ \normof{\avec_{\prods}} + \ceilof{1+\ld \aconst_2}\cdot \normof{\prodvec}\\
 =\ &\ \ceilof{2+\ld \ceilof{\tfrac{\aconst-1}{2}}} \cdot \normof{\prodvec}\\
\leq\ &\ \ceilof{2+\ld \tfrac{\aconst}{2}} \cdot \normof{\prodvec}\\
=\ &\ \ceilof{2 - \ld 2 + \ld \aconst } \cdot \normof{\prodvec}\ .
\end{align*}

The provenance tracking function is defined as expected by combining the provenance tracking functions $\prov_1$ for $\atreeparam{1}{\anonterm}$, $\prov_{\anonterm}$ for $\atreeparam{\aconst_1}{\anonterm}$, and $\prov_{\anontermp}$ for $\atreeparam{\aconst_2}{\anontermp}$. 
More precisely, we shift the output of the latter functions by $+1$ resp. $1+\aconst_1$, and note that the order is invariant under the shift of identities.
We now have $\prov = \prov_1\discup\prov_{\anonterm}^{+1}\discup\prov_{\anontermp}^{1+\aconst_1}$. 
Note that the terminals created by $\anontermp\rightarrow \arun$ still have $\provof{\arun}=1$. 
For the order, consider a prefix $\asentform$ of $\yieldof{\atreeparam{\aconst}{\anonterm}}$.
By the shape of $\atreeparam{\aconst}{\anonterm}$, this prefix either (i) does not contain symbols from $\yieldof{\atreeparam{\aconst_2}{\anontermp\rightarrow \arun}}$ or (ii) it contains the full $\yieldof{\atreeparam{\aconst_1}{\anonterm}}$. 
In the former case, the order is bounded by $1+\ceilof{1+\ld \aconst_1}$. 
The copy of $\prodvec$ in $\atreeparam{1}{\anonterm}$ may be incomplete, and to this we add the maximal number of incomplete copies in a prefix of $\yieldof{\atreeparam{\aconst_1}{\anonterm}}$. 
The latter is bounded by $\ceilof{1+\ld \aconst_1}$ by the induction hypothesis. 
In the latter case, note that $\yieldof{\atreeparam{\aconst_1}{\anonterm}}$ does not contribute to the order, because it only contains complete copies of $\prodvec$. 
Hence, the order is bounded by $1+\ceilof{1+\ld \aconst_2}$. 
We then conclude with an estimation similar to that for the height. 
\end{proof}

\section{Proofs from Sections~\ref{Section:OutlineILProof} and~\ref{Section:IterationLemma}}

In order to make the argument of Section~\ref{Section:IterationLemma} formal, we need to prove many minor claims, for example explain why equation \eqref{Equation:UpdatesReach} holds, as well as minor consequences of \eqref{Equation:EmbeddingProductions}, etc. This appendix is hence a list of many minor lemmas to fill in the gaps, followed by finishing the explanation of pumping (cases 2 and 3 etc.) and the proof of Lemma \ref{Lemma:LowerBound}.

We start with some of these minor properties.
\begin{lemma}
\eqref{Equation:UpdatesReach} holds. 
\end{lemma}
\begin{proof}
Our goal is to show 
\begin{align*}
\vaseffof{\reachrun}\ = \ \updates \cdot \asol'[\avar_{\updates}].
\end{align*}
By \(\ptof{}{\prodvar, \termvar}\) we have $\paramparikhof{\trms}{\asentformreach}=\asol'[\avar_{\trms}]$.
As a first consequence, $\paramparikhof{\updates}{\asentformreach}=\restrictto{\asol'[\avar_{\trms}]}{\updates}$, the number of updates in $\reachrun$ that stem from $\asentformreach$ is as expected by $\asol'$. 
As a second consequence, we know that $\asentformreach$ contains  $\asol'[\termvar[\anngvasp]]$ many instances of~$\anngvasp$. 
Except for the first instance, we used (P4) to construct a run with the effect $\baseeffectdef$. 
For the first instance, we used the run $\arun_{\anngvasp}^{(\maxconst)}$.  
Summing the effects of these updates and runs we obtain
\begin{align*}
\ \vaseffof{\reachrun}=&\updates \cdot \restrictto{\asol'[\termvar]}{\updates}+\\
 &\hspace{1em}\sum_{\anngvasp \in \trms} (\asol'[\termvar[\anngvasp]]-1) \cdot \baseeffectdef + \baseeffectchoicedef  + \\
 &\hspace{7.5em}\maxconst \cdot \periodeffectdef \cdot \periodeffectchoicedef\\
=&\ \argument{Definition of $\baseeffectchoicedef$ and~\(\periodeffectchoicedef\)}\\
&\ \updates \cdot \restrictto{\asol'[\termvar]}{\updates}+ \\
&\hspace{1em}\sum_{\anngvasp \in \trms} \baseeffectdef \cdot \asol'[\termvar[\anngvasp]] + \\
&\hspace{5em}\periodeffectdef \cdot (\asol[\avar_{\anngvasp, \periodeffect}] + \maxconst \cdot \ahomsol[\avar_{\anngvasp, \periodeffect}])\\
=&\ \argument{\(\asol'=\asol+\maxconst \cdot \ahomsol\)}\\
&\argument{Definition of \(\computeupdatesof{\anngvas}{\termvar, \updvar}\), upper line}\\
&\ \updates \cdot \restrictto{\asol'[\termvar]}{\updates}+ \sum_{\anngvasp \in \trms} \updates \cdot \asol'[\avar_{\anngvasp, \updates}]\\
=&\ \argument{Definition of \(\computeupdatesof{\anngvas}{\termvar, \updvar}\), lower line}\\
&\ \updates \cdot \asol'[\avar_{\updates}]\qedhere
\end{align*}
\end{proof}
A simple but important consequence of \eqref{Equation:EmbeddingProductions} is:
\begin{lemma}
The following equation holds:
\begin{align}
\at{(\embedconst\cdot \ahomsol)}{\termvar} - \paramparikhof{\trms}{\asentformpumpleft.\asentformpumpright}\geq 0\ . \tag{$*$}\label{Equation:EmbeddingConsequence}
\end{align} 
\end{lemma}
\begin{proof}
\begin{align*}
&\ (\embedconst\cdot \ahomsol)[\termvar] - \paramparikhof{\trms}{\asentformpumpleft.\asentformpumpright}\\
=&\ \argument{$\ptof{\anngvas}{\prodvar, \termvar}$ and Lemma~\ref{Lemma:EEKConverse}}\\
&\ \effof{\trms}\cdot (\embedconst\cdot \ahomsol)[\prodvar]  - \effof{\trms}\cdot \paramparikhof{\prods}{\aprodseqpumpanngvas}\\
=&\ \argument{Linearity}\\
&\ \effof{\trms}\cdot ((\embedconst\cdot \ahomsol)[\prodvar]  - \paramparikhof{\prods}{\aprodseqpumpanngvas})\\
\geq &\ \argument{Inequality~\eqref{Equation:EmbeddingProductions} and monotonicity of $\effof{\trms}$}\\
&\ \effof{\trms}\cdot 1\\
\geq &\ \argument{Monotonicity of $\effof{\trms}$}\\
&\ 0\ .\qedhere
\end{align*}
\end{proof}

By~\eqref{Equation:EmbeddingConsequence} and $\computeupdatesof{\anngvas}{\termvar, \updvar}$, we in particular know that $\at{(\embedconst\cdot \ahomsol)}{\updvardef}$ contains precisely one copy of the base vector $\baseeffectdef$ for every instance of~$\anngvasp$ in $\asentformpumpleft.\asentformpumpright$, as claimed in the main text.

Similar to equation \eqref{Equation:UpdatesReach} we also have:
\begin{lemma}
\eqref{Equation:UpdatesPumpDiff} holds. 
\end{lemma}
\begin{proof}
We rely on $\computeupdatesof{\anngvas}{\termvar, \updvar}$ to see that 
\begin{align*}
&(\sumconst\cdot\ahomsol)[\updvar]= 
 \restrictto{(\sumconst\cdot\ahomsol)[\termvar]\ }{\updates}+\sum_{\anngvasp\in\trms}(\sumconst\cdot\ahomsol)[\updvardef].
\end{align*}
We first show that
\begin{align*}
\paramparikhof{\updates}{\asentformpumpleft.\asentformpumpright.\asentformdiffleft.\asentformdiffright}\ = \ \restrictto{(\sumconst\cdot\ahomsol)[\termvar]\ }{\updates}\ .
\end{align*}
The equation says that the number of updates produced directly by $\anngvas$ is as expected by the scaled homogeneous solution. 
With Lemma~\ref{Lemma:EEKConverse} and $\ptof{\anngvas}{\prodvar, \termvar}$, it suffices to realize that 
\begin{align*}
&\ \paramparikhof{\prods}{\aprodseqpumpanngvas}+\paramparikhof{\prods}{\aprodseqdiffanngvas}\\ 
=&\ \argument{Definition $\aprodseqdiffanngvas$}\\
&\ \paramparikhof{\prods}{\aprodseqpumpanngvas}+ (\sumconst\cdot \ahomsol)[\prodvar] - \paramparikhof{\prods}{\aprodseqpumpanngvas}\\
 =&\ (\sumconst\cdot \ahomsol)[\prodvar]\ .\\
\end{align*}
We now consider the childNGVAS $\anngvasp$ with restriction $\restrictionsdef=\baseeffectdef+\periodeffectdef^*$. 
The task is to show that $\aprodseqpumpterms$ and $\aprodseqdiffterms$ together create precisely $(\sumconst\cdot\ahomsol)[\updvardef]$ many updates 
in the runs derived from the $\anngvasp$ instances. 
By $\computeupdatesof{\anngvas}{\termvar, \updvar}$, we have
\begin{align*}
(\sumconst\ahomsol)[&\updvardef]\ =\ \\
&\baseeffectdef (\sumconst\ahomsol)[\termvar[\anngvasp]] - \periodeffectdef (\sumconst\ahomsol)[\periodeffectvardef]
\end{align*}
With the argument from the previous paragraph, we know that $\aprodseqpumpanngvas$ and $\aprodseqdiffanngvas$ together create precisely $(\sumconst\cdot\ahomsol)[\termvar[\anngvasp]]$ many instances of $\anngvasp$. 
Moreover, the run we derive from such an instance using $\aprodseqpumpterms$ and $\aprodseqdiffterms$ contains at least the number of updates prescribed by the base vector $\baseeffectdef$. 
The run may contain further updates that together make up copies of the period vectors. \\\\
It remains to check that the expected number of copies of the period vectors $(\sumconst\cdot\ahomsol)[\periodeffectvardef]$  coincides with the number of copies produced by $\aprodseqpumpterms$ and $\aprodseqdiffterms$. 
In $\aprodseqpumpterms$, we produce $\periodeffectpumpdef$ many copies.
In $\aprodseqdiffterms$, there is a single instance of $\anngvasp$ that produces copies of the period vectors. By the induction hypothesis, it produces $(\embedconst\cdot \ahomsol)[\periodeffectvardef]-\periodeffectpumpdef$ copies with $\baseeffectchoicedef'$ plus an additional $\enableconstdef\cdot \periodeffectchoicedef'$ many copies. 
Then 
\begin{align*}
&\ \periodeffectpumpdef\quad +\quad  (\embedconst\cdot \ahomsol)[\periodeffectvardef]-\periodeffectpumpdef\\
&\quad +\quad \enableconst\cdot \periodeffectchoicedef'\\
 =& \ \argument{Definition $\periodeffectchoicedef'$}\\
 &\ (\embedconst\cdot \ahomsol)[\periodeffectvardef] + (\enableconst\cdot \ahomsol)[\periodeffectvardef]\\
 =& \ \argument{Definition $\sumconst$}\\
 &\  (\sumconst\cdot \ahomsol)[\periodeffectvardef]\ .\qedhere 
\end{align*}
\end{proof}

After these minor lemmas, we finish with the missing proofs of the \emph{pumping} part of the extended paper.

\LemmaLowerBoundNonLinear*

\begin{proof}[Proof of Lemma~\ref{Lemma:LowerBound}]
We first bound the negative effect that~$\arun$ may have on $\acounter$. 
The parts of $\arun$ that stem from complete copies of~$\termvec, \termvec'$  contribute $\maxnegdefcompl=\min\set{0, \at{\vaseffof{\diffrunleft.\diffrunright}}{\acounter}}$ for \(\termvec\) and respectively \(\maxnegdefcompl'=\min\set{0, \at{\vaseffof{\diffrunleft'.\diffrunright'}}{\acounter}}\) for \(\termvec'\) to the negative effect.   
For the incomplete copies of~$\termvec$, we iterate over all sentential forms~$\asentform$ that could result from them and over all runs~$\arun'$ that could be derived from $\asentform$ using $\aprodseqdiffterms$. 
Let \(Pos_{\asentform}\) be the (finite) set of all words \(\asentform \in \Sigma^{\ast}\) which fulfill \(\paramparikhof{\trms}{\asentform} \leq \termvec\) or \(\paramparikhof{\trms}{\asentform} \leq \termvec'\). Let \(Pos_{\arun}\) be the (finite) set of runs \(\arun'\) which can be derived from an \(\alpha \in Pos_{\asentform}\) using only the productions in \(\aprodseqdiffterms\) or \(\aprodseqdiffterms'\). Define
\begin{align*}
\maxnegdefincompl:=\min \{0, \min\setcond{\at{\vaseffof{\arun'}}{\acounter}}{\arun' \in Pos_{\arun}}\}.
\end{align*} 
Then $\at{\vaseffof{\arun}}{\acounter}$ can be lower bounded by
\begin{align*} 
j_1 \cdot \maxnegdefcompl + j_2 \cdot \maxnegdefcompl' + \maxnegdefincompl\cdot \ceilof{1+\log (j_1+j_2)}\ .
\end{align*}
The lower bound on the token growth in $\upseq^{j_1+j_2}.\arun$ is by 
\begin{align*}
&\ \at{\vaseffof{\upseq^{j_1+j_2}.\arun}}{\acounter}\\
\geq &\ j_1+j_2+ \maxnegdefincompl\cdot \ceilof{1+\log (j_1+j_2)}\\
\geq &\ j_1+j_2 - (2\cdot \cardof{\maxnegdefincompl}\ +\ \cardof{\maxnegdefincompl}\cdot \log (j_1+j_2))\\
\geq &\ \frac{1}{2} \cdot (j_1+j_2) \quad\text{for $j_1+j_2\in\N$ large enough}.
\end{align*}
The first inequality is the above lower bound on $\at{\vaseffof{\arun}}{\acounter}$, where we immediately lower bound \( \maxnegdefcompl\) and \( \maxnegdefcompl'\) by \(1\) each, which is valid since $\at{\vaseffof{\upseq.\diffrunleft.\diffrunright}}{\acounter}, \at{\vaseffof{\upseq}}{\acounter}\geq 1$.  
We then have $\ceilof{1+\log (j_1+j_2)}\leq 2+\log (j_1+j_2)$. 
We finally use $a_1+a_2\cdot \log (j_1+j_2)\leq \frac{1}{2}\cdot (j_1+j_2)$ for $a_1, a_2\geq 0$ and $j_1+j_2\in\N$ large enough. 
Since $\maxnegdefincompl$ in the last inequality only depends on $\acounter$, and since $\adimset$ is finite, we can define $\iterlb$ as the maximum over the definitions of ``large enough''.
\end{proof}

\textbf{The other cases}: Case 2: Counter \(\acounter\) is concrete only in the input. Then there is a trick here which we did not mention in the main text: We have to ensure that \(\diffrunleft.\diffrunright\) and respectively \(\diffrunleft'.\diffrunright'\) have a positive effect on \(\acounter\). We know that any homogeneous solution has a positive effect on counter \(\acounter\), since \(\acounter\) is concrete in the input and by perfectness \perfectnesscounters in the support in the output. In particular \(\upseq.\diffrunleft.\diffrunright.\downseq\) has a positive effect on \(\acounter\). Hence if we choose \(\sumconst\) large enough, then the effect of \(\upseq.\downseq\) (which is constant) will be less than the effect of \(\upseq.\diffrunleft.\diffrunright.\downseq\), which equals \(\sumconst \cdot \ahomsol\). In particular, the difference \(\diffrunleft.\diffrunright\) and respectively \(\diffrunleft'.\diffrunright'\) will have a positive effect. 

At that point Lemma \ref{Lemma:LowerBound} applies the same way also to this counter \(\acounter\): Since \(\acounter\) is concrete in the input, it is increased by \(\upseq\), and since we applied the above trick, i.e.\ chose \(\diffrunleft.\diffrunright\) to have a positive effect on \(\acounter\), in particular also \(\upseq.\diffrunleft.\diffrunright\) has a positive effect.

Case 3: Counter \(\acounter\) is concrete only in the output. This case is dual to case 2.

Case 4: Counter \(\acounter\) is neither concrete in the input nor in the output: Then we know that the homogeneous solution \(\ahomsol\) increases the value of \(\acounter\) both in the input and output. Via a similar calculation as in case 2, we have to guarantee that the effect \(\ahomsol\) has on the starting value is more than the (possible) negative effect of \(\upseq\) on \(\acounter\). Moreover, similarly for the effect of \(\ahomsol\) on the target value and the (possible) negative effect of \(\downseq\) on \(\acounter\).

\textbf{Finishing the proof}. After this case distinction, we can finally define \(\initconst\) as the maximum of all the ``large enough'' requirements we have encountered thus far, for example including \(\initconst \geq \sumconst^2\), \(\initconst \geq \iterlb\) from Lemma \ref{Lemma:LowerBound}, etc., where \(\sumconst\) has a ``large enough'' requirement from the cases above etc.

It remains to argue that each run $\iterrun{\aconst}$ with $\aconst\geq \initconst$ has the properties formulated in the iteration lemma: 
it solves reachability and uses a number of updates that is as desired.

We first argue that the run is enabled.
The value of the bounded counters is tracked explicitly by the boundedness information, so there is nothing to do. 
For the unbounded counters, we reason as follows.
For the up-pumping sequence, enabledness holds by \perfectnesspumpingnospace. 
For $\upseq^{\aconst\cdot\iterfull}.\arun_{\aconst\cdot\iterfull, 1}$, we use $\iterfull\geq \iterlb$ and Lemma~\ref{Lemma:LowerBound}. 
More precisely, since we have a non-negative effect on all prefixes, we know that we have met the hurdle, using the remark in the preliminaries.
For the remainder of the computation, a similar argument applies. 

\(\Z\)-Reachability and the effect can be studied on a run that has the same Parikh image, i.e.\ on
\begin{align*}
\iterruntildedef=\upseq^{j_1+j_2}.\diffrunleft^{j_1}.\diffrunleft^{`j_2}.\reachrun.\diffrunright^{`j_2}.\diffrunright^{j_1}.\downseq^{j_1+j_2}\ .
\end{align*}

We argue that we solve reachability. 
Since $\asol'$ is a solution to the characteristic equations and $\ahomsol$ is a homogeneous solution,  also \(\asol''\ =\ \asol'\ +\ \aconst \cdot \ahomsol\) solves the characteristic equations.

Due to $\meof{\anngvas}{\invar, \updvar, \outvar}$, this in particular means $\asol''[\invar]+\updates \cdot \asol''[\updvar]=\asol''[\outvar]$. 
By \eqref{Equation:UpdatesReach}, we have \(\vaseffof{\reachrun}\ = \ \updates \cdot \asol'[\avar_{\updates}]\). 
By \eqref{Equation:UpdatesPumpDiff}, the equality \(\vaseffof{\upseq.\diffrunleft.\diffrunright.\downseq}\tinyspace =\tinyspace
\updates \cdot \sumconst\cdot \ahomsol[\updvar]\) holds. Equivalently $\vaseffof{\upseq.\diffrunleft'.\diffrunright'.\downseq}=\updates \cdot \sumconst'\cdot \ahomsol[\updvar]$ holds. By definition of \(j_1, j_2\) we have \(\aconst=j_1 \sumconst + j_2 \sumconst'\). Hence indeed \(\vaseffof{\iterruntildedef}=\updates \cdot \asol''[\updvar]\), i.e.\ the Parikh vector satisfies the marking equation.

\textbf{Iteration Lemma Linear Case}

With $\mydir\in\set{\myleft, \myright}, \mycenter\in\set{\mycenterleft, \mycenterright}$, one can show
\begin{align*}
\vaseffof{\reachrundir}\ &=\ \updates \cdot (\asol+\maxconst\cdot \ahomsol)[\updvardir]\\
\vaseffof{\reachruncenterdef}\ &=\ \updates \cdot (\asol+(\maxconst+\aconst)\cdot\ahomsol)[\updvarcenter]\\
\vaseffof{\upseq.\diffrunleft.\downseqint}\ &=\ \updates \cdot (\sumconst\cdot \ahomsol)[\updvarleft]\\
\vaseffof{\upseqint.\diffrunright.\downseq}\ &=\ \updates \cdot (\sumconst\cdot \ahomsol)[\updvarright]\ .
\end{align*}

It remains to argue that the rest of the run $\iterrundef$ is enabled, and that we reach the claimed target. To this end, we study the effect that 
the homogeneous solution has on a counter. 
The following applies to all directions $\mydir\in\set{\myleft, \mycenterleft, \mycenterright, \myright}$, because the reachability constraint is the same for all of them. 
It also holds for concatenations of directions, say from $\acontextinleft$ to $\acontextoutcenterleft$, due to the equality constraints between the markings reached in the directions. 
(i) If counter $\acounter$ is concrete in input and output, $\at{\acontextindir}{\acounter}, \at{\acontextoutdir}{\acounter}\in\N$, then the updates $\updvardir$ given by the homogeneous solution have effect zero, $\at{(\effof{\updates}\cdot \updvardir)}{\acounter}=0$.
This holds by the definition of $\zeroof{\acontextindir}$ and $\zeroof{\acontextoutdir}$. 
(ii) If the counter is concrete in the input but $\omega$ in the output, the homogeneous solution has a positive effect. 
This holds because the counter is in the support.
(iii) Vice versa, if the counter is $\omega$ in the input but concrete in the output, then the homogeneous solution has a negative effect on the counter. 

It is this argument that made us split the reachability requirement for the center into the left and the right side. 
With the splitting, we could add the variable in the middle to the support, and rely on the pumping in the left child.

Enabledness for $\leftrundef$ is by a pumping argument that already Lambert used for VAS reachability. 
We obtained the center runs with an invokation of the induction hypothesis, and so can rely on the reachability $\anngvaspcenterleft.\acontextin\fires{\reachruncenterleftdef}\anngvaspcenterleft.\acontextout$ respectively 
$\anngvaspcenterright.\acontextin\fires{\reachruncenterrightdef}\anngvaspcenterright.\acontextout$. 
By definition, $\acontextincenterleft=\anngvaspcenterleft.\acontextin$ and similar for the other markings used in the characteristic equations. 
The pumping behavior of the homogeneous solution discussed in the previous paragraph then allows us to pump counters that are $\omega$ in $\acontextincenterleft$ resp. $\acontextincenterright$, and so eventually enable the center runs.
We can concatenate them thanks to $\anngvaspcenterleft.\acontextout = \anngvaspcenterright.\acontextin$ in consistency. 
For the run $\rightrundef$, we again rely on Lambert.

\section{Details for Section~\ref{Section:Decomposition}}\label{Section:DecompositionLevel3}

\newcommand{\makescc}{\mathsf{mkscc}}
\newcommand{\makesccof}[1]{\makescc(#1)}
\newcommand{\lindecomp}{\mathsf{lindecomp}}
\newcommand{\lindecompof}[1]{\lindecomp(#1)}


\paragraph*{Rigidness related facts}
In \Cref{Section:Decomposition}, we use that if $\acounterset$ is rigid for non-linear $\anngvas$, then there are a pair of markings $\amarking_{\asymbol, in}, \amarking_{\asymbol, out}\in\N^{\acounterset}\times\Nomega^{[d]\setminus \acounterset}$ for each symbol $\asymbol\in\anngvas.(\nonterms\cup\trms)$, such that in any run $\arun\in\runsof{\anngvas}$, the part of the run $\arun_{\asymbol}$ derived from $\asymbol$, starts from a counter valuation $\amarkingp_{in, \asymbol}\sqsubseteq\amarking_{in, \asymbol}$, and ends at a counter valuation $\amarkingp_{out, \asymbol}\sqsubseteq\amarking_{out, \asymbol}$.
Here, we show the validity of this statement.

First, we prove that if a counter is fixed on a side in a non-linear NGVAS, then all symbols have the same effect wrt. this counter.
We write $\ceffof{\awordpp}[i]$ to denote $\setcond{\amarkingpp[i]}{\amarkingpp\in\ceffof{\awordpp}}$.
\begin{lemma}\label{Lemma:FixedNLEffects}
    Let $\anngvas$ be a non-linear NGVAS, and $i\in[d]$ fixed on a side.
    Then for all $\asymbol\in\nonterms\cup\trms$, there is a $\amarkingp\in\Z$ with $\ceffof{\aword}[i]=\amarkingp$ for all $\aword\in\trms^{*}$ with $\asymbol\to^{*}\aword$.
\end{lemma}

\begin{proof}
    Suppose that this is not the case.
    Let $i\in[d]$ be wlog. fixed on the left. 
    Then, there is a symbol $\asymbol\in\nonterms\cup\trms$ with $\aword_{0}, \aword_{1}\in\trms^{*}$, $\asymbol\to^{*}\aword_{0}$, $\asymbol\to^{*}\aword_{1}$, $\amarkingp_{0}\in\ceffof{\aword_{0}}$, $\amarkingp_{1}\in\ceffof{\aword_{1}}$ and $\amarkingp_{0}[i]\neq\amarkingp_{1}[i]$.
    Since all symbols are reachable from all non-terminals in an NGVAS, there is a derivation $\anonterm\to^{*}\awordp_{\anonterm, \asymbolp}.\asymbolp.\awordpp_{\anonterm, \asymbolp}$ for all $\anonterm\in\nonterms$ and $\asymbolp\in\nonterms\cup\trms$ with $\awordp_{\anonterm, \asymbolp}, \awordpp_{\anonterm, \asymbolp}\in\trms^{*}$.
    Since $\anngvas$ is non-linear, it has a rule of the form $\anonterm\to\anontermp.\anontermpp$.
    First, we apply the following derivations 
    \begin{align*}
        \startnonterm\to^{*}\awordp_{\startnonterm, \anonterm}.\anonterm.\awordpp_{\startnonterm, \anonterm}\to^{*}
        \awordp_{\startnonterm, \anonterm}.\anontermp.\anontermpp.\awordpp_{\startnonterm, \anonterm}\to^{*}\\
        \awordp_{\startnonterm, \anonterm}.(\awordp_{\anontermp, \asymbol}.\asymbol.\awordpp_{\anontermp, \asymbol}).(\awordp_{\anontermpp, \startnonterm}.\startnonterm.\awordpp_{\anontermpp, \startnonterm}).\awordpp_{\startnonterm, \anontermp}.
    \end{align*}
    Using $\asymbol\to^{*}\aword_{0}$ and $\asymbol\to^{*}\aword_{1}$, we get the cycles 
    \begin{align*}
        \acyc_{0}=\startnonterm\to^{*}\awordpp_{0}.\aword_{0}.\awordpp_{1}.\startnonterm.\awordpp_{2}\ \text{ and }\ 
        \acyc_{1}=\startnonterm\to^{*}\awordpp_{0}.\aword_{1}.\awordpp_{1}.\startnonterm.\awordpp_{2}.
    \end{align*}
    for $\awordpp_{0}, \awordpp_{1}, \awordpp_{2}\in\trms^{*}$ with $\awordpp_{0}=\awordp_{\startnonterm, \anonterm}.\awordp_{\anontermp, \asymbol}$, $\awordpp_{1}=\awordpp_{\anontermp, \asymbol}.\awordp_{\anontermpp, \startnonterm}$, and $\awordpp_{2}=\awordpp_{\anontermpp, \startnonterm}.\awordpp_{\startnonterm, \anontermp}$.
    Since the counter is fixed on the left, $\ceffof{\awordpp_{0}.\aword_{0}.\awordpp_{1}}=\set{0}$.
    Thus, there are $\amarkingpp_{0}\in\ceffof{\awordpp_{0}}$ and $\amarkingpp_{1}\in\ceffof{\awordpp_{1}}$ with $\amarkingpp_{0}[i]+\amarkingp_{0}[i]+\amarkingpp_{1}[i]=0$.
    This means that $\amarkingpp_{0}[i]+\amarkingp_{1}[i]+\amarkingpp_{1}[i]\neq 0$ since $\amarkingp_{1}[i]\neq 0$.
    But $\amarkingpp_{0}[i]+\amarkingp_{1}[i]+\amarkingpp_{1}[i]\in\ceffof{\awordpp_{0}.\aword_{1}.\awordpp_{1}}$.
    So, $\acyc_{1}$ has a non-zero effect on counter $i$.
\end{proof} 

Now, we show that if a counter is fixed on the left (right), then for each symbol, the left (right) effects of all derivations that lead to it are the same.
\begin{lemma}\label{Lemma:FixedEffects}
    Let $\anngvas$ be an NGVAS, and let $i\in [d]$ be fixed on the left (right).
    Then, for all $\anonterm\in\anngvas.(\nonterms\cup\trms)$, there is a $\amarkingp\in\Z$ such that for all $\startnonterm\to^{*}\aword_{\lefttag}.\anonterm.\aword_{\righttag}$ that do not derive center-children, we have $\ceffof{\aword_{\lefttag}}[i]=\set{\amarkingp}$ ($\ceffof{\aword_{\righttag}}[i]=\set{\amarkingp}$).
\end{lemma}

\begin{proof}[Proof Sketch]
    The proof is similar to the proof of \Cref{Lemma:FixedNLEffects}.
    The negation of the lemma yields two derivations with distinct effects on $i$.
    Using the derivations $\anonterm\to^{*}\awordp_{\anonterm, \asymbolp}.\asymbolp.\awordpp_{\anonterm, \asymbolp}$ that go from any $\anonterm\in\nonterms$ to any symbol $\asymbolp\in\trms$, we construct two cycles with distinct effects on $i$.
    Since at most one of them can be non-zero, this contradicts the fixedness of $i$.
\end{proof}

With these results, the validity of our statement becomes clear.
Let $i$ be rigid on the left, the argument for the right-case is similar.
If a counter $i$ is rigid on the left, then it is fixed on the left with $\acontextin[i]\in\N$.
By \Cref{Lemma:FixedEffects}, there is a fixed effect such that all derivations that lead to $\asymbol$ from $\startnonterm$ have this effect on the left side.
Since $\ceffof{-}$ overapproximates the effects of childNGVAS runs, and $\acontextin[i]$ is fixed, all runs reach the same value at the ``enterance'' of $\asymbol$. 
In the non-linear case, we can determine a value for the right-side as well with \Cref{Lemma:FixedNLEffects}.

\paragraph*{Details of $\basisfire$}

We discuss the $\basisfire$ procedure in detail.
Its purpose is to ensure that the base effect of an NGVAS is enabled.
To do this, it repeatedly calls $\perfect$ with various applications of period vectors in $\restrictions$ to determine the minimal applications that lead to enabledness.
This requires that $\perfect$ is reliable up to and for $\rankof{\anngvas}$, but since $\clean$ uses this procedure only on recurring children, termination is still ensured.
\begin{lemma}\label{Lemma:FiresSpec}
    Let $\anngvas$ be a perfect NGVAS, and let $\perfect$ be reliable up to and for $\rankof{\anngvas}$.
    Then $\basisfireof{\anngvas}$ is a perfect deconstruction of $\anngvas$, with $\rankof{\anngvas'}\leq\rankof{\anngvas}$ for all $\anngvas'\in\basisfireof{\anngvas}$, such that all $\anngvas'\in\basisfireof{\anngvas}$ have their base effects enabled. 
\end{lemma}

We develop our notation.
For $z\in\N^{\periodeffect}$ and $z_{\omega}\in\Nomega^{d}$, we let $\anngvas\minrestrto{z}$, and $\anngvas\restrto{\sqsubseteq z_{\omega}}$ be the NGVAS identical to $\anngvas$ except at the restriction, where instead of $\anngvas.\restrictions=(\baseeffect, \periodeffect)$ we have $\anngvas\minrestrto{z}=(\baseeffect + \periodeffect\cdot z, \periodeffect)$ and $\anngvas\restrto{\sqsubseteq z_{\omega}}=(\baseeffect + \sum_{\aperiod\in\periodeffect\setminus\omegaof{z_{\omega}}} z_{\omega}[\aperiod]\cdot \aperiod, \omegaof{z})$.
Intuitively, the former prescribes a minimal amount of applications for each period vector, while the latter prescribes an exact amount of applications for some period vectors. 
The call $\basisfire$ relies on a subcall $\basisfinder:\Nomega^{\periodeffect}\to\powof{\N^{\periodeffect}}$.
Intuitively, $\basisfinderof{\amarking}$ for some $\amarking\in\Nomega^{\periodeffect}$ returns all minimal applications of $\periodeffect$ more specialized than $\amarking$ that correspond to effects of runs in $\anngvas$. 
Using $\basisfinder$, $\basisfire$ returns 
$$\basisfireof{\anngvas}=\setcond{\anngvas\minrestrto{z}}{z\in\basisfinderof{\omega^{\periodeffect}}}.$$  

The call $\basisfinder$ is a recursive function that proceeds as follows.
Let $z\in\Nomega^{\periodeffect}$ be the input of $\basisfinder$.
If $\omegaof{z}=\emptyset$, then $\basisfinder$ checks whether there is a run with the effect described by $z$, by constructing $\perfectof{\cleanof{\anngvas\restrto{\sqsubseteq z}}}$.
If $\perfectof{\cleanof{\anngvas\restrto{\sqsubseteq z}}}\neq\emptyset$, a run has been found and $\basisfinderof{z}$ returns $\set{z}$.
Otherwise, $\basisfinderof{z}$ returns $\emptyset$.
If $\omegaof{z}\neq\emptyset$, $\basisfinder$ constructs $\perfectof{\cleanof{\anngvas\restrto{\sqsubseteq z}}}$ once more.
If $\perfectof{\cleanof{\anngvas\restrto{\sqsubseteq z}}}=\emptyset$, then $\basisfinderof{z}$ returns $\emptyset$.
Else, as a consequence of \Cref{TheoremIterationLemmaNonLinearOverview}, we observe that for $\anngvasp\in
\perfectof{\cleanof{\anngvas\restrto{\sqsubseteq z}}}$, we have $\runsof{\anngvasp}\neq\emptyset$, and we can construct one such $\arun\in\updateseqof{\anngvasp}$.
Then, we construct $z_{\mathsf{conc}}\in\N^{\periodeffect}$ with $z_{\mathsf{conc}}\sqsubseteq z$, and $\updates\cdot\paramparikhof{\updates}{\arun}=\periodeffect\cdot z_{\mathsf{conc}}+\baseeffect$.
Note that such a $z_{\mathsf{conc}}$ must exist.
In this case, the call $\basisfinderof{z}$ returns the union of $\set{z_{\mathsf{conc}}}$ and the sets $\basisfinderof{\setter{z}{i}{l}}$ for all $i\in\omegaof{z}$ and $l< z_{\mathsf{conc}}[i]$.

\begin{proof}[Proof of \Cref{Lemma:FiresSpec}]
    Let $\anngvas$ be perfect.
    As we alluded to in the construction, first assume that $\basisfinderof{\omega^{\periodeffect}}$ terminates, and fulfills the following properties.
    For all $z\in\basisfinderof{\omega^{\periodeffect}}$, (a) there is a sequence $\arun\in\updateseqof{\anngvas}$ with $\updates\cdot\paramparikhof{\updates}{\arun}=\periodeffect\cdot z$, and that (b) for any sequence $\arun\in\updateseqof{\anngvas}$, there is a $z\in\upclsof{\basisfinderof{\omega^{\periodeffect}}}$ where $\updates\cdot\paramparikhof{\updates}{\arun}=\periodeffect\cdot z$.
    In this case, it is clear that $\anngvas\minrestrto{z}$ have enabled base effects for all $z\in\basisfinderof{\omega^{\periodeffect}}$.
    Furthermore, the effect of each run $(\amarking, \arun, \amarkingp)\in\runsof{\anngvas}$ is also captured in $\anngvas\minrestrto{z}.\restrictions$ for some $z\in\basisfinderof{\omega^{\periodeffect}}$, which implies $(\amarking, \arun, \amarkingp)\in\runsof{\anngvas\minrestrto{z}}$ for this $z$.
    Then, $\runsof{\anngvas}=\runsof{\basisfireof{\anngvas}}$.
    Clearly, the construction only modifies the base effect of $\anngvas$, which means that the NGVAS $\anngvas'\in\basisfireof{\anngvas}$ are perfect, and have the same rank as $\anngvas$.
    This shows \Cref{Lemma:FiresSpec}.
    It remains to show our assumption.

    For $z\in\Nomega^{\periodeffect}$, we proceed by an induction on $\cardof{\omegaof{z}}$ to show three claims (a), (b), and (c).
    Namely, we show that (a) for all $z'\in\basisfinderof{z}$, we have $z'\sqsubseteq z$, and a sequence $\arun\in\updateseqof{\anngvas}$ with $\updates\cdot\paramparikhof{\updates}{\arun}=\periodeffect\cdot z'$, (b) for all sequences $\arun\in\updateseqof{\anngvas}$, for which there is a $z'\sqsubseteq z$ with $\updates\cdot\paramparikhof{\updates}{\arun}=\periodeffect\cdot z'$, we have $z'\in\upclsof{\basisfinderof{z}}$, and (c) $\basisfinderof{z}$ terminates. 
    Let $z\in\Nomega^{\periodeffect}$ and assume that $\perfect$ is reliable up to $\rankof{\anngvas}$, and for $\rankof{\anngvas}$.
    
    Before moving on to proving (a), (b) and (c), we note that for any $z'\in\Nomega^{\periodeffect}$, $\rankof{\anngvas}=\rankof{\anngvas\restrto{\sqsubseteq z'}}$ holds by the definition of the rank.
    By the reliability of $\perfect$, the construction $\perfectof{\cleanof{\anngvas\restrto{\sqsubseteq z'}}}$ terminates with the correct return value for all $z'\in\Nomega^{\periodeffect}$.
    Namely, we know that $\perfectof{\cleanof{\anngvas\restrto{\sqsubseteq z'}}}$ is a perfect deconstruction of $\anngvas\restrto{\sqsubseteq z'}$.

    We proceed with the proof of the base case, $\omegaof{z}=\emptyset$.
    For (c), we observe that $\perfectof{\cleanof{\anngvas\restrto{\sqsubseteq z}}}$, and thus $\basisfinderof{z}$ both terminate.
    Towards (a) and (b) consider that $\updateseqof{\anngvas\restrto{\sqsubseteq z}}$ consists exactly of those sequences in $\updateseqof{\anngvas}$, whose effect is comprised of adding $\anngvas.\baseeffect$ once, $\aperiod\in\anngvas.\periodeffect$ exactly $z[\aperiod]$ times if $z[\aperiod]=\omega$, and an arbitrary amount of the remaining period vectors.
    Since $\omegaof{z}=\emptyset$, this constrains the applications of all period vectors.
    We return $\set{z}$ if $\perfectof{\cleanof{\anngvas\restrto{\sqsubseteq z}}}\neq\emptyset$, and $\emptyset$ else.
    Consider the former case.
    Because $\perfectof{\cleanof{\anngvas\restrto{\sqsubseteq z}}}$ is a perfect deconstruction of $\anngvas\restrto{\sqsubseteq z}$, its non-emptiness is a witness for a run with the correct effect $\periodeffect\cdot z$.
    Since $z$ is the only vector with $\sqsubseteq z$, (a) and (b) both hold.
    In the latter case, (a) holds since $\basisfinderof{z}=\emptyset$, and (b) holds since $\updateseqof{\anngvas\restrto{\sqsubseteq z}}=\emptyset$, which implies that there is no such sequence as given in the premise of (b).
    
    We move on to the inductive case.
    For (c), observe that the only procedures called by $\basisfinderof{z}$ are $\perfectof{\cleanof{\anngvas\restrto{\sqsubseteq z}}}$, and $\basisfinderof{z'}$ for some $z'\in\Nomega^{\periodeffect}$, where $\cardof{\omegaof{z'}}<\cardof{\omegaof{z}}$.
    We have already argued that $\perfectof{\cleanof{\anngvas\restrto{\sqsubseteq z}}}$ terminates correctly.
    We use the inner induction hypothesis to observe that $\basisfinderof{z'}$ terminates for all $z'\in\Nomega^{\periodeffect}$ with $\cardof{\omegaof{z'}}<\cardof{\omegaof{z}}$.
    This concludes the proof of (c).
    For the proof of (a) and (b), we make a case distinction.
    First, consider the case that $\perfectof{\cleanof{\anngvas\restrto{\sqsubseteq z}}}=\emptyset$.
    In this case (a) is trivially fulfilled since $\basisfinderof{z}=\emptyset$ as well, and (b) is trivially fulfilled since $\perfectof{\cleanof{\anngvas\restrto{z}}}=\emptyset$ implies $\updateseqof{\anngvas\restrto{\sqsubseteq z}}=\emptyset$.
    Now, assume $\perfectof{\cleanof{\anngvas\restrto{z}}}\neq\emptyset$.
    We show (a).
    Let $z'\in\basisfinderof{z}$.
    By construction, we have one of two cases.
    In the first case, we have constructed $z'\in\N^{\periodeffect}$ by applying \Cref{TheoremIterationLemmaNonLinearOverview} to some $\anngvasp'\in\perfectof{\cleanof{\anngvas\restrto{\sqsubseteq z}}}$, and constructing a sequence $\arun\in\updateseqof{\anngvasp'}$.
    This implies that $z' \sqsubseteq z$ and $\updates\cdot \paramparikhof{\updates}{\arun}=\periodeffect\cdot z'$.
    This shows (a).
    In the second case, $z'\in\basisfinderof{z''}$ for some $z''\in\Nomega^{\periodeffect}$ with $z''\sqsubseteq z$, where $\cardof{\omegaof{z''}}<\cardof{\omegaof{z'}}$.
    Here, the inner induction hypothesis applies to complete the proof of (a).
    Now, we show (b).
    Let $\arun\in\updateseqof{\anngvas}$, and let $z'\in\N^{\periodeffect}$ with $z'\sqsubseteq z$, where $\periodeffect\cdot z'=\updates\cdot\paramparikhof{\updates}{\arun}$.
    Further let $z_{\mathsf{conc}}\in\N^{\periodeffect}$ be the vector with the same name in the construction of $\basisfinderof{z}$. 
    We know that $z_{\mathsf{conc}}\sqsubseteq z$ holds, and by \Cref{TheoremIterationLemmaNonLinearOverview}, we know that there is a sequence $\arun\in\updateseqof{\anngvas}$, with $\periodeffect\cdot z_{\mathsf{conc}}=\updates\cdot\paramparikhof{\updates}{\arun}$.
    We either have $z_{\mathsf{conc}}\leq z'$, or there is a $\aperiod\in\periodeffect$, where $z_{\mathsf{conc}}[\aperiod]>z'[\aperiod]$.
    In the former case, it already holds that $z'\in\upclsof{\set{z_{\mathsf{conc}}}}\subseteq\upclsof{\basisfinderof{z}}$.
    In the latter case, we have $z'\sqsubseteq {\setter{z}{\aperiod}{z'[\aperiod]}}$, and thus $z'\in\upclsof{\basisfinderof{{\setter{z}{\aperiod}{z'[\aperiod]}}}}$ by the innner induction hypothesis.
    Since $z_{\mathsf{conc}}[\aperiod]<z'[\aperiod]$, we have $\upclsof{\basisfinderof{{\setter{z}{\aperiod}{z''[\aperiod]}}}}\subseteq\upclsof{\basisfinderof{z}}$.
    This concludes the proof.
\end{proof}

\newcommand{\asolp}{z}

\newcommand{\awtsymbolp}{\asymbolp_{cg}}
\newcommand{\awtstartnonterm}{S_{\mathsf{cg}}}

\section{Details for Pumpability Decompositions (from Sections~\ref{Section:OutlinePumping} and ~\ref{Section:PumpingMP})}\label{Section:WitnessingGrammarsPlus}

We proceed with a few definitions that we use throughout this section.

\subsection{Definitions}
\subparagraph*{Marked Parse Trees.}
A marked parse tree $\atree$ of $\anngvas$, is a $\Nomega^{d}\times(\nonterms\cup\trms)\times\Nomega^{d}$-labeled tree, that is a parse in the grammar of $\anngvas$ with input and output markings at each node.
We additionally require that the values are continuous between siblings.
Formally, we require that for all $\anode\in\atree$, (i) $\nodemarkingof{\anode}\in\Nomega^{d}\times\trms\times\Nomega^{d}$ holds if and only if $\anode$ is a leaf, and (ii) if $\nodemarkingof{\anode}\in\Nomega^{d}\times\nonterms\times\Nomega^{d}$, then there is a rule $\anonterm\to\asymbol.\asymbolp$, and $\amarking, \amarkingp, \amarkingpp\in\Nomega^{d}$ with $\nodemarkingof{\anode_{\lefttag}}=(\amarking, \asymbol,\amarkingp)$ and $\nodemarkingof{\anode_{\righttag}}=(\amarkingp, \asymbolp, \amarkingpp)$, where $\anode$ has exactly two children $\anode_{\lefttag}$, and $\anode_{\righttag}$, which are the left- resp- right-children of $\anode$.
The first condition ensures that the parse trees are unrolled up to the childNGVAS, and the second condition ensures that the tree follows the production rules.
For notational convenience, we refer to the individual components of the label of $\anode\in\atree$ by $\nodemarkingof{\anode}=(\anode.\inlabel, \anode.\symlabel, \anode.\outlabel)$.
We also write $\atree.\inlabel$, $\atree.\symlabel$, and $\atree.\outlabel$ to denote $\anode.\inlabel$, $\anode.\symlabel$, and $\anode.\outlabel$ where $\anode$ is the root node of $\atree$.
Note that the definition of a marked parse tree does not impose any restrictions on the markings.
This is because different restrictions will be useful in different contexts.

\paragraph*{Reachability Trees.} A reachability tree is a tree $\atree$ is a marked parse tree that captures the derivation that witnesses a given run.
On top of being a marked parse tree, we require markings from $\N^{d}$, and the soundness of the labels wrt. runs.
Formally, $\atree$ is a reachability tree if it is a parse tree, and for all $\anode\in\atree$,  we have (i) $\anode.\inlabel,\anode.\outlabel\in\N^{d}$, (ii) if $\anode.\symlabel\in\trms$, then there is a $\arun\in\updates^{*}$ with $(\atree.\inlabel, \arun, \atree.\outlabel)\in\runsof{\anode.\symlabel}$, and (iii) if $\anode.\symlabel\in\nonterms$, then for left- and right-children of $\anode_{\lefttag}$ and $\anode_{\righttag}$, we have $\anode.\inlabel=\anode_{\lefttag}.\inlabel$ and $\anode.\outlabel=\anode_{\righttag}.\outlabel$.
The first rule ensures the concreteness of the markings, second soundness wrt. childNGVAS, and third continuity between siblings.
 
We denote the set of reachability trees of $\anngvas$ via $\reachtreesof{\anngvas}$.
We write $\reachtreesof{\anngvas, (\amarking, \asymbol, \amarkingp)}$ for $(\amarking, \asymbol, \amarkingp)\in\Nomega^{d}\times(\nonterms\cup\trms)\times\Nomega^{d}$ to denote the set of trees $\atree\in\reachtreesof{\anngvas}$ with $\atree.\inlabel\sqsubseteq\amarking$, $\atree.\symlabel=\asymbol$, and $\atree.\outlabel\sqsubseteq\amarkingp$.
The \emph{runs} $\runsof{\atree}$ of a reachability tree $\atree\in\reachtreesof{\anngvas}$, are the runs whose derivation is captured by $\atree$.
Formally, we let $\runsof{\atree}=\setcond{\runsof{\aword}}{\aword\in\yieldof{\atree}}$ be the runs that are captured by the yield $\yieldof{\atree}\in(\Nomega^{d}\times\trms\times\Nomega^{d})^{*}$ where
$$\runsof{\amarking, \anngvasp, \amarkingp}=\setcond{(\amarking, \arun, \amarkingp)\in\therunset}{(\amarking, \arun, \amarkingp)\in\runsof{\anngvasp}}$$
and the inductive case $\runsof{\aword.\awordp}$ for $\aword, \awordp\in(\Nomega^{d}\times\trms\times\Nomega^{d})^{*}$ is defined the same as in the case of NGVAS, by merging the runs in $\runsof{\aword}$ and $\runsof{\awordp}$ at input resp. output markings.

\paragraph*{Maximal Inputs and Outputs.}
We recall the definitions of the two functions $\postfuncN{\anngvas}, \prefuncN{\anngvas}:\Nomega^{d}\times(\nonterms\cup\trms)\to\powof{\Nomega^{d}}$ from the main paper.
As we discussed, they extract the maximal output (resp. input) values that runs derivable from the input symbol can reach, given an input (resp. output) value. 
\begin{align*}
     \postfuncNof{\anngvas}{\amarking, \asymbol} &=\iddecompof{\downclsof{\setcompact{\amarkingppp\in\N^{d}}{(\amarkingpp, \arun, \amarkingppp)\in\runsof{\asymbol},\;\amarkingpp\sqsubseteq\amarking}}}\\
        \prefuncNof{\anngvas}{\amarkingp, \asymbol}&=\iddecompof{\downclsof{\setcompact{\amarkingpp\in\N^{d}}{(\amarkingpp, \arun, \amarkingppp)\in\runsof{\asymbol},\;\amarkingppp\sqsubseteq\amarkingp}}}.
    \end{align*}
A simple argument shows that $\postfuncN{\anngvas}$ and $\prefuncN{\anngvas}$ are computable when restricted to inputs from $\Nomega^{d}\times\trms$, given reliability assumptions on $\perfect$.
This is a detailed version of the argument for \Cref{Lemma:TermsPostPreComputable} from \Cref{Section:CoverabilityGrammarMP}.
\LemmaTermsPostPreComputableMainPaper*
\begin{proof}
    We only argue for $\postfuncN{\anngvas}$, since the argument for $\prefuncN{\anngvas}$ is similar.
    Let $(\amarking, \anngvasp)\in\Nomega^{d}\times\rectrms$ such that $\unconstrained\subseteq\omegaof{\amarking}$.
    If $\amarking\not\compwith\anngvasp.\acontextin$, then it is clear that there is no run $(\amarking', \arun, \amarkingp')\in\runsof{\anngvasp}$ with $\amarking'\sqsubseteq\amarking$, which implies $\postfuncNof{\anngvas}{\amarking, \anngvasp}=\emptyset$.
    Let $\amarking\compwith\anngvasp.\acontextin$.
    To compute $\postfuncNof{\anngvas}{\amarking, \anngvasp}$, we construct the NGVAS $\anngvasp_{\amarking}$, which is identical to $\anngvasp$, except at the context-information, where we have $\anngvasp_{\amarking}.\acontextin=\amarking\sqcap\anngvasp.\acontextin$ and $\anngvasp_{\amarking}.\acontextout=\anngvasp.\acontextout$.
    Since the children of $\anngvasp$ and $\anngvasp_{\amarking}$ are the same, and $\anngvasp$ is perfect since $\anngvas$ is clean, $\anngvasp_{\amarking}$ has \perfectnesschildren and \perfectnessbasenospace.
    We also have $\srankof{\anngvasp_{\amarking}}=\srankof{\anngvasp}<\srankof{\anngvas}$ because $\anngvasp$ is a child of $\anngvas$ in $\rectrms$.
    Since $\unconstrained\subseteq\omegaof{\amarking}$, $\recrankof{\anngvasp_{\amarking}}<\recrankof{\anngvas}$ also holds.
    Then, $\perfect$ is reliable for and up to $\anngvasp_{\amarking}$.  
    Bringing these together, we know that $\cleanof{\anngvasp_{\amarking}}$ returns a clean deconstruction that is not increasing in rank.
    Then, $\perfect$ is reliable for each $\anngvasp'\in\cleanof{\anngvasp_{\amarking}}$.
    We compute $\perfectof{\cleanof{\anngvasp_{\amarking}}}$, and return $\downclsof{\setcond{\anngvasp'.\acontextout}{\anngvasp'\in\cleanof{\anngvasp_{\amarking}}}}$.
    We argue that this is the desired return value.
    First, note that $\runsof{\anngvasp_{\amarking}}=\setcond{(\amarking', \arun', \amarkingp')}{\amarking'\sqsubseteq\amarking,\; (\amarking', \arun', \amarkingp')\in\runsof{\anngvasp}}$ holds.
    Since deconstruction does not lose any runs, for each $(\amarking', \arun, \amarkingp')\in\runsof{\anngvasp}$ with $\amarking'\sqsubseteq\amarking$, there is a $\anngvasp'\in\perfectof{\cleanof{\anngvasp_{\amarking}}}$ with $\amarkingp'\sqsubseteq\anngvasp'.\acontextout$.
    By \Cref{TheoremIterationLemmaNonLinearOverview}, and the perfectness condition \perfectnesscountersnospace, we know that for each $\anngvasp'\in\perfectof{\cleanof{\anngvasp_{\amarking}}}$, there is a sequence of runs $[(\amarking'_{i}, \arun_{i}, \amarkingp'_{i})]_{i\in\N}\in\runsof{\anngvasp'}^{\omega}\subseteq\runsof{\anngvasp}^{\omega}$, where $\amarkingp'_{i}\sqsubseteq\anngvasp'.\acontextout$ for all $i\in\N$, and $(\amarkingp'_{i}[j])_{i\in\N}\in\N^{d}$ strictly increasing for all $j\in[1,d]$ with $\anngvasp'.\acontextout[j]=\omega$.
    This concludes the proof.
\end{proof}

The decompositions for establishing \perfectnesspumping for the non-linear NGVAS, and \perfectnesspumping resp. \perfectnesspumpingint for linear NGVAS differ significantly.
The linear case consists of an adaptation of classical arguments (the Karp-Miller graph, etc.) to this setting.
In constrast, the non-linear case needs novel techniques.

We first present the Karp-Miller construction for NGVAS, and argue that it allows us to decide whether \perfectnesspumping holds, if $\postfuncN{\anngvas}$ and $\prefuncN{\anngvas}$ can be computed.
Then, we proceed with the linear case, and argue that this construction already gives us what we need for computing a decomposition when \perfectnesspumping resp. \perfectnesspumpingint do not hold for a linear NGVAS.
Finally, we handle the decomposition for the non-linear case.
This is the extention of the argument from \Cref{Section:KarpMillerMP}.

\subsection{The Karp-Miller Graph (Details from \Cref{Section:OutlinePumpingDec})}
We present the formal construction of the Karp-Miller graph, adapted to our setting, which we discussed in \Cref{Section:OutlinePumping} and \Cref{Section:PumpingMP}.
\subparagraph*{Construction.}
The (adapted) Karp-Miller graph $\extkmgraph=(\extkmnodes, \extkmedges)$ is a graph, whose nodes are symbols with input and output markings, and a history component $\extkmnodes\subseteq\Nomega^{d}\times(\nonterms\cup\trms)\times\Nomega^{d}\times\N$.
The history component $\N$ is used to distinguish nodes with different histories, same as in the classical Karp-Miller graph.
We ommit this component in the construction.
Note that we make statements about the properties of this construction wrt. to the variants of $\anngvas$, but here, we state the construction wrt. $\anngvas$ to reduce clutter.
The construction of the Karp-Miller graph starts with nodes and edges $\extkmnodes=\set{(\acontextin, \startnonterm, \acontextout)}$, $\extkmedges=\emptyset$, and iteratively extends it until a fixed point is reached.
At each iteration, the construction obtains a previously unexplored node $(\amarking, \anonterm, \amarkingp)$, and checks whether it has a predecessor $(\amarking', \anonterm, \amarkingp')$ such that $(\amarking', \amarkingp')<(\amarking, \amarkingp)$ holds (in the product order).
If such a pumping situation is present, then it adds a fresh node $(\amarking_{\omega}, \anonterm, \amarkingp_{\omega})$, and an edge $(\amarking, \anonterm, \amarkingp)\to(\amarking_{\omega}, \anonterm, \amarkingp_{\omega})$ to $\extkmedges$ that represents the acceleration, where $\amarking_{\omega}, \amarkingp_{\omega}\in\Nomega^{d}$ are chosen such that $(\amarking_{\omega}, \amarkingp_{\omega})[i]=\omega$ if $(\amarking', \amarkingp')[i]<(\amarking, \amarkingp)[i]$ and $(\amarking_{\omega}, \amarkingp_{\omega})[i]=(\amarking, \amarkingp)[i]$ else for all $i\in [1,d]$.
If it is not present, then it simulates a derivation as follows.
Intuitively, the Karp-Miller construction only tracks one branch of the derivation tree at a time, and incorporates the effect of other branches by $\postfuncN{\anngvas}$ and $\prefuncN{\anngvas}$ calls.
Formally, it considers all rules $\anonterm\to\aword$, and all factorizations $\aword=\aword_{\lefttag}.\anontermp.\awordp_{\righttag}$ with $\aword_{\lefttag}, \aword_{\righttag}\in(\nonterms\cup\trms)^{0}\cup(\nonterms\cup\trms)^{1}$, where $\anontermp\in\nonterms$.
Here, $\anontermp$ is the non-terminal that is on the tracked branch.
If the right-hand side of a rule only contains terminals, the rule is not considered.
Note that the same rule can be considered multiple times, since the choice of factorization may differ.
For each such rule and factorization, we construct $\postfuncNof{\anngvas}{\amarking, \aword_{\lefttag}}$, and $\prefuncNof{\anngvas}{\amarkingp, \aword_{\righttag}}$, where we assume $\postfuncNof{\anngvas}{\amarkingpp, \varepsilon}=\prefuncNof{\anngvas}{\amarkingpp, \varepsilon}=\amarkingpp$ for all $\amarkingpp\in\Nomega^{d}$.
Then, for each $\amarking'\in\postfuncNof{\anngvas}{\amarking, \aword_{\lefttag}}$, and $\amarkingp'\in\prefuncNof{\anngvas}{\amarking, \aword_{\righttag}}$, we add an edge $(\amarking, \anonterm, \amarkingp)\to(\amarking', \anontermp, \amarkingp')$.
The produced $(\amarking', \anontermp, \amarkingp')$ refers to a predecessor of $(\amarking, \anonterm, \amarkingp)$ with the same three components as $(\amarking', \anontermp, \amarkingp')$, and a fresh symbol otherwise.

The termination is also clear by the same arguments as the classical Karp-Miller graph \cite{KM69}.
For this reason, we ommit the proof.
We observe that if $\postfuncN{\anngvas}$ and $\prefuncN{\anngvas}$ are computable for the right domains, then $\extkmgraph$ can be effectively constructed.
\LemmaKarpMillerCompMainPaper*
\begin{proof}
    We only prove the effectiveness, and export the relation to \perfectnesspumpingnospace\ to \Cref{Lemma:ExtPumpChar}.
    Since $\postfuncN{\anngvas}$ and $\prefuncN{\anngvas}$ ensure that all $\omega$ counters on the input remain $\omega$ on their output, it suffices to show termination for the effectiveness. 
    Suppose that the construction does not terminate.
    Let $(\extkmnodes^{(k)}, \extkmedges^{(k)})$ be the graph constructed at the $k$-th iteration.
    Then, by a simple Koenig's Lemma argument, we get an infinite sequence of distinct nodes $[(\amarking_{i}, \anonterm_{i}, \amarkingp_{i})]_{i\in\N}$, where for all $i\in\N$, there is a $k_i\in\N$ where $(\amarking_{i}, \anonterm_{i}, \amarkingp_{i})\in\wtnonterms^{(k_i)}$ and $(\amarking_{i}, \anonterm_{i}, \amarkingp_{i})\to(\amarking_{i+1}, \anonterm_{i+1}, \amarkingp_{i+1})\in\extkmedges$.
    Note that the construction ensures that there are no two distinct symbols with the exact same first three components that can call each other.
    Then, we obtain that $(\amarking_{i}, \anonterm_{i}, \amarkingp_{i})\neq(\amarking_{j}, \anonterm_{j}, \amarkingp_{j})$ for all distinct $i,j\in\N$.
    A standard well-quasi-order argument yields a sequence $(\amarking_{\phi(i)}, \amarkingp_{\phi(i)})_{i\in\N}\in\Nomega^{d}\times\Nomega^{d}$ strictly increasing in the product order, where $\anonterm_{\phi(i)}=\anonterm_{\phi(i+1)}=\anonterm$ for all $i\in\N$.
    However, if $(\amarking_{\phi(i)}, \anonterm, \amarkingp_{\phi(i)})$ can call $(\amarking_{\phi(i+1)}, \anonterm, \amarkingp_{\phi(i+1)})$, and $(\amarking_{\phi(i)}, \amarkingp_{\phi(i)})>(\amarking_{\phi(i+1)}, \amarkingp_{\phi(i+1)})$, then $\omegaof{\amarking_{\phi(i)}, \amarkingp_{\phi(i)}}\subsetneq\omegaof{\amarking_{\phi(i+1)}, \amarkingp_{\phi(i+1)}}$.
    This implies an infinitely increasing chain $\omegaof{\amarking_{\phi(0)}, \amarkingp_{\phi(0)}}\subsetneq\omegaof{\amarking_{\phi(1)}, \amarkingp_{\phi(1)}}\subsetneq\ldots$.
    However, since $\omegaof{\amarking_{\phi(i)}, \amarkingp_{\phi(i)}}\subseteq [1,2d]$ for all $i\in\N$, there can be no such infinite chain.
    This is a contradiction. 
\end{proof}
By relying on the standard arguments, for an NGVAS, we observe that \perfectnesspumping holds if and only if $\extkmgraph$ contains a node labeled $(\inof{\anonterm}, \anonterm, \outof{\anonterm})$ for some $\anonterm\in\nonterms$, completing the proof of \Cref{Lemma:ExtPumpChar}.
\begin{lemma}\label{Lemma:ExtPumpChar}
    Let $\anngvas$ be an NGVAS, and let $\extkmgraph=(\extkmnodes, \extkmedges)$ be its Karp-Miller graph.
    There is an $\anonterm\in\nonterms$ with $(\inof{\anonterm}, \anonterm, \outof{\anonterm})\in\extkmnodes$ if and only if $\anngvas$ has \perfectnesspumpingnospace.
\end{lemma}

Now we argue that each branch of a reachability tree $(\atree, \nodemarking)\in\reachtreesof{\anngvas}$ has an additional counter bounded on one side.
This proof is also standard, but we still include it for the sake of completeness.
\begin{lemma}\label{Lemma:LinearExtPumpingBoundedness}
    Let $\anngvas$ be an NGVAS with the Karp-Miller graph $\extkmgraph=(\extkmnodes, \extkmedges)$, and let $\atree\in\reachtreesof{\anngvas, (\acontextin, \startnonterm, \acontextout)}$.
    Let there be no $\anonterm\in\nonterms$ with $(\inof{\anonterm}, \anonterm, \outof{\anonterm})\in\extkmnodes$, and let $\extkmconst\in\N$ be the largest non-$\omega$ value used in a label of $\extkmnodes$.
    For each branch $(\anode_{i})_{i\leq k}\in\atree$ of $\atree$, one of the following holds:
    (i) there is a $j\in\abdinfoleft$ such that for all $i\leq k$ with $\anode_{i}.\symlabel\in\nonterms$, we have $\anode_{i}.\inlabel[j]\in[0, \extkmconst]$ and (ii) there is a $j\in\abdinforight$ such that for all $i\leq k$ with $\anode_{i}.\symlabel\in\nonterms$, we ahve $\anode_{i}.\outlabel[j]\in[0, \extkmconst]$.
\end{lemma}
\begin{proof}
    Let $\atree\in\reachtreesof{\anngvas}$, and $(\anode_{i})_{i\leq k}$ a branch of $\atree$.
    Further let $(\amarking_{i}, \anonterm_{i}, \amarkingp_{i})_{i\leq k}$ be the labels of these nodes.
    We inductively construct a path $(\amarkingpp_{i}, \anonterm_{i}, \amarkingppp_{i})_{i\leq \ell}\in\extkmnodes$ in $\extkmgraph$ such that there is a surjective non-decreasing map $f:[0, \ell]\to[0, k]$, where for all $i\in[0, \ell]$, $\amarking_{f(i)}\leq\amarkingpp_{i}$ and $\amarkingp_{f(i)}\leq\amarkingppp_{i}$ hold.
    Thanks to the definition of $\postfuncN{\anngvas}$, $\prefuncN{\anngvas}$, and $\extkmgraph$, we know that if a counter on one side becomes labeled $\omega$ at some node along a path in $\extkmgraph$, than it remains $\omega$ for the rest of the path.
    Since there is no node $(\inof{\anonterm}, \anonterm, \outof{\anonterm})\in\extkmnodes$, we know that across all of $(\amarking_{i}, \anonterm_{i}, \amarkingp_{i})_{i\leq k}$, a counter on the input or output side must remain non-$\omega$, and therefore in $[0, \extkmconst]$.
    This shows \Cref{Lemma:LinearExtPumpingBoundedness}.

    We move on to the inductive construction.
    We make an outer induction on $k'$, and an inner induction on $2d-\cardof{\omegaof{\amarkingpp_{\ell}}}-\cardof{\omegaof{\amarkingppp_{\ell}}}$, to show the following statement.
    There is a path $(\amarkingpp_{i}, \anonterm_{i}, \amarkingppp_{i})_{i\leq \ell}\in\extkmnodes$ in $\extkmgraph$ such that there is a surjective non-decreasing map $f:[0, \ell]\to[0, k']$, where for all $i\in[0, \ell]$, $\amarking_{f(i)}\leq\amarkingpp_{i}$ and $\amarkingp_{f(i)}\leq\amarkingppp_{i}$ hold.
    For the base case, we have $k'=0$.
    By the definition of the reachability tree, it must hold that $\amarking_{0}\leq\acontextin$, $\anonterm_{0}=\startnonterm$, and $\amarkingp_{0}\leq\acontextout$.
    The $\extkmgraph$ construction guarantees $(\acontextin, \startnonterm, \acontextout)\in\extkmnodes$.
    We let $(\amarkingpp_{0}, \anonterm_{0}, \amarkingppp_{0})=(\acontextin, \startnonterm, \acontextout)$, and $f(0)=0$.
    
    Now for the inductive case, we are given a path $\rho=(\amarkingpp_{i}, \anonterm_{i}, \amarkingppp_{i})_{i\leq \ell}$ in $\extkmgraph$ and a surjective map $f:[0, \ell]\to[0, k']$.
    We show that if $k'\neq k$, then we can extend $\rho$ and $f$.
    We skip the inner base case.
    We make a case distinction based on how the node $(\amarkingpp_{\ell}, \anonterm_{\ell}, \amarkingppp_{\ell})$ was explored during the construction of $\extkmgraph$.
    If a pumping has occured, then there is an edge $(\amarkingpp_{\ell}, \anonterm_{\ell}, \amarkingppp_{\ell})\to(\amarkingpp', \anonterm, \amarkingppp')$ such that $\amarkingpp_{\ell}\sqsubseteq\amarkingpp'$, $\amarkingppp_{\ell}\sqsubseteq\amarkingppp'$, and $\cardof{\omegaof{\amarkingpp_{\ell}}}+\cardof{\omegaof{\amarkingppp_{\ell}}}<\cardof{\omegaof{\amarkingpp'}}+\cardof{\omegaof{\amarkingppp'}}$.
    In this case, we let $(\amarkingpp_{\ell+1}, \anonterm_{\ell+1}, \amarkingppp_{\ell+1})=(\amarkingpp', \anonterm', \amarkingppp')$, and $f(\ell+1)=f(\ell)=k'$.
    Since the number of $\omega$'s has increased without decreasing $k'$, this concludes the proof.
    Now let the exploration of $(\amarkingpp_{\ell}, \anonterm_{\ell}, \amarkingppp_{\ell})$ have involved no pumping.
    Consider the rule applied at $\anode_{k'}$.
    If it is an exit-rule, we are done, since we must have $k'=k$.
    Let $\anode_{k'+1}$ be the right-child of $\anode_{k'}$ the proof for the left-child case is similar.
    Then, the rule applied at $\anode_{k'}$ must be $\anonterm_{k'}\to\asymbol.\anonterm_{k'+1}$ for some $\asymbol\in\nonterms$.
    Let $\anodep$ be the left-child of $\anode_{k'}$ with the label $\nodemarkingof{\anodep}=(\amarking', \asymbol, \amarkingp')$.
    Since the tree rooted at $\anodep$ is a reachability tree, there is a derivation $\asymbol\to^{*}\awordp$ with $\awordp\in\trms^{*}$ and $(\amarking', \arunp, \amarkingp')\in\runsof{\awordp}$ for some $\arunp\in\updates^{*}$.
    Note that $\amarking'=\amarking_{k'}$ since $\anodep$ is the left-child of $\anode_{k'}$.
    We argue that for all $i<\cardof{\awordp}$, $\awordp[i].\unconstrained=\abdinfoleft$.
    If $\anngvas$ is linear, we have $\cardof{\awordp}=1$ and since it is generated on the left side of a remaining-rule, we have $\awordp[i].\unconstrained=\abdinfoleft$.
    If $\anngvas$ is non-linear, then it is also branching, and it must hold that $\anngvasp.\unconstrained=\abdinfoleft=\abdinforight$ for all $\anngvasp\in\trms$.
    Recall that $\amarking_{k'}\leq\amarkingpp_{\ell}$, so there must be a $\amarkingpp_{diff}$ such that $\amarking_{k'}+\amarkingpp_{diff}\sqsubseteq\amarkingpp_{\ell}$.
    The reachable counter values are constant across all runs for concretely tracked counters, so $\amarking_{k'}, \amarkingpp_{k'}\sqsubseteq \inof{\anonterm}$ hold.
    Then, we know that $\amarkingpp_{diff}[j]=0$ for all $j\in[1,d]\setminus\abdinfoleft$.
    Since the reachability relation of an NGVAS is monotonous on unconstrained counters, we have $(\amarkingpp_{\ell}, \asymbol, \amarkingp'+\amarkingpp_{diff})=(\amarking'+\amarkingpp_{diff}, \asymbol, \amarkingp'+\amarkingpp_{diff})\in\runsof{\awordp}$.
    There must be a $\amarkingpp'\in\postfuncNof{\anngvas}{\amarkingpp_{\ell}, \asymbol}$ where $\amarkingp'\leq\amarkingp'+\amarkingpp_{diff}\leq\amarkingpp'$.
    Thus, we know by the construction of $\extkmgraph$ that there must be a node $(\amarkingpp', \anontermp, \amarkingppp_{\ell})\in\extkmnodes$ with an edge $(\amarkingpp_{\ell}, \anonterm_{\ell}, \amarkingppp_{\ell})\to(\amarkingpp', \anontermp, \amarkingppp_{\ell})$ in $\extkmedges$.
    Letting $(\amarkingpp_{\ell+1}, \anonterm_{\ell+1}, \amarkingppp_{\ell+1})=(\amarkingpp', \anontermp, \amarkingppp_{\ell})$ and $f(\ell+1)=k'+1$ concludes the proof.
\end{proof}

As we discuss in the extended paper, if the NGVAS is linear, then we only need a weaker premise, since all non-terminals appear on the same branch, and the construction never applies $\postfuncN{\anngvas}$ resp. $\prefuncN{\anngvas}$ on non-terminals.
Since we have already proven \Cref{Lemma:TermsPostPreComputable}, we readily get that we can compute $\extkmgraph$ for linear NGVAS.
\begin{lemma}\label{Lemma:LinExtPumpCompHalf}
    If $\anngvas$ is a linear NGVAS, and $\perfect$ is reliable up to $\rankof{\anngvas}$, then $\extkmgraph$ can be effectively constructed.
\end{lemma}

\subsection{Pumpability Decompositions for Linear NGVAS}

The pumpability refinement for linear NGVAS proceeds in two steps.
First, it checks whether the relevant condition (\perfectnesspumping and \perfectnesspumpingintnospace) holds by constructing the Karp-Miller graph.
If it does not hold, as we show in \Cref{Lemma:LinearExtPumpingBoundedness}, the Karp-Miller graph yields a large constant $\extkmconst\in\N$, such that for any reachability tree, and each branch, there is a counter and a side (input resp. output), such that this counter remains bounded across it.
Linear NGVAS have exactly one branch that contains non-terminals, so it suffices to track a counter up to a bound to get a refinement.
Thus, the procedure extends the grammar by a finite counter that counts up to $\extkmconst$, and does not allow derivations to exceed this value.
The specification of the refinement, adapted from \Cref{SectionProofOutline}.
is as follows.
\begin{lemma}\label{Lemma:LinearExtPumpingDecomp}
    Let $\anngvas$ be a clean, linear-NGVAS, and let $\perfect$ be reliable up to $\rankof{\anngvas}$.
    Then $\refinepumpof{\anngvas}$ terminates with $\refinepumpof{\anngvas}=\set{\anngvas}$ if \perfectnesspumping holds, and if not, $\refinepumpof{\anngvas}$ is a decomposition of $\anngvas$.
\end{lemma}

\paragraph*{Construction.}
As its first step, the procedure constructs the (adapted) Karp-Miller graph $\extkmgraph=(\extkmnodes, \extkmedges)$.
If the graph $\extkmgraph=(\extkmnodes, \extkmedges)$ contains a node $(\inof{\anonterm}, \anonterm, \outof{\anonterm})\in\extkmnodes$ for some $\anonterm\in\nonterms$, then the procedure returns $\set{\anngvas}$.
By \Cref{Lemma:ExtPumpChar}, \perfectnesspumping holds in this case.
If the graph does not contain such a node, then we extract the largest non-$\omega$ value $\extkmconst\in\N$ that appears in the input or output marking of a node.
Then for each $i\in\abdinfoleft$, where $\acontextin[i]\neq\omega$, we construct the NGVAS $\anngvas_{\lefttag, i}$, and for each $i\in\abdinforight$, where $\acontextout[i]\neq\omega$, we construct the NGVAS $\anngvas_{\righttag, i}$, and return $\setcond{\anngvas_{\lefttag, i}}{i\in\abdinfoleft,\; \acontextin[i]\neq \omega}\cup\setcond{\anngvas_{\righttag, i}}{i\in\abdinforight,\; \acontextout[i]\neq\omega}$.
The NGVAS $\anngvas_{\lefttag, i}$ tracks the counter $i$ up to $\extkmconst$ on the input side, and $\anngvas_{\righttag, i}$ does the same for $i$ on the output side.
We only present the construction of $\anngvas_{\lefttag, i}$ for $i\in\abdinfoleft$ with $\acontextin[i]\neq\omega$, as the construction for the other side is similar.
We let 
\begin{align*}
    \anngvas_{\lefttag, i}&=(\agram_{\lefttag, i}, \acontext, \restrictions, \unconstrained, (\abdinfoleft\setminus\set{i}, \abdinforight, \infun_{new}, \outfun))\\
    \agram_{\lefttag, i}&=(\nonterms\times[0, \extkmconst], \trms_{\lefttag, i}, (\acontextin[i], \startnonterm), \prods_{\lefttag, i})
\end{align*}
where $\infun_{new}(\anonterm, a)=\setter{\inof{\anonterm}}{i}{a}$ for all $(\anonterm, a)\in\nonterms\times[0, \extkmconst]$, and $\trms_{\lefttag, i}=\setcond{\anngvasp_{a, \omega},\;\anngvasp'}{a\in[0, \extkmconst],\;\anonterm\to\anngvasp.\anngvasp'}\cup\setcond{\anngvasp_{a, b}}{a, b\in[0, \extkmconst],\;\anngvasp\in\trms}$.
We define $\prods_{\lefttag, i}$ shortly,
Here, $\anngvasp_{a, b}$ is the NGVAS that is identical to $\anngvasp$ except at the context information, where we have $\anngvasp_{a, b}.\acontextout[j]=\anngvasp.\acontextout[j]$ and $\anngvasp_{a, b}.\acontextin[j]=\anngvasp.\acontextin[j]$ for all $j\in[1,d]\setminus \set{i}$, along with $\anngvasp_{a, b}.\acontextin[i]=\anngvasp.\acontextin[i]\sqcap a$ and $\anngvasp_{a, b}.\acontextout[i]=\anngvasp.\acontextout[i]\sqcap b$.
Now, we define the rules $\prods_{\lefttag, i}$.
\begin{align*}
    (\anonterm, a)&\to (\anontermp, a).\anngvasp\in\prods_{\lefttag, i},\text{ for all }\anonterm\to\anontermp.\anngvasp\in\prods\\
    (\anonterm, a)&\to\anngvasp_{a, b}. (\anontermp, b)\in\prods_{\lefttag, i},\text{ for all }\anonterm\to\anngvasp.\anontermp\in\prods\\
    (\anonterm, a)&\to \anngvasp_{a, \omega}.\anngvasp'\in\prods_{\lefttag, i},\\
    &\text{ if the exit rule }\anonterm\to\anngvasp.\anngvasp'\text{ has }\anngvasp.\acontextin[i]\compwith a
\end{align*}

\paragraph*{The proof of specification.}

Since $\postfuncN{\anngvas}$ and $\prefuncN{\anngvas}$ are computable for by \Cref{Lemma:TermsPostPreComputable}, the construction is effective by \Cref{Lemma:LinExtPumpCompHalf}.
By \Cref{Lemma:ExtPumpChar}, we get the correctness of the decomposition if \perfectnesspumping holds for $\anngvas$.
\begin{lemma}\label{Lemma:LinearExtPumpingTerminates}
    Let $\anngvas$ be a linear-NGVAS.
    Then, $\refinepumpof{\anngvas}$ terminates.
    If $\refinepumpof{\anngvas}=\anngvas$, then $\anngvas$ has \perfectnesspumpingnospace.
\end{lemma}
Now we bring these lemmas together and prove \Cref{Lemma:LinearExtPumpingDecomp}.

\begin{proof}[Proof of \Cref{Lemma:LinearExtPumpingDecomp}]
    The construction is guaranteed to terminate by \Cref{Lemma:LinearExtPumpingTerminates}.
    It is also clear that each $\anngvas'\in\refinepumpof{\anngvas}$ is indeed a weak-NGVAS.
    By \Cref{Lemma:LinearExtPumpingTerminates}, we know that if $\refinepumpof{\anngvas}=\set{\anngvas}$, then $\anngvas$ fulfills \perfectnesspumpingnospace.
    Now, we argue that $\refinepumpof{\anngvas}$ is a deconstruction of $\anngvas$.
    Since $\refinepump$ does not modify the $\unconstrained$, $\acontextin$, $\acontextout$, and $\restrictions$ components, the deconstruction conditions (i)-(iv) hold.
    Then, it only remains to argue that $\runsof{\refinepumpof{\anngvas}}=\runsof{\anngvas}$.
    We sketch out the argument.
    Since we explicitly restricted the runs, $\runsof{\refinepumpof{\anngvas}}\subseteq\runsof{\anngvas}$ is clear.
    We argue $\runsof{\refinepumpof{\anngvas}}\supseteq\runsof{\anngvas}$ as follows.
    By \Cref{Lemma:LinearExtPumpingBoundedness}, for each reachability tree $\atree\in\reachtreesof{\anngvas, (\acontextin, \startnonterm, \acontextout)}$, there is a side (input resp. output), and a counter $i\in\abdinfoleft$ or $i\in\abdinforight$ depending on the side, such that $i$ remains bounded at this side of $\nodemarkingof{\anode}$ for all $\anode\in\atree$ with $\anode.\symlabel\in\nonterms$.
    Let the counter $i\in\abdinfoleft$ be bounded on the input side, as the other case is similar.
    We argue that all leaves $\anodep$ that are the left-child of their parent have $\anodep.\inlabel[i]\in[0, \extkmconst]$, and if $\anodep$ is the result of applying a remain-rule, we additionally have $\anodep.\outlabel[i]\in[0,\extkmconst]$.
    By the construction of $\anngvas_{\lefttag, i}$, this implies $\runsof{\atree}\subseteq\runsof{\anngvas_{\lefttag, i}}$.
    Let $\anodep$ be a leaf, and the left-child of its parent $\anode$.
    Then, $\anode$ has $\anode.\symlabel\in\nonterms$, and thus $\anode.\inlabel[i]\in [0, \extkmconst]$.
    We get $\anodep.\inlabel[i]=\anode.\inlabel[i]\in [0, \extkmconst]$.
    Let $\anodep$ be produced in a remaining-rule.
    Then, the right-child $\anodepp$ of $\anode$ has $\anodepp.\symlabel\in\nonterms$, and thus $\anodepp.\inlabel[i]\in[0, \extkmconst]$.
    We get $\anodep.\outlabel[i]=\anodepp.\inlabel[i]\in [0, \extkmconst]$.

    Finally, we argue that all $\anngvas'\in\refinepumpof{\anngvas}$ have $\rankof{\anngvas'}<\rankof{\anngvas}$.
    The construction only modifies the context information of the children, so their rank does not change.
    This means that either (i) the system rank or the local index decreased when moving to $\anngvas'$, or (ii) the main branch (beyond the components in $\anngvas'$) remained the same.
    We argue that $\cyclespaceof{\anngvas', \anonterm}\subsetneq\cyclespaceof{\anngvas}$, and thus $\linlrankof{\anngvas', \anonterm}<\linlrankof{\anngvas}$, for all $\anonterm\in\anngvas'.\nonterms$.
    In both cases (i) and (ii), the shows $\rankof{\anngvas'}<\rankof{\anngvas}$.
    We already have $\cyclespaceof{\anngvas', \anonterm}\subseteq\cyclespaceof{\anngvas}$ since we do not introduce new cycles.
    The argument is the same as the one used in \cite{LerouxS19} in the case of coverability.
    Since counter $i$ is concretely tracked on the left, for all cycles $\anonterm\to^{*}\aword.\anonterm.\awordp$ in $\anngvas'$, we have $\ceffof{\aword}=\set{0}$.
    Since $\anngvas$ is clean, the counter $i\not\in\adimsetleft$ is not rigid on the left.
    Since $i\not\in\omegaof{\acontextin}$, the counter $i$ is not fixed on the left.
    Thus, there is a $\amarkingpp\in\cyclespaceof{\anngvas}$ with $\amarkingpp[i]\neq 0$, but for all $\amarkingpp\in\cyclespaceof{\anngvas', \anonterm}$, $\amarkingpp[i]=0$.
    This implies $\cyclespaceof{\anngvas}\neq\cyclespaceof{\anngvas', \anonterm}$, and concludes the proof.
\end{proof}

The decomposition $\refineintpump$, which establishes \perfectnesspumpingint is very similar to $\refinepump$, with only the direction changing.
Instead of simulating pumps from outside to the inside starting at $(\acontextin, \startnonterm, \acontextout)$, it simulates pumps from inside to outside starting at $(\anngvasp.\acontextin, \anontermp, \anngvasp'.\acontextout)$ where $\anontermp\to\anngvasp.\anngvasp'$ is the exit-rule of $\anngvas$.
Here, the rigidity concern is rigidness on the inside-left, and the inside-right.
Since this case is very similar to the case of $\refinepump$, we ommit the proofs, and only state the specification.
\begin{lemma}\label{Lemma:LinearIntPumpingDecomp}
    Let $\anngvas$ be a linear-NGVAS, and let $\perfect$ be reliable up to $\rankof{\anngvas}$.
    Then $\refineintpumpof{\anngvas}$ terminates with $\refineintpumpof{\anngvas}=\set{\anngvas}$ if \perfectnesspumpingint holds, and if not, $\refineintpumpof{\anngvas}$ is a decomposition of $\anngvas$.  
\end{lemma}
Together with \Cref{Lemma:LinearExtPumpingDecomp}, \Cref{Lemma:LinearIntPumpingDecomp} completes the proof of the \Cref{Lemma:ExtPumpCompLin}.
\LemmaExternalPumpCompLinMainPaper*

\subsection{The Coverability Grammar} \label{Section:CoverabilityGrammarL3}
The coverability grammar is a Karp-Miller-graph-like construction, with a novel component called ``promises'' which serves to decompose non-linear NGVAS that do not fulfill \perfectnesspumpingnospace.
The construction is defined relative to the over-approximators of reachability values obtainable on the output, for a given input and vice-versa.
We call such approximators post-approximators and pre-approximators.

In this section, we fill in the details missing from \Cref{Section:CoverabilityGrammarMP}.
For the sake of completeness, text and arguments may repeat.

\paragraph*{Approximators.}

The definition of a pre-approximator is that of the post-approximator definition, but in the reverse direction.
For the sake of completeness, we repeat the \emph{pre-approximator} version of the definition, and refer the reader to \Cref{Section:CoverabilityGrammarMP} for post-approximators.
A \emph{pre-approximator for $\anngvas$} is a function $\preapprox:\Nomega^{d}\times(\nonterms\cup\trms)\to\powof{\Nomega^{d}}$ that always outputs a finite set, and has the following three properties.  
Let the input be $(\amarking, \asymbol)\in\Nomega^{d}\times(\nonterms\cup\trms)$.   
The first property is the soundness wrt. information about concrete values and unboundedness: for every output $\amarking\in\preapproxof{\amarkingp,\asymbol}$, we have  
$\amarking\sqsubseteq\outof{\asymbol}$ and  $\omegaof{\amarkingp}\subseteq\omegaof{\amarking}$. 
The second is the over-approximation of the reachability relation: for every run $(\amarking', \arun, \amarkingp')\in\runsof{\asymbol}$ with $\amarkingp'\sqsubseteq\amarkingp$, there is $\amarking\in\preapproxof{\amarkingp, \asymbol}$ with $\amarking'\sqsubseteq\amarking$.
The last property is the compatibility with the derivation relation: for every $\amarking'\in \preapproxof{\amarkingp, \asentform}$ with $\asymbol\to\asentform$, there is $\amarking\in\preapproxof{\amarkingp, \asymbol}$ with $\amarking'\sqsubseteq\amarking$.
Here, we generalize $ \preapprox$ to sentential forms by defining $\preapproxof{\amarking, \asentform.\asymbolp}=\preapproxof{\preapproxof{\amarking, \asentform}, \asymbolp}$.

For the statements we make in this subsection, it will suffice to consider one pair of post- and pre-approximators.
Further into the paper, we will make use of other approximators.
The approximators we use here are the approximators $\natpostapprox,\natpreapprox:\Nomega^{d}\times(\nonterms\cup\trms)\to\powof{\Nomega^{d}}$, as defined in \Cref{Section:CoverabilityGrammarMP}.
Before moving on to the construction of the coverability grammar, we observe that these are indeed approximators.
\LemmaNatPostPreApproximatorsMainPaper*

\begin{proof}
    It is clear that both functions over approximate the reachability output resp. input values.
    Exactness and the correct unboundedness conditions are also clear by definition.
    Now, we show that $\natpostapprox$ fulfills the precision conditions of post-approximators.
    The proof of the corresponding pre-approximation condition for $\natpreapprox$ is similar, and therefore omitted.
    Let $\amarking\in\Nomega^{d}$, $\anonterm\to\asymbol.\asymbolp\in\prods$, $\amarking_{0}\in\natpostapproxof{\amarking, \asymbol}$, and $\amarking_{1}\in\natpostapproxof{\amarking_{0}, \asymbol}$.
    The membership $\amarking_{0}\in\natpostapproxof{\amarking, \asymbol}$ implies a $\amarking_{0}'\in\postfuncNof{\anngvas}{\amarking, \asymbol}$ with $\amarking_{0}\leq\amarking_{0}'$ and $\omegaof{\amarking_{0}}=\omegaof{\amarking_{0}'}$.
    Then, there is a sequence of runs derivable from $\asymbol$ that reach the $\amarking_{0}'[i]$ at counter $i\in[1,d]$ if $\amarking_{0}[i]'\in\N$, and are unboundedly growing at counter $i\in[1,d]$ if $\amarking_{0}'[i]=\omega$. 
    By the monotonicity of the reachability relation, and by applying the same argument, we observe that there must be some $\amarking_{1}'\in\postfuncNof{\anngvas}{\amarking_{0}', \asymbol}$ with $\amarking_{1}\leq\amarking_{1}'$.
    But, by the definition $\postfuncN{\anngvas}$, all runs derivable from $\asymbol.\asymbolp$ must be also considered for the input $\anonterm$, this yields $\amarking'\in\postfuncNof{\anngvas}{\amarking, \anonterm}$ with $\amarking_{1}\leq\amarking'$.
    Then, we can establish the $\omega$-generalization by obtaining $\amarking''\in\natpostapproxof{\amarking, \anonterm}$ with $\amarking''[i]=\amarking_{1}[i]\leq\amarking'[i]$ for all $i\in[1,d]$ with $\amarking'[i]\neq\omega$, and $\amarking''[i]=\omega$ for all $i\in[1, d]$ with $\amarking'[i]=\omega$.
    This yields $\amarking_{1}\sqsubseteq\amarking''$, and concludes the proof.
\end{proof}

\paragraph*{The Coverability Grammar.}

To keep our presentation unified, we treat the inclusion $\wtprodsinit\subseteq\wtprodspump$ as true.
The construction terminates by the same argument as the Karp-Miller graph \cite{KM69}.
For $\awtsymbol\in\wtnonterms$, we only need to consider the components $\awtsymbol.\wtinlabel$, $\awtsymbol.\wtsymlabel$, and $\awtsymbol.\wtoutlabel$ for termination.
Then, the argument is the same as in the proof of \Cref{Lemma:ExtPumpComp}.
We only need the post- and pre- approximations to be computable for the relevant domain of inputs.
\LemmaCGTerminationMainPaper*

First, by the exactness property of the approximators, we observe that the concretely tracked counters in $\anngvas$ are also concretely tracked in $\wtgrammarof{\anngvas, \postapprox, \preapprox}$.
\begin{lemma}\label{Lemma:CGConcrete}
    Let $\wtgrammarof{\otherctxNGVAS{\anngvas}{\amarking, \anonterm, \amarkingp}, \postapprox, \preapprox}=(\wtnonterms, \wtterms, \wtstartnonterm, \wtprods)$ for post- and pre-approximators $\postapprox$ and $\preapprox$, and let $\awtsymbol\in\wtnonterms\cup\wtterms$.
    Then, $\awtsymbol.\wtinlabel, \awtsymbol.\wtprepromise \sqsubseteq \inof{\awtsymbol.\wtsymlabel}$ and $\awtsymbol.\wtoutlabel, \awtsymbol.\wtpostpromise\sqsubseteq \outof{\awtsymbol.\wtsymlabel}$.
\end{lemma}
Since the approximators do not turn $\omega$-inputs to concrete inputs, and unconstrained counters are $\omega$ all across the NGVAS, we know that all unconstrained counters remain $\omega$ at all markings.
\begin{lemma}\label{Lemma:CGUnconstrained}
    Let $\wtgrammarof{\otherctxNGVAS{\anngvas}{\amarking, \anonterm, \amarkingp}, \postapprox, \preapprox}=(\wtnonterms, \wtterms, \wtstartnonterm, \wtprods)$ for post- and pre-approximators $\postapprox$ and $\preapprox$, and let $\awtsymbol\in\wtnonterms\cup\wtterms$.
    Then we have $\unconstrained\subseteq\omegaof{\awtsymbol.\wtinlabel}\cap\omegaof{\awtsymbol.\wtoutlabel}\cap\omegaof{\awtsymbol.\wtpostpromise}\cap\omegaof{\awtsymbol.\wtprepromise}$.
\end{lemma}

The coverability grammar construction gives us information about pumping derivations, however, the completeness of this information is dependent on the approximation.
This is formulated by \Cref{Lemma:CGPumping}, which follows directly from the construction, and usual Karp-Miller arguments.

\begin{lemma}\label{Lemma:CGPumping}
    Let $\wtgrammarof{\otherctxNGVAS{\anngvas}{\amarkingin, \anonterm, \amarkingout}, \postapprox, \preapprox}=(\wtnonterms, \wtterms, \wtstartnonterm, \wtprods)$ for post- and pre-approximators $\postapprox$ and $\preapprox$.
    For any rule $\awtnonterm\to\awtnontermp\in\wtprodspump$, we have $(\awtnonterm.\inlabel, \awtnonterm.\outlabel)\sqsubseteq(\awtnontermp.\inlabel, \awtnontermp.\outlabel)$, and there are $\amarking, \amarkingp\in\Nomega^{d}$ and $\aword, \awordp\in(\trms\cup\nonterms)^{*}$ where the following holds.
    
    We have $\amarking\sqsubseteq\inof{\awtnonterm.\symlabel}$, $\amarkingp\sqsubseteq\outof{\awtnonterm.\symlabel}$, $(\amarking, \amarkingp)\leq(\awtnonterm.\inlabel, \awtnonterm.\outlabel)$ where $\awtnonterm.\symlabel\to^{*}\aword.(\awtnonterm.\wtsymlabel).\awordp$, $\awtnonterm.\inlabel\in\postapproxof{\amarking, \aword}$, $\awtnonterm.\outlabel\in\preapproxof{\amarkingp, \awordp}$, and $(\amarking, \amarkingp)[i]<(\awtnonterm.\inlabel, \awtnonterm.\outlabel)[i]$ whenever $(\awtnontermp.\inlabel, \awtnontermp.\outlabel)[i]=\omega$ and $(\awtnonterm.\inlabel, \awtnonterm.\outlabel)[i]\neq\omega$ for $i\in[1, 2d]$.
\end{lemma}

If a coverability grammar $\wtgrammarof{\otherctxNGVAS{\anngvas}{\amarking, \anonterm, \amarkingp}, \postapprox, \preapprox}$ does not contain a non-terminal of the form $(\inof{\anontermpp}, \anontermpp, \outof{\anontermpp})$, we say that it \emph{remains bounded}.
If instead, there is such a non-terminal, we say that the grammar \emph{shows unboundedness}.
We observe that, if the grammar shows unboundedness for complete approximators $\natpostapprox$ and $\natpreapprox$, then \perfectnesspumpingnospace\ holds.
This is similar to the case of the Karp-Miller graph. 
If we use the approximators $\natpostapprox$ and $\natpreapprox$, the information is complete enough to expose the boundedness of $\anngvas$-runs, whenever $\anngvas$ does not have \perfectnesspumpingnospace.
\LemmaCGPumpabilityMainPaper*
\begin{proof}[Proof Sketch]
    Let $\awtnonterm\in\nonterms$ with $\omegaof{\awtnonterm.\wtinlabel}=\omegaof{\awtnonterm.\wtoutlabel}=\abdinfomid$.
    There must be a sequence of rules $[\awtnontermp^{(i)}\to\aword_{i}.\awtnontermp^{(i+1)}.\awordp_{i}]_{i\leq k}$ in $\wtprods$ such that $\awtnontermp^{(0)}.\wtsymlabel=\startnonterm$, $\awtnontermp^{(0)}.\wtinlabel=\acontextin$, $\awtnontermp^{(0)}.\wtoutlabel=\acontextout$, and $\awtnontermp^{(k+1)}=\awtnontermp$.
    By induction on $i\leq k$, we show that there is a sequence of derivations $[\startnonterm\to^{*}\aword_{i, j}.(\awtnonterm^{(i)}.\symlabel).\awordp_{i, j}]_{j\in\N}$ that pumps the $\omega$'s obtained when moving from $\startnonterm$ to $\awtnonterm$.
    The base case is trivial, and in the inductive case for $i+1$, we extend the derivations in the sequence from step $i$ to pump towards the markings of $\awtnonterm^{(i+1)}$.
    We make a case distinction on whether $(\awtnontermp^{(i)}\to\aword_{i}.\awtnontermp^{(i+1)}.\awordp_{i})\in\wtprodssim$.
    If this is the case, then the runs that witness the $\natpostapprox$, $\natpreapprox$, $\postfuncN{\anngvas}$ and $\prefuncN{\anngvas}$ images yield the extending derivations we need.
    If instead $(\awtnontermp^{(i)}\to\aword_{i}.\awtnontermp^{(i+1)}.\awordp_{i})\in\wtprodspump$, then \Cref{Lemma:CGPumping} yields the extending derivations.
\end{proof}

\paragraph*{Capturing the Language.}

The coverability grammar has the purpose of overapproximating $\runsof{\anngvas}$.
To make this notion precise, we define a notion of runs captured by a coverability grammar.
Towards a definition, we develop an annotation relation between the parse trees in a coverability grammar, and the reachability trees of $\anngvas$. 
For the rest of the subsection, fix a coverability grammar $\wtgrammar=\wtgrammarof{\anngvas, \postapprox, \preapprox}=(\wtnonterms, \wtterms, \wtstartnonterm, \wtprods)$ for some post- and pre-approximations $\postapprox$ and $\preapprox$.
We say that a complete derivation tree $\atree$ in $\wtgram$ is a $\wtgrammar$-annotation for a reachability tree $\atree\in\reachtreesof{\anngvas}$, denoted $(\atree, \atreep)\in\cgannof{\wtgrammar}\subseteq\reachtreesof{\anngvas}\times\treesof{\wtgrammar}$ if there is a map $\covproj:\atreep\to\atree$, where the child-relations are respected up to $\wtprodspump$ applications, and the labels of $\atreep$ generalize the labels of $\atree$.
Formally $(\atree, \atreep)\in\cgannof{\wtgrammar}$, if there is a map $\covproj:\atreep\to\atree$ where for all $\anode\in\atreep$ it holds that 
\begin{itemize}
    \item (Generalization) $\anode.\wtsymlabel=\covprojof{\anode}.\symlabel$, $\anode.\inlabel\sqsubseteq\covprojof{\anode}.\wtinlabel, \covprojof{\anode}.\wtprepromise$, and $\anode.\outlabel\sqsubseteq\covprojof{\anode}.\wtoutlabel, \covprojof{\anode}.\wtpostpromise$
    \item (Pumping) If $\anode\in\atreep$ is labeled $\awtnonterm\in\wtnontermspump$, then $\childnodesof{\anode}=\anodep\in\atreep^{1}$, and we have $\covprojof{\anode}=\covprojof{\anodep}$.
    \item (Local Bijection) If $\anode\in\atreep$ is labeled $\awtnonterm\in\wtnontermssim$, then for $\childnodesof{\anode}=\anode_{\lefttag}.\anode_{\righttag}$ and $\childnodesof{\anodep}=\anodep_{\lefttag}.\anodep_{\righttag}$, we have $\covprojof{\anode_{\lefttag}}=\anodep_{\lefttag}$ and $\covprojof{\anode_{\lefttag}}=\anodep_{\righttag}$.
\end{itemize}

We define the set 
\begin{align*}
    \cgreachtreesof{\wtgrammar, \awtnonterm}&=\setcond{\atree\in\reachtreesof{\anngvas}}{\\
    &\hspace{2em}(\atree, \atreep)\in\cgannof{\wtgrammar},\;\nodemarkingof{\atreep.\arootnode}=\awtnonterm}
\end{align*}
to be the set of $\anngvas$-reachability trees that can be $\wtgrammar$-annotated by a tree with the given root $\awtnonterm\in\wtnonterms$.
The runs of a coverability grammar, with the root node $\awtnonterm\in\wtnonterms$, is then defined to be 
$$\cgrunsof{\wtgrammar, \anonterm}=\bigcup_{\atree\in\cgreachtreesof{\wtgrammar, \awtnonterm}}\runsof{\atree}$$
We claim that the runs of $\wtgrammar$ overapproximate the runs of $\anngvas$, while still containing only the derivable runs.
In the following, we show that the runs of $\wtgrammar$ overapproximate the runs $\anngvas$, while still consisting of reaching runs.
The latter inclusion is already clear from the fact that the runs of $\wtgrammar$ from $\wtstartnonterm$ are the runs of $\atree\in\reachtreesof{\anngvas}$ that have the root symbol $\startnonterm$ while going from $\atree.\inlabel\sqsubseteq\acontextin$ to $\atree.\outlabel\sqsubseteq\acontextout$ by (Generalization).
Then, we only need to argue the the former inclusion.
\begin{lemma}\label{Lemma:CGLanguage}
    $\runsof{\anngvas}\subseteq\cgrunsof{\wtgrammar, \wtstartnonterm}\subseteq\setcond{(\amarking, \arun, \amarkingp)\in\runsof{\aword}}{\startnonterm\to^{*}\aword,\;\amarking\sqsubseteq\acontextin, \amarkingp\sqsubseteq\acontextout}$.
\end{lemma}

We define three notions of annotations that constitute the parts of a $\wtgrammar$-annotation, namely forward-, and backward-annotations.
We show \Cref{Lemma:CGLanguage} in two steps.
In the first step, we show that forward-, and backward-annotations exist for any reachability tree.
Then, we use these annotations to construct a $\wtgrammar$-annotation for reachability trees.

A marked parse tree $\atreep$ with labels in $\setcond{(\amarking, \asymbol, \amarkingp)\in\Nomega^{d}\times(\nonterms\cup\trms)\times\Nomega^{d}}{\amarkingp\in\postapproxof{\amarking, \asymbol}}$ is called a \emph{forward-annotation} of a $\anngvas$-reachability tree $\atree$, if there is a bijection $\covproj:\atreep\to\atree$ that preserves the child-relation where the following holds.
First, it must hold that $\covprojof{\anode}.\inlabel\sqsubseteq\anode.\inlabel$, $\covprojof{\anode}.\symlabel=\anode.\symlabel$ and $\covprojof{\anode}.\outlabel\sqsubseteq\anodep.\symlabel$.
Second, for the right child $\anodep$ of $\anode$, it must hold that $\anodep.\outlabel\sqsubseteq\anode.\outlabel$.
The \emph{backward-annotation} is defined similarly, with only the direction reversed and the labels stemming from $\setcond{(\amarking, \asymbol, \amarkingp)\in\Nomega^{d}\times(\nonterms\cup\trms)\times\Nomega^{d}}{\amarking\in\preapproxof{\amarkingp, \asymbol}}$.
We argue that all reachability trees have forward-, and backward-annotations that start from any suitable root node.

\begin{lemma}\label{Lemma:CGDirectionalAnnotation}
    Let $\atree\in\reachtreesof{\anngvas}$, and let $\amarking, \amarkingp\in\Nomega^{d}$ with $\atree.\inlabel\sqsubseteq\amarking$, and $\atree.\outlabel\sqsubseteq\amarkingp$.
    Then it has a forward-annotation $\atreep_{fwd}$ with $\atreep_{fwd}.\inlabel=\amarking$, and a backward annotiation $\atreep_{bck}$ with $\atreep_{bck}.\outlabel=\amarkingp$.
\end{lemma}

\begin{proof}
    The proof is by induction on the height of $\atree$.
    We only show the forward case, since the backward case is similiar.
    The base case is trivial, as we only need to choose a $\amarking'$ such that $\amarking'\in\postapproxof{\amarking, \atree.\symlabel}$ and $\atree.\outlabel\sqsubseteq\amarking'$.
    Such a $\amarking'$ must exist by the definitions of the reachability tree and post-approximation.
    We move on to the inductive case.
    Let $\atree_{\lefttag}$ and $\atree_{\righttag}$ be the subtrees rooted on the left and right children of $\atree$, and let $\anode$ be the root node of $\atree$.
    Let $\amarking\in\Nomega^{d}$ with $\atree.\inlabel\sqsubseteq\amarking$.
    By the induction hypothesis, we get the forward-annotation $\atreep_{\lefttag}$ of $\atree_{\lefttag}$ with $\atree_{\lefttag}.\inlabel=\amarking$.
    The forward-annotation property yields $\atree_{\lefttag}.\outlabel\sqsubseteq\atreep_{\lefttag}.\outlabel$, and we have $\atree_{\righttag}.\inlabel=\atree_{\lefttag}.\outlabel$ by the definition of the reachability tree.
    Then, we can call the induction hypothesis with the tree $\atree_{\righttag}$ and the marking $\atreep_{\lefttag}.\outlabel$.
    This yields a forward-annotation $\atreep_{\righttag}$ of $\atree_{\righttag}$, where $\atreep_{\righttag}.\inlabel=\atreep_{\lefttag}.\outlabel$ and $\atree_{\righttag}.\outlabel\sqsubseteq\atreep_{\righttag}.\outlabel$.
    Note that the definition of forward-annotation yields $\atreep_{\lefttag}.\outlabel\in\postapproxof{\atreep_{\lefttag}.\inlabel, \atreep_{\lefttag}.\symlabel}$, and $\atreep_{\righttag}.\outlabel\in\postapproxof{\atreep_{\righttag}.\inlabel, \atreep_{\righttag}.\symlabel}$.
    There must be a rule $\anode.\symlabel\to(\atree_{\lefttag}.\symlabel).(\atree_{\righttag}.\symlabel)=(\atreep_{\lefttag}.\symlabel).(\atreep_{\righttag}.\symlabel)$ by the definition of the reachability tree.
    The precision property yields a $\amarkingp\in\postapproxof{\amarking, \anonterm}$ with $\atreep_{\righttag}.\outlabel\sqsubseteq\amarkingp$.
    It can be readily checked that the labeled tree $\atreep$ with nodes $\anode\cup 0.\atreep_{\lefttag}\cup 1.\atreep_{\righttag}$, where $\atreep_{\lefttag}$ and $\atreep_{\righttag}$ are the left- and right-subtrees, and $\nodemarkingof{\anode}=(\amarking, \atree.\symlabel, \amarkingp)$ fulfills the conditions for a forward-annotation of $\atree$.
\end{proof}

Now, we show the following result.
Note that this implies \Cref{Lemma:CGLanguage}, since $\wtstartnonterm\in\wtnontermspump$.

\begin{lemma}\label{Lemma:CGAnnotation}
    Let $\atree\in\reachtreesof{\anngvas}$, and let $\awtnonterm\in\wtnontermspump$ such that $\atree.\wtinlabel\sqsubseteq\awtnonterm.\inlabel$, $\atree.\symlabel=\awtnonterm.\wtsymlabel$ and $\atree.\outlabel\sqsubseteq\awtnonterm.\wtoutlabel$.
    Then, there is a $\wtgrammar$-annotation $\atreep$ of $\atree$ with $\nodemarkingof{\atree.\arootnode}=\awtnonterm$.
\end{lemma}

\begin{proof}
    Let $\atree\in\reachtreesof{\anngvas}$, $\awtnonterm\in\wtnontermspump$, $\atree.\inlabel\sqsubseteq\awtnonterm.\wtinlabel$, $\atree.\symlabel=\awtnonterm.\wtsymlabel$ and $\atree.\outlabel\sqsubseteq\awtnonterm.\wtoutlabel$.
    The proof is by induction on the number of non-terminals that can be called by $\awtnonterm$ in $\wtgrammar$.
    We construct the annotation top down, and whenever we can, call the induction hypothesis.
    The base case and the inductive cases are similar, so we only handle the inductive case.
    By the construction of the coverability grammar, there are $\amarking, \amarkingp\in\Nomega^{d}$ with $\awtnonterm.\wtinlabel\sqsubseteq\amarking$, $\awtnonterm.\wtoutlabel\sqsubseteq\amarkingp$, where (a) for all $(\awtnonterm\to\awtnontermp)\in\wtprods$, $\awtnontermp.\wtinlabel=\amarking$, $\awtnontermp.\wtoutlabel=\amarkingp$, and (b) for all $\aprommarking\in\postapproxof{\amarking, \awtnontermp.\symlabel}$, $\aprommarkingp\in\preapproxof{\amarkingp, \awtnontermp.\symlabel}$, we have $(\awtnonterm\to\awtnontermp)\in\wtprods$ where $\awtnontermp.\wtpostpromise=\aprommarking$ and $\awtnontermp.\wtprepromise=\aprommarkingp$.
    \Cref{Lemma:CGDirectionalAnnotation} applies for these $\amarking, \amarkingp\in\Nomega^{d}$,  since $\atree.\inlabel\sqsubseteq\awtnonterm.\wtinlabel\sqsubseteq\amarking$ and $\atree.\outlabel\sqsubseteq\awtnonterm.\wtoutlabel\sqsubseteq\amarkingp$.
    By \Cref{Lemma:CGDirectionalAnnotation}, we obtain forward- and backward-annotations $\atreep_{fwd}$ and $\atreep_{bck}$ of $\atree$ with $\atreep_{fwd}.\inlabel=\amarking$ and $\atreep_{bck}.\outlabel=\amarkingp$.
    Let $\pi_{fwd}:\atreep_{fwd}\to\atree$, $\pi_{bck}:\atreep_{bck}\to\atree$ be the bijections that witness the annotations, and $\gamma_{fwd}:\atree\to\atreep_{fwd}$ and $\gamma_{bck}:\atree\to\atreep_{bck}$ the inverses of these bijections.
    For the purposes of this proof, we call the pair $(\atreep, \covproj)$ consisting of a derivation tree $\atreep$, and a map $\covproj:\atreep\to\atree$ a \emph{half-annotation}, if $\covprojof{\atreep.\arootnode}=\atree.\arootnode$, all non-leaf nodes $\anode\in\atreep$ have (Generalization), (Pumping), (Local Bijection),  all leaf nodes $\anode\in\atreep$ have (Generalization), for $\covprojof{\anode}=\anodepp$, we have 
    \begin{align*}
        \nodemarkingof{\anode}&=((\gamma_{fwd}(\anodepp).\inlabel, \gamma_{fwd}(\anodepp).\outlabel), \\
        &\hspace{4em}\anodepp.\symlabel, (\gamma_{bck}(\anodepp).\inlabel, \gamma_{bck}(\anodepp).\outlabel))
    \end{align*}
    and $\nodemarkingof{\anode}\in\wtnontermspump$ implies that $\nodemarkingof{\anode}$ can call less non-terminals than $\awtnonterm$.
    We call it a saturated-annotation, if there is no leaf $\anode\in\atreep$ with $\nodemarkingof{\anode}\in\wtnontermssim$.
    The existence of a saturated-annotation concludes the proof of \Cref{Lemma:CGAnnotation}: In this case, for all leaves $\anode\in\atreep$ we have $\nodemarkingof{\anode}\in\wtterms$ or $\nodemarkingof{\anode}\in\wtnontermspump$ where $\nodemarkingof{\anode}$ can call less non-terminals than $\awtnonterm$.
    For all leaves where the latter case holds, the induction hypothesis applies, and we can extend the annotation towards a complete $\wtgrammar$-annotation.
    
    The construct a saturated-annotation inductively.
    First, we show that there is a half-annotation.
    Then, we show given a half-annotation $(\atreep, \covproj)$, it is a saturated-annotation, or there is a half-annotation $(\atreep', \covproj')$ that has $\covprojof{\atreep}\subsetneq\covprojof{\atreep'}$.
    A simple argument by contradiction shows the existence of a saturated-annotation.
    Otherwise, applying the extention to the initial half-annotation must yield a half-iteration where all leaves are labeled $\wtterms$, which contradicts our assumption, or, there must be an infinite sequence of sets $\covprojof{\atreep_{0}}\subsetneq\covprojof{\atreep_{1}}\ldots$ that are all subsets of a finite set $\atree$.
    Now, we construct a half-annotation.
    Let $\atreep$ consist of a root node $\anode$ and a child $\anodep$, where $\nodemarkingof{\anode}=\awtnonterm$, and $\nodemarkingof{\anodep}=((\atree_{fwd}.\inlabel, \atree_{fwd}.\outlabel), \atree.\symlabel, (\atree_{bck}.\inlabel, \atree_{bck}.\outlabel))$.
    Let $\covproj:\atreep\to\atree$ with $\covprojof{\anode}=\covprojof{\anodep}=\atree.\arootnode$.
    Since $\atree_{fwd}.\inlabel=\amarking$, $\atree_{bck}.\outlabel=\amarkingp$, and $\awtnonterm$ as given in the premise, we know that (Generalization) holds everywhere.
    We also know that (Pumping) and (Local Bijection) holds for $\anode$.
    It is clear that, if $\nodemarkingof{\anodep}\in\wtnontermspump$, then $\anodep$ cannot call $\awtnonterm$ and $\nodemarkingof{\anodep}$ by the nature of the pumping rules.  
    However, $\awtnonterm$ can call all non-terminals $\nodemarkingof{\anodep}$ can.
    Thus, $\nodemarkingof{\anodep}$ can call strictly less non-terminals, showing that $(\atreep, \covproj)$ is indeed a half-annotation.
    Now let $(\atreep, \covproj)$ be a half-annotation.
    Let $\anode\in\atreep$ with $\nodemarkingof{\anode}=\awtnontermp\in\wtnontermssim$, if there is no such $\anode$, then $(\atreep, \covproj)$ is already a saturated-annotation.
    Let $\anodep=\covprojof{\anode}$.
    We know that 
    \begin{align*}
        \nodemarkingof{\anodep}&=((\gamma_{fwd}(\anodep).\inlabel, \gamma_{fwd}(\anodep).\outlabel),\\
        &\hspace{4em}\anode.\symlabel, (\gamma_{bck}(\anodep).\inlabel, \gamma_{bck}(\anodep).\outlabel)).
    \end{align*}
    Let $\childnodesof{\anodep}=\anodep_{\lefttag}.\anodep_{\righttag}$.
    The definition of the forward- and backward-annotations, along with the rules in the $\awtnontermp\in\wtnontermssim$ case, already imply that there are $\awtsymbol, \awtsymbolp\in\wtnonterms\cup\wtterms$, where $\awtnonterm\to\awtsymbol^{\lefttag}.\awtsymbolp^{\righttag}$ and 
    \begin{align*}
        \awtsymbol^{dir}&=((\gamma_{fwd}(\anodep_{dir}).\inlabel, \gamma_{fwd}(\anodep_{dir}).\outlabel), \atree.\symlabel,\\
        &(\gamma_{bck}(\anodep_{dir}).\inlabel, \gamma_{bck}(\anodep_{dir}).\outlabel))
    \end{align*}
    for $dir\in\set{\lefttag, \righttag}$.
    We extend $\atreep$ by two nodes $\anode_{\lefttag}$ and $\anode_{\righttag}$, with $\nodemarkingof{\anode_{dir}}=\awtsymbol^{dir}$ and $\covprojof{\anode_{dir}}=\anodep_{dir}$ for $dir\in\set{\lefttag, \righttag}$.
    The condition (Generalization), (Pumping), and (Local Bijection) clearly hold for all non-leaf nodes, and (Generalization) holds for the new leaves by the definition of the forward- and backward-annotations.
    If $\awtsymbol^{dir}\in\wtnontermspump$ for some $dir\in\set{\lefttag, \righttag}$, it is clear by an argument similar to the construction of the initial half-annotation that $\awtsymbol^{dir}$ can call less non-terminals than $\awtnonterm$.
    Then, the extention is also a half-annotation.
    Clearly, we have extended the image of $\covproj$.
    This concludes the proof.
\end{proof}
\newcommand{\wtomegafun}{\Omega_{cg}}
\newcommand{\wtomegaof}[1]{\wtomegafun(#1)}
\newcommand{\dimofN}{\mathsf{u}}

\subsection{The Pumpability Decomposition, Non-Linear Case}
This section is dedicaded to proving the specification of $\refinepump$ for non-linear NGVAS $\anngvas$.
Our goal is to show that we can construct a decomposition of relevant $\otherctxNGVAS{\anngvas}{\amarking, \anonterm, \amarkingp}$ for $(\amarking, \anonterm, \amarkingp)\in\Nomega^{d}\times\nonterms\times\Nomega^{d}$, given a coverability grammar that exposes boundedness.
This is \Cref{Lemma:DecompGivenCG} from the extended paper.
\LemmaDecompGivenCGMainPaper*
For the rest of this section, fix a $\amarking, \amarkingp\in\Nomega^{d}$ and $\anonterm\in\nonterms$ with $\amarking\sqsubseteq\inof{\anonterm}$, $\amarkingp\sqsubseteq\outof{\anonterm}$, and $\unconstrained\subseteq\omegaof{\amarking}, \omegaof{\amarkingp}$.
Also fix a coverability grammar $\wtgrammarof{\otherctxNGVAS{\anngvas}{\amarking, \anonterm, \amarkingp}, \postapprox, \preapprox}=(\wtnonterms, \wtterms, \wtstartnonterm, \wtprods)$.
To show \Cref{Lemma:DecompGivenCG}, we show \Cref{Lemma:DecompGivenCGLowLevel} that implies it.
We make a few definitions.
For $\awtsymbol\in\wtnonterms$, we let $\nonterms_{\awtsymbol}^{\calltag}\subseteq\wtnonterms\cup\wtterms$ be the set of coverability grammar symbols that can be called by $\awtsymbol$.
Further define $\wtinfun, \wtoutfun:\wtnonterms\to\Nomega^{d}$, $\wtsymfun:\wtnonterms\to\Nomega^{d}$ and $\wtomegafun:\wtnonterms\to\powof{[1,d]}$ as 
\begin{align*}
    \wtinof{\awtsymbol}&=\awtsymbol.\wtinlabel\sqcap\awtsymbol.\wtprepromise,\\
    \wtoutof{\awtsymbol}&=\awtsymbol.\wtoutlabel\sqcap\awtsymbol.\wtpostpromise,\\
    \wtsymof{\awtsymbol}&=\asymbol,\\
    \wtomegaof{\awtsymbol}&=\omegaof{\wtinof{\awtsymbol}}\cap\omegaof{\wtoutof{\awtsymbol}}
\end{align*}
for all $\awtsymbol\in\wtnonterms$.
With these definitions at hand, we are ready to state \Cref{Lemma:DecompGivenCGLowLevel}.
\begin{lemma}\label{Lemma:DecompGivenCGLowLevel}
    Let $\wtgrammarof{\otherctxNGVAS{\anngvas}{\amarking, \anonterm, \amarkingp}, \postapprox, \preapprox}$ remain bounded.
    Then, for each $\awtsymbol\in\wtnonterms\cup\wtterms$, we can construct the decomposition $\set{\anngvas_{\awtsymbol}}$ of $\otherctxNGVAS{\anngvas}{\amarking, \anonterm, \amarkingp}$, with $\recrankof{\anngvas_{\awtsymbol}}<\recrankof{\anngvas}$, where 
    \begin{align*}
        \runsof{\anngvas_{\awtsymbol}}&=\cgrunsof{\otherctxNGVAS{\anngvas}{\amarking, \anonterm, \amarkingp}, \awtsymbol}\\
        \anngvas_{\awtsymbol}.\restrictions&=\Z^{d}\\
        \awtsymbol&\in\wtnonterms\\
        \anngvas_{\awtsymbol}.\unconstrained&=\wtomegaof{\awtsymbol}\\
        \anngvas_{\awtsymbol}.\acontextin&=\wtinof{\awtsymbol}\\
        \anngvas_{\awtsymbol}.\acontextout&=\wtoutof{\awtsymbol}
    \end{align*}
\end{lemma}
We argue that this implies \Cref{Lemma:DecompGivenCG}.
\begin{proof}[Proof of \Cref{Lemma:DecompGivenCG}]
    Using \Cref{Lemma:DecompGivenCGLowLevel}, we get a $\anngvas_{\wtstartnonterm}$, with smaller non-linear rank than $\otherctxNGVAS{\anngvas}{\amarking, \anonterm, \amarkingp}$, and 
    $\runsof{\anngvas_{\wtstartnonterm}}=\cgrunsof{\otherctxNGVAS{\anngvas}{\amarking, \anonterm, \amarkingp}, \wtstartnonterm}$.
    By \Cref{Lemma:CGLanguage}, we get that 
    \begin{align*}
        &\runsof{\otherctxNGVAS{\anngvas}{\amarking, \anonterm, \amarkingp}}\subseteq\cgrunsof{\otherctxNGVAS{\anngvas}{\amarking, \anonterm, \amarkingp}, \wtstartnonterm}\subseteq\\
        &\setcond{(\amarking', \arun, \amarkingp')\in\runsof{\aword}}{\anonterm\to^{*}\aword, \;\amarking'\sqsubseteq\amarking,\;\amarkingp'\sqsubseteq\amarkingp}.
    \end{align*}
    Since $\runsof{\otherctxNGVAS{\anngvas}{\amarking, \anonterm, \amarkingp}}$ does not impose a restriction not imposed by $\anngvas_{\wtstartnonterm}$ we get $\runsof{\otherctxNGVAS{\anngvas}{\amarking, \anonterm, \amarkingp}}=\runsof{\anngvas_{\wtstartnonterm}}$.
    We show the deconstruction conditions (i)-(iv).
    We have $\wtinof{\wtstartnonterm}\sqsubseteq \amarking$ and $\wtoutof{\wtstartnonterm}\sqsubseteq \amarkingp$ by the definition of the coverability grammar.
    Since $\omegaof{\amarking}\cap\omegaof{\amarkingp}\subseteq\omegaof{\wtinof{\wtstartnonterm}}\cap\omegaof{\wtoutof{\wtstartnonterm}}$ by \Cref{Lemma:CGUnconstrained}, and $\anngvas_{\wtstartnonterm}.\restrictions=\Z^{d}$ by \Cref{Lemma:DecompGivenCGLowLevel} as well, deconstruction conditions hold.
\end{proof}
Now we show \Cref{Lemma:DecompGivenCGLowLevel}.
Just as in the previous sections, we first lay out the construction, and then prove that it is correct.
\paragraph*{The Construction.} 
The NGVAS $\agram_{\awtsymbol}$ is constructed inductively on $\cardof{\nonterms_{\awtsymbol}^{\calltag}}$.
As a common component across the construction, let $\acontext_{\awtsymbol}=(\wtinof{\awtsymbol}, \wtoutof{\awtsymbol})$ for $\awtsymbol\in\wtnonterms\cup\wtterms$.
We proceed with the base case $\cardof{\nonterms_{\awtsymbol}^{\calltag}}=0$, which implies $\awtsymbol\in\wtterms$.
For $\awtsymbol.\wtsymlabel=\anngvasp$, we let
\begin{align*}
    \anngvas_{\awtsymbol}=(\anngvasp.\agram, \acontext_{\awtsymbol}, \anngvasp.\restrictions, \wtomegaof{\awtsymbol}, \anngvasp.\abdinfo)
\end{align*}
have the same components as $\anngvasp$, up to the context information and the unconstrained counters.
Note that $\acontext_{\awtsymbol}\sqsubseteq\anngvasp.\acontext$.
For the inductive case, we have $\cardof{\nonterms_{\awtsymbol}^{\calltag}}>0$, and thus $\awtsymbol=\awtnonterm\in\wtnonterms$.
We distinguish between the cases $\awtnonterm\in\wtnontermspump$ and $\awtnonterm\in\wtnontermssim$.
In both cases, we have 
\begin{align*}
    \anngvas_{\awtnonterm}&=(\agram_{\awtnonterm}, \acontext_{\awtnonterm}, \Z^{d}, \wtomegaof{\awtnonterm}, \abdinfo^{\awtnonterm})\\
    \agram_{\awtnonterm}&=(\sccof{\awtnonterm}, \trms_{\awtnonterm}, \awtnonterm, \prods_{\awtnonterm})
\end{align*}
but we handle the definitions of components $\prods_{\awtnonterm}$, $\trms_{\awtnonterm}$ and $\abdinfo^{\awtnonterm}$ differently between the cases.
First consider the case $\awtnonterm\in\wtnontermspump$.
Here, we have 
\begin{align*}
    \prods_{\awtnonterm}&=\setcond{\awtnonterm\to\anngvas_{\awtnontermp}}{(\awtnonterm\to\awtnontermp)\in\wtprods}\\
    \trms_{\awtnonterm}&=\setcond{\anngvas_{\awtnontermp}}{(\awtnonterm\to\awtnontermp)\in\wtprods}\\
    \abdinfo&=(\omegaof{\awtnontermp.\wtinlabel}, \omegaof{\awtnontermp.\wtoutlabel}, \infun_{\awtnonterm}, \outfun_{\awtnonterm})
\end{align*}
where $\infun_{\awtnonterm}(\awtnonterm)=\awtnontermp.\wtinlabel$ and $\outfun_{\awtnonterm}(\awtnonterm)=\awtnontermp.\wtoutlabel$ for some $\awtnontermp\in\wtnonterms$ with $(\awtnonterm\to\awtnontermp)\in\wtprods$.
Note that the choice of $\awtnontermp\in\wtprods$ does not matter here, since all such $\awtnontermp\in\wtnonterms$ agree on the $\wtinlabel$ and $\wtoutlabel$ components by the construction of the coverability grammar.
Also note that our weak-NGVAS definition only allows rules that produce two symbols, and we violate this rule here.
This is a notational choice: we could make the construction strictly formally correct by appending a base-case NGVAS $\anngvas_{\varepsilon}$ with fitting context information to the right hand side of rules in $\prods_{\awtnonterm}$,  where $\anngvas_{\varepsilon}$ only produces $\varepsilon\in\updates^{*}$.
We ommit this component to keep the presentation clear.
Now consider the case $\awtnonterm\in\wtnontermspump$.
In this case, we let $\abdinfo=(\omegaof{\wtinof{\awtnonterm}}, \omegaof{\wtoutof{\awtnonterm}}, \wtinfun, \wtoutfun)$, 
\begin{align*}
    \prods_{\awtnonterm}&=\setcond{\awtnontermp\to\awordp}{(\awtnontermp\to\aword)\in\wtprods,\; \\
&\hspace{4em}\awtnontermp\in\sccof{\awtnonterm},\;\awordp\in\realizationof{}{\aword}},
\end{align*}
and $\trms_{\awtnonterm}=\setcond{\aterm'\in\realizationof{}{\awtsymbol}}{\awtsymbol\in\nonterms_{\awtnonterm}^{\calltag}\setminus\sccof{\awtnonterm}}$, defined with the help of a homeomorphism $\realization{}:(\set{\awtnonterm}\cup\nonterms_{\awtnonterm}^{\calltag})^{*}\to\ngvas^{*}$.
The homeomorphism $\realization{}$ is defined by its base cases $\realizationof{}{\awtsymbol}=\adecomp_{\awtsymbol}$ for all $\awtsymbol\in\nonterms_{\awtnonterm}^{\calltag}\setminus\sccof{\awtnonterm}$, and $\realizationof{}{\awtnontermp}=\set{\awtnontermp}$ for all $\awtnontermp\in\sccof{\awtnonterm}$.

\paragraph*{Proof.}
Before we move on to the proof of \Cref{Lemma:DecompGivenCGLowLevel}, we first show some properties relating to the unboundedness structure in the coverability grammar.
First, we observe that the call-structure imposes certain monotonity conditions on the components of $\awtnonterm\in\wtnonterms$.
Namely, as we apply $\wtprodssim$ rules, then the $\wtinlabel$-$\wtoutlabel$ components get generalized, and $\wtpostpromise$-$\wtprepromise$ have an upper bound on the set of their unbounded counters.
\begin{lemma}\label{Lemma:CGCallMonotonicity}
    Let $\awtnonterm\in\wtnonterms$, $\awtsymbol\in(\wtnonterms\cup\wtterms)$, and let $\awtnonterm\to\aword.\awtsymbol.\awordp\in\wtprodssim$ for some $\aword, \awordp\in(\wtnonterms\cup\wtterms)^{1}\cup(\wtnonterms\cup\wtterms)^{1}$.
    Then, 
    \begin{align*}
        \omegaof{\awtnonterm.\wtinlabel}\subseteq\omegaof{\awtsymbol.\wtinlabel}&\subseteq\omegaof{\awtsymbol.\wtpostpromise}\subseteq\omegaof{\awtnonterm.\wtpostpromise}\\
        \omegaof{\awtnonterm.\wtoutlabel}\subseteq\omegaof{\awtsymbol.\wtoutlabel}&\subseteq\omegaof{\awtsymbol.\wtprepromise}\subseteq\omegaof{\awtnonterm.\wtprepromise}
    \end{align*}
\end{lemma}
\begin{proof}
    We only show $\omegaof{\awtnonterm.\wtinlabel}\subseteq\omegaof{\awtsymbol.\wtinlabel}\subseteq\omegaof{\awtsymbol.\wtpostpromise}\subseteq\omegaof{\awtnonterm.\wtpostpromise}$, since the proofs of other inequalities are similar to these ones.
    We have $\awtsymbol.\wtinlabel=\postapproxof{\awtnonterm.\wtinlabel, \aword.\wtsymlabel}$ where we write $\varepsilon.\wtsymlabel=\varepsilon$.
    By the correct unboundedness property, we observe $\omegaof{\awtnonterm.\wtinlabel}\subseteq\omegaof{\awtsymbol.\wtinlabel}$.
    Note that since $\omegaof{\awtnonterm.\wtinlabel}\subseteq\omegaof{\awtsymbol.\wtinlabel}$ and $\awtsymbol.\wtpostpromise\in\postapproxof{\awtnonterm.\wtinlabel, \awtsymbol.\wtsymlabel}$, we have $\omegaof{\awtsymbol.\wtinlabel}\subseteq\omegaof{\awtsymbol.\wtpostpromise}$.
    Let $\cardof{\awordp}=0$.
    Then, by construction, $\awtsymbol.\wtpostpromise\sqsubseteq\awtnonterm.\wtpostpromise$, which implies $\omegaof{\awtsymbol.\wtpostpromise}\subseteq\omegaof{\awtnonterm.\wtpostpromise}$.
    Let $\cardof{\awordp}=1$ and thus $\awordp=\awtsymbolp\in\wtnonterms\cup\wtterms$.
    Then, $\awtsymbol.\wtpostpromise=\awtsymbolp.\wtinlabel$, $\awtsymbolp.\wtpostpromise\in\postapproxof{\awtsymbolp.\wtinlabel, \awtsymbolp.\wtsymlabel}$, and $\awtsymbolp.\wtpostpromise\sqsubseteq\awtnonterm.\wtpostpromise$.
    By a similar argument to the previous cases, we get $\omegaof{\awtsymbol.\wtpostpromise}\subseteq\omegaof{\awtsymbolp.\wtpostpromise}\subseteq\omegaof{\awtnonterm.\wtpostpromise}$.
    This concludes the proof.
\end{proof}
We also observe that only $\wtnontermssim$ non-terminals can call themselves.

\begin{lemma}\label{Lemma:CGSelfCall}
    Let $\awtnonterm\in\wtnonterms$ be able to call itself. 
    Then, $\awtnonterm\in\wtnontermssim$.
\end{lemma}

\begin{proof}
    Suppose $\awtnonterm\in\wtnontermspump$ can call itself.
    If $\awtnonterm=\wtstartnonterm$, then by construction any non-terminal can call $\awtnonterm$.
    Let $\awtnonterm\neq\wtstartnonterm$ and let $\awtnonterm$ call itself by applying the rules $\aprod_{0}\ldots\aprod_{k}\in\wtprods$, and using non-terminals $\awtnonterm^{0},\ldots,\awtnonterm^{k+1}\in\wtnonterms$, where $\awtnonterm^{0}=\awtnonterm^{k+1}=\awtnonterm$, and $\aprod_{i}$ consumes $\awtnonterm^{i}$ while producing $\awtnonterm^{i+1}$ for all $i<k$.
    Note that here, we use superscripts as indices, and avoid $A^{(-)}$ notation to avoid clutter.
    We claim that $\cardof{\omegaof{\awtnonterm^{0}.\wtinlabel}}+\cardof{\omegaof{\awtnonterm^{0}.\wtoutlabel}}<\cardof{\omegaof{\awtnonterm^{1}.\wtinlabel}}+\cardof{\omegaof{\awtnonterm^{1}.\wtoutlabel}}$ and $\cardof{\omegaof{\awtnonterm^{i}.\wtinlabel}}+\cardof{\omegaof{\awtnonterm^{i}.\wtoutlabel}}\leq\cardof{\omegaof{\awtnonterm^{i+1}.\wtinlabel}}+\cardof{\omegaof{\awtnonterm^{i+1}.\wtoutlabel}}$ for all $i\leq k$.
    This leads to the contradiction $\cardof{\omegaof{\awtnonterm.\wtinlabel}}+\cardof{\omegaof{\awtnonterm.\wtoutlabel}}<\cardof{\omegaof{\awtnonterm.\wtinlabel}}+\cardof{\omegaof{\awtnonterm.\wtoutlabel}}$.
    Note that since $\awtnonterm\neq\wtstartnonterm$ and $\awtnonterm\in\wtnontermspump$, any rule consuming $\awtnonterm$ produces exactly one non-terminal $\awtnontermp$ with $\cardof{\omegaof{\awtnonterm.\inlabel}}+\cardof{\omegaof{\awtnonterm.\outlabel}}<\cardof{\omegaof{\awtnontermp.\inlabel}}+\cardof{\omegaof{\awtnontermp.\outlabel}}$.
    This is because a $\wtprodspump$ rule is only produced when a pumping situation is discovered, which introduces at least one new $\omega$.
    Consider the second inequality $\cardof{\omegaof{\awtnonterm^{i}.\wtinlabel}}+\cardof{\omegaof{\awtnonterm^{i}.\wtoutlabel}}\leq\cardof{\omegaof{\awtnonterm^{i+1}.\wtinlabel}+\omegaof{\awtnonterm^{i+1}.\wtoutlabel}}$ for some $i\leq k$.
    If $\awtnonterm^{i}\in\wtnontermssim$, then the inequality follows from \Cref{Lemma:CGCallMonotonicity}.
    If $\awtnonterm^{i}\in\wtnontermspump$, since no non-terminal can call $\wtstartnonterm$, we get $\awtnonterm^{i}\neq\wtstartnonterm$.
    Thus the same argument we used for the first inequality applies.
    This concludes the proof.
\end{proof}

This yields the following results.

\begin{lemma}\label{Lemma:CGCallStructure}
    Let $\awtnontermp\in\sccof{\awtnonterm}$ for some $\awtnonterm\in\wtnonterms$.
    Then it holds that 
    \begin{align*}
        &\omegaof{\awtnonterm.\wtinlabel}=\omegaof{\awtnontermp.\wtinlabel}\;&\omegaof{\awtnonterm.\wtoutlabel}=\omegaof{\awtnontermp.\wtoutlabel}\\
        &\omegaof{\awtnonterm.\wtpostpromise}=\omegaof{\awtnontermp.\wtpostpromise}\;&\omegaof{\awtnonterm.\wtprepromise}=\omegaof{\awtnontermp.\wtprepromise}
    \end{align*}
    Let $\awtsymbol\in\wtnonterms\cup\wtterms$, $(\awtnonterm\to\aword.\awtnontermp.\awordp)\in\wtprods$ and let $\aword$ or $\awordp$ contain $\awtsymbol$.
    Then, 
    \begin{align*}
        \omegaof{\awtsymbol.\wtinlabel}&=\omegaof{\awtsymbol.\wtpostpromise}=\omegaof{\awtnonterm.\wtinlabel},\;\\
        \omegaof{\awtsymbol.\wtoutlabel}&=\omegaof{\awtsymbol.\wtprepromise}=\omegaof{\awtnonterm.\wtprepromise}\text{, if }\aword\text{ contains }\awtsymbol,\\
        \omegaof{\awtsymbol.\wtinlabel}&=\omegaof{\awtsymbol.\wtpostpromise}=\omegaof{\awtnonterm.\wtpostpromise},\;\\
        \omegaof{\awtsymbol.\wtoutlabel}&=\omegaof{\awtsymbol.\wtprepromise}=\omegaof{\awtnonterm.\wtoutlabel}\text{, if }\awordp\text{ contains }\awtsymbol.
    \end{align*}
\end{lemma}

\begin{proof}
    Let $\awtnontermp\in\sccof{\awtnonterm}$ and $\awtnonterm\in\wtnonterms$.
    Consider the first four equalities.
    If $\awtnonterm=\awtnontermp$, we are done.
    If $\awtnonterm\neq\awtnontermp$, then all symbols in $\sccof{\awtnonterm}$ can call themselves by only deriving symbols in $\sccof{\awtnonterm}$.
    First, by \Cref{Lemma:CGSelfCall}, we get $\sccof{\awtnonterm}\subseteq\wtnontermssim$.
    Then, $\awtnonterm$ can call $\awtnontermp$ by only using rules in $\wtnontermssim$ and vice-versa.
    By applying \Cref{Lemma:CGCallMonotonicity} twice, we get the desired inequalities.
    Now let $\awtsymbol\in\wtnonterms\cup\wtterms$ and let $\awtnonterm\to\aword.\awtnontermp.\awordp\in\wtprods$.
    We only show the case where $\aword$ contains $\awtsymbol$, the proof of the other case is similar.
    Note that, since we only allow each production to generate at most $2$ symbols, we have $\aword=\awtsymbol$ and $\awtnonterm\to\awtsymbol.\awtnontermp\in\wtprodssim$.
    Then, by construction, $\awtnonterm.\wtinlabel=\awtsymbol.\wtinlabel$, $\awtsymbol.\wtpostpromise\in\postapproxof{\awtsymbol.\wtinlabel, \awtsymbol.\wtsymlabel}$, and $\awtsymbol.\wtpostpromise=\awtnontermp.\wtinlabel$.
    We have already shown $\omegaof{\awtnonterm.\wtinlabel}=\omegaof{\awtnontermp.\wtinlabel}$.
    Then, by correct unboundedness of $\postapprox$, we get 
    \begin{align*}
        \omegaof{\awtnonterm.\wtinlabel}=\omegaof{\awtsymbol.\wtinlabel}&\subseteq\omegaof{\awtsymbol.\wtpostpromise}\\
        &=\omegaof{\awtnontermp.\wtinlabel}=\omegaof{\awtnonterm.\wtinlabel}.
    \end{align*}
    This concludes the proof.
\end{proof}

Now, we are ready to prove \Cref{Lemma:DecompGivenCGLowLevel}.

\begin{proof}[Proof of \Cref{Lemma:DecompGivenCGLowLevel}]
    Let $\awtsymbol\in\wtnonterms\cup\wtterms$, and let $\dimofN=d+1-\dimensionof{\lcyclespaceof{\anngvas}}$.
    The claims about the individual components of $\anngvas_{\awtsymbol}$ are already clear by construction.
    For this reason we only show: (ngvas) that the construction soundly produces a weak-NGVAS, (rk) that $\srankof{\anngvas_{\awtsymbol}}<_{\dimofN}\srankof{\anngvas}$, and that (run) $\runsof{\anngvas_{\awtsymbol}}=\runsof{\otherctxNGVAS{\anngvas}{\amarking, \anonterm, \amarkingp}, \awtsymbol}$.
    Here, we mean by $x <_{\dimofN} y$ that $x<y$ lexicographically and $x[\dimofN] < y[\dimofN]$.
    Note that the assumption on the rank suffices to show the main claim $\recrankof{\anngvas_{\awtsymbol}}<\recrankof{\anngvas}$, since $\unconstrained\subseteq\omegaof{\acontextin}\cap\omegaof{\acontextout}\subseteq\omegaof{\amarking}\cap\omegaof{\amarkingp}\subseteq\wtomegaof{\awtsymbol}$.
    The proof is by induction on $\cardof{\nonterms_{\awtsymbol}^{\calltag}}$.

    For the base case, we have $\cardof{\nonterms_{\awtsymbol}^{\calltag}}=0$, and thus $\awtsymbol\in\wtterms$. 
    The condition (ngvas) is already clear, so we proceed with (rk).
    Since $\anngvas$ has \perfectnesschildrennospace, and $\otherctxNGVAS{\anngvas}{\amarking, \anonterm, \amarkingp}$ has the same children, $\awtsymbol.\wtsymlabel\in\trms$ is perfect.
    Then, its children have \perfectnesschildren and \perfectnessbasenospace.
    Since $\srankof{\anngvasp}=\srankof{\anngvas_{\awtsymbol}}$, and $\srankof{\anngvas}=\srankof{\otherctxNGVAS{\anngvas}{\amarking, \anonterm, \amarkingp}}$ by the rank definition, so by \Cref{Lemma:NonLinearChildRank} we get $\srankof{\anngvas_{\awtsymbol}}+\cardof{\anngvas.\nonterms}\cdot 1_{\dimofN}\leq\srankof{\otherctxNGVAS{\anngvas}{\amarking, \anonterm, \amarkingp}}$ and thus $\srankof{\anngvas_{\awtsymbol}}<_{\dimofN}\srankof{\otherctxNGVAS{\anngvas}{\amarking, \anonterm, \amarkingp}}$.
    We show (run).
    The set $\cgrunsof{\otherctxNGVAS{\anngvas}{\amarking, \anonterm, \amarkingp}, \awtsymbol}$ consists of $\awtsymbol.\symlabel$ runs $(\amarkingpp, \arun, \amarkingppp)$ with $\amarkingpp\sqsubseteq\awtsymbol.\wtinlabel, \awtsymbol.\wtprepromise$ and $\amarkingpp\sqsubseteq\awtsymbol.\wtoutlabel, \awtsymbol.\wtpostpromise$.
    This is exactly the membership condition to $\runsof{\anngvas_{\awtsymbol}}$.
    
    We move on to the inductive case $\cardof{\nonterms_{\awtsymbol}^{\calltag}}>0$, which implies $\awtsymbol=\awtnonterm\in\wtnonterms$.
    First, let $\awtnonterm$ not be productive.
    Then, $\adecomp_{\awtnonterm}=\emptyset$, and since $\wtgrammar$-annotations are complete derivation trees, $\runsof{\otherctxNGVAS{\anngvas}{\amarking, \anonterm, \amarkingp}, \awtnonterm}=\emptyset$.
    This shows (run).
    The conditions (ngvas) and (rk) are trivially fulfilled, so we are done.
    Conversely, let $\awtnonterm\in\wtnonterms$ be productive. 
    Same as in the construction, we distinguish between the cases $\awtnonterm\in\wtnontermspump$ and $\awtnonterm\in\wtnontermssim$.
    First consider $\awtnonterm\in\wtnontermspump$.
    The only non-terminal is $\awtnonterm$, and the only rules are exit rules $\awtnonterm\to\anngvas_{\awtnontermp}$ where $(\awtnonterm\to\awtnontermp)\in\wtprodspump$.
    We show (ngvas).
    Since the only rule is an exit rule, we need to verify that the boundedness information is correct.
    That is, we need to show $\wtinof{\awtnonterm}, \anngvas_{\awtnontermp}.\acontextin\sqsubseteq \infun_{\awtnonterm}(\awtnonterm)=\awtnontermp.\wtinlabel$ and $\wtoutof{\awtnonterm}, \anngvas_{\awtnontermp}.\acontextout\sqsubseteq \outfun_{\awtnonterm}(\awtnonterm)=\awtnontermp.\wtoutlabel$.
    We only argue the former.
    By the construction of the coverability grammar, it can be easily verified that $\wtinof{\awtnonterm}, \wtinof{\awtnontermp}\sqsubseteq \awtnontermp.\wtinlabel$ and $\wtoutof{\awtnonterm}, \wtoutof{\awtnontermp}\sqsubseteq\awtnontermp.\wtoutlabel$.
    We know that $\awtnontermp$ can call less non-terminals than $\awtnonterm$, so the induction hypothesis applies.
    Using the induction hypothesis, we assume that $\anngvas_{\awtnonterm}$ has the properties stated in \Cref{Lemma:DecompGivenCGLowLevel}.
    Namely, $\anngvas_{\awtnontermp}.\acontextin=\awtnonterm.\wtinlabel$ and $\anngvas_{\awtnontermp}=\awtnonterm.\wtoutlabel$.
    This concludes the case (ngvas).
    To see (rk), first consider that $\srankof{\anngvas_{\awtnontermp}}<_{\dimofN}\srankof{\anngvas}$ for all $\awtnontermp\in\wtnonterms$ with $(\awtnonterm\to\awtnontermp)\in\wtprods$ by the induction hypothesis.
    Now, since there are no cycles in $\anngvas_{\awtnonterm}$, it is trivially linear, so $\lrankof{\anngvas_{\awtnonterm}}=0$ and $\srankof{\anngvas_{\awtnonterm}}=\srankof{\anngvas_{\awtnontermp}}<_{\dimofN}\srankof{\anngvas}$
    Now we argue (run).
    We have that $\runsof{\anngvas_{\awtnonterm}}$ equals
    \begin{align*}
        &\bigcup_{(\awtnonterm\to\awtnontermp)\in\wtprods}\runsof{\anngvas_{\awtnontermp}}\;\cap\\
        &\hspace{6em}\setcond{(\amarking, \arun, \amarkingp)\in\therunset}{\\
        &\hspace{9em}\amarking\sqsubseteq\wtinof{\awtnonterm},\;\amarkingp\sqsubseteq\wtoutof{\awtnonterm}}.
    \end{align*}
    We show that $\runsof{\otherctxNGVAS{\anngvas}{\amarking, \anonterm, \amarkingp}, \awtnonterm}$ equals
    \begin{align*}
        &\bigcup_{(\awtnonterm\to\awtnontermp)\in\wtprods}\runsof{\otherctxNGVAS{\anngvas}{\amarking, \anonterm, \amarkingp}, \awtnontermp}\cap\\
        &\hspace{1em}\setcond{(\amarking, \arun, \amarkingp)\in\therunset}{\amarking\sqsubseteq\wtinof{\awtnonterm},\;\amarkingp\sqsubseteq\wtoutof{\awtnonterm}},
    \end{align*}
    which yields (run) by the application of the induction hypothesis on $\awtnontermp$ with $\awtnonterm\to\awtnontermp\in\wtprods$.
    Since $\awtnonterm\in\wtnontermspump$, the (Pumping) condition of the $\wtgrammar$-annotation makes sure that all $\atree\in\cgannof{\awtnonterm}$ are also annotatable by $\wtgrammar$-derivations that have a root labeled $\awtnontermp$ for some $(\awtnonterm\to\awtnontermp)\in\wtprods$.
    Then, $\cgannof{\awtnonterm}\subseteq\bigcup_{(\awtnonterm\to\awtnontermp)\in\wtprods}\cgannof{\awtnontermp}$.
    A reachability tree $\atree\in\cgannof{\awtnontermp}$ is in $\cgannof{\awtnonterm}$, if $\atree.\inlabel\sqsubseteq\awtnonterm.\wtinlabel, \awtnonterm.\wtprepromise$ and $\atree.\outlabel\sqsubseteq\awtnonterm.\wtoutlabel, \awtnonterm.\wtpostpromise$.
    This is because we can extend the $\wtgrammar$-annotation with the root labeled $\awtnontermp$, by a root node labeled $\awtnonterm$ that also annotates $\atree.\arootnode$, and in this case (Generalization) has to hold.
    By $\wtinof{\awtnonterm}=\awtnonterm.\wtinlabel\sqcap\awtnonterm.\wtprepromise$ and $\wtoutof{\awtnonterm}=\awtnonterm.\wtoutlabel\sqcap\awtnonterm.\wtpostpromise$, we get
    \begin{align*}
        &\bigcup_{(\awtnonterm\to\awtnontermp)\in\wtprods}\setcond{\atree\in\cgannof{\awtnontermp}}{\\
        &\hspace{4em}\atree.\inlabel\sqsubseteq\wtinof{\awtnonterm},\;\atree.\outlabel\sqsubseteq\wtoutof{\awtnonterm}}.
    \end{align*}
    Since all $\atree\in\cgannof{\awtnonterm}$ must have $\atree.\inlabel\sqsubseteq\wtinof{\awtnonterm}$ and $\atree.\outlabel\sqsubseteq\wtoutof{\awtnonterm}$, we get that $\cgannof{\awtnonterm}$ equals
    \begin{align*}
        &\bigcup_{(\awtnonterm\to\awtnontermp)\in\wtprods}\setcond{\atree\in\cgannof{\awtnontermp}}{\\
        &\hspace{4em}\atree.\inlabel\sqsubseteq\wtinof{\awtnonterm},\;\atree.\outlabel\sqsubseteq\wtoutof{\awtnonterm}}.
    \end{align*}
    This implies the desired run-equality.

    Now assume $\awtnonterm\in\wtnontermssim$.
    To reduce clutter, we argue the following only once, and assume it throughout the proof.
    All $\awtsymbol\in\nonterms^{\calltag}_{\awtnonterm}\setminus\sccof{\awtsymbol}$ have less height than $\awtnonterm$, since they cannot call $\awtnonterm$, but $\awtnonterm$ can call $\awtsymbol$.
    This means that for all $\awtsymbol\in\nonterms^{\calltag}_{\awtnonterm}\setminus\sccof{\awtsymbol}$, the induction hypothesis applies to show that $\anngvas_{\awtsymbol}$ has the properties stated in \Cref{Lemma:DecompGivenCGLowLevel}.
    We proceed with the proof of (ngvas).
    By \Cref{Lemma:CGCallStructure}, it is clear that $\omegaof{\wtinof{\awtnonterm}}=\omegaof{\wtinof{\awtnontermp}}$ and $\omegaof{\wtoutof{\awtnonterm}}=\omegaof{\wtoutof{\awtnontermp}}$ for all $\awtnontermp\in\wtnonterms$, so $\abdinfoleft$ and $\abdinforight$ are soundly defined.
    The condition with the start symbol is also clear.
    Now consider a rule $\awtnontermp\to\aword\in\prods_{\awtnonterm}$ toward showing continuity.
    It must have $\aword=\realizationof{}{\awtsymbol.\awtsymbolp}$, where $(\awtnontermp\to\awtsymbol.\awtsymbolp)\in\wtprodssim$.
    By construction of the coverability grammar, we have $\awtsymbol.\inlabel=\awtnontermp.\inlabel$, $\awtsymbol.\wtpostpromise=\awtsymbolp.\wtinlabel$, $\awtsymbol.\wtoutlabel=\awtsymbolp.\wtprepromise$, and $\awtsymbolp.\wtoutlabel=\awtnontermp.\wtoutlabel$. 
    The construction also requires $\awtsymbol.\wtprepromise\sqsubseteq\awtnontermp.\wtprepromise$ and $\awtsymbolp.\wtpostpromise\sqsubseteq\awtnontermp.\wtpostpromise$.
    Then, $\wtinof{\awtsymbol}=\awtsymbol.\wtprepromise\sqcap\awtsymbol.\wtinlabel\sqsubseteq\wtinof{\awtnontermp}$, $\wtoutof{\awtsymbol}=\wtinof{\awtsymbolp}$, and similarly to the left-side, $\wtoutof{\awtsymbolp}\sqsubseteq\wtoutof{\awtnontermp}$.
    If $\awtsymbol$ or $\awtsymbolp$ is not a non-terminal, then by the induction hypothesis, their images under $\wtinfun$ resp. $\wtoutfun$ is equal to the context information of $\anngvas_{\awtsymbol}$ resp. $\anngvas_{\awtsymbolp}$.
    
    Now, we show loop-consistency.
    Let $\anngvas_{\awtnonterm}$ be branching.
    Then, we show that $\abdinfoleft=\abdinforight=\omegaof{\wtinof{\awtnonterm}}=\omegaof{\wtoutof{\awtnonterm}}=\omegaof{\anngvasp.\acontextin}=\omegaof{\anngvasp.\acontextout}=\anngvasp.\unconstrained$ for all $\anngvasp\in\wtterms$.
    First, we show that $\omegaof{\awtnontermp.\wtinlabel}=\omegaof{\awtnontermp.\wtpostpromise}$ and $\omegaof{\awtnontermp.\wtoutlabel}=\omegaof{\awtnontermp.\wtprepromise}$ for all $\awtnontermp\in\sccof{\awtnonterm}$.
    Since $\anngvas_{\awtnonterm}$ is branching, there must be a rule $(\awtnontermp\to\awtnontermp^{0}.\awtnontermp^{1})$ for some $\awtnontermp, \awtnontermp^{0}, \awtnontermp^{1}\in\sccof{\awtnonterm}$.
    By \Cref{Lemma:CGCallStructure}, $\omegaof{\awtnontermp^{0}.\wtinlabel}=\omegaof{\awtnontermp^{0}.\wtpostpromise}$, since $\awtnontermp^{0}$ is generated on the left of $\awtnontermp^{1}$, and $\omegaof{\awtnontermp^{1}.\wtoutlabel}=\omegaof{\awtnontermp^{1}.\wtprepromise}$, since $\awtnontermp^{1}$ is generated on the right of $\awtnontermp^{0}$.
    However, sincenon-terminals in a strongly connected component agree on all four components $\wtinlabel$, $\wtoutlabel$, $\wtpostpromise$, $\wtprepromise$ by \Cref{Lemma:CGCallStructure}, we get the desired equality.
    Now, consider a rule $\awtnontermp\to\aword.\anngvas_{\awtsymbol}.\awordp$ in $\prods_{\awtnonterm}$ that generates a terminal $\anngvas_{\awtsymbol}$.
    Then, there is a rule $\awtnontermp\to\aword'.\awtsymbol.\awordp'$ in $\wtprodssim$.
    By \Cref{Lemma:CGCallMonotonicity}, we get $\omegaof{\awtnonterm.\wtinlabel}=\omegaof{\awtnontermp.\wtinlabel}\subseteq\omegaof{\awtsymbol.\wtinlabel}\subseteq\omegaof{\awtsymbol.\wtpostpromise}\subseteq\omegaof{\awtnontermp.\wtpostpromise}$, and we have already shown $\omegaof{\awtnonterm.\wtpostpromise}=\omegaof{\awtnontermp.\wtpostpromise}=\omegaof{\awtnontermp.\wtinlabel}$, this yields the equality $\omegaof{\awtsymbol.\wtinlabel}=\omegaof{\awtsymbol.\wtpostpromise}=\omegaof{\awtnonterm.\wtinlabel}$.
    Applying this argument in the other direction, we get $\omegaof{\awtsymbol.\wtoutlabel}=\omegaof{\awtsymbol.\wtprepromise}=\omegaof{\awtnonterm.\wtoutlabel}$.
    Thus, we get $\omegaof{\wtinof{\awtsymbol}}=\omegaof{\wtoutof{\awtsymbol}}=\omegaof{\wtinof{\awtnonterm}}=\omegaof{\wtoutof{\awtnonterm}}$.
    
    Let $\anngvas_{\awtnonterm}$ not be branching, and $\awtnontermp\to\anngvas_{\awtsymbol}.\awtnontermpp$ a rule in $\prods_{\awtnonterm}$ where $\awtnontermp, \awtnontermpp\in\sccof{\awtnonterm}$.
    We ommit the case where the terminal is produced on the right side, since the proofs are similar.
    By the construction of $\anngvas_{\awtnonterm}$, there must be a rule $\awtnontermp\to\awtsymbol.\awtnontermpp$ in $\wtprodssim$.
    We apply \Cref{Lemma:CGCallStructure} and obtain $\omegaof{\awtsymbol.\wtinlabel}=\omegaof{\awtsymbol.\wtpostpromise}=\omegaof{\awtnontermp.\wtinlabel}$, and $\omegaof{\awtsymbol.\outlabel}=\omegaof{\awtsymbol.\wtprepromise}=\omegaof{\awtnontermp.\wtprepromise}$.
    Thus, $\omegaof{\wtinof{\awtsymbol}}=\omegaof{\wtoutof{\awtsymbol}}=\omegaof{\wtinof{\awtnontermp}}$, and we already know $\omegaof{\wtinof{\awtnontermp}}=\omegaof{\wtinof{\awtnonterm}}$.
    This concludes the proof of (ngvas).

    Now we argue that (rk) holds.
    Clearly, all cycles in $\anngvas_{\awtnonterm}$ can be imitated in $\otherctxNGVAS{\anngvas}{\amarking, \anonterm, \amarkingp}$, so we have $\lcyclespaceof{\anngvas_{\awtnonterm}}\subseteq\lcyclespaceof{\otherctxNGVAS{\anngvas}{\amarking, \anonterm, \amarkingp}}=\lcyclespaceof{\anngvas}$.
    We argue that $\lrankof{\anngvas_{\awtnonterm}}<_{\dimofN}\lrankof{\anngvas}$.
    Note that this implies (rk) by the following argument.
    The maximal branch of $\anngvas_{\awtnonterm}$ goes over $\anngvas_{\awtsymbol}$ for $\awtsymbol\in\nonterms_{\awtsymbol}^{\calltag}\setminus\sccof{\awtnonterm}$.
    Then, $\srankof{\anngvas_{\awtnonterm}}=\lrankof{\anngvas_{\awtnonterm}}+\srankof{\anngvas_{\awtsymbol}}$ for such a $\awtsymbol$.
    By the induction hypothesis we already have $\srankof{\anngvas_{\awtsymbol}}<_{\dimofN}\srankof{\anngvas}$, which along with $\lrankof{\anngvas_{\awtnonterm}}<_{\dimofN}\lrankof{\anngvas}$ concludes the proof.
    First, let $\anngvas_{\awtnonterm}$ be linear.
    Then, $\lrankof{\anngvas_{\awtnonterm}}=0$ and $\lrankof{\anngvas}[\dimofN]\neq 0$, so we are done.
    Let $\anngvas_{\awtnonterm}$ be non-linear.
    Because the $\wtgrammar$ shows boundedness, one new counter $i\not\in\abdinfomid$ is concretely tracked in $\anngvas_{\awtnonterm}.\abdinfo$.
    Similarly to the boundedness case of \cite{LerouxS19}, we may only have $\lcyclespaceof{\anngvas_{\awtnonterm}}=\lcyclespaceof{\anngvas}$ if $\anngvas_{\awtnonterm}$ concretely tracks a counter $i\in [1,d]$ in its boundedness information that is already fixed in $\anngvas$.
    Suppose this is the case.
    This implies that $\amarking[i]\in\N$ or $\amarkingp[i]\in\N$.
    Since $\omegaof{\acontextin}\subseteq\omegaof{\amarking}$ and $\omegaof{\acontextout}\subseteq\omegaof{\amarkingp}$, $\acontextin[i]\in\N$ or $\acontextout[i]\in\N$.
    Thus, $i$ is rigid in $\anngvas$.
    However, since $\anngvas$ has all the perfectness conditions except \perfectnesspumpingint and \perfectnesspumpingnospace, we know that all rigid counters are accounted for in the boundedness information of $\anngvas$.
    This is a contradiction.
    Then, $\lcyclespaceof{\anngvas_{\awtnonterm}}\subsetneq\lcyclespaceof{\anngvas}$ which implies $\dimensionof{\lcyclespaceof{\anngvas_{\awtnonterm}}}<\dimensionof{\lcyclespaceof{\anngvas}}$ which implies $\lrankof{\anngvas_{\awtnonterm}}<_{\dimofN}\lrankof{\anngvas}$.
    
    Now, we show (run).
    We argue the inclusion $\cgrunsof{\otherctxNGVAS{\anngvas}{\amarking, \anonterm, \amarkingp}, \awtnonterm}\subseteq\runsof{\anngvas_{\awtnonterm}}$.
    For any tree $\atree\in\cgannof{\awtnonterm}$, with the witnessing $\wtgrammar$-annotation $(\atreep, \covproj)$, we can construct an imitating reachability tree $\atree'\in\reachtreesof{\anngvas_{\awtnonterm}}$ where $\runsof{\atree}\subseteq\runsof{\atree'}$ by extending a root node labeled $(\covprojof{\atreep.\arootnode}.\inlabel, \awtnonterm, \covprojof{\atreep.\arootnode}.\outlabel)$.
    For each node $\anodep\in\atreep$ with $\nodemarkingof{\anodep}=\awtnontermp\in\sccof{\awtnonterm}$, we add $\anodep$ to $\atree'$, and label it $(\covprojof{\anodep}.\inlabel, \awtnontermp, \covprojof{\anodep}.\outlabel)$ if $\awtnontermp\in\sccof{\awtnonterm}$.
    For each node $\anodep\in\atreep$ with $\nodemarkingof{\anodep}=\awtsymbol\not\in\sccof{\awtnonterm}$, we know $\anngvas_{\awtsymbol}\in\trms_{\awtnonterm}$, and we add the leaf node $\anodep$ to $\atree'$ with the label $(\covprojof{\anodep}.\inlabel, \anngvas_{\awtsymbol}, \covprojof{\anodep}.\outlabel)$.
    Thanks to the induction hypothesis, we already have $\runsof{\atree_{\covprojof{\anodep}}}\subseteq\runsof{\covprojof{\anodep}.\inlabel, \anngvas_{\awtsymbol}, \covprojof{\anodep}.\outlabel}$ for the subtree $\atree_{\covprojof{\anodep}}$ of $\atree$ rooted at $\covprojof{\anodep}$. 
    It is easy to verify that $\runsof{\atree}\subseteq\runsof{\atree'}$.
    Finally, we argue the inclusion $\runsof{\anngvas_{\awtnonterm}}\subseteq\cgrunsof{\otherctxNGVAS{\anngvas}{\amarking, \anonterm, \amarkingp}, \awtnonterm}$.
    We argue that for any reachability tree $\atree\in\reachtreesof{\anngvas_{\awtnonterm}}$, we can construct $\atree'\in\cgannof{\otherctxNGVAS{\anngvas}{\amarking, \anonterm, \amarkingp}, \atree.\symlabel}$ with the witnessing annotation $(\atreep, \covproj)$, and $\atree.\inlabel=\atree'.\inlabel$, $\atree.\outlabel=\atree'.\outlabel$.
    We proceed by an induction on the height of $\atree$.
    For the base case, we have $\atree.\symlabel\in\trms_{\awtnonterm}$, where $\atree$ has height $0$.
    Let $\atree.\symlabel=\anngvas_{\awtsymbol}$.
    This implies $\atree.\inlabel\sqsubseteq\wtinof{\awtsymbol}$ and $\atree.\outlabel\sqsubseteq\wtoutof{\awtsymbol}$ by the construction of $\anngvas_{\awtsymbol}$.
    Then, letting $\atree'$ consist of a root node labeled $(\atree.\inlabel, \awtsymbol.\symlabel, \atree.\outlabel)$, and $\atreep$ of a root node labeled $\awtsymbol$ with $\covprojof{\atreep.\arootnode}=\atree'.\arootnode$, we get the annotation we wanted.
    Now consider the inductive case, where we have $\atree.\symlabel\in\sccof{\awtnonterm}$.
    For the left- and right-subtrees $\atree_{\lefttag}$ and $\atree_{\righttag}$, we have $(\atree.\symlabel\to(\atree_{\lefttag}.\symlabel).(\atree_{\righttag}.\symlabel))\in\prods_{\awtnonterm}$.
    We have $\wtinof{\atree.\symlabel}\sqsubseteq\wtinof{\atree_{\lefttag}.\symlabel}$, $\wtoutof{\atree_{\lefttag}.\symlabel}=\wtinof{\atree_{\righttag}.\symlabel}$, and $\wtoutof{\atree_{\righttag}.\symlabel}\sqsubseteq\wtoutof{\atree.\symlabel}$ since $\anngvas_{\awtnonterm}$ is a weak NGVAS.
    We apply the induction hypothesis, and we construct the reachability trees $\atree_{\lefttag}'\in\cgannof{\atree_{\lefttag}.\symlabel}$, $\atree_{\righttag}'\in\cgannof{\atree_{\righttag}.\symlabel}$ for $\atree_{\lefttag}$ and $\atree_{\righttag}$, with the properties we stated, and the witnessing annotations $(\atreep_{\lefttag}, \covproj_{\lefttag})$ and $(\atreep_{\righttag}, \covproj_{\righttag})$.
    We construct $\atree'$ such that $\atree'_{\lefttag}$ and $\atree'_{\righttag}$ are its left- and right-subtrees, and we have $\nodemarkingof{\atree'.\arootnode}=(\atree.\inlabel, \asymbol, \atree.\outlabel)$, where $\awtsymbol=(\atree.\symlabel).\symlabel$ if $\atree.\symlabel\in\sccof{\awtnonterm}$, and $\asymbol=\awtsymbol.\symlabel$ if $\atree.\symlabel=\anngvas_{\awtsymbol}\in\trms_{\awtnonterm}$.
    We also construct $\atreep$ that has $\atreep_{\lefttag}$ and $\atreep_{\righttag}$ as its left- and right-subtrees with $\nodemarkingof{\atreep.\arootnode}=\atree.\symlabel$.
    Further let $\covproj$ map the left- and right-subtrees of $\atree'$ to those of $\atreep$, and $\atree'.\arootnode$ to $\atreep.\arootnode$.
    Clearly, $\atree'$ is a reachability tree with $\runsof{\atree}\subseteq\runsof{\atree'}$, $\atree.\inlabel=\atree'.\inlabel$, $\atree.\outlabel=\atree'.\outlabel$, and $(\atreep, \covproj)$ witnesses $\atree'\in\cgannof{\atree.\symlabel}$.
    The latter statement holds since (Generalization) and (Local Bijection) hold for the root, and $\atree.\symlabel\in\wtnontermssim$, makes (Pumping) hold trivially.
\end{proof}

\paragraph*{The decomposition $\refinepumpof{\anngvas}$.}

The call $\refinepumpof{\anngvas}$ decomposes $\anngvas$ by relying on the insights we gained in \Cref{Lemma:CGTermination}, \Cref{Lemma:CGPumpability}, and \Cref{Lemma:DecompGivenCG}.
First, it constructs $\wtgrammarof{\anngvas, \natpostapprox, \natpreapprox}$.
The soundness of this step relies on the following lemma.
\LemmaPostPreComputableMainPaper*
%
%
Using \Cref{Lemma:CGTermination}, we conclude effectiveness under the assumption of  computability of approximators.
Said computability is provided by \Cref{Lemma:PostPreComputable}.
By \Cref{Lemma:CGPumpability}, we know that if there is a non-terminal $((\amarking, \aprommarking), \anonterm, (\aprommarkingp, \amarkingp))$ in $\wtgrammarof{\anngvas, \natpostapprox, \natpreapprox}$ with $\omegaof{\amarking}=\omegaof{\amarkingp}=\abdinfoleft$, then \perfectnesspumping holds.
In this case, the call returns $\refinepumpof{\anngvas}=\anngvas$.
If this is not the case, then \Cref{Lemma:DecompGivenCG} applies.
The call returns the decomposition provided by this construction, i.e. $\refinepumpof{\anngvas}=\setcond{\anngvas_{\awtnonterm}}{\wtstartnonterm\to\awtnonterm}$.
It is clear by this line of argumentation that, given \Cref{Lemma:CGTermination}, \Cref{Lemma:CGPumpability}, \Cref{Lemma:DecompGivenCG}, and \Cref{Lemma:PostPreComputable}; the specification for $\refinepump$ is fulfilled.
Since the remaining lemmas are already proven, it remains to prove \Cref{Lemma:PostPreComputable}.
This will be the focus of the rest of the section.
\subsection{Computing $\postfunc$ and $\prefunc$, Simple Cases}\label{Section:CoverabilityEasyCasesL3}
Our goal in the rest of the section is to show \Cref{Lemma:PostPreComputable}.
We organize the full domain into parts, and use different arguments and algorithms for computing $\natpostapprox$ and $\natpreapprox$ in each part.
We adopt the notation from the extended paper for the domains, e.g. $\easydomain$, but also define the following.
\begin{align*}
    \afuncdomain_{\acounterset}&=\setcond{(\amarking, \asymbol)\in\Nomega^{d}\times(\nonterms\cup\trms)}{\omegaof{\amarking}=\acounterset}\\
    \afuncdomain_{\trms}&=\inmarkingdomain\times\trms\\
    \afuncdomain_{\nonterms, \acounterset}&=\setcond{(\amarking, \asymbol)\in\Nomega^{d}\times\nonterms}{\omegaof{\amarking}=\acounterset}
\end{align*}
Note that $\trms=\rectrms$ since $\anngvas$ is non-linear.

We proceed by weeding out the simple cases.
We have already seen that we can compute $\postfuncN{\anngvas}$ and $\prefuncN{\anngvas}$ for the domains $\afuncdomain_{\trms}$ and $\afuncdomain_{\acounterset}$ for large enough $\acounterset$, by encoding the query as an NGVAS and relying on $\perfect$.
\LemmaTermsPostPreComputableMainPaper*
\LemmaLowDimPostPreComputableMainPaper*
We make an observation for the different-context variations $\otherctxNGVAS{\anngvas}{\amarking, \anonterm, \amarkingp}$ of $\anngvas$ that will also be useful later in the paper.
Namely, that the resulting NGVAS is of lesser-or-equal to rank compared to $\anngvas$.
This can be easily verified by considering the rank definition.
It also retains the properties relating to the children.
\begin{lemma}\label{Lemma:OtherCTXPropertiesI}
    Let $\otherctxNGVAS{\anngvas}{\amarking, \anonterm, \amarkingp}$ be a variant with $(\amarking, \amarkingp)\in\omegamrkdomainof{\unconstrained}\times\omegamrkdomainof{\unconstrained}$.
    Then $\rankof{\otherctxNGVAS{\anngvas}{\amarking, \anonterm, \amarkingp}}\leq\rankof{\anngvas}$.
    Furthermore, if $\anngvas$ has \perfectnesschildren and \perfectnessbasenospace, then so does $\otherctxNGVAS{\anngvas}{\amarking, \anonterm, \amarkingp}$.
\end{lemma}

Now, we move on the the proof of \Cref{Lemma:LowDimPostPreComputable}.

\begin{proof}[Proof of \Cref{Lemma:LowDimPostPreComputable}]
    Let $(\amarking, \anonterm)\in\afuncdomain_{\acounterset}$ for $\unconstrained\subsetneq\acounterset\subseteq\abdinfomid$.
    Since the case $(\amarking, \anonterm)\in\afuncdomain_{\trms}$ is already covered by \Cref{Lemma:TermsPostPreComputable}, let $(\amarking, \anonterm)\in\afuncdomain_{\acounterset}\setminus\afuncdomain_{\trms}=\afuncdomain_{\nonterms, \acounterset}$.
    Similarly to the previous case, we check if $\amarking\compwith\inof{\anonterm}$, if not we have $\amarking\not\compwith\anngvasp.\acontextin$ and for the first terminal in any word of $\anonterm$, and therefore $\postfuncNof{\anngvas}{\amarking, \anonterm}=\emptyset$ holds.
    Because $\omegaof{\amarking}\subseteq\abdinfomid$, it must be the case that $\amarking\sqsubseteq\inof{\anonterm}$.
    We construct the NGVAS $\otherctxNGVAS{\anngvas}{\amarking, \anonterm, \outof{\anonterm}}$.
    Note that $\otherctxNGVAS{\anngvas}{\amarking, \anonterm, \outof{\anonterm}}.\unconstrained=\omegaof{\amarking}\cap\omegaof{\outof{\anonterm}}=\omegaof{\amarking}.$ 
    Because $\unconstrained\subsetneq\omegaof{\amarking}$, we know that $\cardof{\unconstrained}<\cardof{\omegaof{\amarking}}=\otherctxNGVAS{\anngvas}{\amarking, \anonterm, \outof{\anonterm}}.\unconstrained$.
    Then, $\recrankof{\otherctxNGVAS{\anngvas}{\amarking, \anonterm, \outof{\anonterm}}}<\recrankof{\anngvas}$ clearly holds, by the most significant component of the rank.
    Then, $\perfect$ is reliable up to and for $\rankof{\otherctxNGVAS{\anngvas}{\amarking, \anonterm, \outof{\anonterm}}}$.
    By \Cref{Lemma:OtherCTXPropertiesI}, we know that $\perfectof{\cleanof{\otherctxNGVAS{\anngvas}{\amarking, \anonterm, \outof{\anonterm}}}}$ is a perfect deconstruction of $\otherctxNGVAS{\anngvas}{\amarking, \anonterm, \outof{\anonterm}}$.
    This call is allowed by our programming model.
    The set $\setcond{\anngvas'.\acontextout}{\anngvas'\in\perfectof{\centerdecof{\otherctxNGVAS{\anngvas}{\amarking, \anonterm, \outof{\anonterm}}}}}$ captures the maximal output values.
\end{proof}

It remains to show that $\postfuncN{\anngvas}$ is computable for $\amarking\in\Nomega^{d}$ with $\omegaof{\amarking}=\unconstrained$.
%
%
%
\begin{lemma}\label{Lemma:CurrDimPostComputable}
    Let $\anngvas$ be a non-linear NGVAS with all the perfectness properties excluding \perfectnesspumpingnospace, and let $\perfect$ be reliable up to $\rankof{\anngvas}$.
    Further let $\omegaof{\acontextin}=\unconstrained$.
    We can compute $\postfuncN{\anngvas}$ restricted to the domain $\afuncdomain_{\unconstrained}$.
\end{lemma}

\subsubsection*{Forgetting Counters To Get Decomposition}

We show that we can compute a perfect decomposition for all simply decomposable triples.
\LemmaSimplePerfectDecompositionMainPaper*
\begin{proof}[Proof Sketch]
    Let $(\amarking, \anonterm, \amarkingp)\in\simplydecomps$.
    By the definition of $\simplydecomps$, we know that we can find $i, j\in\abdinfomid\setminus\unconstrained$ such that $\wtgrammarof{\otherctxNGVAS{\anngvas}{\amarking, \anonterm, \amarkingp}, \forgetfulpostapprox{i}, \forgetfulpreapprox{j}}$ does not have a non-terminal $((\amarking', \aprommarking), \anonterm, (\aprommarkingp, \amarkingp'))$ with $\omegaof{\amarking'}=\omegaof{\amarkingp'}=\abdinfomid$.
    Then, by \Cref{Corollary:ForgetfulApproxComputability} and \Cref{Lemma:DecompGivenCG}, we know that we can compute a decomposition $\adecomp'$ of $\otherctxNGVAS{\anngvas}{\amarking, \anonterm, \amarkingp}$ with lesser non-linear rank.
    By the rank-definition, we have $\rankof{\otherctxNGVAS{\anngvas}{\amarking, \anonterm, \amarkingp}}\leq\rankof{\anngvas}$.
    Then, we can call $\perfectof{\cleanof{\adecomp'}}$ within our programming model and get a perfect decomposition.
    This concludes the proof.
\end{proof}
The argument for the converse case is similar to the arguments for \Cref{Lemma:CGPumping}, and \Cref{Lemma:ZApproxPumpDecomp}.
\LemmaSimplyUndecomposableMainPaper*

\subsubsection*{Claims related to $\Z$-pumps}

We show the assumption $\omegaof{\acontextout}\subseteq\omegaoutcount$.
\begin{lemma}\label{Lemma:RelaxedSupport}
    Let $\anngvas$ have all perfectness conditions excluding \perfectnesspumpingnospace.
    Then, $\omegaof{\acontextout}\subseteq\omegaoutcount$.
\end{lemma}
\begin{proof}
    Suppose there is a $i\in\omegaof{\acontextout}\setminus\omegaoutcount$.
    Since $\anngvas$ has \perfectnesscountersnospace, we know that $\outvar[i]$ is in the support of $\anngvas.\homchareq{}$.
    Note that $\omegaof{\outof{\anonterm}}=\omegaof{\outof{\startnonterm}}$, and $\acontextout\sqsubseteq\outof{\startnonterm}$.
    Thus $\outvar[j]=0$ is required in $\anngvas.\homchareq{}$ for all $j\in[1,d]\setminus\abdinfoleft$.
    The constraints that do not require $\outvar[k]=a$ for some $k\in[1,d]$ and $a\in\N$ are identical between $\anngvas.\homchareq{}$ and $\otherctxNGVAS{\anngvas}{\acontextin, \anonterm, \outof{\anonterm}}.\homchareq{}$.
    Furthermore the latter only has the constraints $\outvar[j]=0$ for $j\in[1, d]\setminus\abdinfomid$. 
    Thus, any solution to $\anngvas.\homchareq{}$ is also a solution to $\otherctxNGVAS{\anngvas}{\acontextin, \anonterm, \outof{\anonterm}}.\homchareq{}$.
    This implies that $\outvar[i]$ is in the support of $\otherctxNGVAS{\anngvas}{\acontextin, \anonterm, \outof{\anonterm}}.\homchareq{}$ and thus $i\in\omegaoutcount$.
\end{proof}

Since $\otherctxNGVAS{\anngvas}{\acontextin, \startnonterm, \outof{\startnonterm}}$ and $\otherctxNGVAS{\anngvas}{\amarking, \anonterm, \outof{\anonterm}}$ have the same homogeneous systems if $(\amarking, \anonterm)\in\Nomega^{d}\times\nonterms$, $\amarking\sqsubseteq\inof{\anonterm}$ and $\omegaof{\amarking}=\omegaof{\acontextin}$, we get the following corollary. 
\begin{corollary}\label{Lemma:SameHomEq}
    Let $\anngvas$ have all perfectness conditions excluding \perfectnesspumping and let $\amarking\in\Nomega^{d}$ and $\anonterm\in\nonterms$ with $\omegaof{\amarking}=\omegaof{\acontextin}$ and $\amarking\sqsubseteq\inof{\anonterm}$.
    Then, $\otherctxNGVAS{\anngvas}{\acontextin, \startnonterm, \outof{\startnonterm}}.\homchareq{}$ and $\otherctxNGVAS{\anngvas}{\amarking, \anonterm, \outof{\anonterm}}.\homchareq{}$ are the same homogeneous system.
    In particular, 
    \begin{align*}
        \omegaoutcount&=\setcond{i\in[1,d]}{\\
        &\hspace{2em}\outvar[i]\in\suppof{\otherctxNGVAS{\anngvas}{\amarking, \anonterm, \outof{\anonterm}}.\homchareq{}}}.
    \end{align*}
\end{corollary}

Thanks to \Cref{Lemma:RelaxedSupport}, we also observe that the images of $\intpostapprox$ and $\intpreapprox$ correspond to the outputs of actual words in the grammar.
\begin{lemma}\label{Lemma:SupportAtTheEndOfTheTunnel}
    Let $\anngvas$ be an NGVAS with all perfectness conditions except \perfectnesspumpingnospace.
    Let $\amarking\in\Nomega^{d}$ with $\omegaof{\acontextin}\subseteq\omegaof{\amarking}\subseteq\abdinfoleft$, and $\anonterm\in\nonterms$.
    Then, for any $\amarking'\in\intpostapproxof{\amarking, \anonterm}$, there are the sequences $[\aword_{i}]_{i\in\N}\in\langof{\anonterm}^{\omega}$ and $[\amarking_{i}']_{i\in\N}\in\Nomega^{d}$ with $\amarking_{i}'\in\amarking+\ceffof{\aword_{i}}$, where $\lim_{i\in\N}\amarking_{i}'=\amarking'$.
    Now let $\amarkingp\in\Nomega^{d}$ with $\omegaoutcount\subseteq\omegaof{\amarkingp}\subseteq\abdinfoleft$.
    Then, for any $\amarkingp'\in\intpreapproxof{\amarkingp, \anonterm}$, there are the sequences $[\awordp_{i}]_{i\in\N}\in\langof{\anonterm}^{\omega}$ and $[\amarkingp_{i}']_{i\in\N}\in\Nomega^{d}$ with $\amarkingp_{i}'\in\amarkingp-\ceffof{\aword_{i}}$, where $\lim_{i\in\N}\amarkingp_{i}'=\amarkingp$.
\end{lemma}
\begin{proof}
    Let the NGVAS $\anngvas$, the markings $\amarking, \amarking', \amarkingp, \amarkingp'\in\Nomega^{d}$, and non-terminal $\anonterm\in\nonterms$ as given in the lemma.
    By the definition of $\intpostapprox$ resp. $\intpreapprox$, we know that there are sequences of solutions $[\asol_{i}]_{i\in\N}$ to $\otherctxNGVAS{\anngvas}{\amarking, \anonterm, \outof{\anonterm}}.\chareq{}$ and $[\asol_{i}']_{i\in\N}$ to $\otherctxNGVAS{\anngvas}{\inof{\anonterm}, \anonterm, \amarkingp}.\chareq{}$ such that $\lim_{i\in\N} \asol_{i}[\outvar]=\amarking'$ and $\lim_{i\in\N}\asol_{i}'[\outvar]=\amarkingp'$.
    Since $\perfectnessprods$ holds for $\anngvas$, we get a solution $\ahomsol$ to $\homchareq{}_{\anngvas}$ that takes every production rule at least once.
    Because $\omegaof{\acontextin}\subseteq\omegaof{\amarking}$, the homogeneous system for $\otherctxNGVAS{\anngvas}{\amarking, \anonterm, \outof{\anonterm}}$ is less-restrictive, and therefore $\ahomsol$ is also a solution to it.
    Similarly, $\omegaoutcount\subseteq\omegaof{\amarkingp}$ holds, and we get $\omegaof{\acontextout}\subseteq\omegaoutcount$ by \Cref{Lemma:RelaxedSupport}, so $\ahomsol$ is also a solution to the homogeneous system of $\otherctxNGVAS{\anngvas}{\inof{\anonterm}, \anonterm, \amarkingp}$.
    Then, the sequences $[\asol_{i}+\ahomsol]_{i\in\N}$ and $[\asol_{i}'+\ahomsol]_{i\in\N}$ are also sequences of solutions to their respective equation systems.
    Therefore by \Cref{Theorem:EEK}, each solution can be realized as a derivation with the same effect.
    This concludes the proof. 
\end{proof}
We observe that the $\Z$-approximation can yield one of two results.
Either, we get a decomposition, or we get a special cycle called a $\Z$-pump, which, if we ignore the positivity constraints, pumps $\abdinfomid\setminus\unconstrained$ counters on the input, and $\abdinfomid\setminus\omegaoutcount$ on the output.
We show \Cref{Lemma:ZPumpImpliesUnbounded}.
\LemmaZPumpImpliesUnboundedMainPaper*
\begin{proof}[Proof Sketch]
    Let there be such a non-terminal $\awtnontermp$ with $\omegaof{\awtnontermp.\inlabel}=\omegaof{\awtnontermp.\outlabel}=\abdinfomid$.
    To get a $\Z$-pump in $\anngvas$, we first derive $\anonterm$ from $\startnonterm$, which is possible by strong connectedness, repeat the $\Z$-pump enough times to compensate the effect of this and the next derivation, and then derive $\startnonterm$ from $\anonterm$.
    
    We argue the existence of the $\Z$-pump.
    Since $\awtnontermp$ is a non-terminal in the coverability grammar, there is a derivation that leads from the start non-terminal $\wtstartnonterm$ to $\awtnonterm$. 
    Note that $\wtstartnonterm.\symlabel=\anonterm$.
    By the construction of the coverability grammar, whenever a rule moves from $\awtnontermpp$ to $\awtnontermpp'$, we can derive $\awtnontermpp\to\aword.\awtnontermpp'.\awordp$ where $\awtnontermpp'.\wtinlabel\in\intpostapproxof{\awtnontermpp.\wtinlabel, \aword}$, and $\awtnontermpp'.\wtoutlabel\in\intpreapproxof{\awtnontermpp.\wtoutlabel, \awordp}$.
    By \Cref{Lemma:SupportAtTheEndOfTheTunnel}, the effects assigned to $\aword$ and $\awordp$ by the $\Z$-approximations can be realized as the effects of actual derivations.
    By \Cref{Lemma:CGPumping} and \Cref{Lemma:SupportAtTheEndOfTheTunnel}, we get derivations whose effects pump the counters that become $\omega$ along $\wtprodspump$ derivations.
    We get a $\Z$-pump by combining the derivations with a sufficient amount of pumps, followed by deriving $\anonterm$ back from $\awtnonterm.\symlabel$.
\end{proof}
\subsection{Hard Case 1: Computing $\postfuncN{\anngvas}$ for large inputs}\label{Section:Rackoffesque}
In this section, we include the missing proofs for \Cref{Lemma:FreePumpRackoff}.

We argue the following lemma from the extended paper.
Note that we implicitly assume $\omegaof{\acontextin}=\unconstrained$.
This assumption also carries over to the rest of the appendix.
\LemmaFreePumpNoSDBothSidesRackoffMainPaper*
We assume the following two lemmas, which we show in \Cref{Section:Precalculation}.
\LemmaPrecalculationIMainPaper*
\LemmaPrecalculationIIMainPaper*
\begin{proof}[Proof of \Cref{Lemma:FreePumpNoSDBothSidesRackoff}]
    We argue that \Cref{Lemma:SimpleBoundDetailed}, \Cref{Lemma:PrecalculationI}, \Cref{Lemma:PrecalculationII}, \Cref{Lemma:SimplyUndecomposable} show \Cref{Lemma:FreePumpNoSDBothSidesRackoff}.
    The premise of \Cref{Lemma:FreePumpNoSDBothSidesRackoff} guarantees a large counter on both sides of the triple.
    Then, \Cref{Lemma:SimplyUndecomposable} guarantees the existence of a pumping derivation that ignores the large counters.
    By \Cref{Lemma:PrecalculationI} and \Cref{Lemma:PrecalculationII}, yield upper bounds on the sizes of the free-$\N$-pumps resp. pumping derivations.
    This is precisely what we need to apply \Cref{Lemma:SimpleBoundDetailed}, which concludes the proof.
\end{proof}
Now we extend this result to one side, and show \Cref{Lemma:FreePumpNoSDRackoff}.
\LemmaFreePumpNoSDRackoffMainPaper*
First, we show a result about maximal images in $\intpostapprox$.
\begin{lemma}\label{Lemma:MaximalPull}
    There is a function $\phi_{sol}$, computable with elementary resources, such that the following holds.
    Let $\anngvas$ be an NGVAS that has all perfectness conditions excluding \perfectnesspumpingnospace, $\abigconst\in\N$, $(\amarking, \anonterm)\in\afuncdomain_{\unconstrained, \anonterm}$ with $\amarking\sqsubseteq\inof{\anonterm}$ and $\amarking[i]\geq \phi_{sol}(\abigconst\cdot\cardof{\anngvas})$ for some $i\in\abdinfomid\setminus\unconstrained$.
    Then for any maximal $\amarkingp\in\intpostapproxof{\amarking, \anonterm}$, we have a $j\in\abdinfoleft\setminus\unconstrained$ such that $\amarkingp[j]\geq\abigconst$. 
\end{lemma}

\begin{proof}
    Standard ILP methods give us a function $\theta_{sol}:\N\to\N$, computable in elementary time, such that it that satisfies the following.
    For all NGVAS $\anngvasp$, the size of the largest base-vector or period-vector in the semi-linear set representation of the solution space of $\anngvasp.\chareq{}$ is less than $\theta_{sol}(\cardof{\anngvasp})$.
    Let $\phi_{sol}:\N\to\N$ with $\phi_{sol}(a)=\theta_{sol}(a)+a$.
    Let $\amarking, \amarkingp\in\Nomega^{d}$ and $\abigconst\in\N$ be as defined in the lemma.
    Similarly to the proof of \Cref{Lemma:FreePumpNoSDBothSidesRackoff}, we have $\amarking\sqsubseteq\inof{\anonterm}$ and $\amarkingp\sqsubseteq\outof{\anonterm}$.
    We have $\omegaof{\amarkingp}=\omegaoutcount$ by \Cref{Lemma:SameHomEq} and the definition of $\intpostapprox$.
    If $\omegaoutcount\neq\unconstrained$ and therefore $\unconstrained\subsetneq\omegaoutcount$, we can pick $j\in\concoutcount$ and get $\amarkingp[j]=\omega\geq \abigconst$.
    This shows \Cref{Lemma:MaximalPull}.
    Then let $\omegaoutcount=\unconstrained=\omegaof{\amarking}$.
    Suppose that for all $j\in\omegaoutcount$, $\amarkingp[j]<\abigconst$ holds.
    We show that in this case, there is a $\amarkingp'\in\intpostapproxof{\amarking, \anonterm}$ where $\amarkingp<\amarkingp'$, contradicting maximality.
    Consider the NGVAS $\otherctxNGVAS{\anngvas}{\inof{\anonterm}, \anonterm, \amarkingp}$.
    We can overapproximate the size of this NGVAS by $\cardof{\otherctxNGVAS{\anngvas}{\inof{\anonterm}, \anonterm, \amarkingp}}<\abigconst\cdot d\cdot\cardof{\anngvas}$.
    We observe that there is a solution $\asol$ to $\otherctxNGVAS{\anngvas}{\inof{\anonterm}, \anonterm, \amarkingp}.\chareq{}$, where $\asol[\invar]\sqsubseteq\amarking$.
    We write $\asol=\asol'+\ahomsol$ for some base vector $\asol'$ in the semi-linear set representation of the solution space of $\otherctxNGVAS{\anngvas}{\inof{\anonterm}, \anonterm, \amarkingp}.\chareq{}$, and a solution $\ahomsol$ of $\otherctxNGVAS{\anngvas}{\inof{\anonterm}, \anonterm, \amarkingp}.\homchareq{}$.
    Then, we have $\amarking=\asol'[\invar]+\ahomsol[\invar]$.
    We know $\asol'[\invar][i]<\theta_{sol}(\abigconst\cdot d\cdot\cardof{\anngvas})$ for all $i\leq d$ by the definition of $\theta_{sol}$.
    Because $\amarking[i]\geq \phi_{sol}(\abigconst\cdot d\cdot\cardof{\anngvas})=\theta_{sol}(\abigconst\cdot d\cdot\cardof{\anngvas})+\abigconst\cdot d\cdot\cardof{\anngvas}$ for some $i\in[1,d]\setminus\unconstrained$, we have $\ahomsol[\invar][i]\geq\abigconst\cdot d\cdot\cardof{\anngvas}\geq \abigconst$.
    Since the output counters in $\overline{\omegaoutcount}=\overline{\unconstrained}$ are all constrained in $\otherctxNGVAS{\anngvas}{\inof{\anonterm}, \anonterm, \amarkingp}.\chareq{}$, we know $\ahomsol[\invar][i]=0$ for all $i\in [1,d]\setminus\unconstrained$.
    Since $\asol'$ is also a solution to $\otherctxNGVAS{\anngvas}{\inof{\anonterm}, \anonterm, \outof{\anonterm}}$, and this system does not constrain the counters in $\abdinfomid$, adding $\ahomsol[\invar]+\sizeof{\ahomsol[\invar]}\cdot 1_{\constrained}$ to both input and output of $\asol'$ also yields a solution.
    The term $\sizeof{\ahomsol[\invar]}\cdot 1_{\constrained}$ ensures that the respective values remain positive for $\unconstrained$ counters.
    Let the solution we obtain this way be $\asol''$.
    We have $\asol''[\invar][k]=\asol[\invar][k]=\amarkingp$ for all $k\in [1,d]\setminus\unconstrained$.
    Then we also have $\settoomega{\unconstrained}{\asol''[\outvar]}\in\intpostapproxof{\amarking, \anonterm}$.
    Note that $\asol[\outvar][k]\leq \asol''[\outvar][k]$ for all $k\in [1,d]\setminus\unconstrained$, and $\asol''[\outvar][i]\geq \abigconst$, for $i\in [1,d]\setminus\unconstrained$.
    This implies $\settoomega{\unconstrained}{\asol''[\outvar]}>\amarkingp$, which is the contradiction we wanted.
\end{proof}

We show that \Cref{Lemma:FreePumpNoSDBothSidesRackoff} and \Cref{Lemma:MaximalPull} together imply \Cref{Lemma:FreePumpNoSDRackoff}.

\begin{proof}[Proof of \Cref{Lemma:FreePumpNoSDRackoff}]
    Let $\anngvas$ have the properties stated in the lemma.
    Let $\abigconst\in\N$ be the computed constant from \Cref{Lemma:FreePumpNoSDBothSidesRackoff}.
    Now consider $\abigbound=\phi_{sol}(\abigconst+\cardof{\anngvas})$ for $\phi_{sol}$ given in \Cref{Lemma:MaximalPull}.
    Let $(\amarking, \anonterm)\in\afuncdomain_{\unconstrained, \anonterm}$ with $\amarking\sqsubseteq\inof{\anonterm}$ and maximal $\amarkingp\in\intpostapproxof{\amarking, \anonterm}$.
    Further let $\amarking[j]\geq \abigbound$ and $(\amarking, \anonterm, \amarkingp)\not\in\simplydecomps$.
    Then, by \Cref{Lemma:MaximalPull}, there must be a $k\in [1,d]\setminus\unconstrained$ such that $\amarkingp[k]\geq \abigconst$.
    Thus, by \Cref{Lemma:FreePumpNoSDBothSidesRackoff}, we conclude that $\otherctxNGVAS{\anngvas}{\amarking, \anonterm, \amarkingp}$ has a pumping derivation.
    By \Cref{Lemma:TheZDecomp}, we know that $\otherctxNGVAS{\anngvas}{\amarking, \anonterm, \amarkingp}$ already had all perfectness properties excluding \perfectnesssol and \perfectnesspumpingnospace.
    This concludes the proof.
\end{proof}

Now, we show \Cref{Lemma:FreePumpRackoff}.

\begin{proof}[Proof of \Cref{Lemma:FreePumpRackoff}]
    Let $\anngvas$ be an NGVAS with a $\Z$-pump, and with all perfectness conditions excluding \perfectnesspumpingnospace.
    Let $\perfect$ be reliable up to $\rankof{\anngvas}$.
    We let $\abigbound\in\N$ as stated by \Cref{Lemma:FreePumpNoSDRackoff}.
    Let 
    \begin{align*}
        (\amarking, \anonterm)&\in\setcond{(\amarking, \anonterm)\in\Nomega^{d}\times\nonterms}{\\
        &\hspace{4em}\unconstrained=\omegaof{\amarking},\; \exists i\in\unconstrained.\; \amarking[i]\geq \abigbound}.
    \end{align*} 
    If $\amarking\sqsubseteq\inof{\anonterm}$ does not hold, then we know that $\postfuncNof{\anngvas}{\amarking, \anonterm}=\emptyset$, and we are done.
    Suppose this is not the case.
    By \Cref{Lemma:TheZDecomp}, we have $\runsof{\otherctxNGVAS{\anngvas}{\amarking, \anonterm, \outof{\anonterm}}}=\bigcup_{\amarkingp\in\intpostapproxof{\amarking, \anonterm}}\runsof{\otherctxNGVAS{\anngvas}{\amarking, \anonterm, \amarkingp}}$.
    We iterate over all $\amarkingp\in\intpostapproxof{\amarking, \anonterm}$, starting from the maximal elements, and construct the output set $\mathsf{Out}_{\N}\subseteq\Nomega^{d}$.
    At each step, we argue that we already cover all possible output values in $\runsof{\otherctxNGVAS{\anngvas}{\amarking, \anonterm, \amarkingp}}$.
    We also make sure that $\mathsf{Out}_{\N}$ remains sound, in that it only contains values witnessable by sequences of runs.
    For $\amarkingp\in\intpostapproxof{\amarking, \anonterm}$, we first check whether $(\amarking, \anonterm, \amarkingp)\in\simplydecomps$ holds.
    If yes, then we compute a perfect decomposition $\adecomp_{\amarkingp}$ of $\otherctxNGVAS{\anngvas}{\amarking, \anonterm, \amarkingp}$ by \Cref{Lemma:SimplePerfectDecomposition}.
    We add $\anngvas'.\acontextout$ to $\mathsf{Out}_{\N}$ for each $\anngvas'\in\adecomp_{\amarkingp}$.
    If no, then we check whether $\amarkingp$ is maximal.
    If it is, then \Cref{Lemma:FreePumpNoSDRackoff} applies, and we get that $\otherctxNGVAS{\anngvas}{\amarking, \anonterm, \amarkingp}$ is perfect up to the condition \perfectnesssolnospace.
    The condition \perfectnesssolnospace can be established by restricting (copies of) $\otherctxNGVAS{\anngvas}{\amarking, \anonterm, \amarkingp}$ to the linear sets from the solution space of $\otherctxNGVAS{\anngvas}{\amarking, \anonterm, \amarkingp}.\chareq{}$.
    Since the solution space of $\homchareq{}$ is uneffected by this procedure, all other perfectness conditions keep holding.
    Then, by \Cref{TheoremIterationLemmaNonLinearOverview}, we can construct a sequence of runs that witness $\downclsof{\amarkingp}\subseteq\postfuncNof{\anngvas}{\amarking, \anonterm}$.  
    If this is also not the case, then $(\amarking, \anonterm, \amarkingp)\not\in\simplydecomps$, but there is maximal $\amarkingp'\in\intpostapproxof{\amarking, \anonterm}$ such that $\amarkingp<\amarkingp'$.
    Note that, $(\amarking, \anonterm, \amarkingp)\not\in\simplydecomps$ implies that for all $i, j\in\abdinfoleft\setminus\unconstrained$, $\wtgrammarof{\otherctxNGVAS{\anngvas}{\amarking, \anonterm, \amarkingp}, \forgetfulpostapprox{\unconstrained\cup\set{i}}, \forgetfulpreapprox{\omegaoutcount\cup\set{j}}}$ contains a non-terminal $((\amarking, \aprommarking), \anontermp, (\aprommarkingp, \amarkingp))$ with $\omegaof{\amarking}=\omegaof{\amarkingp}=\abdinfomid$.
    Since the approximators are monotonous, the coverability grammar $\wtgrammarof{\otherctxNGVAS{\anngvas}{\amarking, \anonterm, \amarkingp'}, \forgetfulpostapprox{\unconstrained\cup\set{i}}, \forgetfulpreapprox{\omegaoutcount\cup\set{j}}}$ also contains such a non-terminal for all $i, j\in\abdinfomid\setminus\unconstrained$.
    Then, $(\amarking, \anonterm, \amarkingp')\in\simplydecomps$.
    This means that, by the previous case, and since we process the maximal elements first, we will have already added $\amarkingp'\in\mathsf{Out}_{\N}$.
    Since $\amarkingp''\sqsubseteq\amarkingp\leq\amarkingp'$ for all $(\amarking', \arun, \amarkingp'')\in\runsof{\otherctxNGVAS{\anngvas}{\amarking, \anonterm, \amarkingp}}$, the element $\amarkingp'\in\mathsf{Out}_{\N}$ already covers all possible runs captured by the NGVAS $\otherctxNGVAS{\anngvas}{\amarking, \anonterm, \amarkingp}$.
    This concludes the proof.
\end{proof}

\subsection{Hard Case 2, Witness Trees}\label{Section:WitnessTreesL3}

We provide the details omitted from \Cref{Section:HardCaseTwo}.
We recall the definition of the witness tree, but with the formalism of marked parse trees at hand.
A witness tree is a marked parse tree in the grammar of $\anngvas$, which follows $\postfuncN{\anngvas}$ for the labels on its leaves.
Formally, a \emph{witness tree} $\atree$ is a marked parse tree that satisfies the following.
%
\begin{itemize}
    \item We have $\atree.\inlabel\in\fullinputdom$.
    \item For any node $\anode\in\atree$ where $\childnodesof{\anode}=\anodep.\anodepp$, $\anode.\inlabel=\anodep.\inlabel$ and and $\anodep.\outlabel=\anodepp.\inlabel$.
    \item No node has the same label as one of its successors.
    \item For any node $\anode$ with $\nodemarkingof{\aleaf}=(\amarking,\asymbol,\amarkingp)$, we have
    $(\amarking,\asymbol)\in\fullinputdom\times\trms\cup\largeinputdom\times\nonterms$ iff $\anode$ is a leaf. 
    \item For any leaf $\aleaf$ with $\nodemarkingof{\aleaf}=(\amarking, \asymbol, \amarkingp)$, we have $\amarkingp\in\postfuncNof{\anngvas}{\amarking, \asymbol}$.
    \item For every subtree $\atreep$ of $\atree$, we have $\atreep.\outlabel=\pumpingof{\atreep}$.
\end{itemize}

We show that witness trees are sound and complete with respect to coverability.
\LemmaWTSoundMainPaper*
\begin{proof}
    Let $\atree\in\witnesssetof{h}$.
    First, we show by induction on $h\in\N$ that $\atree.\inlabel\sqsubseteq\inof{\atree.\symlabel}$, $\atree.\outlabel\sqsubseteq\outof{\atree.\symlabel}$, and $\omegaof{\atree.\inlabel}\subseteq\omegaof{\atree.\outlabel}$ hold.
    For the base case, we have $h=0$.
    This means that $\atree$ has only one node, which is a leaf.
    Then, the definition of the witness tree yields $\atree.\outlabel\in\postfuncNof{\anngvas}{\atree.\inlabel, \atree.\symlabel}$.
    It must hold that there is a run $(\amarking, \atree.\symlabel, \amarkingp)\in\runsof{\atree.\symlabel}$ where $\amarking\sqsubseteq\atree.\inlabel$ and $\amarkingp\sqsubseteq\atree.\outlabel$.
    This is only possible if $\amarking\sqsubseteq\inof{\atree.\symlabel}$ and $\amarkingp\sqsubseteq\outof{\atree.\symlabel}$.
    We obtain the first claim.
    This also yields $\omegaof{\atree.\inlabel}\subseteq\abdinfomid$.
    We move on to $\omegaof{\atree.\inlabel}\subseteq\omegaof{\atree.\outlabel}$
    Since terminals (and thus strings of terminals) only constrain the counters in $[1,d]\setminus\abdinfomid$, and $\omegaof{\atree.\inlabel}\subseteq\abdinfomid$, we know that for any $k\in\N$, $(\amarking+k\cdot 1_{\omegaof{\atree.\inlabel}}, \atree.\symlabel,\amarkingp+k\cdot 1_{\omegaof{\atree.\inlabel}})\in\runsof{\atree.\symlabel}$ as well.
    By the definition of $\postfuncN{\anngvas}$, this implies that for any $\amarkingpp\in\postfuncNof{\anngvas}{\atree.\inlabel, \atree.\symlabel}$, it must hold that $\omegaof{\atree.\inlabel}\subseteq\omegaof{\amarkingpp}$.
    This shows the claim.
    We only sketch out the proof for the inductive case.
    The results follow from the induction hypothesis and the witness tree properties. 
    The conditions $\atree.\inlabel\sqsubseteq\inof{\atree.\symlabel}$, $\atree.\outlabel\sqsubseteq\outof{\atree.\symlabel}$ follow from the fact that the witness tree embeds a derivation tree, which, when combined with the induction hypothesis, ensure that the $\infun$ and $\outfun$ assignments are sound.
    The condition $\omegaof{\atree.\inlabel}\subseteq\omegaof{\atree.\outlabel}$ follows from the induction hypothesis, transitivity of $\mathord{\subseteq}$, and the fact that the left and right subtrees of marked parse trees (and thus witness trees) agree on the output resp. input markings.

    Now, we move on to the last condition.
    We show two claims by induction on $h$.
    For convenience, in the case of $h\neq 0$, let $\atree_{\lefttag}$ be the left-subtree of $\atree$, and $\atree_{\righttag}$ the right-subtree.
    We claim
    \begin{itemize}
        \item[(i)] for any successor node $\anode$ labeled $(\amarkingpp, \anontermp, \amarkingppp)$ with $\omegaof{\amarkingp}=\omegaof{\atree.\inlabel}$ and $\anontermp\in\nonterms$, we have a derivation $\atree.\symlabel\to\aword_{\anode}.\anontermp.\awordp_{\anode}$ with sequences $\arun_{\anode}\in\updateseqof{\aword_{\anode}}$, $\arunp_{\anode}\in\updateseqof{\awordp_{\anode}}$ where $\atree.\inlabel\fires{\arun_{\anode}}\amarkingpp$, $\settoomega{\omegaof{\atree_{\righttag}.\outlabel}}{\amarkingppp}\fires{\arunp_{\anode}}\atree_{\righttag}.\outlabel$.
        \item[(ii)] there is a sequence of runs $[(\amarking_{i}, \atree.\symlabel, \amarkingp_{i})]_{i\in\N}\in\runsof{\atree.\symlabel}^{\omega}$ such that $\amarking_{i}\sqsubseteq\atree.\inlabel$ and $\amarkingp_{i}\sqsubseteq\atree.\outlabel$ for all $i\in\N$, and we have $\lim_{i\in\N}\amarkingp_{i}=\atree.\outlabel$.
    \end{itemize}
    The claim (i), in the case of $\atree.\symlabel=\anode.\symlabel$, yields a context that can be repeated to obtain a positive effect on the non-$\omega$ counters of $\atree_{\righttag}.\outlabel$ while keeping the left-side counters stable.
    We do not require $\atree.\symlabel=\anode.\symlabel$ in the claim to keep the induction sound.
    Claim (ii) says that the output markings of a node can be justified by a sequence of runs from the input, that converge to the output.
    We proceed with the base case, $h=0$.
    The tree consists of one leaf label, so claim (i) does not apply, and we have $\atree.\outlabel\in\postfuncNof{\anngvas}{\atree.\inlabel, \atree.\symlabel}$, which shows (ii).
    Now we move to the inductive case $h+1$.
    Note that the trees $\atree_{\lefttag}$ and $\atree_{\righttag}$ both have a maximal height of $h$.
    We show (i).
    Let $\anode$ be a node in $\atree$.
    We show the case where $\anode$ is a non-root node in $\atree_{\lefttag}$, the cases where $\anode$ is a node in $\atree_{\righttag}$, or the root node of $\atree_{\lefttag}$ are similar.
    Let $(\amarkingpp, \anontermp, \amarkingppp)$ be the label of $\anode$, where $\omegaof{\amarkingpp}=\omegaof{\atree.\inlabel}$ and $\anontermp\in\nonterms$ hold.
    Since $\omegaof{\amarkingpp}=\omegaof{\atree.\inlabel}$, and witness trees only gain $\omega$ markings when moving from left to right, it must hold that $\omegaof{\atree_{\lefttag}.\inlabel}=\omegaof{\atree.\inlabel}$.
    By applying the induction hypothesis for (i) to $\atree_{\lefttag}$, we obtain a derivation $\atree_{\lefttag}.\symlabel\to\aword_{\anode}'.\anontermp.\awordp_{\anode}'$ with sequences $\arun_{\anode}'\in\updateseqof{\aword_{\anode}'}$, $\arunp_{\anode}'\in\updateseqof{\awordp_{\anode}'}$ where $\atree.\inlabel\fires{\arun_{\anode}'}\amarkingpp$, and $\settoomega{\omegaof{\atree_{\lefttag, \righttag}.\outlabel}}{\amarkingppp}\fires{\arunp_{\anode}'}\atree_{\lefttag, \righttag}.\outlabel$.
    Here $\atree_{\lefttag, \righttag}$ is the right-subtree of $\atree_{\lefttag}$.
    Since $\omegaof{\atree_{\lefttag, \righttag}.\outlabel}\subseteq\omegaof{\atree_{\lefttag}.\outlabel}$ and $\atree_{\lefttag, \righttag}.\outlabel\sqsubseteq\atree_{\lefttag}.\outlabel$, we simplify the latter condition to the weaker $\settoomega{\omegaof{\atree_{\lefttag}.\outlabel}}{\amarkingppp}\fires{\arunp_{\anode}'}\atree_{\lefttag}.\outlabel$.
    We use the induction hypothesis for (ii) on $\atree_{\righttag}$, and get a run $(\amarking_{\righttag}, \arun_{\righttag}, \amarkingp_{\righttag})\in\runsof{\atree_{\righttag}.\symlabel}$ where $\amarking_{\righttag}\sqsubseteq\atree_{\righttag}.\inlabel$, and $\amarkingp_{\righttag}\sqsubseteq\atree_{\righttag}.\outlabel$.
    Then, $\atree_{\righttag}.\inlabel\fires{\arun_{\righttag}}\atree_{\righttag}.\outlabel$.
    Since $\atree_{\lefttag}.\outlabel=\atree_{\righttag}.\inlabel$, and $\omegaof{\atree_{\righttag}.\inlabel}\subseteq\omegaof{\atree_{\righttag}.\outlabel}$, it holds that $\settoomega{\omegaof{\atree_{\righttag}.\outlabel}}{\atree_{\righttag}.\inlabel}\fires{\arun_{\righttag}}\atree_{\righttag}.\outlabel$ and $\settoomega{\omegaof{\atree_{\righttag}.\outlabel}}{\amarkingppp}\fires{\arunp_{\anode}'}\atree_{\righttag}.\inlabel$.
    Combining these, we get the derivation $\atree.\symlabel\to\aword_{\anode}'.\anontermp.\awordp_{\anode}'.(\atree_{\righttag}.\symlabel)$, with the sequences $\arun_{\anode}'\in\runsof{\aword_{\anode}'}$ where $\atree.\inlabel\fires{\arun_{\anode}'}\amarkingpp$, and $\arunp_{\anode}'.\arun_{\righttag}\in\runsof{\awordp_{\anode}'.(\atree_{\righttag}.\symlabel)}$ where $\settoomega{\omegaof{\atree_{\righttag}.\outlabel}}{\amarkingppp}\fires{\arunp_{\anode}'.\arun_{\righttag}}\atree_{\righttag}.\outlabel$.
    Letting $\aword_{\anode}=\aword_{\anode}'$, $\awordp_{\anode}=\awordp_{\anode}'.(\atree_{\righttag}.\symlabel)$, $\arun_{\anode}=\arun_{\anode}'$, and $\arunp_{\anode}=\arunp_{\anode}'$ concludes the proof of (i).
    Now we show (ii).
    By applying the induction hypothesis for (ii) on the subtrees $\atree_{\lefttag}$ and $\atree_{\righttag}$ and combining the runs, it can be readily verified that we can get a sequence of runs that reach $\atree_{\righttag}.\outlabel$ from $\atree_{\lefttag}.\inlabel=\atree.\inlabel$.
    Formally, we know that there is a sequence of runs $[(\amarking_{i}, \arunpp_{i}, \amarkingp_{i})]_{i\in\N}\in\runsof{\atree.\symlabel}$ where $\amarking_{i}\sqsubseteq\atree.\inlabel$ and $\amarkingp_{i}\sqsubseteq\atree_{\righttag}.\outlabel$ for all $i\in\N$, and $\lim_{i\in\N}\arunpp_{i}=\atree_{\righttag}.\outlabel$.
    We can assume wlog. that $\arunpp_{i}$ adds at least $i$ tokens to the counters that become $\omega$ when moving from $\atree.\inlabel$ to $\atree_{\righttag}.\outlabel$.
    That is, we assume $\updates\cdot\paramparikhof{\updates}{\arunpp_{i}}[a]\geq i$ for all $i\in\N$ and $a\in\omegaof{\atree_{\righttag}.\outlabel}\setminus\omegaof{\atree.\inlabel}$.
    Now, we show that we can construct a sequence of runs that reach the output label after the effect of $\pumpingof{-}$.
    Consider the set of nodes $\set{\anode_{0}, \ldots, \anode_{\ell}}$ that justify the new $\omega$ counters in $\pumpingof{(\atree.\inlabel, \atree.\symlabel, \atree_{\righttag}.\outlabel):\atree_{\lefttag}.\atree_{\righttag}}$.
    Then, for each $j\leq \ell$, the node $\anode_{j}$ is labeled $(\atree.\inlabel, \atree.\symlabel, \amarkingppp_{j})$ where $\amarkingppp_{j}<\atree_{\righttag}.\outlabel$.
    Furthermore, for each $a\in\omegaof{\atree.\outlabel}\setminus\omegaof{\atree_{\righttag}.\outlabel}$, there is a $j\in[1,d]$, where $\amarkingppp_{j}[a]<\atree_{\righttag}.\outlabel[a]$.
    By (i), for each $j\leq \ell$ we get a derivation $\atree.\symlabel\to\aword_{\anode_{j}}.(\atree.\symlabel).\awordp_{\anode_{j}}$ and runs $(\amarking_{\anode}, \arun_{\anode}, \amarkingpp_{\anode})\in\runsof{\aword_{\anode_{j}}}$ and $(\amarkingppp_{\anode}, \arunp_{\anode}, \amarkingp_{\anode})\in\runsof{\awordp_{\anode_{j}}}$ with the properties described in claim (ii).
    Note that this also yields the derivations $\atree.\symlabel\to\aword_{\anode_{0}}^{i}\ldots\aword_{\anode_{\ell}}^{i}.(\atree.\symlabel).\awordp_{\anode_{\ell}}^{i}\ldots\awordp_{\anode_{0}}^{i}$ for each $i\in\N$.
    In the following, we use this derivation, combined with the sequence of runs $[(\amarking_{i}, \arunpp_{i}, \amarkingp_{i})]_{i\in\N}\in\runsof{\atree.\symlabel}$, to construct a sequence of runs in $\runsof{\atree.\symlabel}$ that reaches $\atree.\outlabel$.
    First, note that for any $i\in\N$, we have $\atree.\inlabel\fires{\arun_{\anode_{j}}^{i}}\atree.\inlabel$ for each $j\leq\ell$, and thus $\atree.\inlabel\fires{\arun_{\anode_{0}}^{i}\ldots\arun_{\anode_{\ell}}^{i}}\atree.\inlabel$.
    Let $j\leq\ell$.
    We have $\settoomega{\omegaof{\atree_{\righttag}.\outlabel}}{\amarkingpp_{\anode_{j}}}\fires{\arunp_{\anode_{j}}}\atree_{\righttag}.\outlabel$.
    Since $\amarkingpp_{\anode_{j}}\leq\atree_{\righttag}.\outlabel$, we observe $\updates\cdot\paramparikhof{\updates}{\arunp_{\anode_{j}}}[a]\geq 0$ for all $a\in[1,d]\setminus\omegaof{\atree_{\righttag}.\outlabel}$.
    Then, $\atree_{\righttag}.\outlabel\fires{\arunp_{\anode_{j}}^{i}}$, and $\arunp_{\anode_{j}}^{i}$ has a non-negative effect on counters $[1,d]\setminus\omegaof{\atree_{\righttag}.\outlabel}$ for all $i\in\N$ and $j\leq \ell$.
    If $\amarkingpp_{\anode_{j}}[a]<\atree_{\righttag}.\outlabel[a]$ for $a\in[1,d]\setminus\omegaof{\atree_{\righttag}.\outlabel}$, we observe $\updates\cdot\paramparikhof{\updates}{\arunp_{\anode_{j}}}[a]\geq 1$.
    Thus $\atree_{\righttag}.\outlabel\fires{\arunp_{\anode_{0}}^{i}\ldots\arunp_{\anode_{\ell}}^{i}}$, and $\updates\cdot(\paramparikhof{\updates}{\arunp_{\anode_{0}}^{i}\ldots\arunp_{\anode_{\ell}}^{i}})[a]\geq i$ for all $a\in\omegaof{\atree.\outlabel}\setminus\omegaof{\atree_{\righttag}.\outlabel}$ and $i\in\N$.
    Now, for all $i\in\N$, let $j_{i}=i\cdot(\cardof{\arunp_{\anode_{0}}\ldots\arunp_{\anode_{\ell}}}+1)$.
    Then, we get the enabledness $\atree.\inlabel\fires{\arunpp_{j_{i}}.\arunp_{\anode_{0}}^{i}\ldots\arunp_{\anode_{\ell}}^{i}}$ since $\arunpp_{j_{i}}$ fills all counters in $\omegaof{\atree_{\righttag}.\outlabel}\setminus\omegaof{\atree.\inlabel}$ with an amount of tokens that cannot be exhausted by the suffix.
    We also know that $\arunp_{\anode_{0}}^{i}\ldots\arunp_{\anode_{\ell}}^{i}$ pumps all counters in $\omegaof{\atree.\inlabel}\setminus\omegaof{\atree_{\righttag}.\outlabel}$ by $i$ tokens.
    Combining this with $\atree.\inlabel\fires{\arun_{\anode_{0}}^{i}\ldots\arun_{\anode_{\ell}}^{i}}\atree.\inlabel$, and $\arunp_{\anode_{0}}^{i}\ldots\arunp_{\anode_{\ell}}^{i}$, we observe that 
    $$[\arun_{\anode_{0}}^{i}\ldots\arun_{\anode_{\ell}}^{i}\arunpp_{j_{i}}\arunp_{\anode_{\ell}}^{i}\ldots\arunp_{\anode_{0}}^{i}]_{i\in\N}\in\runsof{\atree.\symlabel}^{\omega}$$
    is the sequence of runs claimed to exist by (ii).
    This concludes the proof.
\end{proof}

Now we show completeness.
\LemmaWTCompleteMainPaper*
\begin{proof}
    Let $\asymbol\in\nonterms\cup\trms$, and $(\amarking, \arun, \amarkingp)\in\runsof{\asymbol}$.
    Further let $\amarking_{\omega}\in\Nomega^{d}$ with $\settoomega{\unconstrained}{\amarking}\sqsubseteq\amarking_{\omega}$.
    We make an induction over the height of the $\anngvas$-reachability tree that witnesses the run $(\amarking, \arun, \amarkingp)\in\runsof{\asymbol}$.
    Let $\atreep\in\reachtreesof{\anngvas}$ be the reachability tree that witnesses  $(\amarking, \arun, \amarkingp)\in\runsof{\asymbol}$.
    For the base case, we have that the height of $\atreep$ is $0$, i.e. that $\atreep$ consists of just one leaf node, labeled $(\amarking, \asymbol, \amarkingp)$.
    Then, it must hold that $\asymbol\in\trms$.
    By definition, we have $\amarkingp\leq\amarkingp_{\omega}$ for some $\amarkingp_{\omega}\in\postfuncNof{\anngvas}{\amarking_{\omega}, \asymbol}$, since $\amarking\sqsubseteq\amarking_{\omega}$.
    Since $\settoomega{\unconstrained}{\amarking}\in\fullinputdom$, we have the $0$-depth witness tree whose sole node is labeled $(\amarking_{\omega}, \asymbol, \amarkingp_{\omega})$.
    
    For the inductive case, assume that $\atreep$ has height $h+1$, and all reachability trees below this height have covering witness trees as specified by the lemma.
    Then, it must hold that $\asymbol\in\nonterms$, we let $\asymbol=\anonterm$ to make this clear.
    We know that $\amarking_{\omega}\in\fullinputdom=\largeinputdom\uplus\smallinputdom$.
    If $\amarking_{\omega}\in\largeinputdom$, then, by a similar argument to the base case, there exists a witness tree whose sole node is labeled $(\amarking_{\omega}, \anonterm, \amarkingp_{\omega})$ for some $\amarkingp_{\omega}\in\postfuncNof{\anngvas}{\amarking_{\omega}, \anonterm}$.
    Now let $\amarking_{\omega}\in\smallinputdom$.
    Let $\atreep_{\lefttag}$ be the subtree of $\atreep$ centered on the left child of the root node, and $\atreep_{\righttag}$ the subtree centered on the right child.
    Note that these trees are of lesser height than $\atreep$, and we have the rule $\anonterm\to(\atreep_{\lefttag}.\symlabel).(\atreep_{\righttag}.\symlabel)$.
    We use the induction hypothesis to construct the witness tree $\atree_{\lefttag}$ with $\amarking\sqsubseteq\atree_{\lefttag}.\inlabel=\amarking_{\omega}$, $\atree_{\lefttag}.\symlabel=\atreep_{\lefttag}.\symlabel$, and $\atree_{\lefttag}.\outlabel\leq\atreep_{\lefttag}.\outlabel$.
    It is easy to see that translating all input and output markings of the reachability tree $\atree_{\righttag}\in\reachtreesof{\anngvas}$ along the same vector $\amarkingpp\in\N^{d}$ with $\amarkingpp[j]=0$ for all $j\in[1,d]\setminus\abdinfomid$ also results in a reachability tree.
    Since $\atreep_{\lefttag}.\outlabel\leq\atree_{\lefttag}.\outlabel$, there is a vector $\amarkingpp\in\N^{d}$, such that $\atreep_{\lefttag}.\outlabel+\amarkingpp\sqsubseteq\atree_{\lefttag}.\outlabel$.
    Because $\atreep_{\lefttag}.\outlabel, \atree_{\lefttag}.\outlabel\sqsubseteq\outof{\atreep_{\lefttag}.\symlabel}$ holds by \Cref{Lemma:WTSound}, we can assume $\amarkingpp[j]=0$ for all $j\in[1,d]\setminus\abdinfomid$. 
    Let $\atreep_{\righttag}'$ be the tree $\atreep_{\righttag}$ translated along such a $\amarkingpp\in\N^{d}$.
    Since $\unconstrained\subseteq\omegaof{\atree_{\lefttag}.\inlabel}\subseteq\omegaof{\atree_{\lefttag}.\outlabel}$, the induction hypothesis applies to $\atreep_{\righttag}'$ with the input label $\atree_{\lefttag}.\inlabel$.
    This yields a witness tree $\atree_{\righttag}$ such that $\atree_{\righttag}.\inlabel=\atree_{\lefttag}.\outlabel$, $\atree_{\righttag}.\symlabel=\atree_{\righttag}.\symlabel$, and $\atreep_{\righttag}.\outlabel\leq\atreep_{\righttag}.\outlabel+\amarkingpp=\atreep_{\righttag}'\leq\atree_{\righttag}.\outlabel$.
    Let $\atree'=(\amarking_{\omega}, \anonterm, \atree_{\righttag}.\outlabel):\atree_{\lefttag}.\atree_{\righttag}$, and let $\atree=(\amarking_{\omega}, \anonterm, \pumpingof{\atree'}):\atree_{\lefttag}.\atree_{\righttag}$.
    Clearly, $\atreep.\outlabel=\atreep_{\righttag}.\outlabel\leq\atree'.\outlabel\leq\atree.\outlabel$.
    If $\atree'$ is readily a witness tree, we are done.
    If this is not the case, then a node repeats the label of one of its successors in $\atree$.
    This can only be the root node, since $\atree_{\lefttag}$ and $\atree_{\righttag}$ are already witness trees.
    In that case, $\atree$ must contain a subtree $\atree''$ with the root label $(\amarking_{\omega}, \atreep.\symlabel, \pumpingof{\atree'})$.
    Then, $\atree''$ must be a subtree of one of $\atree_{\lefttag}$ or $\atree_{\righttag}$, which means that $\atree''$ is a witness tree.
    This shows our claim.
\end{proof}
Witness trees are also effectively constructable.
\LemmaWTEffectivenessMainPaper*
\begin{proof}
    For all $(\amarking, \asymbol)\in\fullinputdom\times(\nonterms\cup\trms)$, we compute the witness trees in $\witnesssetof{h}(\amarking, \asymbol)$ inductively in $h\in\N$.
    Let $(\amarking, \asymbol)\in\fullinputdom\times(\nonterms\cup\trms)$.
    Consider the base case $h=0$.
    There are two cases, depending on whether $(\amarking, \asymbol)\in\smallinputdom\times\nonterms$ holds.
    If $(\amarking, \asymbol)\in\smallinputdom\times\nonterms$, we have $\witnesssetof{0}(\amarking, \asymbol)=\emptyset$, since no leaf with label $\smallinputdom\times\nonterms\times\Nomega^{d}$ is allowed.
    Otherwise, we have $\witnesssetof{0}(\amarking, \asymbol)=\setcond{(\amarking, \asymbol, \amarkingp)\in\Nomega^{d}\times(\trms\cup\nonterms)\times\Nomega^{d}}{\amarkingp\in\postfuncNof{\anngvas}{\amarking, \asymbol}}$,
    This is effective, since in this case we have $(\amarking, \asymbol)\in\fullinputdom\times\trms\;\cup\;\largeinputdom\times\nonterms$, and $\postfuncN{\anngvas}$ can be computed under our assumption \Cref{Corollary:WTSearchAssumption}.
    
    For the inductive case, we assume that $\witnesssetof{h}(\amarking)$ is computable when $(\amarking', \asymbol')\in\fullinputdom\times(\nonterms\cup\trms)$ is given.
    Now, we show that $\witnesssetof{h+1}(\amarking, \asymbol)$ is computable.
    Since any terminal labeled node must be a leaf, we have $\witnesssetof{h+1}(\amarking, \asymbol)=\witnesssetof{0}(\amarking, \asymbol)$ and conclude with the induction hypothesis.
    Let $\asymbol=\anonterm\in\nonterms$.
    We claim that
    \begin{align*}
        \witnesssetof{h+1}(\amarking, \anonterm)&=\setcond{(\amarking, \anonterm, \pumpingof{\atreep'}):\atreep_{\lefttag}.\atreep_{\righttag}}{\\
        &\hspace{2em}\atreep'=(\amarking, \anonterm, \atreep_{\righttag}.\outlabel):\atreep_{\lefttag}.\atreep_{\righttag},\;\\
        &\hspace{2em}\anonterm\to\asymbolp_{0}.\asymbolp_{1}\in\prods,\; \atreep_{\lefttag}\in\witnesssetof{h}(\amarking, \asymbolp_{0}),\; \\
        &\hspace{2em}\atreep_{\righttag}\in\witnesssetof{h}(\atreep_{\lefttag}.\outlabel, \asymbolp_{1})}
    \end{align*}
    holds.
    It is clear that every tree from the set on the righthand side of our claim is a witness tree: we ensure that the rules are correctly followed, and we ensure the pumping property for the top node explicitly, and for the remaining nodes by the definition of $\witnesssetof{h}$.
    Now we argue that the righthand side captures all witness trees in $\witnesssetof{h+1}(\amarking, \anonterm)$.
    Any witness tree $\atree$ with a root $(\amarking, \anonterm, \amarkingp)$, has two subtrees $\atreep_{\lefttag}$ and $\atreep_{\righttag}$ of less height, where $\anonterm\to(\atreep_{\lefttag}.\symlabel).(\atreep_{\righttag}.\symlabel)$, $\atreep_{\lefttag}.\outlabel=\atreep_{\righttag}.\inlabel$, and $\amarkingp=\pumpingof{(\amarking, \anonterm, \amarkingp):\atreep_{\lefttag}.\atreep_{\righttag}}=\pumpingof{(\amarking, \anonterm, \atreep_{\righttag}.\outlabel):\atreep_{\lefttag}.\atreep_{\righttag}}$.
    The latter equality follows from the fact that $\pumping$ ignores the $\atree.\outlabel$ for the input $\atree$.
    Such a witness tree $\atree$ is captured by our description.
    This shows the equivalence.
    Finally, we argue that $\witnesssetof{h+1}$ is computable.
    The sets $\witnesssetof{h}(\amarking, \asymbolp_{0})$ and $\witnesssetof{h}(\atreep_{\righttag}.\inlabel, \asymbolp_{1})$ from the description are computable by the induction hypothesis.
    It remains to argue that $\pumpingof{(\amarking, \anonterm, \atreep_{\righttag}.\outlabel):\atreep_{\lefttag}.\atreep_{\righttag}}$ from the description is computable.
    This is clear, we only need to compare the root label of $\atree'$ to the labels of the successor nodes, and set the dominated counters $\omega$ as prescribed by the definition of $\pumpingof{-}$.
    This concludes the proof.
\end{proof}

\subsection{Computing the constants}\label{Section:Precalculation}

\paragraph*{Computing $\incconst$.}
Now, we show \Cref{Lemma:PrecalculationI} from \Cref{Section:ComputingPostAndPreMP}.
\LemmaPrecalculationIMainPaper*
\begin{proof}
    Let $I, O\subseteq[1,d]$.
    We start from an NGVAS $\anngvas$ with a $(I, O)$-$\Z$-pump, $(\startnonterm\to^{*}\aword_{zp, \startnonterm}.\startnonterm.\awordp_{zp, \startnonterm}, \amarking_{zp, \startnonterm}, \amarkingp_{zp, \startnonterm})$, that has all perfectness conditions excluding \perfectnesspumpingnospace.
    Let $\anonterm\in\nonterms$.
    We construct a free-$\N$-pump $(\anonterm\to^{*}\aword_{fp}.\anonterm.\awordp_{fp}, \arun_{fp}, \arunp_{fp})$.
    Recall the properties of the $\Z$-pump.
    We have $\amarking_{zp, \startnonterm}\in\ceffof{\aword_{zp, \startnonterm}}$, $\amarkingp_{zp, \startnonterm}\in\ceffof{\awordp_{zp, \startnonterm}}$, and $\amarking_{zp, \startnonterm}[i]\geq 1$ for all $i\in \abdinfomid\setminus I$, while $-\amarkingp_{zp, \startnonterm}[i]\geq 1$ for all $i\in\abdinfomid\setminus O$.
    Even though $\Z$-pump has been defined in relation to $\anngvas$, for the purposes of this proof, we refer to $(\anonterm\to^{*}\aword.\anonterm.\awordp, \amarking, \amarkingp)$ as a $\Z$-pump (centered on $\anonterm$), if $\amarking\in\ceffof{\aword}$, $\amarkingp\in\ceffof{\awordp}$, $\amarking[i]\geq 1$ for all $i\in\abdinfomid\setminus I$, and $-\amarkingp[i]\geq 1$ for all $i\in\abdinfomid\setminus O$.
    The proof proceeds as follows.
    First, we move to a $\Z$-pump $(\anonterm\to^{*}\aword_{fp}.\anonterm.\awordp_{fp}, \amarking_{fp}, \amarkingp_{fp})$ that goes from $\anonterm$ to $\anonterm$ instead of $\startnonterm$ to $\startnonterm$.
    Then, we move to a $\Z$-pump $(\anonterm\to^{*}\aword_{full}.\anonterm.\awordp_{full}, \amarking_{full}, \amarkingp_{full})$ that produces each child at least once on both sides, and where $\amarking_{full}$ and $\amarkingp_{full}$ can be realized by taking each child period at least once. 
    Finally, we show that the period vectors can be organized such that we can ensure the existence of runs that witness the effects via \Cref{TheoremIterationLemmaNonLinearOverview}.
    
    Let $b_{base}=\max_{\anngvasp\in\trms}\sizeof{\anngvasp.\baseeffect}$ be the maximum size of the base effect of a terminal.
    Since $\anngvas$ is strongly connected, there are derivations $\anonterm\to\aword_{in}.\startnonterm.\awordp_{in}$ and $\startnonterm\to\aword_{out}.\anonterm.\aword_{out}$ where $\aword_{in}, \aword_{out}, \awordp_{in}, \awordp_{out}\in\trms^{*}$.
    Let $b_{len}=\max\set{\cardof{\aword_{in}}, \cardof{\aword_{out}}, \cardof{\awordp_{in}}, \cardof{\awordp_{out}}}$.
    It holds that there are $\amarking_{in}\in\ceffof{\aword_{in}}$, $\amarking_{out}\in\ceffof{\aword_{out}}$, $\amarkingp_{in}\in\ceffof{\awordp_{in}}$, and $\amarkingp_{out}\in\ceffof{\awordp_{out}}$, with $\sizeof{\amarking_{in}}, \sizeof{\amarking_{out}}, \sizeof{\amarkingp_{in}}, \sizeof{\amarkingp_{out}}\leq b_{len}\cdot b_{base}$.
    Let $b_{mov}=b_{len}\cdot b_{base}$ for brevity.
    Then, we can combine the three derivations into $\anonterm\to^{*}\aword_{in}.\aword_{zp, \startnonterm}^{b_{mov}+1}.\aword_{out}.\anonterm.\awordp_{out}.\awordp_{zp,\startnonterm}^{b_{mov}+1}.\awordp_{in}$.
    For brevity, we write $\aword_{zp}=\aword_{in}.\aword_{zp,\startnonterm}^{b_{mov}+1}.\aword_{out}$, $\awordp_{zp}=\awordp_{out}.\awordp_{zp,\startnonterm}^{b_{mov}+1}.\awordp_{in}$, $\amarking_{zp}=\amarking_{in}+(b_{mov}+1)\cdot\amarking_{zp, \startnonterm}+\amarking_{out}$, and $\amarkingp_{zp}=\amarkingp_{out}+(b_{mov}+1)\cdot\amarkingp_{zp,\startnonterm}+\amarkingp_{out}$.
    Then, $(\anonterm\to^{*}\aword_{zp}.\anonterm.\awordp_{zp}, \amarking_{zp}, \amarkingp_{zp})$ is a $\Z$-pump centered on $\anonterm$.

    As a preparation for the application of \Cref{TheoremIterationLemmaNonLinearOverview}, we move to a $\Z$-pump that takes each child period at least twice.
    Consider the full support homogenous solution $\ahomsol$ of $\anngvas$.
    Since $\anngvas$ has \perfectnessprodsnospace, and $\anngvas$ is non-linear, it holds that $\ahomsol$ has a non-zero value entry for each production rule, and child period.
    We construct $\anonterm\to\aword_{hom}.\anonterm.\awordp_{hom}$, where $\paramparikhof{\trms}{\aword_{hom}}[\anngvasp]\geq 1$ and $\paramparikhof{\trms}{\awordp_{hom}}[\anngvasp]\geq 1$ for all $\anngvasp\in\trms$, with $\amarking_{hom}\in\ceffof{\aword_{hom}}$ and $\amarkingp_{hom}\in\ceffof{\awordp_{hom}}$ both obtained by adding the base effect of $\anngvasp\in\trms$, $\paramparikhof{\trms}{\aword_{hom}}[\anngvasp]$ resp. $\paramparikhof{\trms}{\awordp_{hom}}[\anngvasp]$ times, and each period vector of each $\anngvasp\in\trms$ at least once.
    Note that this is possible by applying the branching rule as often as needed, similarly to the proof of \Cref{TheoremWideTree}, since we impose no reachability or positivity constraints for this construction.
    Now we ensure the $\Z$-pump property on top.
    Let $b_{hlen}=\max\set{\sizeof{\amarking_{hom}}, \sizeof{\amarkingp_{hom}}}$.
    Then, for $\aword_{full}=\aword_{zp}^{b_{hlen}+1}.\aword_{hom}$, $\awordp_{full}=\awordp_{hom}.\awordp_{zp}^{b_{hlen}+1}$, $\amarking_{full}=(b_{hlen}+1)\cdot\amarking_{zp}+\amarking_{hom}$, and $\amarkingp_{full}=(b_{hlen}+1)\cdot\amarkingp_{zp}+\amarkingp_{hom}$, the tuple $(\anonterm\to^{*}\aword_{full}.\anonterm.\awordp_{full}, \amarking_{full}, \amarkingp_{full})$ is a $\Z$-pump with $\paramparikhof{\trms}{\aword_{full}}[\anngvasp], \paramparikhof{\trms}{\awordp_{full}}[\anngvasp]\geq 1$ for all $\anngvasp\in\trms$, where $\amarking_{full}$ and $\amarkingp_{full}$ can be realized by taking each child period at least once.
    We construct the sequences of vectors $[z_{j}]_{j<\cardof{\aword_{full}}}$ and $[m_{j}]_{j<\cardof{\awordp_{full}}}$ which prescribe how often each period vector of child $\aword_{full}[j]$ resp. $\awordp_{full}[j]$ needs to be taken.
    The vector $z_{j}$ is typed $\N^{\aword_{full}[j].\periodeffect}$ and $m_{k}$ is typed $\N^{\awordp_{full}[k].\periodeffect}$ for all $j<\cardof{\aword_{full}}$ and $k<\cardof{\awordp_{full}}$.
    We pack all child periods together, i.e. there is the sequence of vectors $[z_{j}]_{j<\cardof{\aword_{full}}}$, where $z_{j}\in\N^{\aword_{full}[j].\periodeffect}$ has $z_{j}=0$ or $z_{j}\geq 1_{\aword_{full}[j].\periodeffect}$ for all $j<\cardof{\aword_{full}}$, as well as a sequence of vectors $[m_{j}]_{j<\cardof{\awordp_{full}}}$, where $m_{j}\in\N^{\awordp_{full}[j].\periodeffect}$ has $m_{j}=0$ or $m_{j}\geq 1_{\awordp_{full}[j].\periodeffect}$ for all $j<\cardof{\awordp_{full}}$, with
    \begin{align*}
        \amarking_{full}&=\sum_{j<\cardof{\aword_{full}}}\aword_{full}[j].\baseeffect +(\aword_{full}.\periodeffect)\cdot z_{j}\\
        \amarkingp_{full}&=\sum_{j<\cardof{\awordp_{full}}}\awordp_{full}[j].\baseeffect +(\awordp_{full}.\periodeffect)\cdot m_{j}
    \end{align*}
    Since the children NGVAS have their base effects enabled by \perfectnessbasenospace, there is a sequence $\arun_{\anngvasp, 0}\in\updateseqof{\anngvasp}$ for each $\anngvasp\in\trms$ with $\paramparikhof{\updates}{\arun_{\anngvasp, 0}}=\anngvasp.\baseeffect$.
    Since the children NGVAS are perfect by \perfectnesschildrennospace, \Cref{TheoremIterationLemmaNonLinearOverview} applies.
    Let $\arun_{\anngvasp, z}^{(\aconst)}\in\updates^{*}$ be the sequence $\arun_{\anngvasp, z}^{(\aconst)}\in\updateseqof{\anngvasp}$ obtained by \Cref{TheoremIterationLemmaNonLinearOverview}, for some $z\in\N^{\anngvasp.\periodeffect}$ and $\aconst\in\N$ with $z\geq 1_{\anngvasp.\periodeffect}$ and $\aconst\geq\initconst$ ($\initconst$ as described in \Cref{TheoremIterationLemmaNonLinearOverview}) where $\paramparikhof{\updates}{\arun_{\anngvasp, z}^{(\aconst)}}=\anngvasp.\baseeffect+\aconst\cdot \anngvasp.\periodeffect\cdot z$.
    Let $\initconst^{max}\in\N$ be the largest $\initconst\in\N$ imposed by \Cref{TheoremIterationLemmaNonLinearOverview}, when applying it to $\aword_{full}[j]$ with $z_{j}$ for some $j<\cardof{\aword_{full}}$ or $\awordp_{full}[j]$ with $m_{j}$ for some $j<\cardof{\awordp_{full}}$.
    Further, let 
    \begin{align*}
        \arun_{full}^{(0)}&=\arun_{\aword_{full}[0], 0}\ldots\arun_{\aword_{full}[\lastindexof{\aword}],0}\\
        \arun_{full}^{(\aconst)}&=\arun_{\aword_{full}[0], z_{0}}^{(\aconst)}\ldots\arun_{\aword_{full}[\lastindexof{\aword_{full}}],z_{\lastindexof{\aword_{full}}}}^{(\aconst)}\\
        \arunp_{full}^{(0)}&=\arun_{\awordp_{full}[0], 0}\ldots\arun_{\awordp_{full}[\lastindexof{\aword}],0}\\
        \arunp_{full}^{(\aconst)}&=\arun_{\awordp_{full}[0], m_{0}}^{(\aconst)}\ldots\arun_{\awordp_{full}[\lastindexof{\awordp_{full}}],m_{\lastindexof{\awordp_{full}}}}^{(\aconst)}\\
    \end{align*}
    Then, for all $\aconst\geq \initconst^{max}$ or $\aconst=0$, $\arun_{full}^{(\aconst)}\in\aword_{full}$, and $\arunp_{full}^{(\aconst)}\in\awordp_{full}$.
    Let $\arun_{fin}=(\arun_{full}^{(0)})^{\aconst-1}.\arun_{full}^{(\aconst)}\in\runsof{\aword_{full}^{\aconst}}$ and $\arunp_{fin}=\arunp_{full}^{(\aconst)}. (\arunp_{full}^{(0)})^{\aconst-1}\in\runsof{\aword_{full}^{\aconst}}$.
    Note that $(\arun_{full}^{(0)})^{\aconst-1}$ takes all base effects $\aconst-1$ times, and $\arun_{full}^{(\aconst)}$ takes them once.
    Conversely, $(\arun_{full}^{(0)})^{\aconst-1}$ takes no period vectors, and the $j$-th sequence in $\arun_{full}^{(\aconst)}$ takes period effects captured by $z_{j}$, $\aconst$ times.
    Then $\updates\cdot\paramparikhof{\updates}{\arun_{fin}}=\aconst\cdot\amarking_{full}$ and by a similar argument $-\updates\cdot\paramparikhof{\updates}{\arunp_{fin}}=\aconst\cdot\amarkingp_{full}$.
    We have $\aconst\cdot\amarking_{full}[i]\geq 1$ for all $i\in\abdinfomid\setminus I$ and $-\aconst\cdot\amarkingp_{full}[i]\geq 1$ for all $i\in\abdinfomid\setminus O$.
    Since there is a derivation $\anonterm\to\aword_{full}^{\aconst}.\anonterm.\awordp_{full}^{\aconst}$, and we have $\arun_{fin}\in\updateseqof{\aword_{full}^{\aconst}}$ and $\arunp_{fin}\in\updateseqof{\awordp_{full}^{\aconst}}$, this concludes the proof.\\
\end{proof}

\paragraph*{Computing $\extcovconst$.}
Finally, we show the following lemma from \Cref{Section:HardCaseOne}.
\LemmaPrecalculationIIMainPaper*
First using \Cref{Lemma:PrecalculationI}, we fix a $(\unconstrained, \omegaoutcount)$-free-$\N$-pump $\gamma$ of size at most $\incconst$.
We break the task of $\extcovconst$ down, into computing the maximum size of a minimal pumping derivation for each choice of $i,j\in\abdinfomid\setminus\unconstrained$.
To make the proof easier, we also allow ignoring further sets of $\omega$ counters.
We formalize this by the following claim.
\begin{lemma}\label{Lemma:PrecalculationIIBreakdown}
    Let $\unconstrained\subseteq I \subseteq \abdinfomid$, $\omegaoutcount\subseteq O\subseteq\abdinfomid$ with non-empty 
    $I\setminus \unconstrained$ and $O\setminus\unconstrained$.
    Let $\perfect$ be reliable up to $\rankof{\anngvas}$.
    We can compute a constant $\extcovconst_{I, O}\in\N$ such for any $\amarking, \amarkingp\in\Nomega^{d}$ and $\anonterm\in\nonterms$ with $\omegaof{\amarking}=I$ and $\omegaof{\amarkingp}=O$, if $\otherctxNGVAS{\anngvas}{\amarking, \anonterm, \amarkingp}$ has a pumping derivation, then it has a pumping derivation $(\anonterm\to^{*}\aword.\anonterm.\awordp, \arun, \arunp)$ with $\normof{\arun}, \normof{\arunp}\leq\extcovconst_{I, O}$.
\end{lemma}
Clearly, setting 
\begin{align*}
    \extcovconst=\max_{i, j\in\abdinfomid\setminus\unconstrained}\extcovconst_{\unconstrained\cup\set{i}, \omegaoutcount\cup\set{j}}
\end{align*}
yields a constant as claimed by \Cref{Lemma:PrecalculationII}.
In the following, we show \Cref{Lemma:PrecalculationIIBreakdown}.
First, we observe if the markings are given, then the existence of the pumping derivation is decidable, and we can construct the derivation.
\begin{lemma}\label{Lemma:ConstructableSimplePumps}
    Let $\unconstrained\subseteq I \subseteq \abdinfomid$, $\omegaoutcount\subseteq O\subseteq\abdinfomid$ with non-empty 
    $I\setminus \unconstrained$ and $O\setminus\unconstrained$.
    Let $\perfect$ reliable up to $\rankof{\anngvas}$.
    Let $\amarking, \amarkingp\in\Nomega^{d}$ and $\anonterm\in\nonterms$ with $\omegaof{\amarking}=I$, $\omegaof{\amarkingp}=O$, $\amarking\sqsubseteq\inof{\anonterm}$, and $\amarkingp\sqsubseteq\outof{\anonterm}$.
    Then, we can decide whether there is a pumping derivation $(\anonterm\to^{*}\aword.\anonterm.\awordp, \arun, \arunp)$ for $\otherctxNGVAS{\anngvas}{\amarking, \anonterm, \amarkingp}$.
    If there is, we can construct such a pumping derivation.
\end{lemma}

\begin{proof}[Proof Sketch]
    By \Cref{Lemma:LowDimPostPreComputable} and \Cref{Lemma:ExtPumpComp}, we observe that we can decide whether $\otherctxNGVAS{\anngvas}{\amarking, \anonterm, \amarkingp}$ has \perfectnesspumpingnospace, since it needs lesser dimensional reachability checks.
    If it has a pumping derivation, then we can construct it via enumerating possible derivations until we find the shortest one.
    Since $\postfunc$ resp. $\prefunc$ use either (i) lower non-linear rank perfect decompositions or (ii) a witness tree search, the length of this derivation is no longer than the runtimes of these calls, it can be shown that the sequences witnessing the outputs are comparable in size to the run times of these functions. 
    We can argue with the Karp-Miller tree that the shortest pumping derivation is comparable in size to the run time of the Karp-Miller tree.   
\end{proof}

Note that the property of being a pumping-derivation for a triple is monotonous.
That is, we do not lose pumping-derivations by increasing the input and output markings.
\begin{lemma}\label{Lemma:PumpingMonotonous}
    Let $\amarking, \amarking', \amarkingp, \amarkingp'\in\Nomega$ with $\amarking, \amarking'\sqsubseteq\inof{\anonterm}$, $\amarkingp, \amarkingp'\sqsubseteq\outof{\anonterm}$, $\amarking\leq\amarking'$, and $\amarkingp\leq\amarkingp'$.
    If $(\anonterm\to^{*}\aword.\anonterm.\awordp, \arun, \arunp)$ is a pumping derivation for $\otherctxNGVAS{\anngvas}{\amarking, \anonterm, \amarkingp}$, then it is also a pumping derivation for $\otherctxNGVAS{\anngvas}{\amarking', \anonterm, \amarkingp'}$.
\end{lemma}

Since a set of well-quasi-ordered elements can only have finitely many minimal elements, for all $\unconstrained\subseteq I \subseteq \abdinfomid$, and $\omegaoutcount\subseteq O\subseteq\abdinfomid$, there is a finite set $\extcovmap_{I, O}\subseteq\Nomega^{d}\times\nonterms\times\Nomega^{d}$ with
\begin{align*}
    \upclsof{\extcovmap_{I, O}}&=\setcond{(\amarking, \anonterm, \amarkingp)\in\Nomega^{d}\times\nonterms\times\Nomega^{d}}{\\
    &\hspace{2em}\amarking\in\omegamrkdomainof{I},\;\amarkingp\in\omegamrkdomainof{O} \\
    &\hspace{2em}\amarking\sqsubseteq\inof{\anonterm},\; \amarkingp\sqsubseteq\outof{\anonterm},\;\otherctxNGVAS{\anngvas}{\amarking, \anonterm, \amarkingp}\text{ has }\perfectnesspumpingnospace}.
\end{align*}
Remark that, if we could compute $\extcovmap_{I, O}$ for $I, O$ as described by \Cref{Lemma:PrecalculationIIBreakdown}, then by \Cref{Lemma:PumpingMonotonous} we can compute a set of pumping derivations, and thus compute $\extcovconst_{I, O}$ by simply taking the maximum of their lengths.
We show that we can indeed compute $\extcovmap_{I, O}$.
\begin{lemma}\label{Lemma:ExtCovMapComputable}
    Let $\unconstrained\subseteq I \subseteq \abdinfomid$, $\omegaoutcount\subseteq O\subseteq\abdinfomid$ with non-empty 
    $I\setminus \unconstrained$ and $O\setminus\unconstrained$.
    Let $\perfect$ reliable up to $\rankof{\anngvas}$.
    Then, we can compute $\extcovmap_{I, O}$.
\end{lemma}
The challenge in computing $\extcovmap$ and $\extcovconst$ is capturing the upward closure of \emph{all} possible pumping derivations.
We adress this challenge by using the derivation implied by \Cref{Lemma:PrecalculationI}.
Using a similar proof to \Cref{Lemma:FreePumpNoSDBothSidesRackoff}, we obtain an upper bound, where increasing one counter beyond this bound does not make \perfectnesspumping newly hold.  
\begin{proof}
    We compute the sets $\unconstrained\subseteq I \subseteq \abdinfomid$, $\omegaoutcount\subseteq O\subseteq\abdinfomid$ with non-empty $I\setminus \unconstrained$ and $O\setminus\unconstrained$ by an inductive procedure on $2d-(\cardof{I}+\cardof{O})$.
    Consider the base case $I=O=\abdinfomid$.
    It can be readily verified that $\extcovmap_{I, O}=\setcond{(\inof{\anonterm}, \anonterm, \outof{\anonterm})}{\anonterm\in\nonterms}$.
    
    We move on to the inductive case.
    By the induction hypothesis, \Cref{Lemma:PumpingMonotonous}, and \Cref{Lemma:ConstructableSimplePumps}, we know the value $\extcovconst_{I', O'}\in\N$ is computable for all $I\subseteq I' \subseteq \abdinfomid$ and $O\subseteq O'\subseteq\abdinforight$, whenever $I\subsetneq I'$ or $O\subsetneq O'$.
    We define 
    \begin{align*}
        \ecmcalcpreconstof{I, O}&=\max\setcond{\extcovconst_{I', O'}}{\\
    &\hspace{0.75em}I\subseteq I'\subseteq \abdinfomid,\; O\subseteq O'\subseteq\abdinfomid,\; (I\subsetneq I'\text{ or }O\subsetneq O')}
    \end{align*}
    to be the maximal constant $\extcovconst_{I', O'}$ among non-smaller $I'$ and $O'$ where the induction hypothesis applies.
    With this constant at hand, we define 
    \begin{align*}
        \ecmcalcconstof{I, O}=(\incconst+1)\cdot(\ecmcalcpreconstof{I, O}+1)+\cardof{\anngvas}
    \end{align*}
    The addition of $\cardof{\anngvas}$ ensures that $\ecmcalcconstof{I, O}$ is larger than an $\inof{-}$ or $\outof{-}$ image.
    This is meant to make sure that we do not lose markings with too large $\inof{\anontermp}$ or $\outof{\anontermp}$ values by imposing an upper bound on counter values. 
    We claim 
    \begin{align*}
        \upclsof{\extcovmap_{I, O}}=&\upclsof{\setcond{(\amarking, \anonterm, \amarkingp)\in\Nomega^{d}\times\nonterms\times\Nomega^{d}}{\\
        &\hspace{1.5em}I\subseteq\omegaof{\amarking}\subseteq\abdinfomid, \; O\subseteq\omegaof{\amarkingp}\subseteq\abdinfomid, \\
        &\hspace{1.5em}\amarking\sqsubseteq\inof{\anonterm},\; 
        \amarkingp\sqsubseteq\outof{\anonterm},\;
        \otherctxNGVAS{\anngvas}{\amarking, \anonterm, \amarkingp}\text{ has \perfectnesspumpingnospace},\; \\
        &\hspace{1.5em}\amarking, \amarkingp\in \set{0, \ldots, \ecmcalcconstof{I, O}, \omega}^{d}}}.
    \end{align*}
    The inclusion direction $\supseteq$ is clear, we only impose an additional constraint in the form of a maximal concrete value.
    For the direction $\subseteq$, suppose that there is a $(\amarking, \anonterm, \amarkingp)\in\upclsof{\extcovmap_{I, O}}$ that is not included on the right-hand side of the proposed equality.
    Since every other constraint already follows from the membership $(\amarking, \anonterm, \amarkingp)\in\upclsof{\extcovmap_{I, O}}$, one of $\amarking\not\in \set{0, \ldots, \ecmcalcconstof{I, O}, \omega}^{d}$ or $\amarkingp\not\in\set{0, \ldots, \ecmcalcconstof{I, O}, \omega}^{d}$ must hold.
    Then, we construct $\amarking', \amarkingp'\in\Nomega^{d}$ with $\amarking\sqsubseteq\amarking'$, $\amarkingp\sqsubseteq\amarkingp'$, $\amarking', \amarkingp'\in \set{0, \ldots, \ecmcalcconstof{I, O}, \omega}^{d}$ where $\amarking[i]> \ecmcalcconstof{I, O}$ iff $i\in\omegaof{\amarking'}\setminus\omegaof{\amarking}$, and $\amarkingp[i]>\ecmcalcconstof{I, O}$ iff $i\in\omegaof{\amarkingp'}\setminus\omegaof{\amarkingp}$.
    Note that $\omegaof{\amarking'}, \omegaof{\amarkingp'}\subseteq\abdinfomid$, since $\ecmcalcconstof{I, O}$ is larger than $\inof{\anonterm}$ and $\outof{\anonterm}$.
    By \Cref{Lemma:PumpingMonotonous}, we know that $\otherctxNGVAS{\anngvas}{\amarking', \anonterm, \amarkingp'}$ admits a pumping derivation, i.e. $(\amarking', \anonterm, \amarkingp')\in\extcovmap_{\omegaof{\amarking'}, \omegaof{\amarkingp'}}$.
    Since $\amarking\not\in \set{0, \ldots, \ecmcalcconstof{I, O}, \omega}^{d}$ or $\amarkingp\not\in\set{0, \ldots, \ecmcalcconstof{I, O}, \omega}^{d}$ holds, we know that $\cardof{\omegaof{\amarking'}}+\cardof{\omegaof{\amarkingp'}}<\cardof{\omegaof{\amarking}}+\cardof{\omegaof{\amarkingp}}$.
    Thus, the induction hypothesis applies to show that there is indeed a pumping derivation $(\anonterm\to^{*}\aword.\anonterm.\awordp, \arun, \arunp)$ where $\cardof{\arun}, \cardof{\arunp}\leq \ecmcalcpreconstof{I, O}$.
    We let $\amarking''\sqsubseteq\amarking'$ and $\amarkingp''\sqsubseteq\amarkingp'$ with $\omegaof{\amarking''}=\omegaof{\amarking}$, $\omegaof{\amarkingp''}=\omegaof{\amarkingp}$, $\amarking''[i]=\ecmcalcconstof{I, O}$ for all $i\in\omegaof{\amarking'}\setminus\omegaof{\amarking}$, and $\amarkingp''[i]=\ecmcalcconstof{I, O}$ for all $i\in\omegaof{\amarkingp'}\setminus\omegaof{\amarkingp}$.
    Recall the $(\omegaof{\acontextin}, \omegaoutcount)$-free-$\N$-pump $\gamma$ of size $\leq\incconst$ we fixed after restating \Cref{Lemma:PrecalculationII}.
    It allows us to apply \Cref{Lemma:SimpleBoundDetailed}, show that $\otherctxNGVAS{\anngvas}{\amarking'', \anonterm, \amarkingp''}$ 
    has a pumping derivation.
    Clearly, $(\amarking'', \anonterm, \amarkingp'')$ is captured on the right-hand side of the proposed equality.
    By the definition of $\amarking'$, $\amarking''$, $\amarkingp'$, and $\amarkingp''$, we have $(\amarking'', \anonterm, \amarkingp'')\leq(\amarking, \anonterm, \amarking)$.
    This is a contradiction.
\end{proof}

\end{document}